\documentclass[runningheads]{llncs}
\usepackage{geometry}
\usepackage{tikz}
\usetikzlibrary{positioning,arrows.meta,fit,decorations.pathreplacing}
\usetikzlibrary{decorations.shapes}
\usepackage{tkz-graph}
\usepackage{amssymb}
\usepackage{amsfonts}
\usepackage[utf8]{inputenc}
\usepackage[T1]{fontenc}
\usepackage{lmodern}
\usepackage[activate=true,final,
tracking=true,
kerning=true,
spacing=true,
factor=1100,
stretch=10,
shrink=10]{microtype}
\usepackage{bbding}
\usepackage{pifont}
\usepackage{tabularx}
\usepackage{amsmath}
\usepackage{array}
\usepackage{bm}
\usepackage{spverbatim}
\usepackage{breakcites}
\usepackage{enumitem}
\setlist[itemize,enumerate]{leftmargin=*}
\usepackage{mathrsfs}
\usepackage{hyperref}
\hypersetup{
	colorlinks,
	citecolor=black,
	filecolor=black,
	linkcolor=black,
	urlcolor=black
}
\usepackage{footnotehyper} 
\makesavenoteenv{algorithm}
\usepackage{multicol,xparse}
\usepackage{changepage}
\usepackage{mdframed} 
\usepackage{currfile}
\usepackage{booktabs}
\usepackage{mathtools}
\usepackage[font=small]{caption}
\DeclareCaptionFormat{alg-style}{#1#2#3} 
\usepackage[ruled]{algorithm2e}

\SetAlCapFnt{\small\sffamily} 
\SetAlCapNameFnt{\itshape}
\usepackage{etoolbox}

\makeatletter
\let\llncs@orig@addcontentsline\addcontentsline
\renewcommand{\addcontentsline}[3]{%
	\ifstrequal{#1}{toc}{%
		\ifstrequal{#2}{title}{}{%
			\ifstrequal{#2}{author}{}{%
				\llncs@orig@addcontentsline{#1}{#2}{#3}%
			}%
		}%
	}{%
		\llncs@orig@addcontentsline{#1}{#2}{#3}%
	}%
}
\makeatother

\usepackage{tikz}
\usetikzlibrary{positioning, fit, backgrounds, arrows.meta, shadows, shapes.geometric, calc, shadows.blur}
\DeclareUnicodeCharacter{202F}{\,}
\usetikzlibrary{shapes.geometric} 
\renewcommand{\qed}{\hfill$\blacksquare$}
\usepackage[sanserif,basic]{complexity}
\newtheorem{observation}{Observation}
\newtheorem{assumption}{Assumption}

\newcommand{\PubOpen}{\mathsf{PubOpen}}
\newcommand{\GgROCRP}{\mathcal{G}_{\mathrm{gRO\text{-}CRP}}}
\newcommand{\ellregone}{\ell_{\mathsf{reg},1}} 
\newcommand{\ellregtwo}{\ell_{\mathsf{reg},2}} 
\newcommand{\SealToPeer}{\mathsf{SealToPeer}}
\newcommand{\OpenFromPeer}{\mathsf{OpenFromPeer}}

\newcommand{\ok}{\mathsf{ok}}
\newcommand{\dom}{\mathsf{dom}}
\newcommand{\AccSDKG}{\mathsf{Acc}_{\mathsf{SDKG}}}
\newcommand{\TSDKG}{\mathcal T}
\newcommand{\Enc}{\mathsf{Enc}}
\newcommand{\USVrcpt}{\mathsf{USV.rcpt}}
\newcommand{\Dec}{\mathsf{Dec}}
\newcommand{\Bad}{\mathsf{Bad}}
\newcommand{\Tag}{\mathsf{tag}}
\newcommand{\DLEQ}{\mathsf{DLEQ}}
\newcommand{\DL}{\mathsf{DL}}
\newcommand{\sk}{\mathsf{sk}}
\newcommand{\aff}{\mathsf{aff}}
\newcommand{\FDKGstarNXK}{\mathcal{F}^{\star,\mathsf{NXK}}_{\mathsf{DKG}}}
\newcommand{\WDKG}{\mathcal{W}_{\mathsf{DKG}}}
\newcommand{\Ext}{\mathsf{Ext}}
\newcommand{\negl}{\mathrm{negl}}

\newcommand{\isnew}{\mathsf{new}}
\newcommand{\FKeyBox}{\mathcal{F}_{\mathsf{KeyBox}}}
\newcommand{\Fpub}{\mathcal{F}_{\mathsf{pub}}}
\newcommand{\FSDKG}{\mathcal{F}_{\mathsf{SDKG}}}
\newcommand{\Fusv}{\mathcal{F}_{\mathsf{USV}}}
\newcommand{\Fchan}{\mathcal{F}_{\mathsf{channel}}}
\newcommand{\Use}{\mathsf{Use}}
\newcommand{\Cert}{\mathsf{Cert}}
\newcommand{\seed}{\mathsf{seed}}
\newcommand{\Adv}{\mathsf{Adv}}
\newcommand{\gROCRP}{gRO-CRP}
\newcommand{\PubMap}{\mathsf{PubMap}}  
\newcommand{\GetPub}{\mathsf{GetPub}}  
\newcommand{\Pub}{\PubMap}             
\newcommand{\pk}{\mathsf{pk}}
\newcommand{\Derive}{\mathsf{Derive}}
\newcommand{\G}{\mathcal{G}}
\newcommand{\Z}{\mathcal{Z}}
\newcommand{\Zp}{\mathbb{Z}_p}
\newcommand{\F}{\mathcal{F}}
\newcommand{\Zps}{\mathbb{Z}_p^*}

\newcommand{\h}{\mathcal{H}}
\newcommand{\s}{\mathsf{Sim}}
\newcommand{\pp}{\mathsf{pp}}
\newcommand{\cA}{\mathcal{A}}
\newcommand{\cR}{\mathcal{R}}
\newcommand{\cp}{\mathcal{P}}
\newcommand{\cV}{\mathcal{V}}
\newcommand{\cZ}{\mathcal{Z}}
\newcommand{\DeriveK}{\mathsf{Derive}_{K_{1,3}}}
\newcommand{\KBhdl}[1]{\langle \texttt{hdl}, #1\rangle}

\newcommand{\status}{\mathsf{status}}
\newcommand{\Commit}{\mathsf{Commit}}
\newcommand{\Load}{\mathsf{Load}}
\newcommand{\receipt}{\mathsf{receipt}}
\newcommand{\Verify}{\mathsf{Verify}}
\newcommand{\Open}{\mathsf{Open}}
\newcommand{\abort}{\mathsf{abort}}
\newcommand{\Exec}{\mathsf{Exec}}
\newcommand{\Ideal}{\mathsf{Ideal}}
\newcommand{\sid}{\mathsf{sid}}
\newcommand{\Simcert}{\mathsf{Sim}_{\text{cert}}}
\newcommand{\Vcert}{\mathcal{V}_{\text{cert}}}
\newcommand{\cid}{\mathsf{cid}}

\newcommand{\PRF}{\mathsf{PRF}}
\newcommand{\figrefp}[1]{\hyperref[#1]{Fig.~\ref*{#1} (p.~\pageref*{#1})}}
\newcommand{\tabrefp}[1]{\hyperref[#1]{Table~\ref*{#1} (p.~\pageref*{#1})}}
\usepackage{cleveref}
\crefname{tcb@cnt@readerbox}{readerbox}{readerboxes}

\newcommand{\Query}{\mathsf{Query}}

\newcommand{\SimProgramRO}{\mathsf{SimProgram}}

\newcommand{\CtxTEE}{\mathsf{Ctx}_{\mathrm{np}}} 
\newcommand{\CtxUC}{\mathsf{Ctx}_{\mathrm{p}}}   
\newcommand{\FKeyBoxOf}[1]{\FKeyBox^{(#1)}}

\newcommand{\Psig}{\mathcal{P}_{\mathrm{\Sigma}}}
\newcommand{\Vsig}{\mathcal{V}_{\mathrm{\Sigma}}}

\newenvironment{proofsketch}{%
	\par\noindent\textit{Proof Sketch. }\ignorespaces
}{%
	\endproof%
}
\newcommand{\leftarrowdollar}{%
	\mathrel{%
		\leftarrow%
		\mathrlap{\mkern-7mu\vcenter{\hbox{\scalebox{0.7}{\$}}}}%
}%
}
\newcommand{\WKB}{\mathcal{W}_{\mathsf{KB}}}
\newcommand{\KBcmd}{\mathsf{KBcmd}}
\newcommand{\KBret}{\mathsf{KBret}}

\newcommand{\muFS}{\mu_{\mathsf{FS}}}

\newcommand{\SlotSpace}{\{0,1\}^{\ast}}              
\newcommand{\KBslot}[2]{\langle #1,#2\rangle}        
\newcommand{\KBsid}[1]{\KBslot{\sid}{\texttt{#1}}}   

\newcommand{\CtxSDKG}{\mathsf{SDKG.s32}}
\newcommand{\PartySet}{\mathbb{P}} 
\usepackage[most]{tcolorbox}
\tcbuselibrary{breakable}
\newcommand{\ADSDKGreg}[4]{\langle \texttt{SDKG.reg},#1,#2,#3,#4\rangle}

\newcounter{readerboxcount}[section]
\renewcommand{\thereaderboxcount}{\thesection.\arabic{readerboxcount}}

\newtcolorbox{readerbox}[2][]{%
	enhanced,
	breakable,
	colback=gray!3,
	colframe=black,
	boxrule=0.6pt,
	arc=1.5mm,
	left=1.5mm,
	right=1.5mm,
	top=1.2mm,
	bottom=1.2mm,
	fonttitle=\normalsize,
	code={\refstepcounter{readerboxcount}},
	title=\textbf{Reader Note \thereaderboxcount: #2},
	label={#1}
}
\newtcolorbox{protocolbox}[1]{%
	enhanced,
	breakable,
	colback=gray!2,
	colframe=gray!55,
	boxrule=0.6pt,
	arc=1.5mm,
	left=1.5mm,
	right=1.5mm,
	top=1.2mm,
	bottom=1.2mm,
	title=\textbf{#1},
	fonttitle=\normalsize,
}
\newcommand{\DLstmt}[3]{\langle #1,#2,#3\rangle} 

\newcommand{\LblAffOneY}{\texttt{SDKG.aff1.Y}}
\newcommand{\LblAffOneD}{\texttt{SDKG.aff1.D}}
\newcommand{\LblAffTwoY}{\texttt{SDKG.aff2.Y}}
\newcommand{\LblAffTwoD}{\texttt{SDKG.aff2.D}}

\makeatletter

\newcommand{\listofresults}{%
	\section*{List of Results}%
	\addcontentsline{toc}{section}{List of Results}%
	\@starttoc{lor}%
}

\newcommand*\l@result{\@dottedtocline{1}{1.5em}{9em}}

\newcommand{\lor@addentry}[3]{
	\addcontentsline{lor}{result}{%
		\protect\numberline{#1~#2}%
		\ifstrempty{#3}{}{(#3)}%
	}%
}

\newcommand{\lor@maybeadd}[3]{%
	\ifstrequal{#1}{Theorem}{\lor@addentry{#1}{#2}{#3}}{%
	}
		\ifstrequal{#1}{Lemma}{\lor@addentry{#1}{#2}{#3}}{%
		}
			\ifstrequal{#1}{Proposition}{\lor@addentry{#1}{#2}{#3}}{%
			}
				\ifstrequal{#1}{Corollary}{\lor@addentry{#1}{#2}{#3}}{%
				}
					\ifstrequal{#1}{Observation}{\lor@addentry{#1}{#2}{#3}}{}%
}

\let\lor@orig@begintheorem\@begintheorem
\renewcommand{\@begintheorem}[2]{%
	\lor@maybeadd{#1}{#2}{}%
	\lor@orig@begintheorem{#1}{#2}%
}

\let\lor@orig@opargbegintheorem\@opargbegintheorem
\renewcommand{\@opargbegintheorem}[3]{%
	\lor@maybeadd{#1}{#2}{#3}%
	\lor@orig@opargbegintheorem{#1}{#2}{#3}%
}

\makeatother
\makeatletter

\newcommand{\lor@maybeaddsp}[2]{
	\ifstrequal{#1}{theorem}{\lor@addentry{Theorem}{\csname the#1\endcsname}{#2}}{%
		\ifstrequal{#1}{lemma}{\lor@addentry{Lemma}{\csname the#1\endcsname}{#2}}{%
			\ifstrequal{#1}{proposition}{\lor@addentry{Proposition}{\csname the#1\endcsname}{#2}}{%
				\ifstrequal{#1}{corollary}{\lor@addentry{Corollary}{\csname the#1\endcsname}{#2}}{%
					\ifstrequal{#1}{observation}{\lor@addentry{Observation}{\csname the#1\endcsname}{#2}}{}}}}}%
}

\pretocmd{\@spxthm}{\lor@maybeaddsp{#1}{}}{}{}

\pretocmd{\@spythm}{\lor@maybeaddsp{#1}{#5}}{}{}

\makeatother

\makeatletter
\AtBeginDocument{%
	\let\RefOrig\ref
	
	\DeclareRobustCommand{\ref}{\@ifstar{\RefOrig*}{\ref@withpage}}
	\newcommand{\ref@withpage}[1]{%
		\hyperref[#1]{\RefOrig*{#1} (p.~\pageref*{#1})}%
	}
	
	\renewcommand{\eqref}[1]{\hyperref[#1]{\tagform@{\RefOrig*{#1}}}}

}
\makeatother
\usepackage{scrextend}

\author{Vipin Singh Sehrawat \\ \small \texttt{\{vipin.sehrawat.cs@gmail.com\}}}
\authorrunning{Vipin Singh Sehrawat}
\institute{Circle Internet\thanks{The views expressed in this paper are solely those of the author and do not necessarily reflect
		those of Circle Internet or any other affiliated organizations.}}
\title{UC-Secure Star DKG for Non-Exportable Key Shares \\with VSS-Free Enforcement}
\titlerunning{UC-Secure Star DKG for Non-Exportable Key Shares with VSS-Free Enforcement}

\begin{document}
	\maketitle
\begin{abstract}
	Distributed Key Generation (DKG) lets parties derive a common public key while keeping the signing key secret-shared. In UC-secure DKG, the transcript must enforce (i) secrecy against unauthorized corruptions and (ii) uniqueness and affine consistency of the induced sharing. Classically, these obligations are satisfied by a Verifiable-Sharing Enforcement (VSE) layer---realized via Verifiable Secret Sharing (VSS) and/or commitment-and-proof mechanisms---that distributes and later manipulates shares. This paper addresses the Non-eXportable Key (NXK) setting enforced by hardware-backed key-isolation modules (e.g., secure enclaves/TEEs configured as restricted keystores, or HSM-like APIs), formalized via an ideal KeyBox (keystore) functionality $\FKeyBox$ that keeps shares non-exportable, including caller-invertible affine images, and permits only attested KeyBox-to-KeyBox sealing. In this setting, confidentiality can be delegated to the NXK/KeyBox boundary; the remaining challenge in realizing VSE layer, without VSS-style mechanisms (opening/complaints/resharing), is to enforce transcript-defined affine consistency without exporting, opening, or resharing the secret shares. Assuming a key-opaque, state-continuous NXK/KeyBox boundary, classical rewinding/forking-lemma extraction arguments are inapplicable. Hence, straight-line extraction is required. 
	
We present a Universally Composable (UC) DKG design for NXK by combining (i) NXK/KeyBox confidentiality; (ii) Unique Structure Verification (USV), which is a publicly verifiable certificate mechanism intended for tightly coupled NXK deployments where the certified scalar is non-exportable and never leaves the KeyBox, yet the corresponding public group element is deterministically derivable from the transcript. We combine these with (iii) UC-extractable non-interactive zero-knowledge arguments of knowledge via the Fischlin transform in our gRO-CRP-hybrid model, where gRO-CRP denotes a global Random Oracle with Context-Restricted Programmability, to enforce the affine constraints normally certified by VSS-style machinery. Using these tools, we construct a UC-secure Star DKG (SDKG) scheme that is tailored to multi-device wallets with a designated service that must co-sign but can never sign alone. SDKG realizes a $1{+}1$-out-of-$n$ star access structure wherein the center (mandatory) and any leaf of a star graph form a minimal authorized subset. SDKG implements a two-leaf star over roles (primary vs.\ recovery) and supports role-based registration. In the $\FKeyBox$-hybrid and \gROCRP\ models, assuming authenticated confidential channels that leak only message lengths, under DL and DDH hardness assumptions, PRF security of the deterministic nonce-derivation function, and a seed integrity invariant (part of the state-continuity engineering budget), with adaptive corruptions and secure erasures, SDKG UC-realizes a transcript-driven refinement of the standard UC-DKG functionality. Over a prime-order group of size $p$, SDKG incurs $\widetilde{O}(n\log p)$ communication overhead and $\widetilde{O}(n\log^{2.585}p)$ bit-operation cost, while registering a recovery device incurs $\widetilde{O}(\log p)$ communication and $\widetilde{O}(\log^{2.585}p)$ bit-operation costs, respectively. For a 128-bit instantiation with fixed Fischlin parameters, the base transcript (for 1+1-out-of-3 SDKG) is $\approx 11$--$13$~KiB.
\end{abstract}
\tableofcontents
\newpage
\listofresults
\newpage

\section{Introduction}\label{sec:intro}
The growing adoption of cryptocurrencies and decentralized applications has motivated the need for secure and versatile key management solutions. Multiparty Computation (MPC) \cite{Yao[86],GoldMi[87]} wallets, implementing Distributed Key Generation (DKG) \cite{Ped[91],Genn[07]} and threshold signing \cite{Des[87],Des[89]}, have become prevalent because they provide robust security without the need for a trusted party. In practice, some users may prefer managing multiple authorized devices, $\{P_i\}_{i = 2}^n$, in collaboration with a service provider, $P_1$, forming a star topology with $P_1$ as the center and $\{P_i\}_{i = 2}^n$ as the leaves. Furthermore, some deployments require a designated service that must always be involved in signing, yet can never sign alone. Typical examples include: regulated custodians that must co-sign for compliance and auditing; ``recovery-with-friction'' consumer wallets, where a risk engine enforces spending limits and/or anomaly detection; and enterprise wallets, where all transactions must flow through a corporate signing service. In such settings, a uniform threshold structure is ill-suited because it either authorizes subsets $\subseteq \{P_i\}_{i = 2}^n$, excluding $P_1$, or forces $P_1$ to hold a threshold number of shares, effectively authorizing a singleton access structure. This work targets the star access structure in which the minimal authorized sets are $\{P_1,P_i\}$ for $i\ge 2$. Formally, the family of minimal authorized subsets $(\mathrm{\Gamma_0})$ and the corresponding access structure $(\mathrm{\Gamma})$ are defined as: $$\mathrm{\Gamma_0} \coloneqq \{\{P_1,P_i\}\}_{i=2}^n; \quad \mathrm{\Gamma} \coloneqq \{S \subseteq \{P_i\}_{i \in [n]}: \exists i \ge 2, \{P_1,P_i\} \subseteq S \}.$$ This mandatory-center co-signing pattern is reminiscent of mediated / server-supported signature systems \cite{BonehDTW01,BonehDT04,DingMT02} and more recent server-assisted key-use designs \cite{LueksHAT20}. Our focus differs because we need dealerless DKG under secret non-exportability and Universally Composable (UC) composition \cite{Can[01]}---natural requirements for MPC wallets deployed in complex environments.

Hardware-backed key-isolation modules---e.g., secure enclaves/TEEs when configured as restricted keystores, or HSM-like APIs \cite{Shep[24]}---can enforce a Non-eXportable Key (NXK) model which binds secrets to hardware and forbids exporting them, including any caller-invertible affine image, permitting only attested KeyBox-to-KeyBox sealing. Thus, in a DKG performed under NXK, each signing share is generated and stored inside the module in which it was created and can only be accessed via a restricted interface. This constraint rules out classical rewinding/forking-lemma extraction \cite{Fork[1],Fork[2]} for Schnorr/Fiat--Shamir-style proofs \cite{Fiat[86],Schnorr[89],Schnorr[91]} at the hardware boundary. Throughout, we use ``KeyBox'' to denote this kind of hardware-backed key-isolation module. 

We introduce Unique Structure Verification (USV), which is meant to operate precisely at this NXK/KeyBox boundary. In particular, USV targets the tightly coupled setting in which the scalar being committed-to/certified is KeyBox-resident and non-exportable, so it cannot be revealed to the host or appear in the public transcript. The protocol exports only a publicly verifiable certificate, from which any verifier can deterministically derive the associated public group element needed by transcript-defined checks, without ever exporting the scalar or its invertible affine image.

A UC-secure DKG must realize the secrecy and uniqueness guarantees of dealerless (random) Verifiable Secret Sharing (VSS) \cite{ChorGold[85],Feld[87],Ped[91]} as a subtask (formalized in Section~\ref{sec:why-vss}). Therefore, existing UC-secure DKGs include some Verifiable-Sharing Enforcement (VSE) layer that prevents equivocation and enforces affine consistency of parties' contributions. Classically, this layer is instantiated using VSS and its variants (e.g., \cite{Stadler[96],Cac[02]}) via commitments \cite{Blum[83]} and complaint/opening logic \cite{Feld[87],Genn[07]}. In other UC(-style) threshold key-generation components (e.g., \cite{Feld[87],Ped[91],Genn[07]}), VSE layer is realized via commitment-and-proof / (Non)Interactive Zero-Knowledge ((N)IZK)-based mechanisms that enforce the required VSS-style obligations. Either way, traditional UC-secure DKG requires that (linear combinations of) shares can be computed and transmitted outside the device that holds them, an assumption that is incompatible with NXK. 

\paragraph{Roles vs.\ devices.}
We distinguish between cryptographic roles and the
physical devices / secure hardware instances that host those roles. Our base construction
realizes a two-leaf star over roles: the mandatory service $P_1$ (center) must co-sign with
either a primary-role share (held by a designated primary device, say $P_2$) or a recovery-role
share. Our $n$-device extension supports Role-based Device Registration (RDR) under NXK by enrolling additional
recovery devices as redundant front-ends for the same recovery role. Concretely, each recovery device
contains an independent KeyBox instance that ultimately holds the same recovery-role share
$k_{\mathsf{rec}}$, but this replication is not a share export: enrollment is performed via attested
KeyBox-to-KeyBox sealing and a one-shot installation procedure that never places \textit{share-deriving} plaintext
in the public transcript and preserves the KeyBox profile assumptions, namely key-opacity and state continuity.
Thus, the access structure is star-shaped over devices, while being implemented as a two-leaf star over
unique role shares (primary vs.\ recovery); and hence we also call it 1+1-out-of-$n$ (star) access structure.

\subsection{Our Contributions}
We model the NXK/KeyBox boundary as an ideal functionality: the KeyBox provides confidentiality for non-exportable shares via a restricted API. We leverage that, together with other tools summarized in Table \ref{tab:constraint-map}, to realize VSE layer under NXK without VSS-style exported-share machinery. Our contributions are multifold and can be summarized as follows:

\begin{table}[!h]
	\centering
	\begin{tabularx}{\linewidth}{>{\raggedright\arraybackslash}p{3.5cm} >{\raggedright\arraybackslash}p{8.3cm} >{\raggedright\arraybackslash}p{6.2cm}}
		\toprule
		\textbf{Hardware constraint} & \textbf{Standard approaches' shortcomings} & \textbf{Our mechanism} \\
		\midrule
		State continuity (no rollback) &
		Rewinding/forking-lemma extraction requires rolling back a prover to reuse the same commitment under two challenges. &
		Fischlin-based UC-NIZK-AoK with straight-line simulation/extraction formalized in \gROCRP. \\
		\addlinespace
		NXK (non-exportable secrets) &
		VSS/resharing/opening typically assumes exportable shares or share-derived witnesses to enforce polynomial relations and reconfigure trust; NXK blocks exporting the needed witnesses. &
		USV certificates, whose witness remains KeyBox-resident: certify transcript-defined public structure verifiable outside the KeyBox, while the witness never leaves. \\
		\addlinespace
		No caller-decryptable ``wrap'' &
		Exporting encrypted shares under caller-known keys breaks key-opacity \cite{Bortolozzo[10]}. &
		KeyBox-to-KeyBox sealing + RDR: ciphertexts only unwrap inside a KeyBox; the caller learns no decryption handle. \\
		\bottomrule
	\end{tabularx}
	\caption{Mapping KeyBox/NXK constraints to why UC-DKG enforcement via exportable shares fails, and the tools we use instead.}
	\label{tab:constraint-map}
\end{table}

\begin{enumerate}
	\item We introduce Unique Structure Verification (USV) as a non-interactive, publicly verifiable certificate designed for NXK-coupled deployments, where the scalar witness is non-exportable and confined to a KeyBox, and only the certificate is released outside the KeyBox boundary. Conceptually, USV is a publicly extractable commitment-to-group-element abstraction: any party can deterministically derive the canonical public group element, while the committed scalar remains hidden (and non-exportable). Moreover, tags are efficiently simulatable conditioned on the derived opening. We formalize USV both as a primitive and as a handle-bound ideal functionality, and prove its UC security in our global Random Oracle with Context-Restricted Programmability (\gROCRP)-hybrid model (Definition \ref{def:GNPRO}) with adaptive corruptions and secure erasures, under the Discrete Logarithm (DL) and Decision Diffie-Hellman (DDH) hardness assumptions.
	\item We enforce transcript-defined affine relations directly using Fischlin-style \cite{Fischlin[05]} UC-NIZK Arguments of Knowledge (AoKs) with straight-line extraction in the \gROCRP\ model.
	\item We identify and formalize the specific enforcement functionality needed by UC-DKGs in the NXK setting: enforcing transcript-defined linear/affine relations between parties' hidden scalars while keeping designated signing shares non-exportable and resilient to adaptive corruptions, without invoking resharing. We show how to enforce these relations by: 
	\begin{enumerate}
		\item routing sensitive state through a KeyBox interface (confidentiality), 
		\item using USV to obtain canonical, handle-bound public openings to the group elements corresponding to KeyBox-resident scalars, and 
		\item using UC-NIZK-AoKs to certify the cross-party affine constraints that an exported-share enforcement layer would normally check via commitments, and VSS(-like) verification and dispute logic.
	\end{enumerate}
	We also formalize the incompatibility between classical VSE and the NXK/KeyBox boundary: VSS-style complaint/opening/resharing (and related reconfiguration mechanisms) inherently rely on exporting shares or share-derived values, which our NXK model rules out (Sections \ref{sec:why-vss} and \ref{subsec:why-usv-straightline}).
	\item We construct a constant-round UC-secure Star DKG (SDKG) scheme, supporting a 1+1-out-of-3 star access structure as the base case: $P_1$ must co-sign with either a primary device $P_2$ or a recovery role $P_3$. Our $n$-device extension supports Role-based Device Registration (RDR) by enrolling additional recovery devices as redundant front-ends with independent KeyBoxes. Our constant-round UC-secure SDKG for the 1+1-out-of-$n$ star access structure incurs $\widetilde{O}(n\log p)$ communication cost and $\widetilde{O}(n\log^{2.585}p)$ bit-ops computation overhead. The RDR extension adds per device $O(\log p)$ communication bits and $\widetilde{O}(\log^{2.585} p)$ work.
\end{enumerate}

\subsection{Related Work}
Most closely aligned in application model with our ``mandatory service participates'' setting, Snetkov et al. \cite{SnetkovVL24} give a UC treatment of server-supported signatures for smartphones. Their focus is two-party server-assisted signing, whereas we target dealerless DKG for star access structure under NXK with post-DKG RDR. For a recent systematization of DL-based DKG protocols, see \cite{Ranas[25]}. When reviewing the rest of the related literature, we restrict attention to UC-secure DKGs, and evaluate them over two central axes: (i) NXK compatibility and (ii) RDR support. Lindell--Nof~\cite{LN18} give a practical full-threshold ECDSA protocol. Instantiating their ideal functionalities with UC-secure commitments and Zero-Knowledge (ZK) yields a protocol that UC-realizes an ideal ECDSA functionality. However, their scheme targets a fixed party set and manipulates shares via generic MPC over plaintext shares and ciphertexts (no RDR; not NXK). Canetti et al.~\cite{CGGMP21} provide UC threshold ECDSA with proactive key refresh over a fixed party set; dynamic joins are not modeled and several sub-protocols compute or send non-trivial functions of shares (not NXK). Doerner et al. \cite{Jack[24]} give a three-round threshold ECDSA signing protocol and show that shared keys can be generated via a simple commit-release-and-complain procedure (without proofs of knowledge). However, their setting still assumes exportable shares (not NXK) and does not address RDR. Lindell~\cite{Lindell24} gives a three-round, straight-line-simulatable DKG for Schnorr signatures~\cite{Schnorr[89],Schnorr[91]}; again the party set is fixed and the protocol performs linear operations over shares outside the NXK model. Friedman et al.~\cite{Friedman25} support reconfiguration in a two-tier 2PC--MPC framework~\cite{MarOme[24]} via Publicly VSS \cite{Stadler[96]} and threshold additively homomorphic encryption; this achieves RDR for validators but relies on public ciphertexts and homomorphic operations that lie outside the NXK model. Outside the UC framework, Katz \cite{Katz[24]} studies the round complexity of fully secure synchronous DKG in the DL setting.

Unlike \cite{CGGMP21} wherein the security proof is carried out in the strict Global Random Oracle (GRO) setting, we operate under a \gROCRP\ model that provides a single global oracle but with local-call semantics and restricted programmability tailored to NXK/KeyBox environment. In contrast to VSS-based (and VSS-like) UC-DKGs, which require all-to-all distribution of share material and therefore incur $\Theta(n^2)$ aggregate communication in the party count even before accounting for complaint/opening traffic, SDKG incurs only a constant number of proof objects in the base run (1+1-out-of-3 setting) and one additional proof per registered device for the extension to the generic 1+1-out-of-$n$ setting. Asymptotically, we treat the Fischlin parameters as functions of the security parameter and choose them to satisfy the standard Fischlin conditions so that the transform's soundness/knowledge-extraction error is negligible in the security parameter. For concrete 128-bit security, we instantiate with fixed parameters and report explicit concrete
bounds and sizes. Table \ref{tab:comparison} summarizes the resulting asymptotic costs.

\begin{table}[t]
	\centering
	\small
	\begin{tabular}{l @{\hspace{10pt}} c @{\hspace{10pt}} c @{\hspace{5pt}} | @{\hspace{5pt}} c @{\hspace{10pt}} c @{\hspace{5pt}} | @{\hspace{5pt}} c @{\hspace{10pt}} c}
		\toprule
		& \multicolumn{2}{c}{Support}
		& \multicolumn{2}{c}{DKG}
		& \multicolumn{2}{c}{RDR} \\
		Work & NXK & RDR
		& Bit-ops & Comm. (bits)
		& Bit-ops & Comm. (bits) \\
		\midrule
		\textbf{SDKG}
		& \textbf{\checkmark} & \textbf{\checkmark}
		& $\widetilde{O}(n\kappa^{2.585})$
		& $\widetilde{O}(n\kappa)$
		& $\widetilde{O}(\kappa^{2.585})$
		& $\widetilde{O}(\kappa)$ \\
		
		CGGMP21~\cite{CGGMP21}
		& $\times$ & $\times$
		& $\widetilde{O}(n^2\kappa^{2.585})$
		& $\Theta(n^2\kappa)$
		& N/A & N/A \\
		
		Lindell--Nof~\cite{LN18}
		& $\times$ & $\times$
		& $\widetilde{O}(n^2\kappa^{2.585} + n^2\eta^{2.585})$
		& $\Theta(n^2(\kappa+\eta))$
		& N/A & N/A \\
		
		Lindell~\cite{Lindell24}
		& $\times$ & $\times$
		& $\widetilde{O}(n\kappa^{2.585}(t+1+n))$
		& $O(n(t+1+n)\kappa)$
		& N/A & N/A \\
		
		Friedman~\cite{Friedman25}
		& $\times$ & \checkmark
		& $\widetilde{O}(n^2\kappa^{2.585} + n\eta^{2.585})$
		& $O(n^2(\kappa+\eta))$
		& $\widetilde{O}(n^2\kappa^{2.585} + n^2\eta^{2.585})$
		& $\Theta(n^2(\kappa+\eta))$ \\
		\bottomrule
	\end{tabular}
	\caption{Comparing UC-secure DKG schemes on dominant costs (Karatsuba model~\cite{Karatsuba[63]}).
		Notations: $\kappa\coloneqq\log p$ for a prime $p$; $N$
		is a Paillier/class-group modulus; $\eta\coloneqq\log N$; and $t$ is the Shamir polynomial degree.}
	\label{tab:comparison}
\end{table}

\begin{note} 
	SDKG is specialized to a star access structure (1+1-out-of-$n$) in the NXK/KeyBox setting.
	Most prior UC-secure DKGs are analyzed for $t$-out-of-$n$ threshold access structures, whose $\Theta(n^2)$ costs
	largely reflect the all-to-all communication pattern inherent to threshold DKG/VSS-style protocols. For this reason,
	Table~\ref{tab:comparison} is intended as a qualitative comparison of models/techniques and reported asymptotics,
	not a like-for-like complexity comparison across identical access structures.
\end{note}

\subsection{Organization}
The rest of this paper is organized as follows:
Section~\ref{sec:model} formalizes the NXK/KeyBox setting and our UC execution model, including
ideal secure channels and the KeyBox functionality capturing non-exportable long-term secrets and
state continuity. 
Section~\ref{sec:why-vss} isolates the VSE obligations that
any UC-secure DKG must satisfy, and explains why standard exported-share enforcement and resharing
mechanisms clash with NXK. It also recalls relevant cryptographic primitives, namely UC, and the DL and DDH hardness assumptions
(along with Decisional DL Equivalence (DDLEQ) for the equivalent DDH game in the ``same-exponent across $\G,\h$'' form aligned with DLEQ/Chaum--Pedersen statements).
Section~\ref{prelim} collects cryptographic and modeling preliminaries:
the \gROCRP\ model and its local-call semantics, and the UC/NIZK(-AoK)
notions we use, including the optimized Fischlin transform enabling straight-line extraction; in particular, Proposition~\ref{prop:strict-gro-insufficient} explains why strict GRO is not enough for the required simulation interface.
Section~\ref{sec:USV} introduces Unique Structure Verification (USV), gives a concrete instantiation,
and proves UC security of the corresponding handle-bound functionality.
Section~\ref{SSE-NIZK}
develops the UC-extractable NIZK-AoKs used throughout, including DL, DLEQ, and the affine-DL
relations enforced by the transcript.
Section~\ref{SDKG} presents the SDKG protocol: the base
$1{+}1$-out-of-$3$ run and the one-shot role-based device registration (RDR) mechanism used to
enroll additional recovery devices under NXK (Algorithm~\ref{alg:reg}).
Section~\ref{subsec:SDKG-UC} proves UC security for the
base protocol via a transcript-driven ideal functionality and derives the corresponding standard
NXK-star DKG interface, including the following waypoints:
\begin{itemize}
	\item Formal necessity of USV for commit-only alternatives under hardened profiles: Section~\ref{subsec:why-usv-straightline}.
	\item Main theorem: Theorem~\ref{thm:SDKG-UC}.
	\item Compilation (eliminating $\Fusv$ via UC composition): Corollary~\ref{cor:compile-out-fusv}.
\end{itemize}
Section~\ref{sec:n-extension} formalizes the $1{+}1$-out-of-$n$ extension and establishes UC security of
the scalable RDR mechanism.
Section~\ref{sec:complexity} provides a focused complexity and overhead
discussion, including concrete parameterization and transcript sizes.
Section~\ref{sec:conclude}
concludes. Appendix~\ref{App1} discusses an optional tighter integration with programmable KeyBox
implementations (e.g., TEEs and certain HSM/KMS-backed designs). Appendix~\ref{app:keybox-impls} collects concrete candidate implementation classes and a
profile-capture checklist for enforcing a KeyBox API profile in practice.

\section{The NXK and KeyBox Setting}\label{sec:model}
We work in an NXK/KeyBox model wherein long-term shares remain inside state-continuous KeyBoxes (no rewind/fork) and are
\emph{API-non-exportable} in the sense of \textit{Reader Note} \ref{box:export-visibility}. The only permitted cross-KeyBox transfer of share-dependent data is attested KeyBox-to-KeyBox sealing, which returns only ciphertexts to the caller. Any \emph{caller-recoverable} export of a resident share---whether as raw bytes, a caller-invertible affine image, or via any other API-visible behavior that is not simulatable from the corresponding public information---is disallowed. Formally, we assume admissible KeyBox profiles satisfy key-opacity. Informally, key-opacity means that the KeyBox's external outputs are simulatable from the public key alone; the formal statement appears as Assumption~\ref{assump:keybox-opacity} below. We model (i) authenticated confidential point-to-point channels with adversary-controlled scheduling via an ideal functionality $\Fchan$ (Fig. \ref{fig:Fchan}), (ii) an authenticated public dissemination mechanism for transcript-public values via an ideal functionality $\Fpub$ (Fig. \ref{fig:Fpub}), and (iii) a per-party NXK hardware boundary via a KeyBox functionality $\FKeyBox$ (Fig. \ref{Fdskg}) that generates and stores long-term shares internally and exposes only a restricted API. 

\subsection{Terminology and leakage model}\label{subsec:nxk-terms}
\begin{readerbox}[box:export-visibility]{Terminology: exportability vs.\ visibility}
	We distinguish four places where a value may reside or be observed:
	\begin{itemize}[leftmargin=*,nosep]
		\item KeyBox internal state: data stored inside the trusted KeyBox boundary.
		\item Host RAM: volatile party state outside the KeyBox (subject to erasure and adaptive corruption).
		\item Adversary-visible transcript: everything observable to $\cA$ outside honest KeyBoxes, including
		all messages sent over adversary-visible channels, all explicit leakage outputs of ideal functionalities
		(e.g., $\Fchan$'s length leakage and scheduling metadata), and all outputs returned to adversary-controlled ITMs.
		\item Persistent storage (outside the KeyBox): disk/logs/swap outside the KeyBox; we conservatively
		treat any such data as part of the adversary-visible transcript.
	\end{itemize}
	
	We use the following terms throughout:
	\begin{itemize}[leftmargin=*,nosep]
		\item (API-)export / (API-)non-exportable (KeyBox-resident shares):
		A KeyBox-resident share is \emph{API-non-exportable} if no KeyBox API call enables the caller to
		recover the share (or any caller-invertible affine image of it), even when combined with caller-held secrets.
		This is captured formally by \emph{key-opacity} (Assumption~\ref{assump:keybox-opacity}).
		\item Transcript-visible vs.\ transcript-private:
		A value is \emph{transcript-visible} if it appears in the adversary-visible transcript; otherwise it is
		\emph{transcript-private}. Specifically, payloads delivered over $\Fchan$ are transcript-private unless an
		endpoint is corrupted; only $\Fchan$'s explicit leakage (e.g., lengths) is transcript-visible.
		\item Protocol transcript / local views (unqualified ``transcript''):
		When we refer to ``the transcript'' without the qualifier \emph{adversary-visible}, we mean the parties' local protocol views: the tuple of
		messages delivered to the parties (including plaintexts carried over $\Fchan$) together with the public values.
		This full transcript is not necessarily adversary-visible.
		\item NXK-restricted material (share-deriving material):
		NXK-restricted material must be transcript-private and must never be written to persistent storage outside a KeyBox.
		It may be handled transiently in host RAM during an atomic local step and must then be securely erased; under
		adaptive corruptions with secure erasures, such RAM values leak only if corruption occurs before erasure.
	\end{itemize}
	
	In this paper ``non-exportable'' refers to API-non-exportability of KeyBox-resident shares and
	does not mean that related ephemeral/share-deriving material can never appear transiently in host RAM.
	Adversarial access to honest host RAM is modeled only via the UC corruption interface with secure erasures
	(Definition~\ref{def:ace}), instantiated in our $\FKeyBox$-hybrid model below (Definition~\ref{KeyBox}): until corruption, an honest party executes its local steps atomically and may hold
	NXK-restricted material transiently, after which it is securely erased; after corruption, the adversary controls
	the host ITM and learns its current (non-erased) state. We do not model an ``always-on'' adversary that can
	continuously scrape the RAM of an honest host without triggering a corruption event.
\end{readerbox}

The secure-channel functionality $\Fchan$ abstracts both authenticated key establishment and the
subsequent symmetric protection of payloads; its internal session-key material is not part of any
party’s KeyBox state and is not modeled explicitly. In our NXK setting, the only role of $\Fchan$
is to keep NXK-restricted / share-deriving material transcript-private
unless an endpoint is corrupted (\textit{Reader Note} \ref{box:export-visibility}). When mapping to a concrete
implementation, these channel keys can be realized as ordinary ephemeral host state (e.g., via
an AKE or ephemeral IND-CCA KEM establishing per-session AEAD keys) and are subject to the same
adaptive-corruption-with-secure-erasures discipline as other transient values: they may reside in
host RAM during an atomic step and must be erased once no longer needed. This forward-secure/erasure
discipline prevents later corruptions from retroactively decrypting previously recorded ciphertexts,
matching the semantics of $\Fchan$ (Fig.~\ref{fig:Fchan}). Alternatively, one may place channel
cryptography inside the KeyBox/profile adapter, at the cost of extending the admissible profile with
the required symmetric primitives, but our model and proofs do not require this stronger placement.

\begin{definition}[Adaptive corruptions with secure erasures]\label{def:ace}
	\emph{We work in the UC framework with adaptive corruptions and an explicit erasure discipline.
		Each party $P_i$ maintains a local (host) state $\mathsf{st}_i$ outside any KeyBox boundary (cf.\ \textit{Reader note}~\ref{box:export-visibility}).
		A protocol may explicitly erase designated local variables/buffers from $\mathsf{st}_i$ once they are no longer needed.
		Upon corruption of $P_i$, the adversary learns only $P_i$’s current local state $\mathsf{st}_i$ at the moment of corruption;
		any values explicitly erased by the protocol prior to corruption are not revealed.
		Thereafter, the adversary controls $P_i$ and may arbitrarily influence its future actions and state. We assume honest-party activations are atomic with respect to corruption, i.e., corruptions can occur only between activations.
		Consequently, temporary values created and erased within a single honest activation are never revealed by a later corruption.}
\end{definition}

\subsection{Notation, conventions, and ideal channels}\label{subsec:nxk-notation}
We call an algorithm efficient if, for input size $\lambda$, its running time is bounded by $\poly(\lambda)$. Throughout the text, $\lambda \in \mathbb{N}$ denotes the security parameter with $\negl(\lambda)$ denoting a negligible function on it; and $\approx_c$ represents computational indistinguishability. We write $x \leftarrowdollar \mathbb{S}$ to denote uniform sampling from a finite set $\mathbb{S}$. Throughout, for an elliptic curve $\E(\mathbb{F}_q)$, $\mathbb{G} \subseteq \E(\mathbb{F}_q)$ denotes a cyclic subgroup of prime order $p > 3$ written additively, with fixed generator $\G$. We fix an injective, self-delimiting encoding $\langle\cdot\rangle$ of mixed tuples into $\{0,1\}^\ast$. All hash/oracle invocations are applied only to encodings of the form $\langle\cdot\rangle$, and all ``equality of encoded tuples'' statements are with respect to this encoding. Our $\widetilde{O}(\cdot)$ bounds hide factors that are polynomial in the Fischlin parameters
$(t(\lambda),b(\lambda),r(\lambda),S(\lambda))$, which are themselves at most $\polylog(\lambda)$ in the asymptotic
analysis; for any fixed concrete instantiation at a target security level $\lambda=\lambda_0$, these factors become
constants.

The UC security parameter is $\lambda$. Our prime-order group is generated as a function of $\lambda$ and has order
$p=p(\lambda)$. Let $\kappa(\lambda)\coloneqq \lceil \log_2 p(\lambda)\rceil$ denote the bitlength of the group order.
We assume $\kappa(\lambda)=\Theta(\lambda)$; in particular, $\kappa(\lambda)=O(\lambda)$ and, for concrete
instantiations, typically $\kappa(\lambda)\ge \lambda$ up to a constant factor; so that ``PPT in $\lambda$'' and ``PPT in
$\kappa$'' are equivalent up to polynomial factors.
Unless stated otherwise, all auxiliary protocol parameters are deterministic functions of $\lambda$ (equivalently, of $\kappa$ under $\kappa=\Theta(\lambda)$).

\begin{figure}[t]
	\centering
	\setlength{\fboxrule}{0.2pt} 
	\fbox{
\parbox{\dimexpr\linewidth-2\fboxsep-2\fboxrule\relax}{%
			\ding{169} \textsf{Parameters:} session identifier $\sid$; endpoints $(P_s,P_r)$; leakage $\mathrm{\Phi}(c)=|c|$.\\
			\ding{169}  \textsf{State:} $\mathsf{active}\in\{0,1\}$ (init $0$); multiset $\mathsf Q$ of $(\rho,c,\phi)$; set $\mathsf D$ of delivered tickets.\\
			\ding{169}  Upon receiving $(\mathsf{Init},\sid,P_s,P_r)$ from both $P_s$ and $P_r$:
			set $\mathsf{active}\gets 1$; send $(\mathsf{ChReady},\sid,P_s,P_r)$ to $\cA$.\\
			\ding{169}  Upon receiving $(\mathsf{Send},\sid,c)$ from $P_s$ with $\mathsf{active}=1$:
			sample $\rho\leftarrowdollar \{0,1\}^\lambda$, set $\phi\gets\mathrm{\Phi}(c)$, insert $(\rho,c,\phi)$ into $\mathsf Q$,
			send $\rho$ to $P_s$, and send $(\mathsf{Leak},\sid,P_s,P_r,\rho,\phi)$ to $\cA$.
			If $P_s$ is corrupted, additionally reveal $c$ to $\cA$.\\
			\ding{169}  Upon receiving $(\mathsf{Deliver},\sid,\rho)$ from $\cA$:
			if $(\rho,c,\phi)\in\mathsf Q$ and $\rho\notin\mathsf D$, delete it from $\mathsf Q$, add $\rho$ to $\mathsf D$,
			and deliver $(\mathsf{Recv},\sid,P_s,c)$ to $P_r$.
			If $P_r$ is corrupted, reveal $c$ to $\cA$ at delivery time.\\
			\ding{169}  Upon receiving $(\mathsf{Inspect},\sid)$ from $\cA$:
			if both $P_s$ and $P_r$ are corrupted, reveal $c$ for all $(\rho,c,\phi)\in\mathsf Q$.
	}}
	\caption{Ideal authenticated, length-leaking secure channel $\Fchan$}
	\label{fig:Fchan}
\end{figure}

\begin{note}
	Concrete realization for $\Fchan$ from \cite{Tue[09]}: IND-CCA KEM + IND-CPA and INT-CTXT AEAD \cite{Rogaway02AEAD}.
	To match the adaptive-corruption-with-secure-erasures semantics of $\Fchan$, per-session channel keys are treated as ephemeral and
	are securely erased after use (or kept inside a trusted boundary).
\end{note}

\begin{figure}[t]
	\centering
	\setlength{\fboxrule}{0.2pt}
	\fbox{
		\parbox{\dimexpr\linewidth-2\fboxsep-2\fboxrule\relax}{%
			\ding{169}\ \textsf{Parameters:} session identifier $\sid$; sender $P_s\in\PartySet$; leakage $\Phi(c)=|c|$.\\
			\ding{169}\ \textsf{State:} multiset $\mathsf Q$ of $(\rho,P_s,c,\phi)$; set $\mathsf D$ of delivered tickets.\\
			\ding{169}\ \textsf{Upon receiving }$(\mathsf{Publish},\sid,c)\textsf{ from }P_s$: \\
			\hspace*{1.5em}sample $\rho\leftarrowdollar \{0,1\}^\lambda$, set $\phi\gets\Phi(c)$, insert $(\rho,P_s,c,\phi)$ into $\mathsf Q$,\\
			\hspace*{1.5em}send $\rho$ to $P_s$, and send $(\mathsf{Leak},\sid,P_s,\rho,c,\phi)$ to $\cA$.\\
			\ding{169}\ \textsf{Upon receiving }$(\mathsf{Deliver},\sid,\rho)\textsf{ from }\cA$: \\
			\hspace*{1.5em}if $(\rho,P_s,c,\phi)\in\mathsf Q$ and $\rho\notin\mathsf D$, delete it from $\mathsf Q$, add $\rho$ to $\mathsf D$,\\
			\hspace*{1.5em}and deliver $(\mathsf{Recv},\sid,P_s,c)$ to every $P\in\PartySet$.
	}}
	\caption{Ideal authenticated public broadcast $\Fpub$ (adversary-visible, adversary-scheduled).}
	\label{fig:Fpub}
\end{figure}

\subsection{Admissible KeyBox profiles and key-opacity}\label{subsec:nxk-key-opacity}
Let $\mathcal{K}$ be a key space and let $\xi$ be a distribution over $\mathcal{K}$.
Let $\mathcal{F}_{\mathrm{adm}}$ be a set of PPT stateful admissible operations, modeled as
state-transition algorithms
\[
f:\mathcal{K}\times\{0,1\}^*\times\{0,1\}^*\to \{0,1\}^*\times\{0,1\}^*,
\]
where on input $(k,\mathsf{st},m)$, the operation outputs a new private state
$\mathsf{st}'$ and a response $y$.
Let $\PubMap:\mathcal{K}\to\{0,1\}^*$ be an efficiently computable public-information map.
Let $\GetPub\in\mathcal F_{\mathrm{adm}}$ denote the stateless admissible operation
\[
\GetPub(k,\mathsf{st},m) := (\mathsf{st}, \PubMap(k)),
\]
which ignores $m$ and does not change state. Fix a PPT Interactive Turing Machine (ITM) simulator $\s$ and a PPT adversary $\mathcal A$.
Define an experiment $\mathsf{Exp}^{\mathrm{opq}}_{\mathcal A,\s,\mathcal F_{\mathrm{adm}},\xi,\PubMap}(1^\lambda)$:

Some admissible interfaces are multi-stage (e.g., a \textsf{Start}/\textsf{Prove} pair) and therefore
require that later calls read internal state written by earlier calls.
To model this in a profile-centric way, we associate to each admissible operation
a deterministic \emph{state-family} identifier via a map
\[
\digamma:\mathcal F_{\mathrm{adm}}\to \mathsf{Fam},
\]
where $\mathsf{Fam}$ is a finite identifier set fixed by the KeyBox API profile.
The KeyBox maintains one internal state string per family, and all invocations of operations
$f$ with the same family id $\digamma(f)$ share that state component.
Unless stated otherwise, we take singleton families, i.e., $\digamma(f)=f$.

\begin{enumerate}
\item Sample $k\leftarrowdollar \xi$ and set $pk\leftarrow \PubMap(k)$. Sample a hidden bit $\beta \leftarrowdollar \{0,1\}$.
\item Initialize a state table $\mathsf{st}[\cdot]$ by setting $\mathsf{st}[\varphi]\gets \epsilon$
for every family identifier $\varphi\in \mathrm{im}(\digamma)$.
	Initialize a simulator oracle $\s_{pk}$ by running $\s$ on input $(1^\lambda,pk)$. Hence, $\s_{pk}$ may maintain state across oracle calls.
	\item Run $\mathcal A(1^\lambda,pk)$ with oracle access to $\mathcal O_\beta(\cdot,\cdot)$ defined as:
	\[
	\mathcal O_\beta(f,m) :=
	\begin{cases}
		\text{let } \varphi\gets \digamma(f);\ 
		\text{let } (\mathsf{st}',y)\leftarrow f(k,\mathsf{st}[\varphi],m);\ 
		\mathsf{st}[\varphi]\gets \mathsf{st}';\ y 
		& \text{if } \beta=1 \text{ and } f\in\mathcal F_{\mathrm{adm}},\\
		\s_{pk}(f,m)
		& \text{if } \beta=0 \text{ and } f\in\mathcal F_{\mathrm{adm}},\\
		\perp & \text{if } f\notin\mathcal F_{\mathrm{adm}}.
	\end{cases}
	\]
	\item $\mathcal A$ outputs a bit $\beta'\in\{0,1\}$. Output $1$ iff $\beta'=\beta$.
\end{enumerate}

\begin{figure}[t]
	\centering
	\setlength{\fboxrule}{0.2pt} 
	\fbox{
		\parbox{\dimexpr\linewidth-2\fboxsep-2\fboxrule\relax}{%
			\ding{169} \textsf{Parameters:} fixed owner $P_{\mathsf{own}}\in\PartySet$; authenticated sealing-key directory
			$\{\pk_{\mathrm{seal}}^{(P)}\}_{P\in\PartySet}$ with local keypair $(\pk_{\mathrm{seal}},\sk_{\mathrm{seal}})$
			such that $\pk_{\mathrm{seal}}=\pk_{\mathrm{seal}}^{(P_{\mathsf{own}})}$; and a set of public parameters $\pp$.
			
			\smallskip
			\ding{169} \textsf{State:}
			key table $\mathrm{\Lambda}:\SlotSpace\rightharpoonup\mathcal K$ (init $\emptyset$) for local slot $\mu\in\SlotSpace$;
			operation state $\mathsf{st}[\mu,\varphi]\in\{0,1\}^*$ (init $\epsilon$) for each family id
			$\varphi\in \mathrm{im}(\digamma)$;
			private buffer $\mathsf{buf}:\{0,1\}^\lambda\rightharpoonup(\{0,1\}^*\times\{0,1\}^*)$ (init $\emptyset$).
			
			\smallskip
	\ding{169} \textsf{Admissible routines:}
	derivations $\chi_{\mathrm{adm}}$ and operations $\mathcal F_{\mathrm{adm}}$ with
	a designated key-independent subset $\mathcal F_{\mathrm{KI}}\subseteq \mathcal F_{\mathrm{adm}}$
	(e.g., $\OpenFromPeer,\textsf{USV.Cert}\in\mathcal F_{\mathrm{KI}}$ and $\SealToPeer\in\mathcal F_{\mathrm{adm}}\setminus\mathcal F_{\mathrm{KI}}$).
			
			\smallskip
		\ding{169} Procedure $\mathsf{Resolve}(m)$:
		parse $m$ as $\langle m_1,\ldots,m_t\rangle$; for each component that parses as a typed handle
		$\KBhdl{\tau}$ with $\tau\in\dom(\mathsf{buf})$, let $(\mathsf{ad},s):=\mathsf{buf}[\tau]$
		and substitute $\langle \mathsf{ad}, s\rangle$.
		If parsing fails or a referenced typed handle is missing, return $\perp$ and leave $\mathsf{buf}$ unchanged.
		Otherwise delete all used $\tau$ from $\mathsf{buf}$ and return the substituted tuple.
					
			\smallskip
			\ding{169} Upon receiving $(\Load,\mu,g,m)$ from party $P_{\mathsf{own}}$:
			\begin{itemize}[leftmargin=*,nosep]
				\item If $\mu\in\dom(\mathrm{\Lambda})$ or $g\notin \chi_{\mathrm{adm}}$, return $\perp$.
				Let $m'\gets \mathsf{Resolve}(m)$; if $m'=\perp$ return $\perp$.
				\item Compute $k\leftarrow g(1^\lambda,m')$. If $k=\perp$, return $\perp$.
				Set $\mathrm{\Lambda}[\mu]\gets k$ and return $\ok$.
			\end{itemize}
			
			\smallskip
			\ding{169} Upon receiving $(\Use,\mu,f,m)$ from party $P_{\mathsf{own}}$:
			\begin{itemize}[leftmargin=*,nosep]
\item If $f\notin\mathcal F_{\mathrm{adm}}$, return $\perp$.
\item If $f\notin\mathcal F_{\mathrm{KI}}$ and $\mu\notin\dom(\mathrm{\Lambda})$, return $\perp$.
				\item If $f=\SealToPeer$: parse $m=\langle P_{\mathrm{peer}},\mathsf{ad}\rangle$.
				If $P_{\mathrm{peer}}\notin \PartySet$, return $\perp$.
				\item Let $\pk_{\mathrm{peer}}\gets \pk_{\mathrm{seal}}^{(P_{\mathrm{peer}})}$ and $s\gets \mathrm{\Lambda}[\mu]$.
				Return $c\leftarrow \Enc_{\pk_{\mathrm{peer}}}(\mathsf{ad},s)$.
				\item If $f=\OpenFromPeer$: parse $m=\langle c,\mathsf{ad}\rangle$.
				Compute $s\leftarrow \Dec_{\sk_{\mathrm{seal}}}(\mathsf{ad},c)$; if decryption fails return $\perp$.
				Sample $\tau\leftarrowdollar \{0,1\}^\lambda$, set $\mathsf{buf}[\tau]\gets(\mathsf{ad},s)$, and return $\KBhdl{\tau}$.
				\item If $f=\textsf{USV.Cert}$: ignore $\mu$.
				Sample $m_{\mathsf{cert}}\leftarrowdollar \Zp^*$ and compute $\langle C,\zeta\rangle\leftarrow \Cert(\pp,m_{\mathsf{cert}})$ (defined in Section \ref{sec:USV}); erase $m_{\mathsf{cert}}$ and return $\langle C,\zeta\rangle$.
				\item If $f\in\mathcal F_{\mathrm{KI}}$ (general key-independent case): ignore $\mu$.
				Let $\varphi\gets \digamma(f)$; compute
				$(\mathsf{st}',y)\leftarrow f(\bot,\mathsf{st}[\mu,\varphi],m)$,
				set $\mathsf{st}[\mu,\varphi]\gets \mathsf{st}'$, and return $y$.
				(Key-independent operations receive $\bot$ in place of the resident key and must not depend on it.)
				\item Otherwise (key-dependent): let $\varphi\gets \digamma(f)$; compute
				$(\mathsf{st}',y)\leftarrow f(\mathrm{\Lambda}[\mu],\mathsf{st}[\mu,\varphi],m)$,
				set $\mathsf{st}[\mu,\varphi]\gets \mathsf{st}'$, and return $y$.
			\end{itemize}
	}}
	\caption{Per-party KeyBox functionality $\FKeyBox^{(P_{\mathsf{own}})}$ for an NXK KeyBox with KeyBox-to-KeyBox sealing.}
	\label{Fdskg}
\end{figure}

Fig. \ref{Fdskg} specifies a generic per-party KeyBox instance. We denote the instance owned by party $P_i$
as $\FKeyBoxOf{i}$, and sometimes write $\FKeyBox$ when the owner is clear from context.
Within an instance, keys are indexed by the local handle $\mu\in\SlotSpace$; globally we refer to a slot as $(P_i,\mu)$. The $\SealToPeer/\OpenFromPeer$ API assumes that each party’s sealing public key $\pk_{\mathrm{seal}}^{(P)}$ is authenticated and bound to that party’s KeyBox instance (e.g., via attestation). We abstract the attestation/key-distribution mechanism by an authenticated sealing-key directory $\{\pk_{\mathrm{seal}}^{(P)}\}_{P\in\PartySet}$ that is fixed and
not adversary-influenceable. Hence, $\SealToPeer$ encrypts to the directory-defined $\pk_{\mathrm{seal}}^{(P_{\mathrm{peer}})}$
for the designated peer, rather than accepting a raw recipient key as input. $\Enc/\Dec$ denote encryption/decryption
functions of an ideal public-key authenticated-encryption scheme with associated data, \textsf{ad}. We will use that $(\Enc,\Dec)$ is probabilistic and IND-CCA secure, i.e., secure for
polynomially many encryptions under a fixed public key. Thus, even when many
$\SealToPeer$ calls encrypt different slot-resident plaintexts to the same recipient key
$\pk_{\mathrm{seal}}^{(P_{\mathrm{peer}})}$, the resulting ciphertexts are jointly simulatable as
independent encryptions of a fixed dummy plaintext of the appropriate length under the same
$\pk_{\mathrm{seal}}^{(P_{\mathrm{peer}})}$, with fresh and independent randomness per call, and no other
cross-call leakage at the API boundary.

To support multiple sessions on the same state-continuous KeyBox,
we treat the mnemonic slot names as tags rather than literal constants. Concretely, in a session with identifier $\sid$, we define the local slot handle used for $\texttt{tag}$ as
\[
\KBsid{\texttt{tag}} \in \SlotSpace ,
\]
where $\langle\cdot\rangle$ is the fixed injective tuple encoding used throughout.

Our KeyBox abstraction is a narrow, keystore-style resource tailored to NXK. For more general-purpose formal UC abstractions and discussion of subtle composability issues stemming from globally shared attestation keys, see \cite{PassShi[17],MarkLo[25]}. Making the attestation protocol explicit mandates extending our model, which can be done via a separate (sub-)protocol/functionality and instantiated independently, as in UC-style treatments of \cite{MarkLo[25],PassShi[17],ManuCam[16],Xia[24]}. Given that protecting against physically invasive or side-channel attacks necessitates specialized equipment \cite{Carlton[21]}, we consider such threats out of scope.

\begin{definition}[Key-opacity]\label{def:key-opacity-alt}
	\emph{Let $\mathcal K$ be a key space, let $\xi$ be an efficiently sampleable distribution over $\mathcal K$,
		and let $\PubMap:\mathcal K\to\{0,1\}^*$ be an efficiently computable public-information map.
		For security parameter $\lambda$, $\mathcal F_{\mathrm{adm}}$ is \emph{key-opaque} with respect to $(\xi,\PubMap)$
		if for every PPT adversary $\mathcal A$ there exists a PPT ITM simulator $\s$ such that
		\[
		\Adv^{\mathrm{opq}}_{\mathcal A}(\lambda)
		:= \left|\Pr\!\left[\mathsf{Exp}^{\mathrm{opq}}_{\mathcal A,\s,\mathcal F_{\mathrm{adm}},\xi,\PubMap}(1^\lambda)=1\right]-\tfrac12\right|
		\le \negl(\lambda).
		\]
		When $\xi$ is clear from context, we may omit it and simply say “key-opaque w.r.t.\ $\PubMap$.”}
\end{definition}

\begin{assumption}[Key-opacity (profile-level, multi-slot)]\label{assump:keybox-opacity}
	\emph{Fix an admissible KeyBox API profile $(\chi_{\mathrm{adm}},\mathcal F_{\mathrm{adm}},\digamma)$	(Definition~\ref{def:keybox-profile}) and a public-information map
		$\PubMap:\mathcal K\to\{0,1\}^*$.	We assume slot-separable key-opacity for this profile:
		\begin{itemize}[leftmargin=*,nosep]
			\item Single-slot opacity: $\mathcal F_{\mathrm{adm}}$ is key-opaque w.r.t.\ $\PubMap$ in the sense of
			Definition~\ref{def:key-opacity-alt}.			
			\item Slot separability / no cross-slot opacity couplings: In the multi-slot $\FKeyBox$ functionality
			(Fig.~\ref{Fdskg}), every key-dependent admissible operation
			$f\in\mathcal F_{\mathrm{adm}}$ is slot-local: on a call $\Use(\mu,f,m)$ it may read the resident key $k_{i,\mu}:=\mathrm{\Lambda}[\mu]$
			and the per-slot family state $\mathsf{st}[\mu,\digamma(f)]$, but it does not read or update any other slot’s key/state
			$(\mathrm{\Lambda}[\mu'],\mathsf{st}[\mu',\cdot])$ for $\mu'\neq\mu$.
			Moreover, the externally visible randomness/state used by distinct slots can be taken independent.
			Any randomized admissible routine (including $\SealToPeer$) is modeled as using fresh,
			independent coins per invocation, and we do not model any additional cross-slot/cross-call leakage
			(e.g., shared-DRBG artifacts or timing channels) beyond the explicit transcript leakage of the ideal
			functionalities. Although $\SealToPeer$ encrypts different slot-resident values under the same
			recipient directory key $\pk_{\mathrm{seal}}^{(P_{\mathrm{peer}})}$, IND-CCA security of $(\Enc,\Dec)$
			implies that the joint distribution of all such ciphertexts (across slots and across calls) is
			computationally indistinguishable from a product distribution of independently generated dummy
			ciphertexts under $\pk_{\mathrm{seal}}^{(P_{\mathrm{peer}})}$. Hence, these outputs can be simulated
			either by independent per-slot simulators (each using fresh encryption randomness) or by a single
			key-independent sealing simulator shared across slots.			
			Key-independent interfaces $f\in\mathcal F_{\mathrm{KI}}$ are assumed simulatable without knowing any resident key
			(and may be handled by separate key-independent wrapper state).
		\end{itemize}
		For the scalar-share slots used in SDKG we instantiate $\PubMap(k):=k\G$ (so $\GetPub$ returns $k\G$).
		The simulator is understood to run in the same ambient execution model as the surrounding protocol and it may use any simulator-only interface provided by that model.}
\end{assumption}

\begin{remark}[Key-opacity in multi-slot KeyBoxes]
	\label{rem:key-opacity-multislot}
	In the $\FKeyBox$-hybrid model each party’s KeyBox instance stores multiple keys indexed by local slots
	$\mu\in\SlotSpace$, with $k_{i,\mu}:=\mathrm{\Lambda}[\mu]$ and $\pk_{i,\mu}:=\PubMap(k_{i,\mu})$. As implied by Assumption~\ref{assump:keybox-opacity}, in hybrids/proofs we may simulate different slots by running independent copies of the single-slot key-opacity simulator, one per slot $(i,\mu)$, each maintaining its
	own state across queries to that slot. Key-independent calls such as $\OpenFromPeer$ are handled by separate
	key-independent simulation state when needed. When the profile includes $\SealToPeer$, note that ciphertexts across different slots may be under the
	same recipient key $\pk_{\mathrm{seal}}^{(P_{\mathrm{peer}})}$; nevertheless, by IND-CCA security the
	simulator may treat each $\SealToPeer$ output as an independently simulatable ciphertext, so running independent per-slot
	simulators remains sound at the level of the joint external transcript.
\end{remark}

\begin{remark}[Slot separability is an idealization (implementation caveat)]
	\label{rem:slot-sep-idealization}
	Assumption~\ref{assump:keybox-opacity} treats different slots as independent black boxes: admissible
	operations are slot-local and the randomness used across slots and across calls can be
	taken independent, with no additional cross-slot/cross-call leakage at the KeyBox API boundary beyond the
	explicit leakage of our ideal functionalities.
	This is an idealization.
	
	\noindent In concrete KeyBox/HSM implementations, cross-slot correlations can arise from shared entropy/DRBG state,
	nonce/counter reuse, related-key derivation from a common root, or microarchitectural/timing side channels.
	Such correlations may invalidate key-opacity even when the underlying primitive is IND-CCA secure under the
	standard assumption of fresh, independent coins per encryption. The argument that many
	$\SealToPeer$ outputs are jointly simulatable as independent dummy encryptions relies on per-call independent
	randomness and the absence of extra correlated leakage.
	
	\noindent To that end, to instantiate Assumption~\ref{assump:keybox-opacity} in practice, implementations should ensure:
	(i) domain-separated key derivation and per-purpose/per-slot RNG (or otherwise demonstrably independent coins)
	for randomized operations, including sealing; (ii) nonce/coin generation that is robust against reuse/correlation, or the use of misuse-resistant constructions where applicable; and (iii) a strict allowlist/profile adapter that
	prevents cross-mechanism compositions as in the PKCS\#11 API-level attack literature \cite{Bortolozzo[10]}.
\end{remark}

\begin{definition}[KeyBox API profile]\label{def:keybox-profile}
	\emph{A \emph{KeyBox API profile} is a quadruple
		$(\chi_{\mathrm{adm}},\mathcal F_{\mathrm{adm}},\mathcal F_{\mathrm{KI}},\digamma)$ where
		$\chi_{\mathrm{adm}}$ is the set of admissible derivation routines that may be invoked via $\Load$,
		$\mathcal F_{\mathrm{adm}}$ is the set of admissible operations that may be invoked via $\Use$,
		$\mathcal F_{\mathrm{KI}}\subseteq \mathcal F_{\mathrm{adm}}$ is the (profile-fixed) subset of
		\emph{key-independent} admissible operations, and $\digamma:\mathcal F_{\mathrm{adm}}\to \mathsf{Fam}$ is the
		state-family map.
		A KeyBox instance accepts only calls $\Load(\mu,g,\cdot)$ with $g\in \chi_{\mathrm{adm}}$ and
		$\Use(\mu,f,\cdot)$ with $f\in \mathcal F_{\mathrm{adm}}$; for operations $f\notin \mathcal F_{\mathrm{KI}}$ the
		addressed slot $\mu$ must be populated (contain a resident key), while operations in $\mathcal F_{\mathrm{KI}}$
		may be invoked on empty slots and are interpreted as ignoring $\mu$ (and using only key-independent KeyBox state, if any).}
\end{definition}

Opaque buffer handles returned by $\OpenFromPeer$ are type-tagged in the global encoding:
a handle is always represented as $\KBhdl{\tau}=\langle\texttt{hdl},\tau\rangle$ for $\tau\in\{0,1\}^\lambda$,
and $\mathsf{Resolve}$ substitutes only such tagged handles, thereby preventing accidental collisions with ordinary fields.

\begin{remark}[Linear/one-shot handle consumption]\label{rem:resolve-linear}
	In $\FKeyBox$ (Fig.~\ref{Fdskg}), $\mathsf{Resolve}$ is \emph{one-shot}: every typed handle
	$\KBhdl{\tau}$ that is successfully substituted is deleted from $\mathsf{buf}$ before
	$\mathsf{Resolve}$ returns. Hence, handles returned by $\OpenFromPeer$ are linear resources:
	they cannot be reused across multiple $\Load/\Use$ calls. Protocol steps that need the same sealed
	payload more than once must invoke $\OpenFromPeer$ again to obtain fresh handles; Algorithm~\ref{alg:reg}
	does this explicitly.
\end{remark}

\begin{remark}[Profiles must not hide extractor-relevant oracle logs]\label{rem:no-hidden-uc-prover}
	In the \gROCRP\ model, straight-line extraction for our Fischlin-based UC-NIZK-AoKs requires the prover's
	oracle-log $\mathsf{Log}_{\mathcal P^\ast}$ under the corresponding proof context (Definition~\ref{def:NIZK-AoK},
	Remark~\ref{rem:oracle-tape}). Since KeyBox-internal oracle calls are not exposed at the host/API boundary,
	we require that any proof for which the UC argument invokes oracle-log-based extraction is generated by the
	host/party ITM (outside the KeyBox). Concretely, an admissible KeyBox profile must not include any operation that
	produces the UC-context consistency AoKs used by the surrounding protocol, nor any equivalent proof-generation
	interface that would cause the relevant oracle queries to occur inside the KeyBox boundary.
	
	\smallskip
	\noindent	This requirement is sometimes read as being in tension with the intuition that, in a deployment where the host
	OS/hypervisor is treated as adversarial, one would like all witness-bearing scalars to remain inside the
	trusted boundary. Our model does not claim (and does not require) this stronger property. In SDKG, the witnesses
	used for the UC-context consistency AoKs are treated as NXK-restricted material
	(Remark~\ref{rem:transport-vs-nxk}): they may exist transiently in host RAM during an atomic local step and are then securely erased. The NXK guarantee
	enforced by $\FKeyBox$ is the non-exportability of long-term KeyBox-resident shares, not protection against an
	``always-on'' RAM adversary on an otherwise honest host. Achieving the latter would require either strengthening
	the model to expose an extractor-visible oracle log from within the trusted boundary, or replacing oracle-log-based
	straight-line extraction with a different UC proof mechanism.
\end{remark}

\subsection{State continuity and failure modes}\label{subsec:nxk-state-continuity}
\begin{assumption}[State continuity]\label{assump:tee-continuity}
	\emph{We assume that each $\FKeyBox^{(i)}$ instance models a single hardware root whose internal sealed state is
	state-continuous, i.e., a PPT adversary $\cal A$ cannot (a) roll back $\FKeyBox^{(i)}$ to a previous sealed snapshot, nor
	(b) fork/clone $\FKeyBox^{(i)}$ into two independent instances that can be queried in parallel from the same prior state.
	Equivalently, $\cal A$ has no interface that resets $\FKeyBox^{(i)}$'s private state to an earlier value.}
\end{assumption}

\begin{assumption}[Seed integrity invariant]\label{assump:seed-integrity}
	\emph{Each $\FKeyBox^{(i)}$ instance stores a static PRF seed
		$\seed_i\in\{0,1\}^\lambda$, provisioned at enrollment/manufacturing/secure
		setup, that is independent of every KeyBox-resident signing share~$k$.
		We require:
		\begin{enumerate}[leftmargin=*,nosep]
			\item Rollback-invariance: A rollback or reset of
			mutable device/host state (or even of the KeyBox's mutable operation
			state $\mathsf{st}[\cdot]$) does not alter $\seed_i$.
			Concretely, $\seed_i$ resides in a write-once or append-only
			region of the KeyBox's sealed state that is outside the scope of any
			rollback-vulnerable mutable-state snapshot.
			\item Secrecy: $\seed_i$ is never exposed at the KeyBox
			API boundary: no admissible operation returns $\seed_i$ or any
			value from which $\seed_i$ can be efficiently recovered.
			\item Independence from signing shares: $\seed_i$ is
			sampled independently of every KeyBox-resident signing share
			$k_{i,\mu}$ and of all other parties' seeds $\seed_j$ for $j\neq i$.
		\end{enumerate}
		When the KeyBox is realized by a hardware root of trust, $\seed_i$ is typically derived from the hardware's unique
		device secret during secure provisioning via a domain-separated KDF,
		using a purpose tag disjoint from all other key-derivation purposes. If the seed integrity invariant is violated---e.g., $\seed_i$ is corrupted, rolled back to a stale value, or leaked at the API boundary---then deterministic nonce derivation can no longer be relied upon.
		In particular, if $\seed_i$ is revealed then LinOS prover nonces become predictable and (except with
		negligible probability over transcript generation) a single accepting LinOS transcript suffices to recover the
		resident share~$k$.
		More generally, if $\seed_i$ is not rollback-invariant, rollback can again lead to nonce-reuse style
		failures that LinOS is designed to prevent.
		As a consequence, seed provisioning and protection must be treated as part of the state-continuity engineering
		budget.}
\end{assumption}

\begin{remark}[Why an independent seed]
	\label{rem:independent-seed}
	One might consider deriving the nonce seed from the resident signing
	key~$k$ itself via a KDF (as in EdDSA~\cite{RFC8032}), but this
	entangles the PRF security hypothesis with the discrete-log assumption
	on~$k$: the reduction must argue that PRF outputs remain pseudorandom
	even when the adversary sees $K=k\G$ and interacts with $k$ via the
	admissible KeyBox profile.  While such a reduction can be carried out
	under a random-oracle/KDF assumption (cf.\ the analysis of EdDSA
	deterministic nonces), it tightens the bound and introduces an
	additional modeling assumption.  Keeping $\seed_i$ independent of $k$
	yields a clean reduction: the PRF game is played entirely over
	$\seed_i$, whose secrecy is a standalone KeyBox property
	(Assumption~\ref{assump:seed-integrity}), orthogonal to any DL-based
	argument.  If a deployment must derive the seed from~$k$ (e.g., due
	to hardware constraints on independent secret provisioning), then
	the security argument additionally requires modeling the KDF as a
	random oracle (or a dual-PRF) and the resulting bound includes a
	KDF-security term.  We flag this as a deployment caveat and recommend
	the independent-seed instantiation.
\end{remark}

In our NXK setting, a corrupted party may delegate witness-bearing
	computation to a state-continuous KeyBox instance.
	By Assumption~\ref{assump:tee-continuity}, such an instance cannot be rolled back or forked to answer the
	same commitment under two different challenges, and therefore extraction strategies that fundamentally
	depend on rewinding are not implementable at the hardware boundary. Note that state continuity is a strong assumption that real implementations may only approximate, and may require additional mechanisms to prevent rollback/forking-style state-continuity attacks \cite{Lore[25]}. In practice state continuity can be approximated using a monotonic freshness mechanism---e.g., hardware-backed monotonic counters in secure NVRAM/TPM, trusted time, or a server-maintained freshness oracle; details follow below.
	
	Even under approximate state continuity, a single rollback of the KeyBox's
	mutable operation state can---in a na\"{\i}ve design that samples and stores
	prover nonces---enable Schnorr special-soundness extraction and thus
	catastrophic disclosure of the resident share.
	To eliminate this specific failure mode, LinOS (Fig.~\ref{LinOS}) derives
	all Fischlin prover nonces deterministically from a static, rollback-invariant
	seed~$\seed$ stored inside the KeyBox (Assumption~\ref{assump:seed-integrity})
	and session-bound inputs $(\sid,K,i)$.
	Under this derivation, replaying the same session after rollback reproduces the
	same nonces and hence the same proof transcript; and distinct sessions yield
	pseudorandomly independent nonces under PRF security
	(Lemma~\ref{lem:rollback-robust}).
	This is a targeted hardening: it removes the ``one rollback $\Rightarrow$ key
	disclosure'' vector, but does not by itself close all state-continuity gaps
	(Remark~\ref{rem:rollback-scope}).
	
\begin{definition}[State-continuity failure event and parameter]\label{def:bad-sc}
	\emph{Let $\Bad_{\mathsf{sc}}$ denote the event that, in a given execution, some KeyBox instance
	violates Assumption~\ref{assump:tee-continuity}, i.e., its sealed state is rolled back to a prior value
	or forked/cloned into two independently queryable continuations from the same prior state.
	In a concrete deployment that approximates state continuity with an anti-rollback mechanism, let
	$\varepsilon_{\mathsf{sc}}(\lambda)$ be any bound taken over all randomness, faults, and adversarial actions,
	on $\Pr[\Bad_{\mathsf{sc}}]$ for the lifetime of an execution at security parameter $\lambda$.}
\end{definition}

\begin{lemma}[Additive degradation under approximate state continuity]\label{lem:sc-additive}
	Consider any security statement in this paper proved in the $\FKeyBox$-hybrid model under
	Assumption~\ref{assump:tee-continuity}, yielding an advantage bound of the form $\negl(\lambda)$.
	In a concrete realization in which Assumption~\ref{assump:tee-continuity} holds except with probability
	at most $\varepsilon_{\mathsf{sc}}(\lambda)$ (Definition~\ref{def:bad-sc}), the corresponding advantage
	bound becomes $\negl(\lambda)+\varepsilon_{\mathsf{sc}}(\lambda)$.
\end{lemma}

\begin{proofsketch}
	Let $\mathcal{E}$ be the relevant distinguishing/forgery event.
	Then $$\Pr[\mathcal{E}] \le \Pr[\mathcal{E}\wedge \neg\Bad_{\mathsf{sc}}] + \Pr[\Bad_{\mathsf{sc}}]
	\le \Pr[\mathcal{E}\mid \neg\Bad_{\mathsf{sc}}] + \varepsilon_{\mathsf{sc}}(\lambda).$$
	Conditioned on $\neg\Bad_{\mathsf{sc}}$, the execution matches the idealized model with
	state continuity, so $\Pr[\mathcal{E}\mid \neg\Bad_{\mathsf{sc}}]\le \negl(\lambda)$. \qed
\end{proofsketch}

\begin{definition}[Secure-erasure failure event and parameter]\label{def:bad-er}
	\emph{Let $\Bad_{\mathsf{er}}$ denote the event that, in a given execution, the
		secure-erasure / atomic-activation semantics of Definition~\ref{def:ace} is violated for
		some honest party’s host state. Concretely, $\Bad_{\mathsf{er}}$ occurs if there exist an honest
		party $P_i$ and a host-resident value $v$ that the protocol designates as erased (or transient
		within one honest activation) such that $v$ nevertheless becomes available to the adversary
		after the intended erasure point without a prior corruption of $P_i$.  
		In a concrete deployment, let $\varepsilon_{\mathsf{er}}(\lambda)$ be any bound taken over all randomness,
		faults, and adversarial actions, on $\Pr[\Bad_{\mathsf{er}}]$ for the lifetime of an execution at
		security parameter $\lambda$.}
\end{definition}

\begin{lemma}[Additive degradation under approximate secure erasures]\label{lem:er-additive}
	Consider any security statement in this paper proved in the UC model with adaptive corruptions
	with secure erasures (Definition~\ref{def:ace}), yielding an advantage bound of the form $\negl(\lambda)$.
	In a concrete realization in which Definition~\ref{def:ace} holds except with probability at most
	$\varepsilon_{\mathsf{er}}(\lambda)$ (Definition~\ref{def:bad-er}), the corresponding advantage bound becomes
	$\negl(\lambda)+\varepsilon_{\mathsf{er}}(\lambda)$.
\end{lemma}

\begin{proofsketch}
	Let $\mathcal{E}$ be the relevant distinguishing/forgery event.
	Then $$\Pr[\mathcal{E}] \le \Pr[\mathcal{E}\wedge \neg\Bad_{\mathsf{er}}] + \Pr[\Bad_{\mathsf{er}}]
	\le \Pr[\mathcal{E}\mid \neg\Bad_{\mathsf{er}}] + \varepsilon_{\mathsf{er}}(\lambda).$$
	Conditioned on $\neg\Bad_{\mathsf{er}}$, the execution matches the idealized model with secure erasures,
	so $\Pr[\mathcal{E}\mid \neg\Bad_{\mathsf{er}}]\le \negl(\lambda)$. \qed
\end{proofsketch}

	We do not restate $\varepsilon_{\mathsf{sc}}(\lambda)$ or $\varepsilon_{\mathsf{er}}(\lambda)$ in every theorem.
	By Lemmas~\ref{lem:sc-additive} and~\ref{lem:er-additive}, negligible security bounds in the remainder of the paper
	should be read as $\negl(\lambda)+\varepsilon_{\mathsf{sc}}(\lambda)+\varepsilon_{\mathsf{er}}(\lambda)$
	for the relevant deployment/execution.
	
Assumption~\ref{assump:tee-continuity} is a hard safety requirement for any KeyBox interface that is
intended to be one-shot. If the anti-rollback mechanism enforcing state continuity fails even once---e.g.,
a monotonic counter wraps, sealed snapshots can be restored from backup, or a freshness oracle is bypassed---then
rollback/forking becomes possible. In an unhardened design
that samples and stores prover nonces in mutable state, the degradation is
typically not graceful: the adversary may obtain multiple outputs from an
interface intended to be one-shot and can often recover the resident share
via special soundness from two responses under the same commitment. We therefore parameterize the failure of state continuity by the execution-level bad event
$\Bad_{\mathsf{sc}}$ (Definition~\ref{def:bad-sc}). In any deployment where $\Pr[\Bad_{\mathsf{sc}}]\le
\varepsilon_{\mathsf{sc}}(\lambda)$, Lemma~\ref{lem:sc-additive} implies that every advantage bound increases additively by $\varepsilon_{\mathsf{sc}}(\lambda)$. Later, we employ a method to eliminate this specific catastrophic path by deriving all Fischlin prover nonces deterministically from a static, rollback-invariant seed while other rollback-affected state remains under the
general $\varepsilon_{\mathrm{sc}}$.

\paragraph{Concrete accounting for \texorpdfstring{$\varepsilon_{\mathsf{sc}}(\lambda)$}{esc(lambda)} (non-normative).}
The cryptographic results in this paper do not attempt to bound $\varepsilon_{\mathsf{sc}}(\lambda)$; it is a
deployment/engineering parameter that upper-bounds the probability of any rollback/fork event
$\Bad_{\mathsf{sc}}$ over an execution (Definition~\ref{def:bad-sc}).
We record two common approximation patterns to clarify what typically contributes to $\varepsilon_{\mathsf{sc}}$:

\noindent (i) Monotonic counters:
Suppose state continuity is approximated by a $b$-bit monotonic counter that is incremented once per state advance
(e.g., per sealing epoch), and the implementation is engineered to \emph{fail closed} (refuse to unseal/advance)
before wrap-around and to require re-provisioning/rekeying before exhaustion. Then counter wrap-around contributes
zero to $\varepsilon_{\mathsf{sc}}(\lambda)$ for any deployment whose lifetime advance budget $N_{\max}$ satisfies
$N_{\max}<2^b$ (ignoring physical faults). For scale, $2^{64}\approx 1.8\times 10^{19}$ increments; even at
$10^6$ advances/day, exhaustion would occur only after $\approx 5\times 10^{10}$ years. In practice, however,
the effective budget is often dominated by write-endurance limits, rate limits, or administrative rotation,
rather than bit-width; the same accounting applies by replacing $2^b$ with the enforced safe-advance cap.

\medskip
\noindent (ii) Server-/time-based freshness:
For oracle- or time-based approaches, the dominant contribution to $\varepsilon_{\mathsf{sc}}$ is typically
policy, not cryptographic guessing: if the device ever continues to unseal/advance without a fresh token or
a non-decreasing trusted-time reading (i.e., it fails open), we count that execution under $\Bad_{\mathsf{sc}}$.
Conversely, if the mechanism is engineered to fail closed, outages impact availability but do not increase
$\varepsilon_{\mathsf{sc}}(\lambda)$.
Appendix~\ref{App1} discusses engineering trade-offs and re-provisioning strategies.

\subsection{Hybrid execution model and NXK-restricted material}\label{subsec:nxk-hybrid}
\begin{definition}[\(\mathcal F_{\mathsf{KeyBox}}\)-hybrid model]\label{KeyBox}
	\emph{Let $\FKeyBoxOf{i}$ denote the ideal functionality for $P_i$'s KeyBox$_i$. The $\mathcal F_{\text{KeyBox}}$-hybrid model is defined by a PPT environment~$\mathcal Z$,
		a PPT adversary~$\mathcal A$, a set of PPT parties
		$\{P_i\}_{i\in[n]}$, and a collection of ideal
		functionalities $\bigl\{\mathcal F_{\text{KeyBox}}^{(i)}\bigr\}_{i\in[n]}$ that are created as:
		\begin{enumerate}
			\item Instantiation: for every $i\in[n]$, generate a fresh functionality instance
			$\mathcal F_{\text{KeyBox}}^{(i)}$, initialized with an empty state.
			\item Communication between parties is performed via $\Fchan$ for confidential authenticated point-to-point messages,
			and via $\Fpub$ for authenticated transcript-public dissemination.
			\item Adaptive corruptions with secure erasures: parties are corrupted adaptively under Definition~\ref{def:ace}.
			Upon corruption of $P_i$, $\cA$ learns only $P_i$’s current host state (outside the KeyBox boundary);
			any values explicitly erased by the protocol are not revealed. Thereafter, $\cA$ controls $P_i$ and may invoke
			$\FKeyBoxOf{i}.\Load$ and $\FKeyBoxOf{i}.\Use$, but KeyBox-resident secret state (in particular the map
			$\mathrm{\Lambda}$ of resident shares) is never revealed.
	\end{enumerate}}
\end{definition}

In the $\FKeyBox$-hybrid model, ``$P_i$ invokes $\FKeyBox^{(i)}.\Load/\Use(\cdots)$'' denotes a party-local call over the internal channel, executed immediately upon activation if $P_i$ is honest. A corrupted $P_i$ leaves the timing and admissible inputs to the adversary, i.e., the ideal functionality cannot directly write into $\FKeyBox^{(i)}$ or force installation/registration for corrupted parties.

\begin{remark}[KeyBox-driver wrapper for ideal-world KeyBox calls]\label{rem:kb-wrapper}
	In Fig.~\ref{Fdskg}, $\FKeyBox^{(P_{\mathsf{own}})}$ accepts $(\Load,\mu,g,m)$ and $(\Use,\mu,f,m)$ only from its
	owner party $P_{\mathsf{own}}$. Thus, under standard UC semantics an ideal functionality cannot directly issue
	$\Load/\Use$ to $\FKeyBox^{(i)}$. Throughout, whenever an ideal functionality description says
	``have $P_i$ invoke $\FKeyBox^{(i)}.\Load/\Use(\cdot)$,'' this is shorthand for the following fixed mechanism:
	in the ideal execution, the party $P_i$ is composed with a deterministic KeyBox-driver wrapper $\WKB^{(i)}$
	that behaves like the dummy party on all external ports, but additionally implements an internal command port
	from ideal functionalities. Upon receiving $(\KBcmd,\sid,\mathsf{ops})$, where $\mathsf{ops}$ is a list of KeyBox calls
	of the form $(\Load,\mu,g,m)$ and/or $(\Use,\mu,f,m)$, if $P_i$ is honest then $\WKB^{(i)}$ executes the listed calls
	sequentially against its local instance $\FKeyBox^{(i)}$ and returns $(\KBret,\sid,\mathsf{res})$ with the list of
	return values. If $P_i$ is corrupted, $\cA$ controls $\WKB^{(i)}$ and may delay/modify/ignore these commands.
\end{remark}

\begin{remark}[NXK-restricted state]
	\label{rem:transport-vs-nxk}
We call a (possibly structured) bitstring $v\in\{0,1\}^*$ \emph{share-deriving} only
relative to a designated KeyBox share.
Fix a party $P_i$ and a local KeyBox slot $\mu\in\SlotSpace$ in $\FKeyBox^{(i)}$
that contains a scalar long-term share $k_{i,\mu}\in\Zp$.
We say that $v$ is share-deriving for $(i,\mu)$ if there exist PPT-computable functions
$a(\cdot),b(\cdot),y(\cdot)$ (fixed independently of the secret share $k_{i,\mu}$) such that
on input $v$ they output $a(v)\in \Zp^\ast$, $b(v)\in \Zp$, and $y(v)\in \Zp$ with
\[
y(v) = a(v)\cdot k_{i,\mu} + b(v) \pmod p.
\]
Equivalently, from $v$ alone one can compute a \emph{caller-invertible} affine image
$y(v)=L_{a(v),b(v)}(k_{i,\mu})$ together with map parameters $(a(v),b(v))$, where
$L_{a,b}(x):=ax+b$, so the caller can recover
$k_{i,\mu}=a(v)^{-1}(y(v)-b(v))\bmod p$.
When the target slot is clear from context, we omit $(i,\mu)$ and simply say that $v$ is
share-deriving. If $(a,b)$ are fixed by the KeyBox profile or chosen by the caller, the
functions $a(\cdot),b(\cdot)$ may hard-code those public parameters or parse them from $v$. Because we fix an injective, self-delimiting encoding $\langle\cdot\rangle$ of mixed tuples into
	$\{0,1\}^*$, this definition applies equally to any finite collection of values
	$(v_1,\ldots,v_t)$ by taking $v := \langle v_1,\ldots,v_t\rangle$. We will sometimes abuse notation and refer to such a collection as \emph{share-deriving material}.
	
	\noindent We treat share-deriving material as \emph{NXK-restricted}: it must not appear in the adversary-visible transcript beyond
	the explicit leakage modeled by our channels (e.g., $\Fchan$'s length leakage), and it must not be written to persistent
	storage outside a KeyBox. Honest parties may nevertheless handle share-deriving material transiently in host memory and
	may transmit it over authenticated, confidential channels (modeled by $\Fchan$) when required by the surrounding protocol,
	provided that (i) whenever it is used to install a long-term share in $\FKeyBox^{(i)}$ it is delivered only via the
	internal secure channel between $P_i$ and $\FKeyBox^{(i)}$, and (ii) it is securely erased from host memory immediately
	after its last use. This transient handling is protocol-internal and is distinct from a KeyBox export interface: under an
	admissible KeyBox profile (Assumption~\ref{assump:keybox-opacity}), $\FKeyBox$ never returns resident shares (or
	caller-invertible affine images) in the clear.
	Consequently, an adversary can learn enough share-deriving material to recompute an honest device's share only by
	corrupting the relevant endpoint(s) during the protocol before erasure,
	not from the public transcript alone. Share-deriving is intentionally a transcript-only notion: it captures values from which the caller can recover
	$k_{i,\mu}$ from $v$ alone (i.e., without any additional secret inputs).
	An API can still leak a resident share indirectly by returning an output that is not share-deriving by itself but
	becomes share-deriving once combined with caller-held secrets permitted by the profile (e.g., a ciphertext under a
	caller-supplied wrapping key). Such caller-recoverable exports are ruled out by key-opacity, so we explicitly exclude
	caller-decryptable wrapping/export: if the caller knows the decryption key, it can decrypt to obtain share-deriving
	plaintext (or an invertible affine image), which is not simulatable from $\PubMap(k_{i,\mu})$ alone and therefore
	violates key-opacity \cite{Bortolozzo[10]}.
\end{remark}

\subsection{Real-world instantiations of admissible KeyBox profiles}
\label{sec:keybox-realizations}
Our KeyBox idealization is deliberately \emph{profile-centric} (Definition~\ref{def:keybox-profile}): the surrounding
protocol fixes an admissible KeyBox profile \((\chi_{\mathsf{adm}},\mathcal F_{\mathsf{adm}})\), and the KeyBox is
assumed to accept only those derivations and operations. This is essential as API-level non-exportability alone does not suffice for our security arguments. Many deployed keystore/HSM APIs expose operations that either (a) directly leak share-deriving material, or (b) let a
caller recover a resident share indirectly (e.g., via caller-decryptable wrapping/export under a caller-controlled key),
even when raw secret key bytes are nominally non-exportable. In our model this is captured by (i) pinning the profile and (ii) requiring key-opacity
(Assumption~\ref{assump:keybox-opacity}) for the induced external transcript, i.e., that everything observable outside the
KeyBox boundary is simulatable given only the corresponding public key \(\PubMap(k)\) and the public
query inputs. Further, Remark~\ref{rem:transport-vs-nxk} treats all share-deriving material (including caller-invertible affine images
of a stored share) as NXK-restricted, and separately excludes caller-decryptable wrapping/export via key-opacity even
though the ciphertext alone need not be share-deriving under the transcript-only definition.

\subsubsection{Implementation note (non-normative).}
\label{subsubsec:keybox-candidates}
As mentioned earlier, this paper's security proofs assume a profile-centric KeyBox, satisfying key-opacity and state continuity	(Assumption~\ref{assump:tee-continuity}). In practice, such a profile can be enforced
	by construction by interposing a narrow ``profile adapter'' between the protocol and a broader vendor API---e.g., a minimal TEE enclave used purely as a keystore, an
	attested enclave$\leftrightarrow$KMS integration, or an HSM under a strict allowlist.
	The adapter must (i) forbid any vendor-API operation that returns share-deriving outputs of a resident key (i.e., caller-invertible affine images of an already-installed long-term share; cf.\ Remark~\ref{rem:transport-vs-nxk}), while still permitting the protocol's own key-independent operations (such as $\textsf{SDKG.LeafInit}\in\mathcal{F}_{\mathrm{KI}}^{\mathsf{SDKG}}$), whose outputs are share-deriving only prospectively and are handled as NXK-restricted transient host state with secure erasure; (ii) prohibit caller-decryptable wrapping/export, preserving key-opacity; (iii) pin sealing recipients to attested identities while engineering explicit freshness/anti-rollback to approximate state continuity.
	Appendix~\ref{app:keybox-impls} summarizes candidate deployment families
	(Table~\ref{tab:keybox-candidates}) and gives a profile-capture checklist. Throughout, when we say that some computation runs ``inside the KeyBox boundary,'' we include any minimal, pinned ``profile adapter'' that is itself part of the same attested, state-continuous trusted boundary as the KeyBox; and 	(iv) forbid any KeyBox API primitive that can generate the UC-extractable consistency AoKs used by the protocol
	(e.g., Fischlin-based UC-context proofs): these AoKs must be generated by the host/party ITM so that the UC simulator
	can record the prover's \gROCRP\ query log required for straight-line extraction (Remark~\ref{rem:oracle-tape}). It must also ensure that randomized operations (including sealing) use fresh, domain-separated randomness across slots and across calls (Remark~\ref{rem:slot-sep-idealization}); otherwise the slot-separability component
	of key-opacity may fail in spite of IND-CCA security of the abstract scheme.

\section{The Conflict: Verifiable Sharing vs. NXK}
\label{sec:why-vss}
In this section, we isolate the structural reason that mandates UC-secure DKG protocols to enforce the core (R)VSS obligations:
secrecy against unauthorized sets and uniqueness (a single well-defined shared secret) together with
affine consistency of honest parties' local outputs. In the literature, these obligations are ensured by a
VSE layer, which is classically instantiated via (R)VSS and its
variants (PVSS/AVSS) using polynomial sharing \cite{Shamir[79]}, commitments, and complaint/opening logic. In other UC(-style) DKGs,
the same role is alternatively realized via commitment-and-proof / (N)IZK-based authenticated sharing
that prevents equivocation and certifies transcript-defined affine relations (e.g., \cite{LN18,CGGMP21}).

\begin{definition}[DL experiment]\label{def:dl}
	\emph{Let $\mathbb{G}$ be a cyclic group of prime order $p=p(\lambda)$ with generator $\G\in\mathbb{G}$.
		For a PPT adversary $\mathcal{A}$, define
		\[
		\Adv^{\mathrm{dl}}_{\mathcal{A}}(\lambda)
		:= \Pr\Big[x \leftarrowdollar \mathbb{Z}_p^*;\ X:=x\G;\ x' \leftarrow \mathcal{A}(\G,X) : x'=x \Big].
		\]}
\end{definition}

\begin{assumption}[DL]\label{assump:dl}
	\emph{For all PPT adversaries $\mathcal{A}$, $\Adv^{\mathrm{dl}}_{\mathcal{A}}(\lambda)\le \negl(\lambda)$.}
\end{assumption}

\begin{definition}[DDLEQ game]\label{def:ddleq}
	\emph{Let $\mathbb{G}$ be a cyclic group of prime order $p=p(\lambda)$ with (public) generators $\G,\h\in\mathbb{G}$.
	Consider the experiment that samples $r,s\leftarrowdollar \Zps$ and a bit $b\leftarrowdollar\{0,1\}$.
	If $b=1$ set $(A,B):=(r\G,r\h)$; if $b=0$ set $(A,B):=(r\G,s\h)$.
	An adversary $\mathcal{A}$ is given $(\G,\h,A,B)$ and outputs a bit $b'$.
	Define
	\[
	\Adv^{\mathrm{ddleq}}_{\mathcal A}(\lambda):=\left|\Pr[b'=b]-\tfrac12\right|.
	\]}
\end{definition}

\begin{assumption}[DDLEQ]\label{assump:ddleq}
	\emph{For all PPT adversaries $\mathcal A$,  
	$\Adv^{\mathrm{ddleq}}_{\mathcal A}(\lambda)\le \negl(\lambda)$.}
\end{assumption}

\begin{remark}[Terminology: DDLEQ is DDH]\label{rem:ddleq-vs-ddh}
	Definition~\ref{def:ddleq} is exactly the standard DDH distinguishing game on the tuple
	$(\G,\h,A,B)$, written in the “same-exponent across two bases” form matching DLEQ statements.
	We use the name DDLEQ only to align notation with the Chaum--Pedersen relation.
	(If $\log_{\G}(\h)$ were known, the game would be trivial by checking $B=(\log_{\G}\h)\cdot A$.)
\end{remark}

We use the standard UC framework with ITMs as in Canetti~\cite{Can[01]}.
Let $\Exec(\mathrm{\Psi},\cA,\cZ,\lambda)$ and $\Ideal(\F,\s,\cZ,\lambda)$ denote the standard real and ideal execution
ensembles (Exec/Ideal experiments).

\begin{definition}[UC realization]
	\label{def:uc-realization}
	\label{UC1} 
	\label{UC2} 
	\emph{A protocol $\mathrm{\Psi}$ UC-realizes an ideal functionality $\F$ if for every PPT adversary $\cA$
	there exists a PPT simulator $\s$ such that for every PPT environment $\Z$,
	\[
	\Exec(\mathrm{\Psi},\cA,\cZ,\lambda)\ \approx_c\ \Ideal(\F,\s,\cZ,\lambda).
	\]}
\end{definition}

\begin{note}
	Our only protocol security notion is UC realization. Other experiment‑based definitions in the paper are standard assumptions or local properties of underlying resources/primitives used solely as hypotheses to establish UC realization of our ideal functionalities.
\end{note}

Informally, a DL-based DKG outputs:
(i) a public key $K$ to everyone and
(ii) a secret sharing of its discrete log $k$ among the parties according to the access structure $\mathrm{\Gamma}$,
so that authorized sets $A \in \mathrm{\Gamma}$ can reconstruct $k$ or jointly sign using their shares while unauthorized sets learn no (non-negligible) information about $k$ beyond $K$. In a UC formulation, the (DL-based) ideal functionality $\F_{\mathsf{DKG}}^\mathrm{\Gamma}$ samples a fresh secret
$k \leftarrowdollar \mathbb{Z}_p$, distributes shares $(k_1,\dots,k_n)$ consistent with $\mathrm{\Gamma}$,
and outputs $K := k\G$. This implies three tightly coupled obligations that are classically associated with (R)VSS:

\begin{itemize}
	\item Secrecy: for any corruption set $B \notin \mathrm{\Gamma}$, the adversary learns no (non-negligible) information
	about $k$ beyond $K$.
	\item Uniqueness (strong correctness): whenever honest parties accept completion, there exists a single
	value $k$ such that every authorized set $A\in\mathrm{\Gamma}$ can reconstruct that same $k$.
	\item Affine consistency: honest parties' local outputs are consistent with one global sharing instance
	of $k$ under $\mathrm{\Gamma}$: the transcript cannot induce incompatible sharings across different honest subsets.
\end{itemize}

\subsection{DKG subsumes dealerless (R)VSS}
\label{sec:dkg-implies-vss}

We capture the verifiable-sharing subtask by an ideal functionality for dealerless random (R)VSS that is
exactly the UC-DKG functionality with its public-key output suppressed.
Concretely, for a fixed access structure $\mathrm{\Gamma}$, the functionality $\F_{\mathsf{RVSS}}^\mathrm{\Gamma}$ is defined as follows:
on session identifier $\sid$ and upon receiving $\mathsf{init}$ from all parties, it samples
$k \leftarrowdollar \Zp$ and distributes shares $(k_1,\dots,k_n)$ consistent with $\mathrm{\Gamma}$, and outputs
no additional public value. Equivalently, $\F_{\mathsf{RVSS}}^\mathrm{\Gamma}$ is obtained from
$\F_{\mathsf{DKG}}^\mathrm{\Gamma}$ by locally dropping the public output $K:=k\G$\footnote{An RVSS functionality that also outputs $K:=k\G$ can be denoted by $\F_{\mathsf{RVSS}}^{\mathrm{\Gamma},\mathsf{pk}}$; we will not use it here.}.

\begin{theorem}[UC-DKG implies UC-RVSS]
	\label{thm:dkg-implies-rvss}
	Let $\mathrm{\Psi}$ be any protocol that UC-realizes $\F_{\mathsf{DKG}}^\mathrm{\Gamma}$ in some model $\cal M$. Then there exists a protocol $\mathrm{\Psi}'$ that UC-realizes $\F_{\mathsf{RVSS}}^\mathrm{\Gamma}$ in the same model $\cal M$.
\end{theorem}

\begin{proofsketch}
	This is immediate from UC closure under efficient post-processing: $\mathrm{\Psi}'$ is obtained from $\mathrm{\Psi}$ by locally
	suppressing a public output, and $\F_{\mathsf{RVSS}}^\mathrm{\Gamma}$ is obtained from $\F_{\mathsf{DKG}}^\mathrm{\Gamma}$
	by the same transformation. Formally, given any $\cA$ for $\mathrm{\Psi}'$, build $\cA^\star$ for $\mathrm{\Psi}$ that forwards
	messages unchanged and drops $K$, and apply the simulator for $\mathrm{\Psi}$ with the same post-processing. \qed
\end{proofsketch}

Hence, any UC-secure DKG protocol must already satisfy the (R)VSS properties embodied by
$\F_{\mathsf{RVSS}}^\mathrm{\Gamma}$: secrecy against unauthorized sets and a unique, well-defined shared secret
underlying honest outputs. Thus, in the standard model (without trusted hardware), a DKG protocol needs a
mechanism that enforces exactly these VSS-style guarantees. Whether it is realized via
(R)VSS/PVSS/AVSS and/or via commitment-and-proof / ZK-based authenticated sharing, it plays the
same conceptual role which is that of a VSE layer.

\begin{lemma}[Uniqueness is necessary for UC-DKG]
	\label{lem:uniqueness-necessary}
Let $\mathrm{\Psi}$ be a protocol intended to UC-realize $\F_{\mathsf{DKG}}^\mathrm{\Gamma}$. Suppose there exists a non-negligible probability event wherein an execution of $\mathrm{\Psi}$ terminates without honest abort, yet there exist authorized sets $A, B \in \mathrm{\Gamma}$ whose respective reconstructions yield $k_A \neq k_B$, then $\mathrm{\Psi}$ does not UC-realize $\F_{\mathsf{DKG}}^\mathrm{\Gamma}$.
\end{lemma}

\begin{proofsketch}
	Define a PPT environment $\Z$ that runs one execution and then adaptively corrupts all parties in $A \cup B$.
	From their revealed states, $\Z$ computes $k_A$ and $k_B$ using the prescribed reconstruction algorithm and
	outputs $1$ iff $k_A \neq k_B$.
	By assumption, $\Pr[\Z \text{ outputs }1]$ is non-negligible in the real world. In the ideal world, $\F_{\mathsf{DKG}}^\mathrm{\Gamma}$ samples a single secret $k$ and distributes shares
	consistent with that $k$, so any authorized set must reconstruct the same $k$.
	Thus, $\Pr[\Z \text{ outputs }1]=0$ (up to negligible reconstruction error) in the ideal world, contradicting
	UC indistinguishability. \qed
\end{proofsketch}

Lemma~\ref{lem:uniqueness-necessary} is precisely the ``verifiability'' obligation that VSS packages:
malicious parties must not be able to make different honest, authorized subsets accept incompatible sharings.
Therefore, any UC-secure DKG must implement a mechanism that prevents (or detects and neutralizes)
equivocation in the distribution of share material and/or in the public data that defines the sharing. For additional discussion of DKG security notions and constructions, see \cite{KomloGS23}.

\subsection{Exported-share enforcement vs. NXK}
\label{sec:why-vss-machinery}

In the traditional, transcript-visible share (a.k.a. ``exportable-share'') model, parties exchange explicit
share-derived values over the network (so they become part of the adversary-visible transcript), and an adaptive adversary
can later corrupt parties and inspect their local states.

\begin{enumerate}
	\item Hiding of each contribution: parties must contribute randomness to the final key without
	revealing their secret contribution to unauthorized corruptions.
	\item Binding/consistency of each contribution: a malicious party must not be able to send
	inconsistent information to different recipients in a way that makes honest parties accept incompatible
	sharings. Moreover, the transcript must support the simulation/extraction requirements demanded by UC.
\end{enumerate}

This is exactly what classical VSS-based mechanisms provide: each party acts as a dealer, shares a random
secret in a verifiable way, and the final secret is the sum of the non-disqualified contributions. In
DL-based constructions, Feldman/Pedersen-style commitments \cite{Feld[87],Ped[91]} additionally expose the public group
element corresponding to each dealer's secret, enabling computation of the public key. The same enforcement role can also be realized via ``authenticated sharing'' based on commitments and (N)IZK proofs: parties commit to contributions and prove knowledge/consistency of the
relations that the transcript induces.

For the following Proposition, fix a prime field $\mathbb{F}_p$ and an integer $t\ge 1$.
Let $f(X)\in\mathbb{F}_p[X]$ be uniformly random of degree at most $t$, and define $k_i:=f(i)$ for $i\in[n]$.
Let $x_{\isnew}\in\mathbb{F}_p$ with $x_{\isnew}\notin[n]$ denote the evaluation point for a joining device.
Let $\pp$ denote the public parameters and any fixed transcript-public information.
We work in the NXK/$\FKeyBox$ setting from Section~\ref{sec:model}.
Let $\tau_{\mathrm{ext}}$ denote the external enrollment view/transcript outside all KeyBox instances.
Assume $f(x_{\isnew})$ is computationally unpredictable given $\pp$: there exists a function
$\varepsilon=\varepsilon(\lambda)$ such that for every PPT predictor $\mathcal{B}$ and every realization of $\pp$ in the support,
\[
\Pr\!\big[\mathcal{B}(\pp)=f(x_{\isnew}) \,\big|\, \pp\big]\ \le\ \varepsilon(\lambda).
\]

\begin{proposition}[External fresh-share enrollment and NXK]\label{prop:nxk-resharing-impossible-updated}
	Let the setup be as above. Then for any PPT strategy that computes an output
	$\widehat{k}_{\isnew}\in\mathbb{F}_p$ from $(\pp,\tau_{\mathrm{ext}})$,
	\[
	\Pr\!\big[\widehat{k}_{\isnew}=f(x_{\isnew}) \,\big|\, \pp\big]\ \le\ \varepsilon(\lambda)+\negl(\lambda).
	\]
\end{proposition}

\begin{proof}
	Define Hybrid $\Game_0$ as the real enrollment execution and hybrid $\Game_1$ by modifying $\Game_0$ as follows:
	maintain a table of simulator instances indexed by KeyBox slots. Whenever a slot $(P,\mu)$ is first queried and has an installed key
	$k_{P,\mu}:=\mathrm{\Lambda}[P,\mu]$, set $\pk_{P,\mu}\gets \Pub(k_{P,\mu})$ and initialize an independent simulator instance
	$\s_{P,\mu}$ by running $\s$ on input $(1^\lambda,\pk_{P,\mu})$.
	Thereafter, for every call to $\FKeyBox.\Use(\mu,f,m)$ initiated by owner $P$ with
	$f\in\mathcal F_{\mathrm{adm}}$, respond as follows:
	\begin{itemize}[leftmargin=*,nosep]
		\item Key-dependent interfaces (slot-bound):
		If $f\neq \OpenFromPeer$, replace the real reply produced by
		$f(k_{P,\mu},\mathsf{st}[P,\mu,\digamma(f)],m)$ with the output of $\s_{P,\mu}(f,m)$.
		\item Key-independent interfaces (sealing-only):
		If $f=\OpenFromPeer$, answer using a separate key-independent simulator state
		$\s^{\mathsf{KI}}_{P}$ as allowed by Assumption~\ref{assump:keybox-opacity}.
		Concretely, $\s^{\mathsf{KI}}_{P}$ is stateful across $\OpenFromPeer$ calls and maintains its own
		buffer state for typed handles so that subsequent uses of those handles (via $\mathsf{Resolve}$)
		are answered consistently with the handles it issued.
	\end{itemize}
	The instance $\s_{P,\mu}$ is reused across all queries to the same slot so it may maintain state,
	while different slots use independent instances (cf.\ Assumption~\ref{assump:keybox-opacity} and
	Remark~\ref{rem:key-opacity-multislot}).
			
	By key-opacity of $\mathcal{F}_{\mathrm{adm}}$ w.r.t.\ $\Pub$, the external transcript/view $\tau_{\mathrm{ext}}$ in $\Game_0$
	and $\Game_1$ are computationally indistinguishable; hence for any event $\mathcal{E}$,
	\[
	\big|\Pr[\mathcal{E}\mid \Game_0]-\Pr[\mathcal{E}\mid \Game_1]\big|\le \negl(\lambda).
	\]
	
In $\Game_1$, by construction, all externally visible key-dependent outputs are generated by the
slotwise simulators from $\pp$ and the corresponding public values $\pk_{P,\mu}$ only, while the
key-independent sealing-only interface $\OpenFromPeer$ is handled by the separate simulator state
$\s^{\mathsf{KI}}_{P}$ (which, per Assumption~\ref{assump:keybox-opacity}, does not require any resident key).
	
	Let $\mathcal{A}_{\mathrm{out}}$ be the (PPT) procedure that outputs $\widehat{k}_{\isnew}$ from $(\pp,\tau_{\mathrm{ext}})$.
	Define a PPT predictor $\mathcal{B}$ that on input $\pp$ samples $\tau_{\mathrm{ext}}$ according to $\Game_1$
	(using the same slotwise simulators $\{\s_{P,\mu}\}$ and key-independent sealing simulator states
	$\{\s^{\mathsf{KI}}_{P}\}$ as above) and outputs $\mathcal{A}_{\mathrm{out}}(\pp,\tau_{\mathrm{ext}})$. Then for every fixed $\pp$ in the support,
	\[
	\Pr\!\big[\widehat{k}_{\isnew}=f(x_{\isnew}) \,\big|\, \pp,\Game_1\big]
	=\Pr\!\big[\mathcal{B}(\pp)=f(x_{\isnew}) \,\big|\, \pp\big]
	\le \varepsilon(\lambda).
	\]
	Transferring back to $\Game_0$ via indistinguishability yields
	$\Pr[\widehat{k}_{\isnew}=f(x_{\isnew})\mid \pp]\le \varepsilon(\lambda)+\negl(\lambda)$.
	\qed
\end{proof}

\begin{remark}[Min-entropy vs.\ public group elements in $\pp$]\label{rem:nxk-resharing-minentropy}
	In a prime-order group with fixed generator $\G$, the map $a\mapsto a\G$ is a bijection on $\Zp$.
	Thus, if $\pp$ contains a group element that is a deterministic one-to-one function of $f(x_{\isnew})$, then the information-theoretic conditional min-entropy of $f(x_{\isnew})$ given $\pp$
	is $0$, even though recovering $f(x_{\isnew})$ from $\pp$ may be computationally hard under DL.
	Accordingly, Proposition~\ref{prop:nxk-resharing-impossible-updated} is stated in terms of
	computational unpredictability. As a special case, if $\pp$ is such that $\max_{a\in\mathbb{F}_p}\Pr[f(x_{\isnew})=a\mid \pp]\le 2^{-h}$ for some $h(\lambda)$
	(e.g., when $\pp$ is statistically independent of $f(x_{\isnew})$), then the unpredictability hypothesis holds with
	$\varepsilon(\lambda)=2^{-h}$, recovering the corresponding information-theoretic bound.
\end{remark}

Note that Proposition~\ref{prop:nxk-resharing-impossible-updated} is a statement about external derivation/export:
	it rules out computing a clear value $\widehat{k}_{\isnew}$ intended to equal $f(x_{\isnew})$
	using only $(\pp,\tau_{\mathrm{ext}})$, where $\tau_{\mathrm{ext}}$ is the view outside all KeyBox instances.
	It does not rule out enrollment mechanisms in which share-derived material is transferred only via attested
	KeyBox-to-KeyBox sealing and is decrypted/consumed inside a KeyBox. Dynamic joins that assign a fresh independent Shamir share typically require some party outside KeyBoxes to
compute, reveal, or otherwise export share-derived values.
Proposition~\ref{prop:nxk-resharing-impossible-updated} captures this obstruction under key-opacity: no PPT strategy can
externally compute such a fresh share from $(\pp,\tau_{\mathrm{ext}})$. Therefore, join protocols that require
exporting share-derived values are incompatible with NXK. Therefore, in our SDKG protocol, we instead realize post-DKG enrollment via RDR, where
additional devices are enrolled as redundant front-ends for an existing role/share using $\SealToPeer/\OpenFromPeer$,
avoiding exported share-derived plaintexts and preserving the public key.

\section{Cryptographic Primitives for State-Continuous KeyBoxes}\label{prelim}
Our end-to-end security claim is a UC realization theorem for SDKG. To keep the exposition modular, we state required properties of underlying proof/certificate mechanisms using standard game-based definitions in the \gROCRP\ global-setup model, and we instantiate them in disjoint, domain-separated contexts so these guarantees apply within the surrounding UC execution. We analyze our protocols in the UC framework augmented with a global resource shared across all sessions: a random-oracle-like functionality $H$ that is sampled once per UC execution and used by all ITMs. This is in the spirit of UC formulations with a global random oracle (e.g., \cite{Global1,Global2,LysyanskayaRosenbloom22,DoernerKR24,Manu[18]}). Formally, $H$ is joint state that persists across sessions and sub-protocols (cf. \cite{Global3}).

\begin{figure}[t]
	\centering
	\setlength{\fboxrule}{0.2pt}
	\fbox{
		\parbox{\dimexpr\linewidth-2\fboxsep-2\fboxrule\relax}{%
\ding{169}\ \textsf{Global setup:} A nonempty context set $\mathsf{Ctx}\subseteq\{0,1\}^*$ and a finite range $\mathcal Y := \{0,1\}^{\lambda}$. Fix a partition $\mathsf{Ctx}=\CtxTEE\ \dot\cup\ \CtxUC$ into non-programmable contexts $\CtxTEE$ and restricted-programmable contexts $\CtxUC$.\\
			\ding{169}\ \textsf{State:} a lazy table $\mathsf T:\mathsf{Ctx}\times\{0,1\}^*\rightharpoonup \mathcal Y$ (init empty).\\[0.5mm]
\ding{169}\ \textsf{Semantics (direct access):} All interfaces below are delivered locally to the invoking ITM
(i.e., not as network messages). If the invoking ITM is adversary-controlled, then $\cA$ is notified of the full query/answer transcript:
for each invocation it learns $(\mathsf{ctx},x,y)$ where $y$ is the returned value.\\[0.5mm]
				\ding{169}\ \textsf{Interfaces:}
			\begin{itemize}[leftmargin=*,nosep]
				\item $\Query(\mathsf{ctx},x)$: if $(\mathsf{ctx},x)\notin\dom(\mathsf T)$, sample $y\leftarrowdollar \mathcal Y$ and set $\mathsf T[\mathsf{ctx},x]\gets y$.
				Return $\mathsf T[\mathsf{ctx},x]$.
				\item $\SimProgramRO(\mathsf{ctx},x,y)$ \textsf{(simulator-only)}:
				If $\mathsf{ctx}\notin\CtxUC$ or $(\mathsf{ctx},x)\in\dom(\mathsf T)$, return $\perp$.
				Set $\mathsf T[\mathsf{ctx},x]\gets y$ and return $\ok$.
			\end{itemize}
	}}
	\caption{Global setup functionality $\GgROCRP$ implementing \gROCRP.}
	\label{fig:GgROCRP}
\end{figure}

\begin{remark}[Oracle-tape convention]\label{rem:oracle-tape}
	When we say that an extractor or simulator inspects $\mathsf{Log}_{\mathcal P^\ast}$, we mean that it runs the ITM
	$\mathcal P^\ast$ with explicit oracle access to the global functionality $\GgROCRP$ and records the transcript of its
	local calls to $\Query(\mathsf{ctx},x)$ together with the corresponding replies.
	When a set of contexts is relevant, $\mathsf{Log}_{\mathcal P^\ast}$ is understood to be restricted to those contexts.
	
	\noindent\emph{UC usage.}
	In our UC proofs, whenever straight-line extraction is applied to a proof produced by an adversary-controlled
	host prover (i.e., outside any KeyBox boundary), the simulator obtains the required oracle transcript by
	recording the $\Query$ calls made by that adversary-controlled ITM (cf.\ Fig.~\ref{fig:GgROCRP}).
	We never assume access to $\Query$ traces issued inside an honest KeyBox instance; such KeyBox-internal
	oracle calls are not exposed at the host/API boundary under local-call semantics.
\end{remark}

\begin{definition}[\gROCRP]\label{def:gROCRP}\label{def:GNPRO}
	\emph{The \gROCRP-hybrid model is the UC model augmented with the single global setup functionality $\GgROCRP$ (Fig.~\ref{fig:GgROCRP}).
		It implements an oracle $H:\mathsf{Ctx}\times\{0,1\}^*\to\{0,1\}^{\lambda}$ with local-call semantics. The context set $\mathsf{Ctx}$ is partitioned as $\mathsf{Ctx}=\CtxTEE\ \dot\cup\ \CtxUC$.
		All ITMs may invoke $\Query(\mathsf{ctx}, x)$. In programmable contexts $\mathsf{ctx} \in \CtxUC$, the simulator additionally has access to $\SimProgramRO(\mathsf{ctx}, x, y)$ as specified in Fig. \ref{fig:GgROCRP}; no other ITM can invoke $\SimProgramRO$.}
\end{definition}

Thus, unlike existing approaches (e.g., \cite{CGGMP21}) wherein the UC proof is carried out in a strict GRO setting, we work in
\gROCRP: a single global oracle resource with local-call semantics and explicit
context-based domain separation. Our \gROCRP\ model can be characterized as a context-partitioned instance of the restricted-programmable GRO variants in the taxonomy of Camenisch et al. \cite{Manu[18]}: it coincides with strict GRO on the non-programmable contexts $\CtxTEE$ and adds only fresh-point, non-overwriting programmability for the simulator on the designated proof contexts $\CtxUC$. On non-programmable contexts $\CtxTEE$ (used for transcript digests/receipts), \gROCRP\ coincides with the standard non-programmable GRO. For proof contexts $\CtxUC$, \gROCRP\ additionally exposes a simulator-only programming hook $\SimProgramRO$ to realize the
universal simulation interface required by our UC-NIZK(-AoK)s (Definition~\ref{def:NIZK-AoK}) instantiated via the
optimized Fischlin transform (Definition~\ref{FSdef}). On proof contexts $\CtxUC$, \gROCRP\ exposes only a \emph{simulator-only} programming hook $\SimProgramRO$.
No protocol party (and hence no adversary) can program oracle outputs. Thus, in $\CtxUC$ contexts, the simulator can fail only if an external ITM \emph{pre-queries} an input that the simulator intends to program.
Lemma \ref{lem:grocrp-prequery} bounds this pre-query event and implies that, for our uses (Fischlin-based UC-NIZKs in contexts in $\CtxUC)$,
the simulator programs only fresh points except with negligible probability. 

For any PPT machine making $\poly(\lambda)~\Query$ calls, each fresh $(\mathsf{ctx},x)$ returns an independent uniform $y \gets \mathcal{Y}$. Moreover, for any $(\mathsf{ctx}^*,x^*)$ that was neither queried nor simulator-programmed, $H(\mathsf{ctx}^*,x^*)$ is uniform conditioned on the machine's view. This aligns with the standard
programmable-random-oracle abstraction used to express the universal simulation interface for NIZKs: the simulator may set oracle values at a (negligible) set of fresh inputs associated with simulated
proofs, while all non-proof uses of hashing (contexts in $\CtxTEE$) remain strictly non-programmable. The programming hook is a simulator interface in the idealized model; it has no concrete analogue for a fixed hash; it is included solely to realize the universal simulation interface for RO-based NIZKs.
Accordingly, instantiating $\Query$ by a fixed domain-separated hash function should be read as the usual (global) RO heuristic for programmable-RO-based UC-NIZKs, with instantiation caveats as studied in the GRO literature (e.g., \cite{Manu[18]}) and in recent work on limitations for distributing RO-based proofs (e.g., \cite{DoernerKR24}).

Because $\GgROCRP$ has local-call semantics, the simulator can obtain an oracle query/answer log $\mathsf{Log}_{\mathcal P^\ast}$ only for provers that are themselves adversary-controlled ITMs. Specifically, KeyBox ITM is never adversary-controlled (even when its owner party is corrupted), so any $\Query(\cdot,\cdot)$ calls issued inside a KeyBox are transcript-private and are unavailable to the UC simulator. Therefore, whenever our UC proofs rely on straight-line extraction from $\mathsf{Log}_{\mathcal P^\ast}$ (e.g., for
Fischlin-based UC-context AoKs in contexts in $\CtxUC$), the corresponding prover must be the host/party ITM outside any KeyBox boundary. Equivalently, the admissible KeyBox API profile must not expose any operation that outputs those UC-context proofs (or any artifact that would verify as such a proof) in a way that keeps the relevant oracle queries inside the KeyBox. KeyBox-resident proof-generation primitives, if present, must be instantiated in disjoint contexts and are used only under ZK/simulation (no extraction).

\paragraph{Relation to standard models.}
It is useful to compare three increasingly strong global-oracle abstractions:
\begin{itemize}[leftmargin=*]
	\item Strict GRO: only $\Query(\mathsf{ctx},x)$ is available in all contexts. In this model the universal
	simulation interface required by our Fischlin-based UC-NIZKs is unattainable (Proposition~\ref{prop:strict-gro-insufficient}).
	\item \gROCRP\: strict GRO on $\CtxTEE$ plus a simulator-only, fresh-point programming hook on $\CtxUC$.
	Programming is \emph{non-overwriting} (fails on pre-queried points) and is used only to realize the UC-NIZK
	simulation interface in proof contexts.
	\item Fully programmable global RO: a simulator can program arbitrary points (potentially even overwrite)
	in all contexts. We do not assume this: \gROCRP\ disallows programming outside $\CtxUC$ and disallows
	overwriting anywhere.
\end{itemize}
Thus, \gROCRP\ is strictly stronger than strict GRO but strictly weaker than a fully programmable GRO. For a broader study of GRO formulations (strict, programmable, restricted programmable/observable), see \cite{Manu[18]}.

\begin{lemma}[Pre-query bound for \gROCRP\ programming]\label{lem:grocrp-prequery}
	Fix any context $\mathsf{ctx}\in\CtxUC$. Consider a UC execution (possibly involving many concurrently interleaved protocol sessions) in which the ideal-world simulator makes at most
	$m=m(\lambda)$ calls to $\SimProgramRO(\mathsf{ctx},x_j,y_j)$ at (possibly adaptive) inputs
	$x_1,\ldots,x_m\in\{0,1\}^*$. Let $\Bad_{\mathsf{pre}}$ denote the event that for some $j$,
	the $j$-th call to $\SimProgramRO(\mathsf{ctx},x_j,y_j)$ returns $\perp$
	(equivalently, $(\mathsf{ctx},x_j)\in\dom(\mathsf T)$ at the time of that call). Suppose that immediately before each $x_j$ is fixed, conditioned on the complete external view
	$\mathrm{view}_j$ of all non-simulator ITMs, the point $x_j$ has conditional min-entropy at least $h_j(\lambda)$ in the sense that 
	\[
	\max_{u\in\{0,1\}^*}\Pr[x_j=u\mid \mathrm{view}_j] \ \le\ 2^{-h_j(\lambda)}.
	\]
	If the total number of $\Query(\mathsf{ctx},\cdot)$ calls made by all non-simulator ITMs before the $j$-th programming attempt
	is at most $Q_j(\lambda)$, then
	\[
	\Pr[\Bad_{\mathsf{pre}}] \le\ \sum_{j=1}^m Q_j(\lambda)\cdot 2^{-h_j(\lambda)}.
	\]
\end{lemma}

\begin{proofsketch}
	For a fixed $j$, let $S_j$ be the set of inputs queried via $\Query(\mathsf{ctx},\cdot)$ by non-simulator ITMs prior to the $j$-th call to $\SimProgramRO$.
	By assumption $|S_j|\le Q_j$. Conditioned on $\mathrm{view}_j$, we have
	\[
	\Pr[x_j\in S_j \mid \mathrm{view}_j]
	\le \sum_{u\in S_j}\Pr[x_j=u\mid \mathrm{view}_j]
	\le |S_j|\cdot 2^{-h_j}
	\le Q_j\cdot 2^{-h_j}.
	\]
	A union bound over $j\in[m]$ yields the claim. \qed
\end{proofsketch}

	Hence, if $m,Q_j=\poly(\lambda)$ and
	$h_j(\lambda)=\omega(\log\lambda)$ for all $j$, then $\Pr[\Bad_{\mathsf{pre}}]=\negl(\lambda)$.

\begin{lemma}[No cross-context influence in \gROCRP]
	\label{lem:grocrp-noninterference}
	Fix any execution in the \gROCRP-hybrid model. For every non-programmable context $\mathsf{ctx}\in\CtxTEE$, the joint
	distribution of all replies to calls $\Query(\mathsf{ctx},\cdot)$ is identical to that of a standard (non-programmable)
	GRO for that context, even conditioned on an arbitrary sequence of simulator calls
	$\SimProgramRO(\mathsf{ctx}',\cdot,\cdot)$ in contexts $\mathsf{ctx}'\in\CtxUC$.
\end{lemma}

\begin{proofsketch}
	In Fig.~\ref{fig:GgROCRP}, the oracle table $\mathsf T$ is indexed by pairs $(\mathsf{ctx},x)$.
	A call to $\SimProgramRO(\mathsf{ctx}',x',y')$ can write only the entry $\mathsf T[\mathsf{ctx}',x']$ and only when
	$\mathsf{ctx}'\in\CtxUC$. Therefore no entry with $\mathsf{ctx}\in\CtxTEE$ is ever written by $\SimProgramRO$; such entries
	are populated only by lazy sampling upon the first corresponding $\Query(\mathsf{ctx},x)$ call.
	Thus, $H(\mathsf{ctx},\cdot)$ for $\mathsf{ctx}\in\CtxTEE$ is distributed exactly as in a standard global random oracle and
	is unaffected by programming in other contexts. \qed
\end{proofsketch}

\begin{remark}[Per-context domain separation in \gROCRP]\label{rem:grocrp-domain-sep}
	In $\GgROCRP$ (Fig.~\ref{fig:GgROCRP}) the oracle table $\mathsf T$ is indexed by pairs
	$(\mathsf{ctx},x)$. Consequently, oracle activity in one context cannot affect any other
	context: a call $\Query(\mathsf{ctx}',\cdot)$ (resp.\ $\SimProgramRO(\mathsf{ctx}',\cdot,\cdot)$) reads/writes only
	entries of the form $\mathsf T[\mathsf{ctx}',\cdot]$ and never touches $\mathsf T[\mathsf{ctx},\cdot]$ for
	$\mathsf{ctx}\neq \mathsf{ctx}'$. In particular:
	\begin{itemize}[leftmargin=*,nosep]
		\item for $\mathsf{ctx}\in\CtxTEE$, the induced oracle $H(\mathsf{ctx},\cdot)$ is a strict (non-programmable) random
		oracle for that context even conditioned on arbitrary simulator programming in other contexts
		(cf.\ Lemma~\ref{lem:grocrp-noninterference});
		\item for $\mathsf{ctx}\in\CtxUC$, the induced oracle $H(\mathsf{ctx},\cdot)$ is a restricted-programmable random
		oracle for that context, where programming is non-overwriting and confined to that same $\mathsf{ctx}$.
	\end{itemize}
	Thus, when logically distinct uses of hashing are assigned disjoint contexts (and inputs are injectively encoded),
	they behave as domain-separated uses of independent per-context oracles. This is the precise sense in which we
	apply ROM/\gROCRP security arguments modularly inside an arbitrarily interleaved UC execution.
\end{remark}

\begin{remark}[Concrete entropy sources for \gROCRP\ programming in this paper]
	\label{rem:grocrp-entropy-sources}
	Lemma~\ref{lem:grocrp-prequery} reduces simulator programming failure to the conditional min-entropy of the
	programmed inputs $x_j$. In all of our uses of Fischlin-based UC-NIZK simulation (DL and DLEQ), every
	$\SimProgramRO$ call is on an input that includes at least one fresh simulator-sampled prover response
	$z_i\in\Zp$ (or the analogous response component for the underlying $\mathrm{\Sigma}$-protocol (Definition \ref{def:Sigma})), embedded in the
	injective tuple encoding $\langle\cdot\rangle$ as in Definition~\ref{FSdef}. Table \ref{tab:grocrp-entropy} summarizes the concrete entropy sources for the relevant inputs. Even under concurrency and adaptive corruptions, the ``external view'' $\mathrm{view}_j$ in
	Lemma~\ref{lem:grocrp-prequery} may include (a) the entire public transcript so far, (b) all corruption-revealed
	host state, subject to the protocol's explicit erasures, and (c) all prior \gROCRP\ replies to non-simulator ITMs;
	nevertheless, the simulator chooses each $z_i$ after $\mathrm{view}_j$ is fixed, so $z_i$ remains uniform
	conditioned on $\mathrm{view}_j$. Since $\log p=\Theta(\lambda)$, this yields $h_j\ge \log p - O(1)=\Theta(\lambda)$.
	\begin{table}[t]
		\centering
		\small
		\setlength{\tabcolsep}{4pt}
		\renewcommand{\arraystretch}{1.15}
		\begin{tabular}{p{0.3\linewidth}p{0.15\linewidth}p{0.5\linewidth}}
			\hline
			\textbf{Programming event} & \textbf{Proof context} & \textbf{Fresh high-entropy component in programmed input} \\
			\hline
			Simulating UC-context DL proofs &
			$\mathsf{ctx}_{\mathsf{UC}}\in\CtxUC$ &
			Programmed inputs include a fresh response $z_i\in\Zp$ inside
			$\langle x,\mathbf a,i,e_i,z_i\rangle$ (Definition~\ref{FSdef}); conditioned on the full external view,
			$z_i$ is uniform, so $H_\infty(x)\ge \log p - O(1)$. \\[0.4em]
			
			Simulating UC-context DLEQ proofs &
			$\mathsf{ctx}_{\mathsf{DLEQ}}\in\CtxUC$ &
			Programmed inputs contain a fresh simulator-chosen response component $z_i\in\Zp$ inside the
			encoded Fischlin query input, giving $H_\infty(x)\ge \log p - O(1)$. \\[0.4em]
			
			Simulating KeyBox-context DL proofs &
			$\mathsf{ctx}_{\mathsf{KeyBox}}\in\CtxUC$ &
Each programmed point includes a fresh $z_i\in\Zp$ chosen by the simulator and embedded in the programmed Fischlin input.
Under LinOS (Fig.~\ref{LinOS}), this input is of the form
$\langle \sid, K,\mathbf a,i,e_i,z_i\rangle$, so conditioned on the external view, $z_i$ is uniform and contributes
$H_\infty=\Theta(\lambda)$ to that input. \\
			\hline
		\end{tabular}
		\caption{Concrete entropy sources: in all cases, a fresh $z_i\in\Zp$ contributes $\Theta(\lambda)$
			conditional min-entropy.}
		\label{tab:grocrp-entropy}
	\end{table}
	
	Plugging $h_j=\Theta(\lambda)$ and $Q_j=\poly(\lambda)$ into Lemma~\ref{lem:grocrp-prequery} yields
\[
\Pr[\Bad_{\mathsf{pre}}] \ \le\ \sum_{j=1}^{m} Q_j(\lambda)\cdot 2^{-\Theta(\lambda)} \ =\ \negl(\lambda),
\]
where $m=m(\lambda)$ is the total number of simulator programming attempts $\SimProgramRO(\mathsf{ctx},\cdot,\cdot)$
in the given proof context $\mathsf{ctx}\in\CtxUC$ over the entire UC execution (i.e., summed over all concurrently
interleaved sessions). Since the simulator and all non-simulator ITMs are PPT, we have $m(\lambda)=\poly(\lambda)$ and
$Q_j(\lambda)=\poly(\lambda)$ even under $\poly(\lambda)$ concurrent sessions; hence a union bound over all programmed points
(across all sessions) remains negligible.
\end{remark}

In the UC simulations for USV and SDKG, the simulator never invokes
$\SimProgramRO(\mathsf{ctx}_{\mathsf{DLEQ}},\cdot,\cdot)$.
Intuitively, USV certificates are generated honestly (with witnesses) and are only verified/opened
by the simulator in those UC hybrids; no UC hybrid ever requires simulating a new accepting
USV/DLEQ proof without a witness. 

UC-NIZKs in the adaptive-corruption setting have been studied since Groth et al. \cite{UCNIZK1,UCNIZK2}. For practical adaptive UC-NIZK-PoK constructions in (G)RO via straight-line compilation of $\mathrm{\Sigma}$-protocols, see \cite{Anna[25]}. Our use requires the stronger AoK interface with straight-line extraction in \gROCRP\, defined as:

\begin{definition}[NIZK proofs/arguments and AoKs]\label{def:NIZK-AoK}
	\emph{Let $\cal R\subseteq \mathcal{X}\times \mathcal{W}$ be an NP relation and let
		$\mathcal{L}_\mathcal{R}:=\{x\in\mathcal{X}:\exists w\in\mathcal{W}\ \text{s.t.}\ (x,w)\in \cal R\}$.
		A Non-Interactive Zero-Knowledge (NIZK) argument system for $\cR$ is a triple of PPT algorithms $(\cal K,P,V)$: 
		\[
		\text{Setup: } \mathcal{K}(1^\lambda)\to \pp, \quad \text{Prove: } \mathcal{P}(\pp,x,w)\to \pi \text{ for } (x,w)\in \cR, \quad \text{Verify: } \mathcal{V}(\pp,x,\pi)\to\{0,1\},	
		\]
		satisfying:
		\begin{itemize}
			\item Completeness: for all $(x,w)\in \cal R$,
			$\Pr[\mathcal{V}(\mathsf{pp},x,\pi)=1\mid \mathsf{pp}\leftarrow \mathcal{K}(1^\lambda),\ \pi\leftarrow \mathcal{P}(\pp,x,w)] \ge 1 - \negl(\lambda)$.
			\item (Computational) soundness: for all PPT $\mathcal{P}^\ast$ and all $x\notin \cal L_R$,
			$$\Pr[\mathcal{V}(\mathsf{pp},x,\pi)=1\mid \mathsf{pp}\leftarrow \mathcal{K}(1^\lambda),\ \pi\leftarrow \mathcal{P}^\ast(\mathsf{pp},x)]\le \negl(\lambda).$$
			\item Zero-knowledge with universal simulation interface (in the \gROCRP\ model):
			There exists a PPT simulator $\s=(\s_1,\s_2)$ and (possibly empty) simulator-only state $\tau$ such that:
			\begin{enumerate}[label=(\roman*),nosep]
				\item $\{\pp\leftarrow \mathcal{K}(1^\lambda)\}\approx_c \{\pp : (\pp,\tau)\leftarrow \s_1(1^\lambda)\}$.
				\item (Indistinguishability on true statements) For all $(x,w)\in \cal R$,
				\[
				\{(\pp,\pi): \pp\leftarrow \mathcal{K}(1^\lambda),\ \pi\leftarrow \mathcal{P}(\pp,x,w)\}\ \approx_c\
				\{(\pp,\pi): (\pp,\tau)\leftarrow \s_1(1^\lambda),\ \pi\leftarrow \s_2(\pp,\tau,x)\}.
				\]
				\item (Universal simulation interface) For every statement $x\in\mathcal X$, letting $\pi\leftarrow \s_2(\pp,\tau,x)$ we have
				\[
				\Pr\big[\,\pi\neq\perp\ \wedge\ \mathcal V(\pp,x,\pi)=1\,\big]\ \ge\ 1-\negl(\lambda).
				\]
			\end{enumerate}
			The simulator $\s_2$ may invoke the simulator-only interface $\SimProgramRO$ in the proof's
			\gROCRP\ context(s), which must lie in $\CtxUC$. If a required call to $\SimProgramRO$ returns $\perp$, then $\s_2$ outputs $\perp$\footnote{Throughout this paper we use this universal notion. Thus, the real setup is transparent and no protocol party ever learns $\tau$ or can program $H$.}.
		\end{itemize}
Additionally, $(\mathcal K,\mathcal P,\mathcal V)$ is a NIZK-AoK for $\mathcal R$ if
for every PPT prover $\mathcal P^\ast$ (with oracle access to $H$) there exists a PPT
straight-line extractor $\Ext_{\mathcal P^\ast}$ such that, in the experiment
\[
\pp \leftarrow \mathcal K(1^\lambda);\quad (x,\pi)\leftarrow \mathcal P^{\ast\,H}(\pp);\quad
w\leftarrow \Ext_{\mathcal P^\ast}(\pp,x,\pi;\mathsf{Log}_{\mathcal P^\ast}),
\]
we have
\[
\Pr\big[\mathcal V(\pp,x,\pi)=1 \wedge (x,w)\notin\mathcal R\big]\le \negl(\lambda).
\]
}
\end{definition}

\begin{definition}[Simulation soundness]
	\label{def:ss-lang}
	\emph{Let $\mathrm{\Pi}=(\cal K,P,V)$ be a NIZK for relation $\cR$ with simulator $\s=(\s_1,\s_2)$.
		We say $\mathrm{\Pi}$ is \emph{simulation-sound} (for $\cR$) if for every PPT adversary $\cA$, the following
		experiment has negligible success probability: sample $(\pp,\tau)\leftarrow \s_1(1^\lambda)$, give $\pp$ to $\cA$,
		and grant $\cA$ oracle access to $\s_2(\pp,\tau,\cdot)$ producing simulated proofs for adaptively chosen statements.
		Let $Q$ be the set of statements queried by $\cA$ to this oracle. $\cA$ outputs $(x^*,\pi^*)$
		and wins only if
		(i) $\cV(\pp,x^*,\pi^*)=1$,
		(ii) $x^*\notin Q$, and
		(iii) $(x^*,w)\notin \cR$ for all $w$.}
\end{definition}

\begin{definition}[Simulation-extractability / simulation-sound AoK]
	\label{def:simext}
	\emph{Let $\mathrm{\Pi}=(\cal K,P,V)$ be a NIZK for relation $\cR$ in the \gROCRP\ model with simulator
		$\s=(\s_1,\s_2)$. We say $\mathrm{\Pi}$ is \emph{simulation-extractable} (a.k.a.\ \emph{simulation-sound AoK})
		if there exists a PPT extractor $\mathsf{Ext}$ such that for every PPT adversary $\cA$, the following holds
		with all but negligible probability:
		$(\pp,\tau)\leftarrow \s_1(1^\lambda)$. Run $\cA$ on input $\pp$ with oracle access to
		(i) the proof-simulation oracle $\s_2(\pp,\tau,\cdot)$ and
		(ii) the global oracle interface $\Query(\cdot,\cdot)$, while recording the
		oracle query/answer transcript $\mathsf{Log}_{\cA}$ (restricted to the \gROCRP\ context(s) used by $\mathrm{\Pi}$,
		cf.\ Remark~\ref{rem:oracle-tape}). Let $Q$ be the set of statements queried by $\cA$ to $\s_2(\pp,\tau,\cdot)$.
		When $\cA$ outputs $(x^*,\pi^*)$ with $\cV(\pp,x^*,\pi^*)=1$ and $x^*\notin Q$,
		the extractor outputs
		\[
		w^*\leftarrow \mathsf{Ext}_{\cA}(\pp,\tau,x^*,\pi^*;\mathsf{Log}_{\cA})
		\]
		such that $(x^*,w^*)\in\cR$.}
\end{definition}

\begin{definition}[Public-coin protocol]
	\emph{\cite{Babai[88]} A proof system $\mathrm{\Pi} = (\mathcal{P}, \mathcal{V})$ is public-coin if each verifier message consists
		only of freshly sampled public randomness. Concretely, in round $i$, $\mathcal{V}$
		samples $c_i \leftarrowdollar C_i$ uniformly from some finite set $C_i$ and sends $c_i$.
		If $C_i=\{0,1\}^t$ for all verifier rounds, we say that $\mathrm \Pi$ has $t$-bit challenges.}
\end{definition}

Next, we recall the standard folklore definition of $\mathrm \Sigma$-protocol.

\begin{definition}[$\mathrm{\Sigma}$-protocol with $t$-bit challenge]\label{def:Sigma}
	\emph{Let $\mathcal{R}$ be an NP relation. A $\mathrm{\Sigma}$-protocol for $\mathcal R$ with $t$-bit challenges is
		a three-move protocol $\mathrm{\Sigma}=(\Psig,\Vsig)$: (i) $\Psig$ sends $a$; (ii) $\Vsig$ samples $e\leftarrowdollar \{0,1\}^t$ and sends $e$; (iii) $\Psig$ responds with $z$.
		The protocol satisfies the following properties:
		\begin{itemize}
			\item Completeness: for all $(x,w)\in \mathcal R$, an honest interaction accepts
			with probability 1.
			\item Special soundness: there exists a PPT extractor $\Ext$ such that
			from any two accepting transcripts $(a,e,z)$ and $(a,e',z')$ with $e \neq e'$
			(one and the same first message $a$), $\Ext$ outputs $w$ with $(x,w)\in \mathcal R$.
			\item (Computational) Special Honest-Verifier Zero-Knowledge: there exists a PPT simulator $\s$
			such that for all $(x,w)\in \mathcal R$ and all $e\in\{0,1\}^t$, the simulated transcript
			$\s(x,e)$ is computationally indistinguishable from the verifier's view
			in an honest execution with challenge fixed to $e$.
	\end{itemize}}
\end{definition}

We assume $2^{t(\lambda)}<p(\lambda)$ for every security parameter $\lambda$.
Hence each $t$-bit challenge $e\in\{0,1\}^t$ is interpreted as the corresponding integer
$\bar e\in\{0,\dots,2^t-1\}\subset\Zp$ via the natural injection.
All prover responses and verifier checks treat $\bar e$ as an element of $\Zp$. 

Throughout, the public parameters $\pp$ and the \gROCRP\ oracle $H$ (together with the relevant context string)
are fixed once per UC execution and are treated as implicit inputs to all algorithms.
When this improves readability, we omit $\pp$ from signatures and write
$\cp(x,w)$ for $\cp(\pp,x,w)$ and $\cV(x,\pi)$ for $\cV(\pp,x,\pi)$.
Likewise, for a $\mathrm{\Sigma}$-protocol we write $\Vsig(x;a,e,z)=1$ to denote acceptance of transcript $(a,e,z)$
for statement $x$ (with $\pp$ implicit in $x$).

The Fischlin transform \cite{Fischlin[05]} is a method to convert interactive public-coin proof systems into non-interactive ones. Unlike the Fiat--Shamir transform \cite{Fiat[86]}, which directly derives challenges from a random oracle, the Fischlin transform enforces an output structure criterion on the prover's transcript. Specifically, the prover must generate outputs that satisfy a rare structural condition (e.g., several trailing zeros), which is deliberately chosen to be rare, requiring the prover to perform multiple attempts to find a valid output. By observing the prover's queries to the random oracle, an extractor can identify at least two transcripts with the same initial commitment but different challenges. Then, \textit{special soundness} property of the underlying $\mathrm{\Sigma}$-protocol allows straight-line extraction of the witness. In \gROCRP, ``observing oracle queries'' means emulating the adversarial prover ITM and recording its local $\Query$ calls under the proof context (cf. Remark~\ref{rem:oracle-tape}). Because witness-bearing computation may be delegated to a state-continuous KeyBox, we require proof systems with straight-line extractors (no rewinding). Hence, we employ Fischlin-based UC-NIZKs. See \cite{Anna[25]} for an explicit formalization of adaptive straight-line compilation of $\mathrm{\Sigma}$-protocols and a proof that the randomized Fischlin transform satisfies it.

\begin{definition}[Fischlin transform in \gROCRP]\label{FSdef}
	\emph{
		Let $\mathrm{\Sigma}=(\Psig,\Vsig)$ be a three-move protocol for an NP relation $\mathcal R$
		with $t=O(\log \lambda)$-bit challenges. For security parameter $\lambda$, fix parameter functions
		$b=b(\lambda)$, $r=r(\lambda)$, and $S=S(\lambda)$ satisfying
		$b,r=O(\log\lambda)$, $b\le t$, $br=\omega(\log\lambda)$, $2^{t-b}=\omega(\log\lambda)$, and $S=\Theta(r)$.
Let $H:\mathsf{Ctx}\times\{0,1\}^*\rightarrow\{0,1\}^{\lambda}$ be the global oracle provided by $\GgROCRP$
in the \gROCRP\ model. Assume $b(\lambda)\le \lambda$. For $u\in\{0,1\}^*$, define
$H_b(\mathsf{ctx},u) \coloneqq \mathrm{lsb}_b\!\big(H(\mathsf{ctx},u)\big) \in\{0,\ldots,2^b-1\}$.
		When the Fischlin transform $F[\mathrm{\Sigma}]$ is used as a (UC-)NIZK (including when run inside an honest KeyBox),
		instantiate it with a proof context $\mathsf{ctx}\in\CtxUC$. Let $\langle\cdot\rangle$ be any fixed injective encoding of tuples into $\{0,1\}^*$. $F[\mathrm{\Sigma}]$ produces a non-interactive proof system $(\mathcal{P}', \mathcal{V}')$ as:\\[2mm]
		Prover $\mathcal P'(x,w)$:
		Run $r$ independent first moves of $\mathrm{\Sigma}$ on $(x,w)$ to obtain commitments $\mathbf a:=(a_1,\dots,a_r)$.
		For each $i\in[r]$ do:
		\begin{enumerate}
			\item Initialize $\widehat s_i \gets 2^b - 1$ and $(\widehat e_i,\widehat z_i)\gets(\perp,\perp)$.
			\item For $e=0,1,\ldots,2^t-1$ do:
			\begin{itemize}
				\item Compute the $\mathrm{\Sigma}$-response $z_{i,e}$ for challenge $e$ (using $w$), and set $s_{i,e} \coloneqq H_b\!\big(\mathsf{ctx},\langle x,\mathbf a, i, e, z_{i,e}\rangle\big).$
				\item If $s_{i,e}=0$, set $(e_i,z_i)\gets(e,z_{i,e})$ and break.
				\item Else, if $s_{i,e} \le \widehat s_i$, set $\widehat s_i\gets s_{i,e}$ and $(\widehat e_i,\widehat z_i)\gets(e,z_{i,e})$.
			\end{itemize}
			\item If the loop ended without $s_{i,e}=0$, set $(e_i,z_i)\gets(\widehat e_i,\widehat z_i)$.
		\end{enumerate}
		Let $s_i\coloneqq H_b\!\big(\mathsf{ctx},\langle x,\mathbf a, i, e_i, z_i\rangle\big)$.
		Output proof as $\pi \coloneqq \big((a_i,e_i,z_i)\big)_{i=1}^r.$\\[2mm]
		Verifier $\mathcal V'(x,\pi)$:
		Parse $\pi=((a_i,e_i,z_i))_{i=1}^r$, set $\mathbf a=(a_1,\dots,a_r)$, and accept iff:
		\begin{enumerate}[label=(\roman*)]
\item $\forall i\in[r]$, $\Vsig(x;a_i,e_i,z_i)=1$ and (ii) $\displaystyle \sum_{i=1}^r H_b\!\big(\mathsf{ctx}, \langle x,\mathbf a, i, e_i, z_i\rangle\big) \le S$.
		\end{enumerate}}
\end{definition}

\begin{remark}[Asymptotic vs.\ concrete Fischlin parameters]
	\label{rem:fischlin-asymptotic-vs-concrete}
	Definition~\ref{FSdef} fixes parameter functions $(t(\lambda),b(\lambda)$, $r(\lambda),S(\lambda))$ and imposes
	asymptotic growth conditions (e.g., $2^{t-b}=\omega(\log\lambda)$ and $br=\omega(\log\lambda)$) that are used only to
	derive negligible completeness and soundness/extraction error as $\lambda\to\infty$ (Lemma~\ref{lem:fischlin-negl}). When we later quote a concrete tuple $(t,b,r,S)$, this is shorthand for an instantiation at a fixed target security level $\lambda=\lambda_0$, i.e., $(t,b,r,S)=(t(\lambda_0),b(\lambda_0),r(\lambda_0),S(\lambda_0))$.
	In that concrete setting, we evaluate the corresponding explicit bounds from Fischlin’s analysis, rather than claiming that a constant tuple satisfies the asymptotic growth conditions for all $\lambda$.
\end{remark}

Throughout this paper, every UC-NIZK(-AoK) we use is instantiated via the optimized Fischlin transform \cite{Lindell[24]} in the \gROCRP\ model (Definition \ref{def:GNPRO}).

\paragraph{Early-break rarity search \cite{Lindell[24]}.}
For each $i$, the prover evaluates $H_b(\mathsf{ctx},\langle x,\mathbf a,i,e,z_{i,e}\rangle)$ on successive challenges and
stops at the first $e$ with $s_{i,e}=0$. Since $H_b$ is uniform over $\{0,\ldots,2^b-1\}$, the expected number of trials per $i$ is $2^b$. For a cap of $2^t$, the per-repetition ``no hit'' probability is
\[
p_{\mathrm{miss}} := \Pr[~\forall e\in\{0,1\}^t:\ s_{i,e}\neq 0] = (1-2^{-b})^{2^t}\approx e^{-2^{t-b}}.
\]
Under Definition~\ref{FSdef}, we require $2^{t-b}=\omega(\log\lambda)$; hence $p_{\mathrm{miss}}=\negl(\lambda)$.
With fixed concrete parameters, $p_{\mathrm{miss}}$ is an explicit completeness probability.
The honest-rejection probability of an $r$-fold proof is at most $r \cdot p_{\mathrm{miss}}$ by a union bound.
This event is an operational / liveness failure: an honest prover may abort and retry. We present the general slack-$S$ formulation in Definition~\ref{FSdef}. By ``optimized Fischlin'', we mean the optimized variant from \cite{Lindell[24]}, that fixes parameters and streamlines the prover/verification logic for efficiency.

We say that a $\mathrm{\Sigma}$-protocol $\mathrm{\Sigma}=(\Psig,\Vsig)$ has \emph{unique responses} if for every statement $x$,
first message $a$, and challenge $e$, there exists at most one response $z$ such that
$\Vsig(x;a,e,z)=1$. If $\mathrm \Sigma$ has special soundness and unique responses, then in the random-oracle/\gROCRP\ model
$F[\mathrm{\Sigma}]$ is a NIZK-AoK with a straight-line (online) extractor \cite[Thm. 2]{Fischlin[05]}. For any PPT adversary making at most $Q$ distinct queries to $H$ under
a context, the soundness / extraction-error is bounded by
\[
\Pr[\text{$\cV'(x,\pi)=1$ for $x\notin \mathcal{L}$}]
\le (Q+1)\cdot \frac{N_{r,S}}{2^{br}} + \negl(\lambda),
\]
where
\[
N_{r,S} := \sum_{T=0}^{S} \binom{T+r-1}{r-1} = \binom{S+r}{r}
\]
is the number of $r$-tuples $(s_1,\dots,s_r)\in\{0,\dots,2^b-1\}^r$ with
$\sum_i s_i \le S$.
In particular, for any $Q=\poly(\lambda)$, the soundness/extraction error is negligible whenever
\[
br - \log N_{r,S} = \omega(\log\lambda),
\]
e.g., under $S=\Theta(r)$ this is implied by $br=\omega(\log\lambda)$ since
$\log N_{r,S}=O(r)$.

\begin{lemma}[Negligible Fischlin error for admissible parameters]\label{lem:fischlin-negl}
	Let $(t(\lambda),b(\lambda),r(\lambda),S(\lambda))$ satisfy Definition~\ref{FSdef}.
	Then for every $Q=\poly(\lambda)$, the soundness/knowledge-extraction error of $F[\mathrm{\Sigma}]$
	against any PPT prover that makes at most $Q$ distinct queries under the transform context
	is negligible in $\lambda$. Moreover, the honest-rejection probability is $\negl(\lambda)$.
\end{lemma}

\begin{proofsketch}
	We know that the soundness error is at most
	$(Q+1)\cdot N_{r,S}/2^{br} + \negl(\lambda)$.
	Under Definition~\ref{FSdef}, $br-\log N_{r,S}=\omega(\log\lambda)$; hence the term is negligible for any $Q=\poly(\lambda)$.
	The honest-rejection bound $r\cdot(1-2^{-b})^{2^t}$ is negligible since $2^{t-b}=\omega(\log\lambda)$. \qed
\end{proofsketch}

For completeness, note that in an honest execution each repetition $i$ yields $s_i=0$ whenever the prover's rarity search finds some challenge $e$ with $H_b(\cdot)=0$. In that case the verifier's
sum test holds with $\sum_{i=1}^r s_i = 0 \le S$.
Thus, an explicit (union-bound) upper bound on honest-rejection is
\[
\Pr[\cV' \text{ rejects an honest proof}]
\le r\cdot (1-2^{-b})^{2^t} + \negl(\lambda)
\approx r\cdot \exp(-2^{t-b}) + \negl(\lambda).
\]
Soundness/knowledge-extraction/security error of the Fischlin transform, for a prover
making at most $Q$ distinct oracle queries, under the proof context is bounded by
\[
\varepsilon_{\mathsf{sec}}(Q) \;:=\; (Q+1)\cdot \frac{N_{r,S}}{2^{br}} \;+\; \negl(\lambda).
\]
Consequently, choosing parameter functions $t(\lambda),b(\lambda),r(\lambda),S(\lambda)$ as in
Definition~\ref{FSdef} yields negligible soundness and completeness error in $\lambda$. When quoting concrete numbers, we always label which bound is
liveness and which is security.

The next proposition is an immediate corollary of Fischlin soundness: in a strict GRO model the
UC simulator has no programming capability and is therefore just another PPT prover for $F[\mathrm{\Sigma}]$.
We nevertheless state it explicitly because our UC-NIZK(-AoK) definition requires a universal simulation interface
(Definition~\ref{def:NIZK-AoK}) that must succeed even on off-language statements, and strict GRO rules this out.
This motivates the move to \gROCRP\ (Definition~\ref{def:gROCRP}) and clarifies the model separation relative to works
proved in strict GRO (e.g.,~\cite{CGGMP21}).

\begin{proposition}[Strict GRO is insufficient for Fischlin universal simulation]
	\label{prop:strict-gro-insufficient}
	Fix an NP relation $\cR\subseteq\mathcal X\times\mathcal W$ such that
	$\mathcal X\setminus \mathcal L_\cR \neq \emptyset$, where
	$\mathcal L_\cR:=\{x\in\mathcal X:\exists w\ (x,w)\in\cR\}$.
	Let $\mathrm{\Sigma}=(\Psig,\Vsig)$ be a $\mathrm{\Sigma}$-protocol for $\cR$ with special soundness and unique responses.
	Let $F[\mathrm{\Sigma}]=(\mathcal K,\mathcal P',\mathcal V')$ denote the (optimized) Fischlin transform
	(Definition~\ref{FSdef}) instantiated with parameter functions
	$(t(\lambda),b(\lambda),r(\lambda),S(\lambda))$ satisfying Definition~\ref{FSdef}. Consider the UC model augmented with a strict global random oracle (GRO) that provides only
	$\Query(\cdot,\cdot)$ and no simulator-only programming interface. Assume the public parameters
	$\pp\leftarrow \mathcal K(1^\lambda)$ are generated by the honest/transparent setup of $F[\mathrm{\Sigma}]$.
	Then $F[\mathrm{\Sigma}]$ cannot satisfy the universal simulation interface of
	Definition~\ref{def:NIZK-AoK} by any PPT simulator. Concretely, for every PPT simulator $\s=(\s_1,\s_2)$ that has access only to $\Query$,
	for every fixed statement $x^\star\in \mathcal X\setminus \mathcal L_\cR$, if
	$(\pp,\tau)\leftarrow \s_1(1^\lambda)$ and $\pi^\star\leftarrow\s_2(\pp,\tau,x^\star)$, then
	\[
	\Pr\big[\mathcal V'(\pp,x^\star,\pi^\star)=1\big]\ \le\ \negl(\lambda),
	\]
	where the probability is over the coins of $\s$ and the strict GRO.
\end{proposition}

\begin{proofsketch}[Immediate from soundness (model-separation point)]
	Fix any $x^\star\in\mathcal X\setminus \mathcal L_\cR$.
	In the strict GRO model the simulator has access only to $\Query(\cdot,\cdot)$ and has no programming interface.
	Consequently, $\s_2(\pp,\tau,\cdot)$ is simply a PPT prover for the non-interactive proof system $F[\mathrm{\Sigma}]$,
	with oracle access to $\Query$ in the proof context used by $F[\mathrm{\Sigma}]$.
	
	Let $Q(\lambda)$ upper bound the number of distinct oracle queries that $\s_2$ makes under that proof context; since
	$\s_2$ is PPT, $Q(\lambda)=\poly(\lambda)$.
	By Fischlin soundness for admissible parameters (Lemma~\ref{lem:fischlin-negl}), any PPT prover making at most
	$Q(\lambda)$ distinct oracle queries outputs an accepting proof for an off-language statement with probability at most
	$\negl(\lambda)$. Applying this to $\s_2$ yields
	\[
	\Pr\big[\mathcal V'(\pp,x^\star,\pi^\star)=1\big]\le \negl(\lambda),
	\quad\text{where }\pi^\star\leftarrow\s_2(\pp,\tau,x^\star).
	\]
	
	However, the universal simulation interface in Definition~\ref{def:NIZK-AoK} requires that for every
	statement $x\in\mathcal X$ (specifically, $x^\star$), the simulator outputs an accepting proof with probability
	$1-\negl(\lambda)$.
	This contradiction shows that $F[\mathrm{\Sigma}]$ cannot satisfy the universal simulation interface in strict GRO. \qed
\end{proofsketch}

For DLEQ proofs we restrict the statement space to
\[
\mathcal X_{\mathsf{DLEQ}} \ :=\ (\mathbb G\setminus\{0_\mathbb G\})\times(\mathbb G\setminus\{0_\mathbb G\}).
\]
Henceforth, we refer to $\mathrm{\Pi}_{\mathsf{DL}}$ and $\mathrm{\Pi}_{\mathsf{DLEQ}}$ for the Fischlin transforms of the Schnorr and Chaum--Pedersen $\mathrm{\Sigma}$-protocols \cite{Schnorr[91],ChaumPedersen93Wallet} (for $\cR_{\mathsf{DL}}$ and $\cR_{\mathsf{DLEQ}}$ resp.), instantiated in disjoint \gROCRP\ contexts; concrete details appear in Section \ref{SSE-NIZK}.

\begin{lemma}[Fischlin-based UC-NIZK-AoKs for DL and DLEQ in \gROCRP]
	\label{lem:gnpro-uc-nizk}
	Assume DL hardness in the prime-order group fixed by the security parameter $\lambda$.
	Fix Fischlin parameters $(t(\lambda),b(\lambda),r(\lambda),S(\lambda))$ satisfying Definition~\ref{FSdef}.
	Let $\mathrm{\Pi}_{\mathsf{DL}}$ denote the (optimized) Fischlin instantiation described in Section~\ref{SSE-NIZK}
	for relation $\cR_{\mathsf{DL}}$, instantiated in a \gROCRP\ context in $\CtxUC$. If, in addition, DDLEQ hardness holds, let $\mathrm{\Pi}_{\mathsf{DLEQ}}$ denote the corresponding (optimized)	Fischlin instantiation for relation $\cR_{\mathsf{DLEQ}}$, instantiated in a disjoint \gROCRP\ context in $\CtxUC$. Then the following hold for $\mathrm{\Pi}_{\mathsf{DL}}$; and, under the additional DDLEQ assumption, the same items also hold for $\mathrm{\Pi}_{\mathsf{DLEQ}}$.
	
	\medskip
	\noindent (i) Zero-knowledge (universal simulation interface): 
	The proof system satisfies computational zero-knowledge with the universal simulation interface of
	Definition~\ref{def:NIZK-AoK} in its corresponding context in $\CtxUC$.
	
	\medskip
	\noindent (ii) NIZK-AoK with straight-line extraction: 
	The proof system is a NIZK-AoK in the sense of Definition~\ref{def:NIZK-AoK}, with an online/straight-line
	extractor that may inspect the adversary's \gROCRP\ query/answer log under the corresponding proof context.
	For any $Q=\poly(\lambda)$ distinct \gROCRP\ queries under that proof context, the resulting
	knowledge/soundness error is negligible (Lemma~\ref{lem:fischlin-negl}).
	
	\medskip
	\noindent (iii) Simulation-extractability:
	The proof system is simulation-extractable for fresh statements in the sense of Definition~\ref{def:simext}.
	
	\medskip
	\noindent (iv) Composability under domain separation in \gROCRP:
	The guarantees in (i)--(iii) continue to hold when the proof system is invoked as a subroutine inside an
	arbitrarily interleaved UC execution in the \gROCRP-hybrid model, provided that every use of the global oracle that serves a distinct protocol-level purpose is domain-separated by disjoint gRO-CRP contexts and injective encodings.
\end{lemma}

\begin{proofsketch}
	Items (i) and (ii) follow by instantiating Fischlin's analysis \cite[Thm.~2]{Fischlin[05]}
	for $\mathrm{\Sigma}$-protocols with special soundness and unique responses (here, Schnorr for $\cR_{\mathsf{DL}}$
	and Chaum--Pedersen for $\cR_{\mathsf{DLEQ}}$), with negligible error by Lemma~\ref{lem:fischlin-negl}.
	The optimized prover/verifier organization of \cite{Lindell[24]} affects efficiency but not the underlying
	transform security argument. Item (iii) follows from Fischlin's simulation-soundness / simulation-extractability theorem
	\cite[Thm.~3]{Fischlin[05]}. In our \gROCRP\ formulation, the simulator's universal simulation interface is realized via $\SimProgramRO$ in the corresponding contexts in $\CtxUC$.
	Since $\SimProgramRO$ fails on already-defined points, the simulator’s only failure mode is a pre-query collision on a programmed input;
	by Lemma~\ref{lem:grocrp-prequery}, this occurs with negligible probability for the Fischlin-based proofs whose programmed inputs include fresh high-entropy material. 
	
	For item (iv), fix any UC execution in the \gROCRP-hybrid model with arbitrarily many concurrently
	interleaved sessions and sub-protocols. By construction of $\GgROCRP$, the oracle
	table is indexed by pairs $(\mathsf{ctx},x)$, so oracle activity (including simulator programming) in any
	context $\mathsf{ctx}'\neq \mathsf{ctx}_\mathrm{\Pi}$ cannot affect the distribution of replies in
	$\mathsf{ctx}_\mathrm{\Pi}$ (cf.\ Lemma~\ref{lem:grocrp-noninterference} and Remark~\ref{rem:grocrp-domain-sep}).
	Injective encodings together with disjoint contexts ensure that logically distinct uses do not collide on
	the same $(\mathsf{ctx},x)$, so each invocation of $\mathrm{\Pi}$ is subject to the same per-context oracle behavior as in the
	standalone \gROCRP\ analysis.
	
	Under UC concurrency, the total number of oracle queries and simulator programming attempts in any fixed
	programmable context remains $\poly(\lambda)$. The simulator’s only failure mode for the universal
	simulation interface is a pre-query collision on a point it intends to program; Lemma~\ref{lem:grocrp-prequery},
	together with a union bound over all simulator programming attempts in that context across the entire UC
	execution, bounds this event by $\negl(\lambda)$. Conditioned on the complement, the proofs’ ZK, AoK, and
	simulation-extractability guarantees in (i)--(iii) apply unchanged inside the UC execution. \qed
\end{proofsketch}

\section{Enforcing Public Structure without Export: USV Certificates}\label{sec:USV}
Unique Structure Verification (USV) is a non-interactive, publicly verifiable certificate that lets anyone
derive a unique public opening for a commitment to a hidden scalar without exporting that scalar\footnote{In UC, we will model handle binding as verifier-scoped; see Section \ref{UC:usv}}.
Extraction is public and straight-line: it is a deterministic function of the certificate and uses no trapdoor or
rewinding. 

\begin{assumption}[Transparent generator derivation]\label{assump:transparent-pp}
	\emph{Let \(\pounds\) be a public randomness source (e.g., \cite{Beacon[19]}) that is sampled outside the protocol and
			is not adversary-influenceable: conditioned on the adversary's view prior to publication, \(\pounds\) has min-entropy
			at least \(\lambda\). Fix a prime-order group $\mathbb G$ of size $p$ with
	canonical generator $\mathcal G$.
	For a deterministic, publicly specified hash-to-group map, $\mathsf{H2G}$ \cite{RFC9380},
	whose output distribution (over random $\pounds$) is computationally indistinguishable
	from uniform over $\mathbb G\setminus\{0_\mathbb G\}$, define $\mathcal H:=\mathcal H_{c^\star}$ where
	$\mathcal H_c := \mathsf{H2G}(\texttt{USV.H}\parallel \pounds\parallel \mathsf{enc}(c))$,
	$\mathsf{enc}(c)$ is the 4-byte big-endian encoding of $c$, and $c^\star$ is the smallest $c\ge 0$
	such that $\mathcal H_c\notin\{0_\mathbb G,\mathcal G\}$. Set $\pp:=(\mathbb G,p,\mathcal G,\mathcal H)$, and define the deterministic public setup procedure
	\[
	\mathsf{Setup}(1^\lambda,\pounds)\to \pp
	\]}
\end{assumption}

Definitions \ref{def:USV-cert}--\ref{def:usv-eqv} describe the stand‑alone primitive; our composable statement is the UC realization of $\Fusv$ (Theorem \ref{thm:usv-real-ideal}).

\begin{definition}[USV certificate scheme]\label{def:USV-cert}
	\emph{A USV certificate scheme consists of the deterministic public setup procedure
		\[
		\mathsf{Setup}(1^\lambda,\pounds)\to \pp
		\]
		specified in Assumption~\ref{assump:transparent-pp}, together with the following algorithms
		\[
		\Cert(\pp,m)\to(C,\zeta),\qquad
		\Vcert(\pp,C,\zeta)\in\{0,1\},\qquad
		\Derive(\pp,C,\zeta)\to(\mathrm{\Upsilon}\ \text{or}\ \perp),
		\]
		and a deterministic \emph{public-opening projection}
		\[
		\mathsf{PubOpen}(\pp,\cdot):\ \mathcal O \to \mathbb G\cup\{\perp\},
		\]
		where $\mathcal O$ is the opening space of $\Derive$ (and we adopt the convention $\mathsf{PubOpen}(\pp,\perp)=\perp$). We define a \emph{verified opening} algorithm as
		\[
		\Open(\pp,C,\zeta):=
		\begin{cases}
			\Derive(\pp,C,\zeta) & \quad \text{if }\Vcert(\pp,C,\zeta)=1,\\
			\perp & \quad \text{otherwise}.
		\end{cases}
		\]
		We also define the associated \emph{public opening} as
		\[
		\Open_M(\pp,C,\zeta)\ :=\ \mathsf{PubOpen}\bigl(\pp,\Open(\pp,C,\zeta)\bigr).
		\]
	There exists a polynomial-time decidable relation $\mathcal R_{\pp}\subseteq \mathcal C\times\mathcal O$
		such that the following hold:
		\begin{enumerate}
			\item Completeness:
			For every $m \neq 0$,
			\begin{align*}
			\Pr\Big[
			\Vcert(\pp,C,\zeta)=1 \wedge
			\mathrm{\Upsilon}:=\Open(\pp,C,\zeta)\neq\perp \wedge
			(C,\mathrm{\Upsilon})\in\mathcal R_{\pp} \wedge
			\Open_M(\pp,&C,\zeta)\neq\perp 
			: \\ &(C,\zeta)\leftarrow \Cert(\pp,m)
			\Big] \ge 1-\negl(\lambda).
		\end{align*}
			\item Deterministic verified opening (uniqueness):
			For any fixed $(C,\zeta)$, $\Derive(\pp,C,\zeta)$ is deterministic and returns either $\perp$ or a single value
			$\mathrm{\Upsilon}\in\mathcal O$. Moreover, $\Open(\pp,C,\zeta)=\perp$ iff $\Vcert(\pp,C,\zeta)=0$, and whenever
			$\Open(\pp,C,\zeta)=\mathrm{\Upsilon}\neq\perp$, we have $(C,\mathrm{\Upsilon})\in\mathcal R_{\pp}$.
			Further, $\mathsf{PubOpen}(\pp,\mathrm{\Upsilon})$ is deterministic; hence $\Open_M(\pp,C,\zeta)$ is deterministic
			and returns either $\perp$ or a single $M\in\mathbb G$.
			\item Opening-conditional tag simulatability:
			There exists a PPT simulator $\s_{\mathsf{cert}}$ such that:
			\begin{enumerate}[label=(\alph*),nosep]
				\item Correctness of simulated tags:
				For every $(C,\mathrm{\Upsilon})\in\mathcal R_{\pp}$, if
				$\tilde\zeta\leftarrow\s_{\mathsf{cert}}(\pp,C,\mathrm{\Upsilon})$ then,
				except with probability $\negl(\lambda)$,
				\[
				\Vcert(\pp,C,\tilde\zeta)=1
				\quad\text{and}\quad
				\Derive(\pp,C,\tilde\zeta)=\mathrm{\Upsilon}.
				\]
				Equivalently, except with probability $\negl(\lambda)$,
				$\Open(\pp,C,\tilde\zeta)=\mathrm{\Upsilon}$, and therefore
				$\Open_M(\pp,C,\tilde\zeta)=\mathsf{PubOpen}(\pp,\mathrm{\Upsilon})$.
				\item Indistinguishability conditioned on the opening:
				For every $m$,
				\begin{align*}
					\Big\{(C,\mathrm{\Upsilon},\zeta):\ &(C,\zeta)\leftarrow\Cert(\pp,m);\ \mathrm{\Upsilon}\leftarrow\Open(\pp,C,\zeta)\Big\}
					\ \approx_c\ \\
					&\Big\{(C,\mathrm{\Upsilon},\tilde\zeta):\ (C,\zeta)\leftarrow\Cert(\pp,m);\ \mathrm{\Upsilon}\leftarrow\Open(\pp,C,\zeta);\
					\tilde\zeta\leftarrow\s_{\mathsf{cert}}(\pp,C,\mathrm{\Upsilon})\Big\}.
				\end{align*}
			\end{enumerate}
	\end{enumerate}}
\end{definition}

For the UC realization of $\Fusv$ and for SDKG, we rely on completeness and deterministic verified
openings (Items~1--2 of Definition~\ref{def:USV-cert}) together with equivocation resistance
(Definition~\ref{def:usv-eqv}). The stronger opening-conditional tag-simulatability interface
(Item~3) is not required by any UC hybrid in this paper; we include it because it is useful in
variants where the ideal world fixes an opening first and the simulator must later backfill a
compatible accepting tag, and because it holds for our instantiation
(Lemma~\ref{lem:usv-tag-sim}).

\begin{definition}[Equivocation experiment]\label{def:usv-eqv-exp}
	\emph{Work in the \gROCRP-hybrid model of Definition~\ref{def:gROCRP} with global oracle
		\(\GgROCRP\) (Fig.~\ref{fig:GgROCRP}). The experiment and the adversary share access
		to the same oracle interface \(\Query(\cdot,\cdot)\). Fix public parameters
		\(\pp=(\mathbb{G},p,\mathcal G,\mathcal H)\) and the USV algorithms
		\((\Cert,\Derive,\Vcert,\Open,\mathsf{PubOpen})\) from Definition~\ref{def:USV-cert}. Experiment \(\mathsf{Exp}^{\mathrm{eqv}}_{\mathcal A}(1^\lambda)\) is defined as:
	\begin{enumerate}
		\item Sample \(\pounds\) according to the public randomness source/beacon distribution
		(Assumption~\ref{assump:transparent-pp}), and set \(\pp \gets \mathsf{Setup}(1^\lambda,\pounds)\).
		\item Run \(\mathcal A^{\Query}(1^\lambda,\pp)\), i.e., \(\mathcal A\) on input \((1^\lambda,\pp)\) with oracle access
		to \(\Query(\cdot,\cdot)\), and obtain \((C,\zeta,\zeta')\) with \(\zeta\neq \zeta'\).
		\item Let \(b\leftarrow 1\) iff \(\Vcert(\pp,C,\zeta)=1\ \wedge\ \Vcert(\pp,C,\zeta')=1\ \wedge\
		\Open_M(\pp,C,\zeta)\neq \Open_M(\pp,C,\zeta')\), where $\Vcert$ (and the embedded verifier $\cV_{\DLEQ}$) evaluates any needed oracle calls as
		$\Query(\mathsf{ctx}_{\mathsf{DLEQ}},\cdot)$ in the dedicated USV/DLEQ proof context
		$\mathsf{ctx}_{\mathsf{DLEQ}}\in\CtxUC$ (disjoint from $\mathsf{ctx}_{\mathsf{UC}}$ and
		$\mathsf{ctx}_{\mathsf{KeyBox}}$). Otherwise set \(b\leftarrow 0\). Output \(b\).
	\end{enumerate}
Define \(Adv^{\mathrm{eqv}}_{\mathcal A}(\lambda):=\Pr[\mathsf{Exp}^{\mathrm{eqv}}_{\mathcal A}(1^\lambda)=1]\).}
\end{definition}

\begin{definition}[Equivocation resistance]\label{def:usv-eqv}
	\emph{A USV certificate scheme is \emph{equivocation resistant} (in the \gROCRP-hybrid model)
		if for every PPT adversary \(\mathcal A\) with oracle access to \(\Query(\cdot,\cdot)\),
		\(Adv^{\mathrm{eqv}}_{\mathcal A}(\lambda)\le \negl(\lambda)\).}
\end{definition}

In the NXK/KeyBox setting targeted by this paper, USV certificate generation is mediated by the KeyBox API. Concretely, we include a key-independent admissible operation $\textsf{USV.Cert}\in\mathcal F_{\mathrm{adm}}$ that samples an internal witness scalar $m\leftarrow \Zp^*$, computes $\langle C,\zeta\rangle\leftarrow \Cert(\pp,m)$, erases $m$, and returns only $\langle C,\zeta\rangle$ to the host. Formally, a party obtains a certificate by invoking
\[
\langle C,\zeta\rangle \leftarrow \FKeyBox^{(P)}.\Use(\mu,\textsf{USV.Cert},\langle \mathsf{lbl}\rangle),
\]
where the slot argument $\mu$ is ignored (as for the key-independent $\OpenFromPeer$ interface in Fig.~\ref{Fdskg}).

\subsubsection{Relation to standard notions.}
USV can be viewed as a small ``commit-and-certify'' primitive whose interface differs from similar abstractions as:
\begin{itemize}[leftmargin=*,nosep]
	\item A standard commitment is opened by revealing the full witness (message and randomness).
	USV never reveals the scalar; instead it yields a deterministic public opening to the induced group element that is sufficient for the transcript-defined affine consistency checks used under NXK.
	
	\item Extractable commitments typically provide a privileged extractor that outputs the committed
	message (using trapdoors and/or rewinding). USV instead makes extraction public and deterministic,
	but extracts only the canonical group element (not the scalar), aligning with NXK.
	
	\item Equivocable commitments enable opening a fixed commitment to different messages.
	In contrast, USV is \textit{opening-unique} and explicitly requires equivocation resistance. The simulator’s power is orthogonal:
	it can simulate tags conditioned on a chosen opening, which is exactly what we need in the UC hybrids.
	
	\item PVSS-style objects certify well-formed encrypted shares or allow recovery by a set of parties.
	USV provides neither ciphertexts nor recoverable secrets; it certifies only public structure needed to
	replace exported-share enforcement under NXK.
	
	\item The tag $\zeta$ can be interpreted as a compact NIZK-style certificate of well-formedness, but with the special
	feature that it induces a canonical and deterministic public opening that the transcript can reference.
	
\item USV supports publishing commitment-shaped material $(C,\zeta)$ that deterministically defines the
canonical public point
\[
M \ :=\ \Open_M(\pp,C,\zeta)\ =\ m\G.
\]
The scalar \(m\) remains non-exportable and is hidden computationally (under DL): since
\(M\) is public and \(m\mapsto m\G\) is a bijection in a prime-order group, \(m\) is not information-theoretically
hidden once a valid certificate is published. In the SDKG base run (Section~\ref{SDKG}), the leaf transmits
$(C,\zeta)$ in Round~1, so $M$ is transcript-defined from the outset.
\end{itemize}

\subsection{Why USV is needed under hardened NXK profiles (overview)}\label{subsec:usv-need}
Several later checks (in SDKG verification and in the transcript-driven idealization) require certain group
elements to be deterministic functions of the public transcript in straight-line. The principal example is a
leaf-defined point $M=m\G$ (and derived auxiliaries): verifiers and the UC simulator must be able to compute
these points from the transcript alone.

In our NXK setting, long-term shares (and any other KeyBox-resident secrets) are API-non-exportable
(Assumption~\ref{assump:keybox-opacity}). Any share-deriving material is NXK-restricted:
it must remain transcript-private and must not be written to persistent storage outside a KeyBox, though it may be
handled transiently in host RAM during an atomic local step and must then be securely erased
(\textit{Reader Note}~\ref{box:export-visibility}; Remark~\ref{rem:transport-vs-nxk}). State continuity furthermore
rules out rewinding/forking-based extraction at the hardware boundary (Assumption~\ref{assump:tee-continuity}).
Finally, KeyBox-local oracle calls are not visible to the UC simulator (local-call semantics).

Under hardened/minimal profiles, the scalar $m$ underlying $M=m\G$ is generated inside a state-continuous KeyBox and
is not exported. A plain hiding commitment $C=\mathsf{Commit}(m;r)$ therefore does not determine $M$ unless one
can extract $m$ (or otherwise obtain $m\G$) in straight-line. This leaves three design options:

\begin{enumerate}[leftmargin=*,nosep]
	\item Publish $M=m\G$ directly:
	This removes the need for USV, but assumes the profile allows computing/exporting $m\G$ for fresh ephemeral scalars.
	
	\item Commit to $m$ and extract $m$ from an opening proof:
	If the opening proof is generated inside the KeyBox (to avoid exporting $m$), then straight-line extraction fails in our
	model: there is no rewinding/forking (state continuity) and no simulator access to the KeyBox’s oracle-log.
	If the opening proof is generated outside the KeyBox, then the leaf must materialize $(m,r)$ (or an equivalent
	caller-invertible image sufficient to derive $M$) in non-KeyBox state, which contradicts the hardened/minimal-profile
	design point.
	
	\item Publish commitment-shaped material with a publicly verifiable certificate that deterministically yields $M$:
	USV implements exactly this: from $(C,\zeta)$ anyone can compute the unique public opening
	$M=\Open_M(\pp,C,\zeta)$, while $m$ remains non-exportable.
\end{enumerate}

\paragraph{Least-privilege motivation (why Option 1 may be disallowed).}
This restriction can arise in deployments where the protocol principal is authorized to use a non-exportable
asymmetric key (e.g., sign/derive), but is explicitly denied the separate capability to retrieve its public key.\footnote{Concretely, several cloud KMS products gate ``get public key'' behind distinct permissions; e.g., AWS KMS \texttt{kms:GetPublicKey} \cite{aws-kms-getpublickey}, Google Cloud KMS \texttt{cloudkms.cryptoKeyVersions.viewPublicKey} \cite{gcp-kms-roles-perms}, and Azure Key Vault \texttt{get} reads the public part of a key \cite{azure-kv-keys-details}.}
This is an interface/profile assumption, not a claim that KeyBoxes in general cannot export public points.

For the formal necessity statement for the commit-only alternative, see Section~\ref{subsec:why-usv-straightline}
(Lemma~\ref{lem:commit-only-no-M2}).

\subsection{An Instantiation}\label{subsec:USV-inst}
Let $\mathcal R_{\pp} := \{(C,(M,R))\in \mathbb G^3 : C=M+R,\ M\neq 0_{\mathbb G},\ R\neq 0_{\mathbb G}\}$, and define the DLEQ relation
\[
\cR_{\mathsf{DLEQ}} :=
\left\{ \big((\pp,A,B),r\big)\ :\ (A,B)\in\mathcal X_{\mathsf{DLEQ}} \ \wedge\ r\in\mathbb Z_p^* \ \wedge\ A=r\G \ \wedge\ B=r\h \right\}.
\]
Let $\mathrm{\Pi}_{\mathsf{DLEQ}}$ be the UC-NIZK-AoK
(via optimized Fischlin in the \gROCRP\ model; see Section~\ref{SSE-NIZK}) for $\cR_{\mathsf{DLEQ}}$. We instantiate $\mathrm{\Pi}_{\mathsf{DLEQ}}$ in the dedicated context
$\mathsf{ctx}_{\mathsf{DLEQ}}\in\CtxUC$, which is disjoint from
$\mathsf{ctx}_{\mathsf{UC}}$ and $\mathsf{ctx}_{\mathsf{KeyBox}}$.

Concrete algorithms for an instantiation follow:

\begin{itemize}
	\item $\Cert(\pp,m)$: on input $m\leftarrowdollar  \mathbb Z_p^*$ sample $r\leftarrowdollar \mathbb Z_p^*\setminus\{-m\}$.
	Set $M:=m\G$, $R:=r\h$, $C:=M+R$, $\nu:=mr^{-1}\bmod p$, and $\upsilon:=(m+r)\G$.
	Define $A:=\upsilon-M$ and $B:=C-M$.
	Compute $\pi_\DLEQ \leftarrow \cp_{\mathsf{DLEQ}}(\pp,(A,B),r)$
	 and output $\zeta:=(\nu,\upsilon,\pi_\DLEQ)$ together with $C$. Erase $m, r$.
	
	\item $\Derive(\pp,C,\zeta)$: parse $\zeta=(\nu,\upsilon,\pi_\DLEQ)$.
	If $\nu\in\{-1,0\}\bmod p$, output $\perp$.
	Else set $M:=\frac{\nu}{\nu+1}\upsilon$ and $R:=C-M$ and output $\mathrm{\Upsilon}:=(M,R)$.
	
	\item $\PubOpen(\pp,\mathrm{\Upsilon})$: parse $\mathrm{\Upsilon}=(M,R)$ and output $M$. If parsing fails, output $\perp$.
	
	\item $\Vcert(\pp,C,\zeta)$: parse $\zeta=(\nu,\upsilon,\pi_\DLEQ)$.
Compute $\mathrm{\Upsilon}\leftarrow \Derive(\pp,C,\zeta)$; if $\mathrm{\Upsilon}=\perp$ output $0$. Else parse $\mathrm{\Upsilon}=(M,R)$, set $A:=\upsilon-M$ and $B:=C-M$.
	Output $1$ iff $\cV_{\mathsf{DLEQ}}(\pp,(A,B),\pi_\DLEQ)=1$.
		
\item $\Open(\pp,C,\zeta)$: $\Open(\pp,C,\zeta):=\Derive(\pp,C,\zeta)$ if $\Vcert(\pp,C,\zeta)=1$, and $\perp$ otherwise.
\end{itemize}

\subsubsection{Properties of the USV certificate scheme}
\begin{itemize}
	\item Correctness: For honest generation, $\upsilon=(m+r)\mathcal G$ and $\nu=mr^{-1}$. Hence,
	$$\tfrac{\nu}{\nu+1}\upsilon=\tfrac{mr^{-1}}{mr^{-1}+1}(m+r)\mathcal G=m\mathcal G, \text{ and } R=C-M.$$
	Therefore, $\Vcert(\pp,C,\zeta)=1$ and $\Open(\pp,C,\zeta)=(M,R)\in\mathcal R_{\pp}$.
	
	\item Unique verified opening: For any fixed $(C,\zeta)$, $\Derive(\pp,C,\zeta)$ is deterministic. Hence, there is at most
	one candidate opening $\mathrm{\Upsilon}$ it can output. Moreover, whenever $\Open(\pp,C,\zeta)\neq\perp$, we have
	$\Open(\pp,C,\zeta)=\Derive(\pp,C,\zeta)=(M,R)$ with $M=\tfrac{\nu}{\nu+1}\upsilon$ uniquely determined (for $\nu\neq -1$),
	which leads to uniquely determined $R=C-M$.
	
	\item Straight-line verified public extraction: Extraction of the public opening is $\Open(\pp,C,\zeta)$. It is deterministic,
	uses no trapdoor, and is non-rewinding. 
	
	\item Opening-conditional tag simulatability: Given $(C,\mathrm{\Upsilon})$ with $\mathrm{\Upsilon}=(M,R)$ and $(C,\mathrm{\Upsilon})\in\mathcal R_{\pp}$:
	sample $\bar\nu \leftarrowdollar \mathbb Z_p^*\setminus\{-1\}$ and set
	$\bar\upsilon := (1+\bar\nu^{-1})M$.
	Define $\bar A:=\bar\upsilon-M$ and $\bar B:=C-M$.
	Compute $\bar{\pi}_\DLEQ \leftarrow \mathsf{Sim}_{\mathrm{DLEQ}}(\pp,(\bar A,\bar B))$
	using the UC-NIZK simulator, and output $\tilde\zeta:=(\bar\nu,\bar\upsilon,\bar{\pi}_\DLEQ)$.
	Note that for a random $\bar\nu$, the derived instance $(\bar A,\bar B)$ need not lie in $\mathcal{R}_{\mathrm{DLEQ}}$.
	Accordingly, the simulator relies on the NIZK simulator to produce an accepting proof even for off-language statements,
	and indistinguishability holds under the DDLEQ assumption; see Lemma~\ref{lem:usv-tag-sim}.
	By construction, $\Derive(\pp,C,\tilde\zeta)=(M,R)$, and thus $\Open(\pp,C,\tilde\zeta)=\mathrm{\Upsilon}$ whenever
	$\Vcert(\pp,C,\tilde\zeta)=1$.
\end{itemize}

\begin{lemma}[Opening-conditional tag simulatability]\label{lem:usv-tag-sim}
Assume DDLEQ hardness and that the Fischlin-based UC-NIZK $\mathrm{\Pi}_{\mathsf{DLEQ}}$ in the \gROCRP\ model
has the universal simulation interface and straight-line extraction guarantees of Lemma~\ref{lem:gnpro-uc-nizk}. Then the simulator $\Simcert$ satisfies Item 3 (opening-conditional tag simulatability) of Definition~\ref{def:USV-cert}.
\end{lemma}

\begin{proof}
	Fix $\lambda$ and let $\mathcal D$ be any PPT distinguisher.
	We compare the following two experiments, which both output a tuple
	$(\pp,C,\mathrm{\Upsilon},\zeta)$ where $\mathrm{\Upsilon}=(M,R)$.
	Note that $\Cert(\pp,m)$ samples $m\in\Zp^*$ and $r\in\Zp^*\setminus\{-m\}$, so the resulting
	$M=m\G\neq 0_{\mathbb G}$ and $R=r\h\neq 0_{\mathbb G}$; hence
	$(C,\mathrm{\Upsilon})\in\mathcal R_{\pp}$ and the simulator's derived instance
	$(\bar A,\bar B)\in\mathcal X_{\mathsf{DLEQ}}$ in every invocation.
	
	\smallskip\noindent
	\textbf{Experiment $\mathsf{RealTag}(1^\lambda)$:}
	\begin{enumerate}
		\item Sample $\pp \gets \mathsf{Setup}(1^\lambda,\pounds)$, $m\leftarrowdollar \Zp^*$, and $r\leftarrowdollar \Zp^*\setminus\{-m\}$.
		\item Define
		$M:=m\G$, $R:=r\h$, $C:=M+R$, $\nu:=mr^{-1}\bmod p$, $\upsilon:=(m+r)\G$,
		$A:=\upsilon-M$ and $B:=C-M$.
		\item Compute $\pi_\DLEQ\leftarrow \cp_{\DLEQ}(\pp,(A,B),r)$ and set $\zeta:=(\nu,\upsilon,\pi_\DLEQ)$. Output $(\pp,C,\mathrm{\Upsilon},\zeta)$ where $\mathrm{\Upsilon}:=(M,R)$.
	\end{enumerate}
	
	\smallskip\noindent
	\textbf{Experiment $\mathsf{SimTag}(1^\lambda)$:}
	\begin{enumerate}
		\item Sample $\pp\leftarrow \mathsf{Setup}(1^\lambda,\pounds)$, $m\leftarrowdollar \Zp^*$, and $r\leftarrowdollar \Zp^*\setminus\{-m\}$.
		\item Define $M:=m\G$, $R:=r\h$, $C:=M+R$ and $\mathrm{\Upsilon}:=(M,R)$.
		\item Sample $\bar\nu\leftarrowdollar \Zp^*\setminus\{-1\}$ and set
		$\bar\upsilon := (1+\bar\nu^{-1})M$.
		Let $\bar A:=\bar\upsilon-M$ and $\bar B:=C-M$.
		\item Compute $\bar{\pi}_\DLEQ \leftarrow \s_{\DLEQ}(\pp,(\bar A,\bar B))$
		and set $\tilde\zeta := (\bar\nu,\bar\upsilon,\bar{\pi}_\DLEQ)$. Output $(\pp,C,\mathrm{\Upsilon},\tilde\zeta)$.
	\end{enumerate}
	
	We prove that $\Bigl|\Pr[\mathcal D(\mathsf{RealTag}(1^\lambda))=1]-\Pr[\mathcal D(\mathsf{SimTag}(1^\lambda))=1]\Bigr|
	\le \negl(\lambda).$
	
	\smallskip\noindent
	Hybrid $\Game_0$:
	This is identical to $\mathsf{RealTag}$.
	
	\smallskip\noindent
	Hybrid $\Game_1$:
	Modify $\Game_0$ by replacing the real proof $\pi_\DLEQ\leftarrow\cp_{\DLEQ}(\pp,(A,B),r)$
	with a simulated proof $\tilde{\pi}_\DLEQ\leftarrow \s_{\DLEQ}(\pp,(A,B))$ for the same statement $(A,B)$.
	All other values are unchanged. Since $(A,B)$ is a true DLEQ statement in $\Game_0$ (indeed $A=r\G$ and $B=r\h$),
	the zero-knowledge property of $\mathrm{\Pi}_{\DLEQ}$ for true statements implies $	\Game_0 \approx_c \Game_1.$
	
	\smallskip\noindent
	$\Game_1$ vs. $\mathsf{SimTag}$:
	Assume for contradiction that there exists a PPT distinguisher $\mathcal D$
	and a non-negligible function $\epsilon(\lambda)$ such that
	\[
	\Bigl|\Pr[\mathcal D(\Game_1(1^\lambda))=1]-\Pr[\mathcal D(\mathsf{SimTag}(1^\lambda))=1]\Bigr|
	\ge \epsilon(\lambda)
	\]
	for infinitely many $\lambda$. We construct a PPT distinguisher $\mathcal B$ for DDLEQ.
	
	\smallskip\noindent
	\textbf{Distinguisher $\mathcal B$ (for DDLEQ):}
	$\mathcal B$ receives $\pp=(\G,\h)$ and a challenge pair $(A^\star,B^\star)\in\G\times\G$ sampled as:
	either $(A^\star,B^\star)=(r\G,r\h)$ (DDLEQ-true) for uniform $r\leftarrow\Zp^*$,
	or $(A^\star,B^\star)=(\rho\G,r\h)$ (DDLEQ-false) for uniform independent
	$\rho,r\leftarrow\Zp^*$. $\mathcal B$ performs:
	\begin{enumerate}
		\item Sample $\nu\leftarrowdollar \Zp^*\setminus\{-1\}$.
		\item Define $M:=\nu A^\star$, $R:=B^\star$, $C:=M+R$, and $\upsilon:=M+A^\star$. Compute $\pi_\DLEQ\leftarrow \s_{\DLEQ}(\pp,(A^\star,B^\star))$.
		\item Output $\mathcal B$'s guess as $\mathcal D(\pp,C,(M,R),(\nu,\upsilon,\pi_\DLEQ))$.
	\end{enumerate}
	We claim that the distribution fed to $\mathcal D$ by $\mathcal B$ matches
	$\Game_1$ in the DDLEQ-true case, and matches $\mathsf{SimTag}$ up to negligible statistical distance
	in the DDLEQ-false case.
	
	\smallskip\noindent
	\emph{DDLEQ-true case:}
	Here $(A^\star,B^\star)=(r\G,r\h)$ for uniform $r\in\Zp^*$.
	$\mathcal B$ samples $\nu\in\Zp^*\setminus\{-1\}$ and sets
	$M=\nu r\G$, $R=r\h$, $C=M+R$, and $\upsilon=M+r\G=(\nu+1)r\G$.
	Let $m:=\nu r \bmod p$; then $M=m\G$ and $\upsilon=(m+r)\G$ and $\nu=mr^{-1}$.
	Moreover $r=-m$ would imply $\nu=-1$, which is excluded; thus the support condition
	$r\in\Zp^*\setminus\{-m\}$ holds automatically.
	Finally, the proof is generated as $\s_{\DLEQ}(\pp,(A^\star,B^\star))=\s_{\DLEQ}(\pp,(r\G,r\h))$,
	which is exactly the proof distribution in $\Game_1$.
	Hence, the output distribution is identical to $\Game_1$.
	
	\smallskip\noindent
	\emph{DDLEQ-false case:}
	Here $(A^\star,B^\star)=(\rho\G,r\h)$ for independent uniform $\rho,r\in\Zp^*$.
	$\mathcal B$ samples $\nu\in\Zp^*\setminus\{-1\}$ and defines
	$M=\nu\rho\G$ and $R=r\h$ and $\upsilon=M+\rho\G$.
	Let $m:=\nu\rho \bmod p$. Then $M=m\G$, and also 
	\[
	A^\star=\rho\G=\nu^{-1}M,\qquad
	\upsilon=M+A^\star=(1+\nu^{-1})M,
	\]
	so the tuple $(\nu,\upsilon)$ is distributed exactly as in $\s_\Tag$,
	and the proof is generated by the same simulator $\s_{\DLEQ}$.
	The only difference from $\s_\Tag$ is that $\s_\Tag$ samples
	$r\leftarrow \Zp^*\setminus\{-m\}$ whereas here $r\leftarrow \Zp^*$ independently of $m$.
	These two $r$-distributions differ by statistical distance at most
	\[
	\Pr[r=-m]\le \frac{1}{p-1},
	\]
	which is negligible in $\lambda$ since $p$ is exponential in $\lambda$.
	Therefore the DDLEQ-false output distribution is negligibly close to $\s_\Tag$. Thus, $\mathcal B$ distinguishes DDLEQ-true from DDLEQ-false with advantage at least
	$\epsilon(\lambda)-\negl(\lambda)$, contradicting DDLEQ.
	Hence, $\Game_1 \approx_c \mathsf{SimTag}$. Combining $\Game_0 \approx_c \Game_1$ (ZK) and $\Game_1 \approx_c \mathsf{SimTag}$
	(DDLEQ) yields $\mathsf{RealTag}\approx_c \mathsf{SimTag}$, as required. \qed
\end{proof}

\begin{remark}[Status of opening-conditional tag simulatability]
	\label{rem:tag-sim-status}
	Lemma~\ref{lem:usv-tag-sim} proves a stronger interface than what our UC hybrids need:
	no UC proof in this paper requires fabricating an accepting USV tag $\zeta$ for a fixed commitment
	$C$ without knowing a witness. Consequently, the UC simulators constructed for
	Theorems~\ref{thm:usv-real-ideal} and~\ref{thm:SDKG-UC} never call the DLEQ proof simulator and never
	program $\mathsf{ctx}_{\mathsf{DLEQ}}$.
	We retain the lemma because it is a natural property of USV as a standalone primitive and is useful
	for other compositions (e.g., ``commit now, open-to-$M$ later'' designs under NXK).
\end{remark}

\begin{lemma}[Equivocation resistance of the USV instantiation]\label{lem:usv-eqv}
	Assume further that the embedded proof system $\mathrm{\Pi}_{\mathrm{DLEQ}}$ used inside $\Vcert$ is a NIZK-AoK for
	$\mathcal R_{\mathrm{DLEQ}}$ in the \gROCRP\ model with an online extractor $\Ext_{\mathrm{DLEQ}}$
	(as in Section~\ref{SSE-NIZK}). Then, under the DL assumption (Assumption~\ref{assump:dl}) and the transparent generator derivation assumption (Assumption~\ref{assump:transparent-pp}), the USV certificate scheme of Section~\ref{subsec:USV-inst}
	is equivocation resistant (Definition~\ref{def:usv-eqv}).
\end{lemma}

\begin{proof}
	Fix any PPT adversary $\mathcal A$.
	We construct a PPT algorithm $\mathcal B$ that receives a DL challenge
	$(\G,X)$ with $X=x\G$ for $x\leftarrowdollar\Zps$ (Definition~\ref{def:dl})
	and outputs $x$ with non-negligible probability whenever $\mathcal A$ wins
	the equivocation game. Specifically, $\mathcal B$ sets
	$\pp:=(\mathbb{G},p,\G,X)$---embedding the DL challenge $X$ as the
	secondary generator~$\h$---and simulates the equivocation experiment
	for~$\mathcal A$ on these parameters. By
	Assumption~\ref{assump:transparent-pp}, the Setup-derived $\h$ is
	computationally indistinguishable from uniform over
	$\mathbb{G}\setminus\{0_{\mathbb G}\}$, and the DL challenge $X$ is
	exactly uniform over $\mathbb{G}\setminus\{0_{\mathbb G}\}$ (since
	$x\leftarrowdollar\Zps$ and the map $z\mapsto z\G$ is a bijection on
	$\Zp$). Hence $\mathcal A$'s view under $\pp=(\mathbb{G},p,\G,X)$ is
	computationally indistinguishable from its view in the real equivocation
	experiment, and so
	$\Adv^{\mathrm{dl}}_{\mathcal B}(\lambda)\ge
	\Adv^{\mathrm{eqv}}_{\mathcal A}(\lambda)-\negl(\lambda)$.
	Since $\Adv^{\mathrm{dl}}_{\mathcal B}(\lambda)\le\negl(\lambda)$ by
	Assumption~\ref{assump:dl}, it follows that
	$\Adv^{\mathrm{eqv}}_{\mathcal A}(\lambda)\le \negl(\lambda)$.

	\smallskip\noindent
	\textbf{Algorithm $\mathcal B^{\mathcal A}(\G,X)$.}
	$\mathcal B$ sets $\h:=X$ and $\pp:=(\mathbb{G},p,\G,\h)$, then
	internally simulates the equivocation experiment for $\mathcal A$:
	\begin{enumerate}
		\item Give $\pp$ to $\mathcal A$ and answer $\mathcal A$'s \gROCRP\ queries by lazy sampling,
		while recording the complete query/answer log $\mathsf{Log}_{\mathcal A}$ under each DLEQ-proof context.
		\item Receive $(C,\zeta,\zeta')$ from $\mathcal A$, where $\zeta=(\nu,\upsilon,\pi_\DLEQ)$ and
		$\zeta'=(\nu',\upsilon',\pi_\DLEQ')$.
		\item Compute $\mathrm{\Upsilon}\leftarrow \Derive(\pp,C,\zeta)$ and $\mathrm{\Upsilon}'\leftarrow \Derive(\pp,C,\zeta')$.
		If either is $\perp$, output $\perp$.
		Otherwise parse $\mathrm{\Upsilon}=(M,R)$ and $\mathrm{\Upsilon}'=(M',R')$.
		\item If $\Vcert(\pp,C,\zeta)=0$ or $\Vcert(\pp,C,\zeta')=0$, output $\perp$. If $M=M'$, output $\perp$.
		\item Define $A:=\upsilon-M, B:=C-M, A':=\upsilon'-M',$ and $B':=C-M'$.
		\item Run the AoK extractor to obtain witnesses:
		\[
		r \leftarrow \Ext_{\DLEQ}\bigl(\pp,(A,B),\pi_\DLEQ; \mathsf{Log}_{\mathcal A}\bigr),\qquad
		r' \leftarrow \Ext_{\DLEQ}\bigl(\pp,(A',B'),\pi_\DLEQ'; \mathsf{Log}_{\mathcal A}\bigr).
		\]
		If either extraction fails (outputs $\perp$) or the extracted values do not satisfy
		$A=r\G,B=r\h$ and $A'=r'\G,B'=r'\h$, output $\perp$.
		\item Compute $m:=\nu r \bmod p$ and $m':=\nu' r'\bmod p$.
		\item Compute $\hat{r}:=r'-r \bmod p$. If $\hat{r}=0$, output $\perp$.
		Otherwise output
		\[
		x := (m-m')\cdot (\hat{r})^{-1}\bmod p.
		\]
	\end{enumerate}
	
	\smallskip\noindent
	Assume $\mathcal A$ wins, i.e.,
	$\Vcert(\pp,C,\zeta)=\Vcert(\pp,C,\zeta')=1$ and $M\neq M'$ where
	$\Derive(\pp,C,\zeta)=(M,R)$ and $\Derive(\pp,C,\zeta')=(M',R')$. Since $\Vcert(\pp,C,\zeta)=1$, it follows that $\nu\neq -1\bmod p$.
	Moreover, by definition of $\Vcert$ and our setting of $(A,B)$, we have
	$\cV_{\DLEQ}(\pp,(A,B),\pi_\DLEQ)=1$.
	By the NIZK-AoK property, except with negligible probability the extractor returns a witness
	$r\in\Zp^*$ such that $A=r\G$ and $B=r\h$.
	Similarly, except with negligible probability, the extractor returns $r'\in\Zp^*$ such that
	$A'=r'\G$ and $B'=r'\h$. Using $\Derive$'s defining equations, for $\zeta=(\nu,\upsilon,\pi_\DLEQ)$ we have
	\[
	M=\frac{\nu}{\nu+1}\upsilon,
	\qquad\text{hence}\qquad
	A=\upsilon-M=\frac{1}{\nu+1}\upsilon,
	\]
	so $\upsilon=(\nu+1)A$ and therefore $M=\frac{\nu}{\nu+1}\upsilon=\nu A=\nu r\,\G = m\,\G \text{ where } m:=\nu r \bmod p.$
	Likewise $M'=m'\G$ where $m':=\nu'r'\bmod p$. Also, by definition, $B=C-M=r\h, B'=C-M'=r'\h,$
	so we obtain two decompositions of $C$:
	\[
	C = M + r\h = m\G + r\h
	\qquad\text{and}\qquad
	C = M' + r'\h = m'\G + r'\h.
	\]
	Subtracting yields $(m-m')\G = (r'-r)\h.$
	Since $M\neq M'$ and the map $z\mapsto z\G$ is a bijection on $\Zp$ (prime-order group),
	we have $m\neq m'$.
	Therefore, the above equality implies $r'\neq r$ (otherwise $(m-m')\G=0$ would force $m=m'$).
	Hence, $\hat{r}:=r'-r\in\Zp^*$ and is invertible, and the value
	\[
	x := (m-m')\cdot(\hat{r})^{-1}\bmod p
	\]
	satisfies $x\G=\h=X$. $\mathcal B$ outputs such an $x$ whenever $\mathcal A$ wins the equivocation experiment and both extractor
	invocations succeed. The extractor failure probability is negligible by the AoK guarantee, and
	the hybrid loss from replacing the Setup-derived~$\h$ with the DL
	challenge~$X$ is negligible by Assumption~\ref{assump:transparent-pp}.
	Thus, $\Adv^{\mathrm{dl}}_{\mathcal B}(\lambda)\ge\Adv^{\mathrm{eqv}}_{\mathcal A}(\lambda)-\negl(\lambda)$.
	Since $\Adv^{\mathrm{dl}}_{\mathcal B}(\lambda)\le\negl(\lambda)$ by Assumption~\ref{assump:dl},
	it follows that $\Adv^{\mathrm{eqv}}_{\mathcal A}(\lambda)\le \negl(\lambda)$. \qed
\end{proof}

\subsection{UC Security}\label{UC:usv}
We model certificates with a handle-bound ideal functionality that stores the
\emph{verified} opening implied by $(C,\zeta)$, namely the (unique) value returned by
$\Open(\pp,C,\zeta)$, and binds a handle to $(C,M)$ via a \gROCRP\ receipt digest, where $M:=\Open_M(\pp,C,\zeta)$ is the
verified public opening implied by $(C,\zeta)$.

In our applications, each USV handle is intended for a single committer--verifier pair.
Accordingly, we define $\Fusv$ (Fig.~\ref{fig:Fusvcert}) so that each handle is scoped to
both endpoints: $\Fusv$ maintains independent state keyed by $(\sid,\cid,P_s,P_r)$.
This prevents a party $P_s'\neq P_s$ from invalidating another sender's handle by reusing
the same $(\sid,\cid)$ toward the same relying party $P_r$, while still keeping recipient-local
state (so a corrupted committer may equivocate across different relying parties, as in the real protocol).
In $\Fusv$, the leakage to~$\s$ includes the public certificate $(C,\zeta)$ along with the receipt digest~$d$: all three values are public (the certificate $(C,\zeta)$ is intended for delivery to~$P_r$ over~$\Fchan$, and $d$ is deterministic from the committed tuple and the global gRO-CRP oracle). This enables~$\s$ to simulate the $\Fchan$ delivery of $(C,\zeta,d)$ to a corrupted recipient.

\begin{figure}[t]
	\centering
	\setlength{\fboxrule}{0.2pt} 
	\fbox{
		\parbox{\dimexpr\linewidth-2\fboxsep-2\fboxrule\relax}{%
			\ding{169} \textsf{State:} table $\mathsf T$ indexed by $(\sid,\cid,P_s,P_r)$ with entries $(M,d,\status,\mathsf{del})$ where
			$M\in\mathbb G\cup\{\perp\}$, $d\in\{0,1\}^\lambda$, $\status\in\{\textsf{pending},\textsf{invalid}\}$,
			and $\mathsf{del}\in\{0,1\}$ (init $0$).

			\smallskip
			\ding{169} Upon receiving $(\Commit,\sid,\cid,P_r,C,\zeta)$ from party $P_s$:
			\begin{itemize}[leftmargin=*,nosep]
				\item Compute the verified public opening
				\[
				M^\star \coloneqq \Open_M(\pp,C,\zeta)
				=
				\mathsf{PubOpen}\bigl(\pp,\Open(\pp,C,\zeta)\bigr)\in \mathbb G\cup\{\perp\}.
				\]
				Compute receipt digest
				\[
				d \coloneqq H\!\bigl(\USVrcpt,\langle \sid,\cid,P_s,P_r,C,M^\star\rangle\bigr),
				\qquad \text{where }\USVrcpt \in \CtxTEE .
				\]
				\item If $(\sid,\cid,P_s,P_r)\in\dom(\mathsf T)$:
				send $(\mathsf{DupCommit},\sid,\cid,P_s,P_r,C,\zeta,d)$ to~$\s$,
				send $(\receipt,\sid,\cid,P_r,d)$ to $P_s$, and return.
				\item If $M^\star=\perp$, store $\mathsf T[\sid,\cid,P_s,P_r]\gets(\perp,d,\textsf{invalid},0)$.
				\item Otherwise store $\mathsf T[\sid,\cid,P_s,P_r]\gets(M^\star,d,\textsf{pending},0)$.
				\item Send $(\receipt,\sid,\cid,P_r,d)$ to $P_s$ and $(\receipt,\sid,\cid,P_s,P_r,C,\zeta,d)$ to $\s$.
			\end{itemize}

			\smallskip
			\ding{169} Upon receiving $(\mathsf{Deliver},\sid,\cid,P_s,P_r)$ from $\s$:
			if $(\sid,\cid,P_s,P_r)\in\dom(\mathsf T)$, set $\mathsf T[\sid,\cid,P_s,P_r].\mathsf{del}\gets 1$.

			\smallskip
			\ding{169} Upon receiving $(\mathsf{Invalidate},\sid,\cid,P_s,P_r)$ from $\s$:
			if $(\sid,\cid,P_s,P_r)\in\dom(\mathsf T)$, set $\mathsf T[\sid,\cid,P_s,P_r].\status\gets\textsf{invalid}$.

			\smallskip
			\ding{169} Upon receiving $(\Verify,\sid,\cid,P_s,C,\zeta)$ from party $P_r$:
			\begin{itemize}[leftmargin=*,nosep]
				\item If $(\sid,\cid,P_s,P_r)\notin\dom(\mathsf T)$ or $\mathsf T[\sid,\cid,P_s,P_r].\mathsf{del}=0$, return $\perp$. If $\mathsf T[\sid,\cid,P_s,P_r].\status=\textsf{invalid}$, return $0$.
				\item Let $(M,d,\textsf{pending},1) := \mathsf T[\sid,\cid,P_s,P_r]$.
				Compute $M' \coloneqq \Open_M(\pp,C,\zeta)$.
				Return $1$ iff
				\[
				M'\neq\perp \ \wedge\ M'=M \ \wedge\
				H\!\bigl(\USVrcpt,\langle \sid,\cid,P_s,P_r,C,M'\rangle\bigr)=d,
				\]
				else return $0$.
			\end{itemize}

			\smallskip
			\ding{169} Upon receiving $(\Open,\sid,\cid,P_s)$ from party $P_r$:
			if $(\sid,\cid,P_s,P_r)\notin\dom(\mathsf T)$ or $\mathsf T[\sid,\cid,P_s,P_r].\mathsf{del}=0$ or $\mathsf T[\sid,\cid,P_s,P_r].\status=\textsf{invalid}$ then
			send $(\abort,\sid,\cid)$ to $P_r$ and $\s$; else send $(\Open,\sid,\cid,\mathsf T[\sid,\cid,P_s,P_r].M)$ to $P_r$ and $\s$.
	}}
	\caption{Two-party USV certificate functionality $\Fusv$ with \gROCRP\ receipt binding and verified openings.}
	\label{fig:Fusvcert}
\end{figure}

\begin{remark}[Game-based equivocation vs.\ UC-level non-malleability]
	\label{rem:eqv-vs-uc}
	Lemma~\ref{lem:usv-eqv} is a standalone game-based statement about the USV certificate scheme
	itself: it rules out producing two verifying tags for the same commitment $C$ that induce different
	public openings. In the UC setting, the relevant ``no substitution under a fixed handle'' guarantee is
	instead enforced by $\Fusv$’s receipt binding in the non-programmable context $\USVrcpt\in\CtxTEE$ and
	captured by Lemma~\ref{lem:hb-nm} (handle-bound non-malleability), which is what the UC proof of
	Theorem~\ref{thm:usv-real-ideal} relies on.
\end{remark}

\begin{lemma}[Handle-bound non-malleability]\label{lem:hb-nm}
Let $M:=\Open_M(\pp,C,\zeta)\neq\perp$ be the USV opening for the honest commit and let $d:=H(\USVrcpt,\langle \sid,\cid,P_s,P_r,C,M\rangle)$.
Then for every $(C',\zeta')$ such that $M':=\Open_M(\pp,C',\zeta')\neq\perp$ and $(C',M')\neq(C,M)$,
\[
\Pr\Bigl[\Fusv.\Verify(\sid,\cid,P_s,C',\zeta')=1 \Bigr]\le \negl(\lambda).
\]
\end{lemma}

\begin{proofsketch}
	Let $H_{\mathrm{rcpt}}(u)\coloneqq H(\USVrcpt,u)$.
	Fix the honest receipt input
	$u:=\langle \sid,\cid,P_s,P_r,C,M\rangle$ and receipt
	$d:=H_{\mathrm{rcpt}}(u)$.
	For any distinct pair $(C',\zeta')\neq(C,\zeta)$, injectivity of $\langle\cdot\rangle$ implies that
	$u':=\langle \sid,\cid,P_s,P_r,C',M'\rangle \neq u$.
	A successful substitution (i.e., $\Fusv.\Verify(\sid,\cid,C',\zeta')=1$) therefore implies a
	second-preimage for the fixed target input $u$, namely
	$H_{\mathrm{rcpt}}(u') = H_{\mathrm{rcpt}}(u)=d$ with $u'\neq u$.
	
	In the random-oracle/\gROCRP\ model, suppose the adversary makes at most $Q=\poly(\lambda)$ queries to
	$H_{\mathrm{rcpt}}$ before the $\Verify$ call. The $\Verify$ procedure itself evaluates
	$H_{\mathrm{rcpt}}(u')$ once. Hence
	\[
	\Pr\big[\,H_{\mathrm{rcpt}}(u')=d\,\big]\ \le\ \frac{Q+1}{|\mathcal Y|}
	\quad (\text{in particular } \le (Q+1)/2^\lambda \text{ when } |\mathcal Y|=2^\lambda),
	\]
	which is negligible. \qed
\end{proofsketch}

The PPT simulator $\s$ only needs $\mathsf{pp}$ and the public leakage from $\Fusv$.
In the ideal world, on receiving each \((\Commit,\sid,\cid,C,\zeta)\), \(\Fusv\) computes the receipt
\(d\coloneqq H\!\bigl(\USVrcpt,\langle \sid,\cid,P_s,P_r,C,M\rangle\bigr)\) and stores the public opening
\(M:=\Open_M(\pp,C,\zeta)\) only if the certificate verifies, i.e., only if \(\Vcert(\pp,C,\zeta)=1\).
$\Fusv$ sends $(\receipt,\sid,\cid,P_s,P_r,C,\zeta,d)$ to~$\s$; $\s$ uses $(C,\zeta,d)$ to reproduce the $\Fchan$ traffic: $\s$ sends $(\mathsf{Send},\sid,(C,\zeta,d))$ to the $\Fchan$ instance for $(P_s,P_r)$, generating the same $\mathsf{Leak}$ ticket and length leakage as in the real world. When $\cA$ issues $(\mathsf{Deliver},\sid,\rho)$, $\Fchan$ delivers $(C,\zeta,d)$ to~$P_r$, and $\s$ concurrently sends $(\mathsf{Deliver},\sid,\cid,P_s,P_r)$ to $\Fusv$ to activate the recipient-side entry.
If $P_r$ is corrupted, $\s$ additionally forwards $(C,\zeta,d)$ to the corrupted~$P_r$.
Upon receiving $(\mathsf{DupCommit},\sid,\cid,P_s,P_r,\ldots)$ from $\Fusv$, $\s$ sends the duplicate over $\Fchan$ identically and, when $\cA$ delivers the duplicate, sends $(\mathsf{Invalidate},\sid,\cid,P_s,P_r)$ to $\Fusv$.
When a party is adaptively corrupted after a successful \textsf{Commit},
\(\s\) reconstructs \(\mathrm{\Upsilon}\) deterministically as \(\mathsf{Open}(\mathsf{pp},C,\zeta)\).
This matches the real world since the verified opening is a fixed deterministic function of \((C,\zeta)\).

\begin{theorem}[UC security of $\Fusv$]\label{thm:usv-real-ideal}
	Let $\mathrm{\Pi}_{\mathsf{USV}}$ denote the concrete protocol that maintains a local
	table $\mathsf{T}^{\mathrm{\Pi}}$ indexed by $(\sid,\cid,P_s,P_r)$ with entries
	$(M,d,\status)$, tracking the $(M,d,\status)$ components of $\Fusv$'s state (the delivery flag $\mathsf{del}$ is implicit in $\mathrm{\Pi}_{\mathsf{USV}}$: the recipient's entry is created only upon $\Fchan$ delivery), and whose per-interface behavior is:
	\begin{itemize}[nosep]
		\item \textbf{Commit}$(\sid,\cid,P_r,C,\zeta)$: party $P_s$ receives $(C,\zeta)$
		(e.g., from a prior $\Cert(\pp,m)$ invocation).
		Compute $M:=\Open_M(\pp,C,\zeta)$ and the receipt
		$d:=H\!\bigl(\USVrcpt,\langle \sid,\cid,P_s,P_r,C,M\rangle\bigr)$
		(in context $\USVrcpt\in\CtxTEE$).
		If $(\sid,\cid,P_s,P_r)\in\dom(\mathsf{T}^{\mathrm{\Pi}})$, mark
		$\mathsf{T}^{\mathrm{\Pi}}[\sid,\cid,P_s,P_r].\status\gets\textsf{invalid}$
		and return $(\receipt,\sid,\cid,P_r,d)$ to $P_s$.
		If $M=\perp$, store $\mathsf{T}^{\mathrm{\Pi}}[\sid,\cid,P_s,P_r]\gets(\perp,d,\textsf{invalid})$;
		else store $\mathsf{T}^{\mathrm{\Pi}}[\sid,\cid,P_s,P_r]\gets(M,d,\textsf{pending})$.
		Send $(C,\zeta,d)$ to $P_r$ over $\Fchan$.
		\item \textbf{Receive} (at $P_r$): Upon receiving $(C,\zeta,d)$ from $P_s$ via $\Fchan$,
		$P_r$ computes $M':=\Open_M(\pp,C,\zeta)$ and
		$d':=H\!\bigl(\USVrcpt,\langle \sid,\cid,P_s,P_r,C,M'\rangle\bigr)$.
		If $(\sid,\cid,P_s,P_r)\in\dom(\mathsf{T}^{\mathrm{\Pi}})$, set
		$\mathsf{T}^{\mathrm{\Pi}}[\sid,\cid,P_s,P_r].\status\gets\textsf{invalid}$
		and return.
		If $M'=\perp$ or $d'\neq d$, store
		$\mathsf{T}^{\mathrm{\Pi}}[\sid,\cid,P_s,P_r]\gets(\perp,d,\textsf{invalid})$;
		else store
		$\mathsf{T}^{\mathrm{\Pi}}[\sid,\cid,P_s,P_r]\gets(M',d,\textsf{pending})$.
		\item \textbf{Verify}$(\sid,\cid,P_s,C',\zeta')$: recipient $P_r$ looks up
		$\mathsf{T}^{\mathrm{\Pi}}[\sid,\cid,P_s,P_r]$.
		If the entry does not exist, return~$\perp$.
		If $\mathsf{T}^{\mathrm{\Pi}}[\sid,\cid,P_s,P_r].\status=\textsf{invalid}$, return~$0$.
		Let $(M,d,\textsf{pending}):=\mathsf{T}^{\mathrm{\Pi}}[\sid,\cid,P_s,P_r]$.
		Locally compute $\Vcert(\pp,C',\zeta')$;
		if it rejects, return~$0$.
		Compute $M':=\Open_M(\pp,C',\zeta')$ and
		$d':=H\!\bigl(\USVrcpt,\langle\sid,\cid,P_s,P_r,C',M'\rangle\bigr)$.
		Return $1$ iff $M'\neq\perp\ \wedge\ M'=M\ \wedge\ d'=d$.
		\item \textbf{Open}$(\sid,\cid)$: $P_r$ reads the stored $M$ (or $\perp$ if none).
	\end{itemize}
	All verification/derivation is deterministic from $(C,\zeta)$ and the public algorithms $(\Vcert,\Derive,\PubOpen)$.
	Then $\mathrm{\Pi}_{\mathsf{USV}}$ UC-realizes $\Fusv$ in the ($\Fchan$, \gROCRP)-hybrid model; i.e., for every PPT adversary
	$\cA$ there exists a PPT simulator $\s$ such that for every PPT environment $\cZ$,
	\[
	\Exec(\mathrm{\Pi}_{\mathsf{USV}},\cA,\cZ,\lambda)\ \approx_c\ \Ideal(\Fusv,\s,\cZ,\lambda).
	\]
\end{theorem}

\begin{proof} Fix any PPT adversary $\cA$ and PPT environment $\cZ$.
	We construct a PPT simulator $\s$ and employ the following sequence of hybrids:
	\begin{itemize}
		\item Hybrid \(\Game_0\) (real execution): Each honest party $P_s$ receives $(C,\zeta)$
		(generated via $\Cert(\pp,m)$ by the calling protocol or locally),
		computes $M:=\Open_M(\pp,C,\zeta)$ and the receipt
		$$d=H\!\bigl(\USVrcpt,\langle \sid,\cid,P_s,P_r,C,M\rangle\bigr),$$
		maintains the handle table $\mathsf{T}^{\mathrm{\Pi}}$ (including duplicate-handle invalidation),
		and sends $(C,\zeta,d)$ to $P_r$.
		
	\item Hybrid \(\Game_1\) (syntactic resampling / explicit \(\Cert\) randomness):
	Proceed identically to \(\Game_0\), except that for each honest certificate we
	(re-)sample the full random tape of \(\Cert\) and re-evaluate it.
	Concretely, we sample
	\(r \leftarrowdollar \mathbb Z_p^*\setminus\{-m\}\) together with the internal
	randomness used by the randomized prover \(\cp_{\DLEQ}\) (and hence by the Fischlin transform),
	and then compute \((C,\zeta)\) by running \(\Cert(\pp,m)\) with that sampled randomness
	(and the fixed global oracle \(H\)).
	All other external inputs and protocol randomness are unchanged.
	This is a purely syntactic refactoring of how \(\Cert\)'s coins are sampled, so
	\(\Game_0 \overset{\mathrm{d}}{=} \Game_1\).
	
		\item Hybrid \(\Game_2\) (ideal execution), switch to \(\Fusv\) and $\s$:
		For each honest handle, \(\Fusv\) stores the same receipt \(d=H\!\bigl(\USVrcpt,\langle \sid,\cid,P_s,P_r,C,M\rangle\bigr)\). Since honest certificates verify, \(\Fusv\) stores \(M:=\Open_M(\pp,C,\zeta)\). By Lemma~\ref{lem:hb-nm}, any in-flight substitution under a fixed handle is detected in both worlds.
		For every \textsf{Commit} (regardless of $P_r$'s corruption status), $\s$ receives $(\receipt,\sid,\cid,P_s,P_r,C,\zeta,d)$ from $\Fusv$ and sends $(\mathsf{Send},\sid,(C,\zeta,d))$ to the $\Fchan$ instance for $(P_s,P_r)$, generating the same $\mathsf{Leak}$ ticket as in $\Game_0$.
		When $\cA$ issues $\mathsf{Deliver}$ for this ticket, $\Fchan$ delivers $(C,\zeta,d)$ to~$P_r$ and $\s$ concurrently sends $(\mathsf{Deliver},\sid,\cid,P_s,P_r)$ to $\Fusv$ to activate the recipient-side entry.
		If $P_r$ is honest, $P_r$ then processes $(C,\zeta,d)$ identically to $\Game_0$ and any subsequent $\Fusv.\Verify$ or $\Fusv.\Open$ call returns the same result as the local computation in $\Game_0$ (since $\mathsf{del}=1$ after delivery and all checks are deterministic).
		If $P_r$ is corrupted, $\s$ additionally forwards $(C,\zeta,d)$ to the corrupted~$P_r$.
		Upon receiving $(\mathsf{DupCommit},\sid,\cid,P_s,P_r,\ldots)$ from $\Fusv$, $\s$ reproduces the duplicate over $\Fchan$ and, upon delivery, sends $(\mathsf{Invalidate},\sid,\cid,P_s,P_r)$ to $\Fusv$, matching the real-world recipient-side invalidation timing.
		The adversary may invoke $\Fusv.\Verify$ or $\Fusv.\Open$ directly; $\s$ simply forwards the deterministic public computations $\Vcert$ and $\Open_M$ that the adversary could perform itself from $(C,\zeta)$.
		No USV/DLEQ proof simulation is needed because the adversary holds $(C,\zeta)$ (via the simulated $\Fchan$ delivery) and can verify using the public algorithms.
		Adaptive corruptions reveal identical internal states because the verified opening
		\(\mathsf{Open}(\pp,C,\zeta)\) is public and deterministic.
		Hence, \(\Game_1 \approx_c \Game_2\), and therefore \(\Game_0 \approx_c \Game_2\).
		Interpreting \(\Game_0\) as the real execution
		\(\Exec(\mathrm{\Pi}_{\mathsf{USV}},\cA,\cZ,\lambda)\) and \(\Game_2\) as the ideal execution
		\(\Ideal(\Fusv,\s,\cZ,\lambda)\), we conclude
		\(\Exec(\mathrm{\Pi}_{\mathsf{USV}},\cA,\cZ,\lambda)\approx_c \Ideal(\Fusv,\s,\cZ,\lambda)\). \qed
	\end{itemize}
\end{proof}

The simulator in this UC proof never invokes $\SimProgramRO$ in
$\mathsf{ctx}_{\mathsf{DLEQ}}$ (indeed it never needs to simulate USV/DLEQ proofs); it only forwards
verifier oracle queries via $\Query$.

\paragraph{Broader applicability of USV} Although USV is introduced here to make the SDKG transcript ``structure-complete'' under NXK, the primitive is not specific to DKG since it turns a KeyBox-resident, non-exportable scalar into a transcript-defined public group element via a publicly verifiable certificate. As an orthogonal illustration, Appendix \ref{App2} shows how to build NXK-compatible commit–reveal randomness beacons using USV.

\section{Enforcing Consistency via Straight-Line Extraction}\label{SSE-NIZK}
Since $\G$ generates $\mathbb G$, every $M\in\mathbb G$ can be written as $M=m\G$ for some $m\in\mathbb Z_p$. Accordingly, for DL-style statements the relevant security notion in our use-cases is \emph{knowledge soundness} (extractability of the witness), rather than language non-membership soundness. We use standard Schnorr (DL) and Chaum--Pedersen (DLEQ) $\mathrm{\Sigma}$-protocols,
which have special soundness and unique responses in prime-order groups.
Applying the optimized Fischlin transform in \gROCRP\ yields UC-NIZK-AoKs with straight-line extractors for both relations. Define the NP relation
\[
\cR_{\mathsf{DL}} := \{ ((\pp,M),m)\ :\ m\in\Zp\ \wedge\ M=m\G\}.
\]
We intentionally allow $m=0$ (and hence $M=0_{\mathbb G}$), since protocol-derived witnesses
(e.g., affine-relation witnesses and derived shares) can be $0$ under adversarial choice.
For the hardness assumption (Definition~\ref{def:dl}) we sample $x\leftarrow\Zp^\ast$ to avoid
the trivial instance $X=0_{\mathbb G}$; including $x=0$ would change success probability by at most $1/p=\negl(\lambda)$.

\paragraph{Tagged DL statements .}
To rule out replay/mauling of UC-context proofs across sessions and to justify the
use of simulation-extractability for fresh statements in concurrent executions,
we bind every UC-context DL statement to the session identifier $\sid$ and a fixed
label $\ell\in\{0,1\}^*$ describing the proof position. Concretely, we use the tagged
relation
\[
\cR_{\mathsf{DL}}^{\mathsf{tag}}
:= \{ ((\pp,\sid,\ell,M),m)\ :\ m\in\Zp \ \wedge\ M=m\G \}.
\]
In SDKG, every invocation of $\mathrm{\Pi}^{\mathsf{UC}}_{\mathsf{DL}}$ is on a tagged
statement of the form $x:=\DLstmt{\sid}{\ell}{M}$ (with $\pp$ implicit), so the Fischlin
hash input in Definition~\ref{FSdef} includes $(\sid,\ell)$ along with $M$.
We will slightly abuse notation and keep writing $\cp_{\mathsf{DL}}$ and $\cV_{\mathsf{DL}}$
for this tagged variant.

In this section, all proofs contexts are in $\CtxUC$. Given statement $(\pp,M)$ and witness $m$ such that $M=m\G$, the standard Schnorr $\mathrm{\Sigma}$-protocol for $\cR_{\mathsf{DL}}$ is:
\begin{enumerate}
	\item Commit: Prover samples $j\leftarrowdollar \mathbb{Z}_p$ and sends $J:=j\G$.
	\item Challenge: Verifier samples $e\leftarrowdollar \{0,1\}^t$, interprets it as an integer $\bar e\in[0,2^t-1]$,
	embeds it into $\mathbb{Z}_p$ (e.g., require $2^t<p)$, and sends $\bar e$ to the prover.
	\item Response: Prover sends $z:=j+\bar e\cdot m \bmod p$.
	\item Verify: accept iff $z\G = J+\bar e\cdot M$.
\end{enumerate}
Special soundness holds since from two accepting transcripts $(J,\bar e,z)$ and $(J,\bar e',z')$ with $\bar e\neq \bar e'$
one extracts $m=(z-z')(\bar e-\bar e')^{-1}\bmod p$.

\begin{lemma}[Unique responses for Schnorr]
	The Schnorr $\mathrm{\Sigma}$-protocol for $\cR_{\mathrm{DL}}$ has unique responses.
\end{lemma}
\begin{proofsketch}
	For fixed $(M,J,e)$, verification requires $z\G = J + eM$ in a prime-order group.
	Since $\G$ generates $\mathbb{G}$, this equation has a unique solution $z\in\mathbb{Z}_p$. \qed
\end{proofsketch}

Recall the public parameters $\pp=(\mathbb{G},p,\G,\h)$ and the NP relation
\[
\cR_{\mathsf{DLEQ}} := \Big\{ \big((\pp,A,B),r\big)\ :\ (A,B)\in\mathcal X_{\mathsf{DLEQ}} \ \wedge\ r\in\mathbb{Z}^*_p\ \wedge\ A=r\G\ \wedge\ B=r\h \Big\}.
\]

\paragraph{Chaum--Pedersen $\mathrm{\Sigma}$-protocol for $\cR_{\mathsf{DLEQ}}$.}
Given statement $(\pp,A,B)$ and witness $r$ such that $A=r\G$ and $B=r\h$, the Chaum--Pedersen
$\mathrm{\Sigma}$-protocol (with $t$-bit challenge) is:
\begin{enumerate}
	\item Commit: Prover samples $j\leftarrowdollar \mathbb{Z}_p$, and sends $J_1 := j\G \text{ and } J_2 := j\h.$
	\item Challenge: Verifier samples $e\leftarrowdollar \{0,1\}^t$, interprets it as $\bar e\in[0,2^t-1]$,
	embeds it into $\mathbb{Z}_p$ (e.g., require $2^t<p)$, and sends $\bar e$.
	\item Response: Prover sends $z := j+\bar e\cdot r \bmod p$.
	\item Verify: reject if $A=0_\mathbb G$ or $B=0_\mathbb G$; otherwise accept iff
	$z\G = J_1+\bar e\cdot A \text{ and } z\h = J_2+\bar e\cdot B.$
\end{enumerate}

Special soundness for Chaum--Pedersen: From any two accepting transcripts $(J_1,J_2,\bar e,z)$ and $(J_1,J_2,\bar e',z')$ with $\bar e\neq \bar e'$,
	one can efficiently extract a witness
	\[
	r = (z-z')\cdot(\bar e-\bar e')^{-1}\bmod p
	\]
	satisfying $A=r\G$ and $B=r\h$.

\begin{lemma}[Unique responses for Chaum--Pedersen]\label{lem:cp-unique}
	The Chaum--Pedersen $\mathrm{\Sigma}$-protocol for $\cR_{\mathsf{DLEQ}}$ has unique responses.
\end{lemma}
\begin{proofsketch}
	Fix any statement $(A,B)$, first message $(J_1,J_2)$, and challenge $\bar e$.
	If an accepting response $z$ exists, it must satisfy $z\G=J_1+\bar e\cdot A$.
	Since $\G$ generates a prime-order group, the map $z\mapsto z\G$ is a bijection on $\mathbb{Z}_p$,
	so $z$ is uniquely determined. Hence, there is at most one accepting response.\qed
\end{proofsketch}

\paragraph{UC-NIZK-AoK for DLEQ from Fischlin in \gROCRP.}
Applying the (optimized) Fischlin transform to the above Chaum--Pedersen $\mathrm{\Sigma}$-protocol yields
a UC-NIZK-AoK for $\cR_{\mathsf{DLEQ}}$ in the \gROCRP\ model.
We denote it by $\mathrm{\Pi}_{\mathsf{DLEQ}}=(K_{\mathsf{DLEQ}},\cp_{\mathsf{DLEQ}},\cV_{\mathsf{DLEQ}}),$
with a straight-line extractor $\Ext_{\mathsf{DLEQ}}$ that may inspect the adversary's \gROCRP\
query/answer log under the DLEQ-proof context(s).

\paragraph{UC-NIZK-AoK for DL from Fischlin in \gROCRP.}
Applying the (optimized) Fischlin transform to the above Schnorr $\mathrm{\Sigma}$-protocol yields a
UC-NIZK-AoK for $\cR_{\mathsf{DL}}$ in the \gROCRP\ model.
We denote this system by $\mathrm{\Pi}_{\mathsf{DL}}=(K_{\mathsf{DL}},\cp_{\mathsf{DL}},\cV_{\mathsf{DL}})$ and summarize its properties (which also apply to $\mathrm{\Pi}_{\mathsf{DLEQ}})$:
\begin{itemize}
	\item Completeness: honest proofs produced using witness $m$ verify.
	\item (Computational) ZK: real proofs are computationally indistinguishable from simulated proofs.
	\item Simulation-extractability (fresh statement) in the sense of Definition~\ref{def:simext} \cite[Thm. 3]{Fischlin[05]}.
	\item Knowledge soundness / extraction: there exists a straight-line PPT extractor $\Ext_{\mathsf{DL}}$ such that for any PPT
	adversarial prover $\cp^\ast$ that outputs $(M,\pi)$ with $\cV_{\mathsf{DL}}(\pp,M,\pi_\DL)=1$, the extractor outputs
	$m\leftarrow \Ext_{\mathsf{DL}}(\pp,M,\pi;\mathsf{Log}_{\cp^\ast})$ satisfying $M=m\G$, except with negligible probability,
	where $\mathsf{Log}_{\cp^\ast}$ is the prover's logged \gROCRP\ query/answer transcript under the proof context.
\end{itemize}

Our use of Fischlin-style straight-line compilation of DL-based $\mathrm{\Sigma}$-protocols into UC‑NIZK-PoK/AoK is closely aligned with \cite{Anna[25]}, which targets adaptive UC security in global random-oracle models; we further tailor the model/usage to gRO‑CRP and KeyBox local-call semantics.

\subsection{Affine DL relation}
Let $\Open_M(\pp,C,\zeta)$ denote the public opening derived from a USV certificate as in Definition~\ref{def:USV-cert}. Fix $\gamma\in\mathbb{Z}_p$ and public points $X,M,B,\mathrm{\Delta}\in\mathbb G$. Define affine DL NP relation 
\[
\cR_{\mathsf{aff}} :=
\Big\{ \big((X,\gamma,M,B,\mathrm{\Delta}),(\alpha,\delta)\big)\ :\
\alpha\G = M+\gamma B\ \wedge\ \delta\G = X-\mathrm{\Delta}-(M+\gamma B) \Big\}.
\]
Equivalently, define $Y:=M+\gamma B$ and $D:=X-\mathrm{\Delta}-Y$. Then $(\alpha,\delta)$ is a witness iff
$Y=\alpha\G$ and $D=\delta\G$, where $Y$ and $D$ are deterministically derived from the public tuple $(X,\gamma,M,B,\mathrm{\Delta})$. 

\begin{remark}[Realizing a UC-NIZK-AoK for $\cR_{\mathsf{aff}}$ from $\mathrm{\Pi}_{\mathsf{DL}}$]
	Let $\mathrm{\Pi}_{\mathsf{DL}}$ be the UC-NIZK-AoK for $\cR_{\mathsf{DL}}$.
	A UC-NIZK-AoK for $\cR_{\mathsf{aff}}$ is obtained by parallel composition of two instances of $\mathrm{\Pi}_{\mathsf{DL}}$:
	\begin{itemize}
		\item prove knowledge of $\alpha$ such that $Y=M+\gamma B=\alpha \G$;
		\item prove knowledge of $\delta$ such that $D=X-\mathrm{\Delta}-(M+\gamma B)=\delta \G$.
	\end{itemize}
	Concretely, the proof is $\pi_\aff:=(\pi_Y,\pi_D)$ where
	$\pi_Y\leftarrow \cp_{\mathsf{DL}}(\pp,Y,\alpha)$ and $\pi_D\leftarrow \cp_{\mathsf{DL}}(\pp,D,\delta)$.
	Verification of $\pi_\aff$ checks both $\pi_Y \text{ and }\pi_D$. The corresponding extractor outputs $(\alpha,\delta)$ by running the two
	$\mathrm{\Pi}_{\mathsf{DL}}$ extractors.
\end{remark}

As noted earlier, we instantiate UC-NIZK(-AoK)s via the Fischlin transform in the \gROCRP\ model. For efficiency we assume the optimized prover/verifier organization and batch-verification techniques of \cite{Lindell[24]}. Accordingly, we omit low-level optimization details and refer to \cite{Lindell[24]} for the optimized algorithms, while relying on \cite{Fischlin[05]} for the core transform security proof.

\section{Star DKG (SDKG)}\label{SDKG}
Fix Fischlin parameter functions $t=t(\lambda),b=b(\lambda),r=r(\lambda),S=S(\lambda)$ as in Definition~\ref{FSdef}.
Whenever the protocol refers to $t,b,r,S$, it means their values at the current security parameter $\lambda$. Let $H_{\mathrm{s32}}$ denote the \gROCRP\ $H$ under context $\CtxSDKG \in \CtxTEE$, i.e., $H_{\mathrm{s32}}(u):=H(\CtxSDKG,u)$. All invocations of $H_{\mathrm{s32}}$ are on an encoded tuple, i.e., $H_{\mathrm{s32}}(\langle\cdot\rangle)$.

For the sake of readability, we begin by providing our protocol, $\mathrm{\Psi}^{(3)}_{\mathsf{SDKG}}$, for $\mathrm{\Gamma}_0 = \left\{ \{P_1, P_2\}, \{P_1, P_3\} \right\}$ and then extend it to a 1+1-out-of-$n$ SDKG scheme. Hence, we first present $\mathrm{\Psi}^{(3)}_{\mathsf{SDKG}}$ in the
$(\FKeyBox,\Fusv,\Fchan,\Fpub)$-hybrid and \gROCRP\ models. Accordingly, the protocol description contains explicit invocations of the ideal certificate functionality $\Fusv$.
The corresponding real-world protocol $\widehat{\mathrm{\Psi}}^{(3)}_{\mathsf{SDKG}}$ is obtained by replacing each call to $\Fusv$ with an execution of its concrete realization $\mathrm{\Pi}_{\mathsf{USV}}$ (Corollary~\ref{cor:compile-out-fusv}). In Section \ref{subsec:SDKG-UC}, Theorem~\ref{thm:SDKG-UC} establishes UC security of the hybrid protocol $\mathrm{\Psi}^{(3)}_{\mathsf{SDKG}}$; Corollary~\ref{cor:compile-out-fusv} then compiles out $\Fusv$ via UC composition to obtain UC security of $\widehat{\mathrm{\Psi}}^{(3)}_{\mathsf{SDKG}}$ in the $(\FKeyBox,\Fchan,\Fpub)$-hybrid (and \gROCRP) model.

\begin{remark}
	We use three proof contexts in $\CtxUC$ (all mutually disjoint) for domain separation:
	\begin{itemize}
		\item $\mathrm{\Pi}^{\mathsf{UC}}_{\mathrm{DL}}$: UC-NIZK-AoK for DL-based consistency statements,
		instantiated under $\mathsf{ctx}_{\mathsf{UC}}\in\CtxUC$. This is the only proof type from which the
		UC proof extracts witnesses.
		\item $\mathrm{\Pi}^{\mathsf{KeyBox}}_{\mathrm{DL}}$: the sealed one-shot DL prover inside a KeyBox,
		instantiated under the disjoint context $\mathsf{ctx}_{\mathsf{KeyBox}}\in\CtxUC$. These proofs are
		treated under ZK/simulation only; the UC proof never extracts from them.
		\item $\mathrm{\Pi}_{\mathsf{DLEQ}}$: the DLEQ NIZK-AoK embedded in USV certificates, instantiated under
		the disjoint context $\mathsf{ctx}_{\mathsf{DLEQ}}\in\CtxUC$.
	\end{itemize}
	In the UC simulations for $\Fusv$ and SDKG, the simulator never programs
	$\mathsf{ctx}_{\mathsf{DLEQ}}$.
\end{remark}

\begin{figure}[t]
	\centering
	\setlength{\fboxrule}{0.2pt} 
	\fbox{
		\parbox{\dimexpr\linewidth-2\fboxsep-2\fboxrule\relax}{%
			            \ding{169}\ \textsf{Fixed Fischlin parameters:}
			Let $(t,b,r,S):=(t(\lambda),b(\lambda),r(\lambda),S(\lambda))$ be the canonical parameter tuple
			fixed for this KeyBox profile (Definition~\ref{FSdef}). Callers do not supply (and cannot influence) these values.
			
			\smallskip
			\ding{169}\ \textsf{Static seed:}
			Let $\seed\in\{0,1\}^\lambda$ be the KeyBox-resident static PRF seed
			(Assumption~\ref{assump:seed-integrity}), provisioned independently
			of any resident signing share.
			Let $\PRF:\{0,1\}^\lambda\times\{0,1\}^*\to\{0,1\}^{\lceil\log_2 p\rceil+\lambda}$
			be a secure pseudorandom function.
			
			\smallskip
			\ding{169} \textsf{FS.Start}$(\mathsf{sid},K)$:
			Check $K = \Pub(k)$; allocate a fresh handle $\muFS\in\mathsf{slot}_{fs}$, where $\mathsf{slot}_{fs} \in \{0,1\}^\lambda$.
			For each $i\in[r]$, derive the prover nonce deterministically:
			\[
			j_i\ \coloneqq\ \PRF\!\bigl(\seed,\,
			\langle \texttt{"LinOS/nonce/v1"},\,\mathsf{sid},\,K,\,i\rangle\bigr)
			\bmod p.
			\]
			Set $J_i\coloneqq j_i\mathcal G$, $a_i\coloneqq J_i$ for $i\in[r]$; bind the tuple
			$(\muFS\mapsto(\mathsf{sid},K,\{j_i\}_{i\in[r]},\{a_i\}_{i\in[r]},t,b,r,S))$ and return $(\muFS,\mathbf a),$
			where $\mathbf a\coloneqq (a_1,\ldots,a_r)$.
			
			\smallskip
			\ding{169} \textsf{FS.Prove}$(\muFS)$:
			If $\muFS$ is sealed, return $\perp$. Else, for each $i\in[r]$, the KeyBox runs an early-break rarity search:
			initialize $\widehat s_i\gets 2^b-1$ and $(\widehat e_i,\widehat z_i)\gets(\perp,\perp)$.
			For $e=0,1,\ldots,2^t-1$ it internally computes
			\[
			z\gets j_i+e\,k \bmod p,\quad
			s\gets H_b\!\big(\mathsf{ctx}_{\mathsf{KeyBox}},\langle \mathsf{sid},K,\mathbf a,i,e,z\rangle\big).
			\]
			If $s=0$, it sets $(e_i,z_i)\gets(e,z)$ and breaks; otherwise, if $s \le \widehat s_i$ it updates
			$\widehat s_i\gets s$ and $(\widehat e_i,\widehat z_i)\gets(e,z)$.
			If the loop ends without $s=0$, it sets $(e_i,z_i)\gets(\widehat e_i,\widehat z_i)$.
			It discards all per-$e$ temporary values other than the selected $(e_i,z_i)$, outputs the optimized-Fischlin proof
			$\pi_\DL=((a_i,e_i,z_i))_{i=1}^r$, and seals $\muFS$. Subsequent calls on $\muFS$ return $\perp$.
			
			\smallskip
			\ding{169} \textsf{FS.Verify}$(\mathsf{sid},K,\pi_\DL)$:
			Accept iff
			(i) $\forall i$, $z_i \mathcal{G}=a_i+e_i K$, and
			(ii) $\sum_{i=1}^r H_b(\mathsf{ctx}_{\mathsf{KeyBox}},\langle \mathsf{sid},K,\mathbf a,i,e_i,z_i\rangle)\le S.$
	}}
	\caption{LinOS-Fischlin API with profile-fixed parameters.}
	\label{LinOS}
\end{figure}

\begin{definition}[PRF for nonce derivation]\label{def:linos-prf}
	\emph{Let $\PRF:\{0,1\}^\lambda\times\{0,1\}^*\to\{0,1\}^{\lceil\log_2 p\rceil+\lambda}$
		be a keyed function family. We say $\PRF$ is a secure pseudorandom
			function if for every PPT distinguisher~$\cal D$ that makes at most
		$Q=\poly(\lambda)$ adaptive queries,
		\[
		\left|\Pr_{K\leftarrowdollar\{0,1\}^\lambda}\!\bigl[\mathcal{D}^{\PRF(K,\cdot)}(1^\lambda)=1\bigr] -
		\Pr_{F\leftarrowdollar \mathsf{Func}}\!\bigl[\mathcal{D}^{F(\cdot)}(1^\lambda)=1\bigr]\right|
		\;\le\; \negl(\lambda),
		\]
		where $\mathsf{Func}$ is the set of all functions from $\{0,1\}^*$
		to $\{0,1\}^{\lceil\log_2 p\rceil+\lambda}$.}
\end{definition}

Fig. \ref{LinOS} defines LinOS (Linear One-Shot), a sealed,
handle-bound prover for Schnorr DL that runs the optimized Fischlin transform inside the KeyBox. In LinOS, the PRF is keyed by the static seed
$\seed$ (Assumption~\ref{assump:seed-integrity}), and the derived
nonce is reduced modulo~$p$:
$$j_i := \PRF(\seed,\langle\texttt{"LinOS/nonce/v1"},\sid,K,i\rangle)\bmod p.$$

\begin{remark}[Injective encoding and domain separation for LinOS nonces]
	\label{rem:linos-encode}
	The PRF input in Fig.~\ref{LinOS} is the injective tuple encoding
	$\langle\texttt{"LinOS/nonce/v1"},\sid,K,i\rangle$ under the paper's
	fixed self-delimiting encoding $\langle\cdot\rangle$
	(Section~\ref{subsec:nxk-notation}).
	Because $\langle\cdot\rangle$ is injective and self-delimiting,
	distinct tuples $(\sid,K,i)\neq(\sid',K',i')$ map to distinct PRF
	inputs.  The domain-separation tag \texttt{"LinOS/nonce/v1"} prevents
	collisions with any other use of $\seed$ within the KeyBox profile.
	
	\smallskip\noindent
	Session-identifier requirement:
	The rollback-robustness argument (Lemma~\ref{lem:rollback-robust})
	relies on the property that distinct signing/proving instances use
	distinct $\sid$ values.  In the SDKG protocol of this paper, each
	base-run session has a unique $\sid$ by construction, and LinOS is
	invoked at most once per $(\sid,K)$ pair.
	In LinOS, the statement being proved is $K=\Pub(k)$ (the public key
	bound to the resident share), and $K$ is already part of the PRF
	input.  Therefore, for the uses in this paper, distinct signing
	instances always produce distinct PRF inputs, and no additional
	message binding beyond $(\sid,K)$ is needed.
	
	\smallskip\noindent
	If LinOS were adapted for a setting in which the same $(\sid,K)$
	could be used to produce proofs for different external messages~$m$
	(e.g., as part of a threshold signing protocol), then the PRF input
	must additionally include a commitment to~$m$---for instance,
	$\langle\texttt{"LinOS/nonce/v1"},\sid,K,i,H(m)\rangle$---to prevent
	nonce reuse across different messages.  This extension is
	straightforward but is not required by our SDKG construction.
\end{remark}

In contrast to a na\"{\i}ve design that
samples prover nonces $j_i\leftarrowdollar\Zp$ and stores them in
mutable KeyBox state, LinOS derives each nonce deterministically as
$j_i:=\PRF(\seed,\langle\texttt{"LinOS/nonce/v1"},\sid,K,i\rangle)\bmod p$
from a static KeyBox-resident seed~$\seed$
(Assumption~\ref{assump:seed-integrity}) and session-bound inputs.
This deterministic derivation follows the design principle of
RFC~6979~\cite{RFC6979} and EdDSA~\cite{RFC8032} (deterministic
nonces for (EC)DSA/Schnorr signatures) and of resettable
zero-knowledge~\cite{CGGM00}, adapted to the Fischlin-transform
setting. LinOS ensures:
(i) all per-challenge linear evaluations $j_i+e k$ remain internal; (ii) at most one $(e_i,z_i)$ ever leaves the KeyBox per commitment $a_i$, enforced by sealing and state continuity; and (iii) all rare-structure oracle queries are issued internally under a dedicated \gROCRP\ context.
In the UC proof, LinOS proofs are treated only as $\mathrm{\Pi}^{\mathsf{KeyBox}}_{\mathsf{DL}}$ proofs, whereas all protocol-consistency proofs
that require straight-line extraction use $\mathrm{\Pi}^{\mathsf{UC}}_{\mathsf{DL}}$ under a disjoint context.
For corrupted parties, the simulator mediates \gROCRP\ access. We get the following security semantics:
\begin{itemize}
	\item The statement $K$ is bound at \textsf{FS.Start} to the KeyBox’s resident key (share); the host cannot request a proof for any other statement.
	\item For each commitment $a_i$ produced by \textsf{FS.Start}, at most one $(e_i,z_i)$ ever leaves the KeyBox (enforced by sealing), preventing an external party from obtaining two responses to the same $a_i$.
	\item No intermediate $z_{i,e}$ is exposed beyond the selected $(e_i,z_i)$; hence an adversary cannot solve for $k$ by computing the difference of two responses with the same $a_i$.
\item The ``one-shot'' sealing of $\muFS$ assumes state continuity
(Assumption~\ref{assump:tee-continuity}).
Even if the one-shot seal is circumvented by a rollback of
mutable KeyBox state, the deterministic nonce derivation
(Lemma~\ref{lem:rollback-robust}) ensures that replaying the
same $(\sid,K)$ regenerates the same nonces $j_i$ and hence the
same commitments and (given the same \gROCRP\ oracle answers) the
same proof transcript---no new information is obtained.
For distinct sessions with different~$\sid$, the derived nonces
are pseudorandomly independent under PRF security, so the
probability of a commitment collision is negligible.
Consequently, the catastrophic ``one rollback reveals~$k$'' path
via Schnorr special soundness is eliminated, and key-opacity is preserved even
under approximate state continuity.
This hardening follows the established principle of deterministic
nonce generation~\cite{RFC6979,RFC8032} and resettable
ZK~\cite{CGGM00}.
Note that this does not fully close the
$\varepsilon_{\mathsf{sc}}$ gap for all rollback-affected state;
see Remark~\ref{rem:rollback-scope} for the precise scope.
	\item All \gROCRP\ oracle calls $H_b(\mathsf{ctx},\cdot)$ made during \textsf{FS.Prove} are local to the KeyBox ITM and are not visible at the host/API boundary. The host observes only $(\mathbf a,\pi_\DL)$ and the public verification outcome.
\end{itemize}

When LinOS is run on a KeyBox-resident share stored in a local KeyBox slot $\mu\in\SlotSpace$
(with public key $K=\Pub(k)$), the interfaces are invoked via the KeyBox API:
\[
(\muFS,\mathbf a)\leftarrow \FKeyBox^{(P)}.\Use(\mu,\textsf{FS.Start},\langle \sid,K\rangle),
\qquad
\pi_{\DL}\leftarrow \FKeyBox^{(P)}.\Use(\mu,\textsf{FS.Prove},\muFS).
\]
The verifier-side predicate $\textsf{FS.Verify}(\sid,K,\pi_{\DL})$ is public and is evaluated outside the KeyBox.

\begin{lemma}[LinOS preserves key-opacity]\label{lem:linos-key-opacity}
	Work in the \gROCRP-hybrid model. Assume the baseline KeyBox profile is key-opaque w.r.t.\ $\Pub$
	(Definition~\ref{def:key-opacity-alt}), and extend the admissible operations by adding
	\textsf{FS.Start}/\textsf{FS.Prove}/\textsf{FS.Verify} as in Fig.~\ref{LinOS}. Instantiate LinOS under a dedicated
	Fischlin proof context $\mathsf{ctx}_{\mathsf{KeyBox}}\in\CtxUC$ that is disjoint from all other contexts. Then the
	extended profile remains key-opaque w.r.t.\ $\Pub$.
\end{lemma}

\begin{proofsketch}
	Fix any KeyBox slot holding secret $k$ with public key $K=\Pub(k)$.
	Extend the slot-wise key-opacity simulator $\s_{K}$ (Remark~\ref{rem:key-opacity-multislot})
	to answer LinOS queries as follows.
	
\smallskip
\noindent Simulating \textsf{FS.Start}:
Maintain a cache $\mathsf T$ keyed by $(\mathsf{sid},K)$ whose entries store an accepting Fischlin transcript
$\pi_{\DL}$ and its first-message tuple $\mathbf a$.
On input $\textsf{FS.Start}(\mathsf{sid},K')$, if $K'\neq K$ output $\perp$.
Otherwise sample a fresh handle $\muFS$.
If $(\mathsf{sid},K)\in\mathsf T$, retrieve $(\mathbf a,\pi_{\DL})\gets \mathsf T[\mathsf{sid},K]$.
Else run the Fischlin-based UC-NIZK simulator for the DL statement
``$\exists k:\ K=\Pub(k)$'' in context $\mathsf{ctx}_{\mathsf{KeyBox}}$ to obtain an accepting proof
$\pi_{\DL}=((a_i,e_i,z_i))_{i=1}^r$ (with $\mathbf a:=(a_1,\ldots,a_r)$), and set
$\mathsf T[\mathsf{sid},K]\gets(\mathbf a,\pi_{\DL})$.
This simulation may invoke $\SimProgramRO$ only in $\mathsf{ctx}_{\mathsf{KeyBox}}$.
Return $(\muFS,\mathbf a)$ and store $\pi_{\DL}$ under $\muFS$.
	
	\smallskip
	\noindent Simulating \textsf{FS.Prove}:
	On input $\textsf{FS.Prove}(\muFS)$, if $\muFS$ is unknown or already sealed output $\perp$; otherwise output the stored
	$\pi_{\DL}$ and mark $\muFS$ sealed, matching LinOS one-shot semantics.
	
	\smallskip
	\noindent Simulating \textsf{FS.Verify}:
	On input $\textsf{FS.Verify}(\mathsf{sid},K',\pi)$, output the deterministic verification result, which is simulatable
	without $k$.
	
\smallskip
\noindent Indistinguishability:
In the real LinOS execution the prover nonces $j_1,\ldots,j_r$ are derived deterministically via
$\PRF(\seed,\cdot)$ (Fig.~\ref{LinOS}). In particular, for fixed $(\mathsf{sid},K)$ the first message $\mathbf a$
(and thus the full transcript released by \textsf{FS.Prove}) repeats across repeated calls.
We argue indistinguishability in two steps:\\
\emph{Step~(a):} Replace all $\PRF(\seed,\cdot)$ evaluations with outputs of a truly random
function~$F$. By PRF security (Definition~\ref{def:linos-prf}) and secrecy of $\seed$
(Assumption~\ref{assump:seed-integrity}), this introduces at most $\negl(\lambda)$ distinguishing advantage.
After this replacement, for each fixed $(\mathsf{sid},K)$ the tuple $(j_1,\ldots,j_r)$ is (up to the negligible
statistical distance of mod-$p$ reduction) uniform over~$\Zp^r$ and independent across distinct inputs
$\langle \texttt{"LinOS/nonce/v1"},\mathsf{sid},K,i\rangle$; repeated evaluations on the same input repeat the same
nonce, matching the real correlation across repeated \textsf{FS.Start} queries.\\
\emph{Step~(b):} In the resulting hybrid, for each distinct pair $(\mathsf{sid},K)$ queried, the externally visible
output is a Fischlin-based UC-NIZK transcript for the DL statement $K=\Pub(k)$ under context
$\mathsf{ctx}_{\mathsf{KeyBox}}$ with truly random prover nonces (and replay on repeats of $(\mathsf{sid},K)$).
By the computational ZK guarantee of the Fischlin transform in \gROCRP\ (Lemma~\ref{lem:gnpro-uc-nizk}),
we can replace each such real transcript with a simulated one generated by the Fischlin simulator, while preserving
replay behavior by caching the simulator's output per $(\mathsf{sid},K)$ as above (up to the simulator's programming
failure event). Combining Steps~(a) and~(b) yields the claim.
	
	\smallskip
	\noindent Simulator failure event:
	The only simulator failure mode is that a $\SimProgramRO(\mathsf{ctx}_{\mathsf{KeyBox}},x,y)$ attempt returns $\perp$
	because the host/adversary pre-queried $(\mathsf{ctx}_{\mathsf{KeyBox}},x)$. By Lemma~\ref{lem:grocrp-prequery},
	this pre-query event is negligible because the programmed inputs $x$ include fresh high-entropy components. Finally,
	by \gROCRP\ local-call semantics (Fig.~\ref{fig:GgROCRP}), the host/API boundary never observes the KeyBox's internal
	oracle-query trace; only the released transcript must be simulated. Hence, adding LinOS preserves key-opacity. \qed
\end{proofsketch}

\begin{lemma}[Rollback robustness of LinOS commitments]
	\label{lem:rollback-robust}
	Work in the \gROCRP-hybrid model.
	Consider the LinOS interface (Fig.~\ref{LinOS}) with deterministic
	nonce derivation via $\PRF$ keyed by $\seed$
	(Definition~\ref{def:linos-prf}).
	Assume that $\PRF$ is a secure pseudorandom function
	(Definition~\ref{def:linos-prf}), and that $\seed$ satisfies the
	seed integrity invariant (Assumption~\ref{assump:seed-integrity}).
	Let $\cA$ be a PPT adversary that may:
	\begin{enumerate}[leftmargin=*,nosep]
		\item invoke $\textsf{FS.Start}(\sid,K)$ and
		$\textsf{FS.Prove}(\muFS)$ via the KeyBox API;
		\item roll back or reset the mutable operation state
		$\mathsf{st}[\cdot]$ of the KeyBox to any prior value
		(violating state continuity for mutable state, but
		not altering $\seed$ or the resident key~$k$);
		\item make arbitrary \gROCRP\ queries.
	\end{enumerate}
	Then $\cA$ cannot obtain two accepting Fischlin transcripts
	$\pi=(a_i,e_i,z_i)_{i=1}^r$ and $\pi'=(a_i',e_i',z_i')_{i=1}^r$
	for the same $(\sid,K)$ with $a_j=a_j'$ and $e_j\neq e_j'$ for
	some~$j$, except with probability negligible in~$\lambda$.
\end{lemma}

\begin{proofsketch}
	We argue via a two-step hybrid.
	
	\smallskip\noindent
	Step 1:
	Suppose $\cA$ rolls back the KeyBox's mutable state and re-invokes
	$\textsf{FS.Start}(\sid,K)$ with the same $(\sid,K)$.
	Because $\seed$ is rollback-invariant
	(Assumption~\ref{assump:seed-integrity}) and $k$ is unchanged, the
	PRF derivation $$j_i=\PRF(\seed,\langle\texttt{"LinOS/nonce/v1"},\sid,K,i\rangle)\bmod p$$
	produces the same nonces $j_1,\ldots,j_r$ and hence the same
	commitments $a_1,\ldots,a_r$.
	In $\textsf{FS.Prove}$, the rarity-search queries
	$H_b(\mathsf{ctx}_{\mathsf{KeyBox}},\langle\sid,K,\mathbf a,i,e,z_{i,e}\rangle)$
	are deterministic functions of $(j_i,k,H)$; since $j_i$ and $k$ are
	unchanged, and $H$ is a (lazy-sampled) global random oracle whose
	table is not affected by KeyBox rollback, the rarity search produces
	the same selected $(e_i,z_i)$ for each~$i$.
	Therefore, the replayed proof transcript is identical to the
	original---no new information is obtained.
	
	\smallskip\noindent
	Step 2:
	Suppose $\cA$ invokes $\textsf{FS.Start}(\sid',K)$ with $\sid'\neq\sid$.
	We show that the commitments $a_i'=j_i'\G$ are independent of
	$a_i=j_i\G$ by a PRF-to-random replacement.
	Consider the hybrid experiment $\Game_1$ in which all PRF outputs
	$\PRF(\seed,\cdot)$ are replaced by outputs of a truly random
	function~$F$.
	By PRF security (Definition~\ref{def:linos-prf}) and secrecy of
	$\seed$ (Assumption~\ref{assump:seed-integrity}), $\Game_0\approx_c\Game_1$.
	In $\Game_1$, because $\langle\cdot\rangle$ is injective,
	$\langle\texttt{"LinOS/nonce/v1"},\sid,K,i\rangle\neq
	\langle\texttt{"LinOS/nonce/v1"},\sid',K,i'\rangle$ whenever
	$(\sid,i)\neq(\sid',i')$, so the derived nonces are independent
	uniform values.
	Hence, the probability that $a_j=a_j'$ for any~$j$ (i.e., a
	commitment collision $j_j\G=j_j'\G$) is at most $r^2/p=\negl(\lambda)$.
	Without a shared commitment, Schnorr special-soundness extraction
	cannot be applied, so the adversary cannot recover~$k$.
	
	\smallskip\noindent
	In either case (same or different $\sid$),
	$\cA$ cannot obtain two transcripts sharing a commitment $a_j$ with
	distinct challenges $e_j\neq e_j'$ except with negligible probability.
	Therefore, the Schnorr witness-extraction attack via nonce reuse is
	thwarted. \qed
\end{proofsketch}

Lemma~\ref{lem:rollback-robust} rules out Schnorr witness extraction on the resident share $k$ via rollback-induced nonce reuse, except with negligible probability.

\begin{remark}[Fischlin-transform compatibility]
	\label{rem:fischlin-compat}
	Because the repetition index~$i$ is an explicit input to the PRF
	derivation, the nonces $j_1,\ldots,j_r$ within a single session
	remain mutually independent (under PRF security).
	Therefore, the Fischlin rarity-search procedure is unchanged: each
	repetition~$i$ independently searches over challenges
	$e\in\{0,\ldots,2^t-1\}$ using its own independently derived $j_i$.
	The completeness, soundness, and knowledge-extraction analyses of
	the Fischlin transform (Definition~\ref{FSdef},
	Lemma~\ref{lem:fischlin-negl}) apply without modification.
	Only the \emph{source} of prover randomness changes---from
	``sample uniformly and store in mutable state'' to ``derive on demand
	from a static seed''---while the statistical properties required by
	the transform (near-uniform, mutually independent nonces) are
	preserved.
\end{remark}

\begin{remark}[Scope of the rollback-robustness hardening]
	\label{rem:rollback-scope}
	Lemma~\ref{lem:rollback-robust} is a surgical hardening
	that removes one specific catastrophic failure mode: the path from
	``single KeyBox rollback'' to ``full secret-key disclosure'' via
	Schnorr nonce reuse / special-soundness extraction.
	It does not, by itself, close the broader
	$\varepsilon_{\mathsf{sc}}$ gap (Definition~\ref{def:bad-sc}) or
	address all rollback/state-continuity issues for other protocol
	state (e.g., monotonic counters, sealed records, session metadata,
	or the one-shot sealing flag~$\muFS$).
	A rollback that restores the mutable one-shot seal could still allow
	the adversary to obtain a second proof transcript; however,
	by Lemma~\ref{lem:rollback-robust}, this second transcript is
	identical to the first, and therefore reveals no new information about~$k$.
	General state-continuity remains an engineering requirement for
	other protocol-level invariants (e.g., preventing replay of
	registration or re-execution of protocol rounds), and
	$\varepsilon_{\mathsf{sc}}$ continues to parameterize those
	residual failure modes as before.
\end{remark}

\begin{remark}[Scope of un-hardened randomness inside \textsf{USV.Cert} and \textsf{SDKG.LeafInit}]
	\label{rem:unhardened-cert-nonces}
	The PRF-based nonce derivation of LinOS (Fig.~\ref{LinOS}) is applied specifically to the interface whose
	rollback failure mode is catastrophic: reuse of a Schnorr commitment across two distinct challenges enables
	special-soundness extraction of the long-term resident share~$k$. Two other KeyBox-internal randomized operations are not hardened in this way and remain covered by the general	$\varepsilon_{\mathsf{sc}}$ accounting (Definition~\ref{def:bad-sc}):
	\begin{enumerate}[leftmargin=*,nosep]
		\item DLEQ prover randomness inside \textsf{USV.Cert} (Section~\ref{subsec:USV-inst}):
		The USV certificate generator $\Cert$ invokes the UC-NIZK prover $\cp_{\mathsf{DLEQ}}$ for the Chaum--Pedersen
		relation using fresh randomness sampled from the KeyBox DRBG (both for the witness~$r$ and for the internal Fischlin
		nonces of $\cp_{\mathsf{DLEQ}}$). A rollback that resets the DRBG can therefore cause reuse or replay of this internal
		randomness, resulting in repeated certificates and/or in multiple certificates being generated from the same
		pre-rollback state.
		In our Fischlin-based instantiation (Definition~\ref{FSdef}), re-executing $\cp_{\mathsf{DLEQ}}$ with the same
		randomness yields an identical transcript. However, in any alternative deployment in which the DLEQ prover could output two accepting Chaum--Pedersen transcripts
		for the same statement and first message but with different challenges (e.g., via an interactive DLEQ interface under
		rollback), special soundness would recover the witness~$r$; combined with the public~$\nu$ in~$\zeta$, this recovers
		$m=\nu r$. This does not by itself expose the long-term signing share~$k$, but may matter for applications that require
		confidentiality of ephemeral USV witnesses.
		
		\item Ephemeral sampling in \textsf{SDKG.LeafInit}:
		The key-independent leaf routine samples $m_2\leftarrowdollar \Zp^*$ and $b_2\leftarrow \Zps$ using ordinary KeyBox
		DRBG randomness. A rollback that resets the DRBG and re-invokes \textsf{SDKG.LeafInit} can produce different
		$(m_2,b_2)$ and hence a different Round~1 tuple, leading to inconsistent protocol state if the rest of the session
		proceeds using the original values. As above, this is subsumed by $\varepsilon_{\mathsf{sc}}$ and does not by itself
		disclose~$k$.
	\end{enumerate}
	In summary, the PRF-based hardening is targeted at the single interface whose rollback failure mode is catastrophic, while other KeyBox-internal randomized operations remain under the general $\varepsilon_{\mathsf{sc}}$ accounting as before.
\end{remark}

For the star access structure $\mathrm{\Gamma}_0 := \bigl\{\{P_1,P_2\},\{P_1,P_3\}\bigr\}$, our end-to-end claim is UC realization of an \emph{NXK-star DKG} interface for $\mathrm{\Gamma}_0$:
applications should see only a public key (or abort) and the induced installation of
non-exportable long-term shares inside the parties' KeyBoxes.

\begin{figure}[t]
	\centering
	\setlength{\fboxrule}{0.2pt}
	\fbox{
		\parbox{\dimexpr\linewidth-2\fboxsep-2\fboxrule\relax}{%
			\ding{169}\ \textsf{SDKG.LeafInit}$(\sid)$ \textsf{(key-independent)}:
			sample $m_2\leftarrowdollar \Zp^*$ and $b_2\leftarrow \Zps$; define $f_2(x):=m_2+b_2x$.
			Compute $B_2:=b_2\G$ and
			$\sigma_{2,1}:=f_2(2)$, $\sigma_{2,2}:=f_2(3)$, $\sigma_{2,3}:=f_2(1)$ in $\Zp$.
			Compute $(C_2,\zeta_2)\leftarrow \Cert(\pp,m_2)$.
			Erase $(m_2,b_2,f_2)$ and return $(C_2,\zeta_2,B_2,\sigma_{2,1},\sigma_{2,2},\sigma_{2,3})$.}}
	\caption{Key-independent leaf routine for hardened/minimal NXK deployments of SDKG.}
	\label{fig:SDKGLeafInit}
\end{figure}

\begin{readerbox}[reader: app-vis]{Application-visible functionality}
	The functionality exposed to applications is $\FDKGstarNXK$ (Definition~\ref{def:fdkg-star-nxk}), which is obtained by wrapping the proof-oriented transcript-driven functionality $\FSDKG$ (Fig. \ref{fig:FSDKG}) as $\FDKGstarNXK = \WDKG \circ \FSDKG.$ We prove UC realization for $\FSDKG$ and then lift it to $\FDKGstarNXK$ via interface restriction (Lemma \ref{lem:sdkg-interface-refinement}). $\FDKGstarNXK$ exposes only:
	\begin{itemize}[leftmargin=*,nosep]
		\item DKG output: upon successful completion, make a public key $K\in\mathbb G$ available to all parties ($P_1$ and $P_2$ receive $K$ directly at finalization; $P_3$ receives $K$ via the real protocol's $\Fpub$ publications---preserved under UC emulation---and/or the registration payload from $P_1$); as usual in UC-DKG, no output is provided on abort or selective-abort.
		\item NXK share installation: install long-term shares into the parties' local KeyBoxes (via $\FKeyBox$),
		with no export interface for the underlying scalars or caller-invertible affine images.
		\item Optional post-finalization registration (RDR): if invoked, output the corresponding
		$\mathsf{registered}$ event(s) and install the recovery-role share in the joining device's KeyBox.
	\end{itemize}
	All transcript bookkeeping and the simulator-only \textsf{Program} hook in $\FSDKG$
	are not part of the application interface and are hidden by the wrapper $\WDKG$ (Definition~\ref{def:fdkg-star-nxk}).
\end{readerbox}

We prove UC realization first for the proof-oriented functionality $\FSDKG$
(Fig.~\ref{fig:FSDKG}), and then obtain the application-visible interface $\FDKGstarNXK$
by local interface restriction (Lemma~\ref{lem:sdkg-interface-refinement}). In the protocol and definition below, $\sigma_{i,j}$ denote the scalar polynomial evaluations for SDKG and KeyBox installation/registration, respectively.

\begin{definition}[SDKG KeyBox derivation routines]\label{def:sdkg-deriv}
	\emph{Fix public parameters $\pp=(\mathbb{G},p,\G,\h)$ and let $\langle\cdot\rangle$ be a tuple encoding. All arithmetic below is in $\mathbb{Z}_p$.
		\begin{itemize}
			\item $g_{1,2}(1^\lambda,l)$:
			parse $l=\langle \sigma_{1,1},\sigma_{2,1},\sigma_{3,1}\rangle$ with each $\sigma_{i,j}\in\mathbb{Z}_p$.
			If parsing fails, output $\perp$.
			Set $x_1 := \sigma_{1,1}+\sigma_{2,1}+\sigma_{3,1}$ and output $k_{1,2} := 3x_1.$
			\item $g_{1,3}(1^\lambda,l)$:
			parse $l=\langle \sigma_{1,1},\sigma_{2,1},\sigma_{3,1},\sigma_{1,3}\rangle$ with each component in $\mathbb{Z}_p$.
			If parsing fails, output $\perp$.
			Set $x_1 := \sigma_{1,1}+\sigma_{2,1}+\sigma_{3,1}$ and output $k_{1,3} := 2\sigma_{1,3}-x_1.$
			\item $g_{2}(1^\lambda,l)$:
			parse $l=\langle \sigma_{1,2},\sigma_{2,2},\sigma_{3,2}\rangle$ with each component in $\mathbb{Z}_p$.
			If parsing fails, output $\perp$.
			Set $x_2 := \sigma_{1,2}+\sigma_{2,2}+\sigma_{3,2}$ and output $k_{2} := -2x_2.$
			\item $g_{3}^{(P_r,P_{\mathsf{sp}},\sid)}(1^\lambda,l)$ (parameterized by joiner $P_r$, sponsor $P_{\mathsf{sp}}$, and session $\sid$):\\
			For the base case, set $P_r=P_3$ and $P_{\mathsf{sp}}=P_2$. For the $n$-party extension, $P_r=P_i$ and $P_{\mathsf{sp}}=P_j$ for
			arbitrary joiner $P_i$ ($i\ge 3$) and sponsor $P_j\in\{2\}\cup\mathsf{Reg}$.
			Parse $l$ as either
			\begin{enumerate}[label=(\alph*),nosep]
				\item $\langle \sigma_{2,3},\sigma_{3,2},\sigma_{3,1},K,K_{1,3}\rangle$
				with $\sigma_{2,3},\sigma_{3,2},\sigma_{3,1}\in\mathbb{Z}_p$ and $K,K_{1,3}\in\mathbb{G}$, or
				\item $\langle \langle \mathsf{ad}_{23},\sigma_{2,3}\rangle,\ \langle \mathsf{ad}_{32},\sigma_{3,2}\rangle,\ \langle \mathsf{ad}_{31},\sigma_{3,1}\rangle,\ K,\ K_{1,3}\rangle$
				where each $\mathsf{ad}_{\cdot}\in\{0,1\}^*$ and each $\sigma_{\cdot}\in\mathbb{Z}_p$.
				In case (b), additionally parse each $\mathsf{ad}_{xy}$ as
				\[
				\mathsf{ad}_{xy}=\langle \texttt{SDKG.reg},\sid',P_s',P_r',\texttt{tag}\rangle
				\]
				and require the associated-data fields are mutually consistent and match the closure parameters:
				all three $\sid'$ fields equal $\sid$, all three $P_r'$ fields equal $P_r$, and
				$P_s'=P_{\mathsf{sp}}$ with $\texttt{tag}=\texttt{k23}$ for $\mathsf{ad}_{23}$,
				$P_s'=P_{\mathsf{sp}}$ with $\texttt{tag}=\texttt{k32}$ for $\mathsf{ad}_{32}$,
				and $P_s'=P_1$ with $\texttt{tag}=\texttt{k31}$ for $\mathsf{ad}_{31}$.
			\end{enumerate}
If parsing fails, including the case-(b) associated-data checks, output $\perp$.
			Compute $k_3 := 2\cdot(\sigma_{2,3}+2\sigma_{3,1}-\sigma_{3,2}) \bmod p.$
			Let $K_3 := k_3\G$. If $K_{1,3}+K_3 \neq K$, output $\perp$, else output $k_3$.
			\item For each role pair $(a,b)\in\{(\mathsf{c},\mathsf{sp}),(\mathsf{sp},\mathsf{c2}),(\mathsf{sp2},\mathsf{c3})\}$
			and corresponding concrete base-case instantiation
			$(i,j)\in\{(3,1),(3,2),(2,3)\}$, define $g^{\textsf{reg},(P_r,P_{\mathsf{sp}},\sid)}_{i,j}(1^\lambda,l)$
			(parameterized by joiner $P_r$, sponsor $P_{\mathsf{sp}}$, session $\sid$) as:
			parse $l$ as either
			\begin{enumerate}[label=(\alph*),nosep]
				\item $\langle \sigma_{i,j}\rangle$ with $\sigma_{i,j}\in\mathbb{Z}_p$; in this case output $k^{\textsf{reg}}_{i,j}:=\sigma_{i,j}$, or
				\item $\langle \langle \mathsf{ad},\sigma_{i,j}\rangle\rangle$ with $\mathsf{ad}\in\{0,1\}^*$ and $\sigma_{i,j}\in\mathbb{Z}_p$.
				In this case, parse $\mathsf{ad}=\langle \texttt{SDKG.reg},\sid',P_s',P_r',\texttt{tag}\rangle$ and require the
				fields match the closure parameters:
				$\sid'=\sid$, $P_r'=P_r$, and
				$\texttt{tag}=\texttt{k31}$ with $P_s'=P_1$ if $(i,j)=(3,1)$,
				$\texttt{tag}=\texttt{k32}$ with $P_s'=P_{\mathsf{sp}}$ if $(i,j)=(3,2)$, and
				$\texttt{tag}=\texttt{k23}$ with $P_s'=P_{\mathsf{sp}}$ if $(i,j)=(2,3)$.\\
If the $\mathsf{ad}$ parse or associated-data check fails, output $\perp$; else output $k^{\textsf{reg}}_{i,j}:=\sigma_{i,j}$.
			\end{enumerate}
			If parsing fails, output $\perp$.
			\smallskip
			\noindent For the base case ($P_3$ joining with sponsor $P_2$), we write $g_3$ for $g_{3}^{(P_3,P_2,\sid)}$
			and $g^{\textsf{reg}}_{i,j}$ for $g^{\textsf{reg},(P_3,P_2,\sid)}_{i,j}$ when the closure parameters are clear from context.
	\end{itemize}}
\end{definition}

\begin{remark}[Associated-data fields vs.\ $g$-routine checks]
	\label{rem:sdkg-ad-fields}
	In SDKG registration we use AEAD associated data strings of the form
	$\mathsf{ad}=\langle \texttt{SDKG.reg},\sid,P_s,P_r,\texttt{tag}\rangle$.
	The AEAD layer (via $\SealToPeer/\OpenFromPeer$ in $\FKeyBox$) binds the entire tuple
	$(\sid,P_s,P_r,\texttt{tag})$ to the ciphertext: in Fig.~\ref{Fdskg}, $\OpenFromPeer$ computes
	$s\leftarrow \Dec_{\sk_{\mathrm{seal}}}(\mathsf{ad},c)$ and returns a handle that later resolves to
	$\langle \mathsf{ad},s\rangle$ only if the ciphertext authenticates under that exact $\mathsf{ad}$.
	
	\noindent Importantly, the protocol does not accept $\mathsf{ad}$ strings from the network.
	In Algorithm~\ref{alg:reg}, the receiver $P_3$ deterministically recomputes the intended
	$\mathsf{ad}_{31},\mathsf{ad}_{32},\mathsf{ad}_{23}$ from $(\sid,P_s,P_r,\texttt{tag})$ and supplies them to
	$\OpenFromPeer$. Therefore, when $P_3$ is honest, any cross-session/cross-peer mix-and-match attempt would require opening a ciphertext
	under an $\mathsf{ad}$ that does not match its sealing $\mathsf{ad}$, which causes $\OpenFromPeer$ to return $\perp$.
	Hence, our derivation routines $g_3$ and $g^{\textsf{reg}}_{i,j}$ verify that the parsed associated-data tuple
	$(\sid',P_s',P_r',\texttt{tag})$ matches the intended session/peers/slot. The \texttt{tag} component enforces \emph{slot disjointness}, while the extra
	$\sid,P_s,P_r$ checks provide interface-level robustness if a buggy/adversarial host supplies inconsistent associated data to
	$\OpenFromPeer$.
	
	\smallskip
	\noindent If the sponsor is corrupted during registration, it may send arbitrary ciphertexts. Any in-transit modification of
	a ciphertext (or cross-session/slot substitution) is rejected by $\OpenFromPeer$ due to AEAD integrity under the
	receiver-supplied associated data (Fig.~\ref{Fdskg} and Algorithm~\ref{alg:reg}).
	Moreover, even if a corrupted sponsor generates fresh ciphertexts under the correct associated data but encrypting
	incorrect scalars, the joiner’s recovery-share derivation $g_3$ uniquely determines $k_3$ via
	$K_{1,3}+K_3=K$ where $K_3=k_3\G$ (Definition~\ref{def:sdkg-deriv}). Because $g_3$ constrains only the linear
	combination $\sigma_{2,3}-\sigma_{3,2}$ (not the individual values), the sponsor may shift both slots by a common
	offset without affecting $k_3$.
	Hence, the sponsor can at most force an abort (denial of registration) or shift the per-device sponsor-state
	slots, but cannot cause installation of a recovery share inconsistent with~$K$.
\end{remark}

	Fix the KeyBox ideal functionality $\FKeyBox$ (Fig.~\ref{Fdskg}) and the LinOS interface (Fig.~\ref{LinOS}).
	We fix the admissible KeyBox profile
	$(\chi_{\mathrm{adm}}^{\mathsf{SDKG}},\mathcal F_{\mathrm{adm}}^{\mathsf{SDKG}},\mathcal F_{\mathrm{KI}}^{\mathsf{SDKG}},\digamma^{\mathsf{SDKG}})$
	(Definition~\ref{def:keybox-profile}) as follows: the only derivation routines invoked via $\Load$ are $	\chi_{\mathrm{adm}}^{\mathsf{SDKG}}
	:= \{ g_{1,2}, g_{1,3}, g_{2}, g_{3}, g_{3,1}^{\textsf{reg}}, g_{3,2}^{\textsf{reg}}, g_{2,3}^{\textsf{reg}} \}.$ The only operations invoked via $\Use$ are
\begin{align*}
&\mathcal{F}_{\mathrm{adm}}^{\mathsf{SDKG}}
:= \{ \GetPub,\ \textsf{SDKG.LeafInit},\ \textsf{USV.Cert},\ \textsf{FS.Start},\ \textsf{FS.Prove},\ \SealToPeer,\ \OpenFromPeer \},
\\
&\mathcal{F}_{\mathrm{KI}}^{\mathsf{SDKG}}
:= \{ \OpenFromPeer,\ \textsf{USV.Cert},\ \textsf{SDKG.LeafInit}\},
\end{align*}
	where $\textsf{FS.*}$ are as in Fig.~\ref{LinOS} and $\SealToPeer/\OpenFromPeer$ are the KeyBox-to-KeyBox sealing operations specified as part of $\FKeyBox$ in Fig.~\ref{Fdskg}. Recall that the predicate $\textsf{FS.Verify}(\sid,K,\pi)$ is public/deterministic and is evaluated outside the KeyBox. We fix the profile’s family map $\digamma^{\mathsf{SDKG}}$ by identifying the LinOS start/prove pair:
	\[
	\digamma^{\mathsf{SDKG}}(\textsf{FS.Start})=\digamma^{\mathsf{SDKG}}(\textsf{FS.Prove})=: \textsf{FS},
	\qquad
	\digamma^{\mathsf{SDKG}}(f)=f\ \text{for all other } f\in\mathcal{F}_{\mathrm{adm}}^{\mathsf{SDKG}}.
	\]
	
\begin{note}
	The profile intentionally omits any operation for producing the UC-extractable
	consistency AoKs (e.g., $\mathrm{\Pi}^{\mathsf{UC}}_{\mathsf{DL}}$ / $\cR_{\mathsf{aff}}$) inside the KeyBox:
	those proofs are generated by the host/party ITM so that the simulator can log the prover's
	\gROCRP\ queries for straight-line extraction.
\end{note}	

\subsection{The Protocol}
Our $\mathrm{\Psi}^{(3)}_{\mathsf{SDKG}}$ protocol is inspired by the DKG from \cite{Batta[22]}. Specifically, the high-level idea of mapping polynomial evaluations into compatible shares is shared by the two solutions. However, the settings and techniques differ: \cite{Batta[22]} constructs a $(2,3)$ threshold ECDSA scheme with an offline party in a standalone, game-based model, and relies on Feldman VSS, Paillier‑based MtA/MtAwc, and interactive ZK proofs over freely manipulable shares. By contrast, our scheme targets a star-shaped access structure in the NXK model. We obtain UC security in this setting by combining USV, NXK, and UC-NIZK-AoKs to replace the role of VSS and homomorphic share manipulations in \cite{Batta[22]}.

\paragraph{\textbf{Placement / profile convention.}}
The protocol steps below are written at the level of algebraic relations (e.g., ``sample $m_2$ and set
$f_2(x)=m_2+b_2x$'') and should not be read as fixing where a scalar is materialized.
In the minimal/hardened NXK profile that motivates USV
(Section~\ref{subsec:why-usv-straightline}),
the leaf’s Round~1 linear-polynomial setup is realized by the \emph{key-independent} KeyBox operation
$\textsf{SDKG.LeafInit}\in\mathcal F_{\mathrm{KI}}^{\mathsf{SDKG}}$ (Fig.~\ref{fig:SDKGLeafInit}), which samples
$m_2$ and $b_2$ inside the KeyBox boundary, outputs only $(C_2,\zeta_2,B_2,\sigma_{2,1},\sigma_{2,2},\sigma_{2,3})$,
and erases $(m_2,b_2)$ as named variables before returning. Note that $m_2$ remains trivially recoverable from the returned $\sigma$-values (e.g., $m_2=2\sigma_{2,3}-\sigma_{2,1}$); however, these values are NXK-restricted transient host state (Remark~\ref{rem:transport-vs-nxk}), not adversary-visible, and are securely erased after their last use. No generic ``export $m_2$'' or ``export $m_2\G$'' interface is assumed.
If, instead, one permits a transient-host-RAM deployment in which $m_2$ is sampled/held outside the KeyBox as
NXK-restricted material with secure erasure, then one may simplify Round~1 by publishing $M_2:=m_2\G$ directly and
omitting the USV instance; we keep USV so the same transcript format covers the hardened profile in which
the adversary-visible transcript must still deterministically fix $M_2$.

\smallskip
\noindent Even in the minimal (KeyBox-hardened long-term-share) profile, $\textsf{SDKG.LeafInit}$ returns
$\sigma_{2,1},\sigma_{2,2},\sigma_{2,3}$ to the host. These values are share-deriving and therefore NXK-restricted
(Remark~\ref{rem:transport-vs-nxk}); they are held only as transient host state needed to (i) generate the UC-context
consistency AoKs and (ii) invoke the internal $\FKeyBox.\Load$ calls, and are then securely erased. This is
intentional: the UC-context AoKs are host-generated so that the simulator can observe the prover's \gROCRP\ query log
(Remark~\ref{rem:no-hidden-uc-prover}).

\smallskip\noindent\textbf{Reconciliation with implementation guidance.}
The implementation note in Section~\ref{subsubsec:keybox-candidates} instructs the profile adapter to ``forbid share-deriving outputs.'' This refers to vendor-API operations that would export caller-invertible affine images of an \emph{already-installed} resident share. It does not apply to $\textsf{SDKG.LeafInit}$, which is a key-independent operation ($\textsf{SDKG.LeafInit}\in\mathcal{F}_{\mathrm{KI}}^{\mathsf{SDKG}}$) that samples fresh scalars internally: its outputs become share-deriving only prospectively, after the subsequent $\Load$ calls install the corresponding long-term shares, and are erased immediately thereafter.

Next, we detail our $\mathrm{\Psi}^{(3)}_{\mathsf{SDKG}}$ protocol. Session setup designates $P_2$ as the initiating leaf (primary role). All computations are performed in $\mathbb{Z}_p.$ Let $\mathsf{pp} \coloneqq (\mathbb{G},p,\mathcal G,\mathcal H)$ be the set of public parameters. Each step of the protocol is atomic for the local state. Detailed description follows:

\medskip
\noindent\textbf{Round 1} $P_2 \to P_1:$
\begin{itemize}
\item Leaf setup (hardened/minimal profile):
$P_2$ obtains its Round~1 values by invoking the key-independent KeyBox operation $\textsf{SDKG.LeafInit}$:
\[
(C_2,\zeta_2,B_2,\sigma_{2,1},\sigma_{2,2},\sigma_{2,3})
\leftarrow \FKeyBox^{(2)}.\Use(\KBsid{\texttt{ki}},\textsf{SDKG.LeafInit},\langle \sid\rangle).
\]
Define the leaf’s transcript-defined group element $M_2 \coloneqq \Open_M(\pp,C_2,\zeta_2)$ (by USV correctness,
$M_2=m_2\G$ for the internal $m_2$ sampled inside $\textsf{SDKG.LeafInit}$).
In the transient-host-RAM deployment, this step may equivalently be implemented by sampling $(m_2,b_2)$ in host state,
setting $f_2(x)=m_2+b_2x$, and computing the same outputs, treating all $\sigma$-scalars as NXK-restricted and
securely erased after their last use.
	\item For a fresh commitment identifier $\cid_2$, $P_2$ sends $(\Commit, \sid, \cid_2, P_1, C_2, \zeta_2)$ to $\Fusv$ and receives receipt digest $d$ from $\Fusv$. 
	\item $P_2$ generates $\sigma_{3,2} \leftarrowdollar \Zp$; computes $h_{3,2} \coloneqq H_{\mathrm{s32}}(\langle \sid, \cid_2, \sigma_{3,2}\G \rangle)$ and sends \((C_2, \zeta_2, B_2, h_{3,2}, \sigma_{2,1},\sid,\cid_2,d)\) to \(P_1\).
\end{itemize}
\textbf{Round 2} $P_1 \to P_2:$
\begin{itemize}
	\item $P_1$ verifies that $\sid$ and $\cid_2$ are fresh, calls $\Fusv.\Verify(\sid,\cid_2,P_2,C_2,\zeta_2)$, and requires it returns 1; else, $P_1$ aborts. 
	\item $P_1$ computes $M_2 \gets \Open_M(\pp,C_2,\zeta_2)$ (abort if $M_2=\perp$) and checks that
	\[
	d = H\!\bigl(\USVrcpt, \langle \sid, \cid_2, P_2, P_1, C_2, M_2\rangle\bigr).
	\]
	It aborts if the check fails.
	\item $P_1$ parses $\zeta_2 = (\nu_2, \upsilon_2, \pi_{\DLEQ_2})$, checks $$(i) ~\nu_2 \neq -1 \bmod p \quad \text{ and } \quad (ii)~ \sigma_{2,1}\mathcal{G}-2B_2 = M_2.$$ It aborts if any verification fails. 
	\item $P_1$ samples \(m_1 \leftarrowdollar \mathbb Z_p^* \text{ and } b_1, \sigma_{3,1}\!\leftarrowdollar \!\mathbb{Z}_p\); sets \(f_1(x)\coloneqq m_1+b_1x\), and computes $$M_1 \coloneqq m_1 \G, B_1 \coloneqq b_1\G, \sigma_{1,1}\coloneqq f_1(2), \sigma_{1,2}\coloneqq f_1(3), \sigma_{1,3}\coloneqq f_1(1), X_1 \coloneqq \sigma_{1,1}\G + \sigma_{2,1}\G + \sigma_{3,1}\G.$$ Equivalently, $X_1=x_1\G$ for the conceptual scalar $x_1:=\sigma_{1,1}+\sigma_{2,1}+\sigma_{3,1}$, which is never materialized in host memory. Then, deletes $m_1, b_1, \text{and } f_1$.
	\item $P_1$ sends $(\sid, X_1,M_1,B_1,\sigma_{1,2},\pi_{\aff_1})$ to $P_2$, where $\pi_{\aff_1}$ is a UC-NIZK-AoK for the affine relation $\cal R_{\mathsf{aff}}$ on statement
	\[
	(X_1,\gamma_1,M_1,B_1,\mathrm{\Delta}_1)\quad\text{with}\quad \gamma_1:=2,\ \mathrm{\Delta}_1:=\sigma_{2,1}\G,
	\]
	using witness\footnote{The witnesses $(\alpha_i,\delta_i)$ are $\sigma$-values and hence share-deriving material
		(Remark~\ref{rem:transport-vs-nxk}). They are held only transiently for UC-proof generation and the
		subsequent KeyBox installation steps, then securely erased.} $(\alpha_1,\delta_1):=(\sigma_{1,1},\sigma_{3,1})$, i.e., $P_1$ generates UC-NIZK-AoKs for the DL relation $\cR_{\mathsf{DL}}$:
	\begin{align*}
	&\pi_{Y_1}\leftarrow \cp_{\mathsf{DL}}(\pp,\DLstmt{\sid}{\LblAffOneY}{Y_1},\alpha_1)\ \text{with witness } \alpha_1:=\sigma_{1,1}; \\
	&\pi_{D_1}\leftarrow \cp_{\mathsf{DL}}(\pp,\DLstmt{\sid}{\LblAffOneD}{D_1},\delta_1)\ \text{with witness } \delta_1:=\sigma_{3,1},
\end{align*}
	then $\pi_{\aff_1} = (\pi_{Y_1}, \pi_{D_1}).$ 
\end{itemize}
\textbf{Round 3} $P_2 \to P_1:$
\begin{itemize}
	\item $P_2$ verifies that $\sigma_{1,2}\G = M_1 + 3 B_1$, and aborts if the check fails. 
	\item $P_2$ computes $Y_1= M_1 + 2B_1, D_1 = X_1 - \sigma_{2,1}\G - Y_1$, and accepts iff
\begin{align*}
	\cV_{\mathsf{DL}}(\pp,\DLstmt{\sid}{\LblAffOneY}{Y_1},\pi_{Y_1})=1\ \wedge\
	\cV_{\mathsf{DL}}(\pp,\DLstmt{\sid}{\LblAffOneD}{D_1},\pi_{D_1})=1
\end{align*}
	\item $P_2$ computes the public group element
	$X_2 \coloneqq \sigma_{1,2}\G + \sigma_{2,2}\G + \sigma_{3,2}\G.$ Equivalently, $X_2=x_2\G$ for the conceptual scalar $x_2:=\sigma_{1,2}+\sigma_{2,2}+\sigma_{3,2}$, which is never materialized in host memory. Compute $M_2 \gets \Open_M(\pp,C_2,\zeta_2)$; define $Y_2 := M_2 + 3B_2$ and $D_2 := X_2 - \sigma_{1,2}\G - Y_2.$
	\item $P_2$ generates two UC-NIZK-AoKs for the discrete-log relation $\cR_{\mathsf{DL}}$:
	\begin{align*}
	&\pi_{Y_2}\leftarrow \cp_{\mathsf{DL}}(\pp,\DLstmt{\sid}{\LblAffTwoY}{Y_2},\alpha_2)\ \text{ with witness }\ \alpha_2:=\sigma_{2,2};\\
	&\pi_{D_2}\leftarrow \cp_{\mathsf{DL}}(\pp,\DLstmt{\sid}{\LblAffTwoD}{D_2},\delta_2)\ \text{ with witness }\ \delta_2:=\sigma_{3,2}.
\end{align*}
	Equivalently, $\pi_{\aff_2}:=(\pi_{Y_2},\pi_{D_2})$ is a UC-NIZK-AoK for the affine relation $\cR_{\mathsf{aff}}$
	on statement $(X_2,\gamma_2,M_2,B_2,\mathrm{\Delta}_2)$ with $\gamma_2:=3$ and $\mathrm{\Delta}_2:=\sigma_{1,2}\G$.
	\item $P_2$ computes the public key $K = 3 X_1 - 2 X_2,$ and sends $(X_2,\pi_{\aff_2}, K)$ to $P_1$. Additionally, $P_2$ publishes $K$ via the authenticated public broadcast functionality:
	it sends $(\mathsf{Publish},\sid,K)$ to $\Fpub$.
\end{itemize}
\paragraph{Verification by $P_1$:}
Parse $(X_2,\pi_{\aff_2},K)$ as $(X_2, (\pi_{Y_2},\pi_{D_2}), K_{\mathsf{rec}})$ and compute
\[M_2 \gets \Open_M(\mathsf{pp},C_2,\zeta_2), \quad Y_2 := M_2 + 3B_2, \quad D_2 := X_2 - \sigma_{1,2}\G - Y_2, \quad K = 3X_1 - 2X_2,\]
and accepts iff: 
\begin{align*}
&\cV_{\mathsf{DL}}(\pp,\DLstmt{\sid}{\LblAffTwoY}{Y_2},\pi_{Y_2})=1\ \wedge \\ &\cV_{\mathsf{DL}}(\pp,\DLstmt{\sid}{\LblAffTwoD}{D_2},\pi_{D_2})=1 \wedge \\ &h_{3,2} = H_{\mathrm{s32}}(\langle \sid, \cid_2, D_2 \rangle) \wedge \\ &K = K_{\mathsf{rec}}.
\end{align*}

\paragraph{Post-accept public-key confirmation.}
If $P_1$ accepts the transcript above, $P_1$ also publishes the verified key via $\Fpub$:
it sends $(\mathsf{Publish},\sid,K)$ to $\Fpub$.
This ensures that $P_3$ (and any other non-base-run party) can learn $K$ from an honest
source regardless of $P_2$'s behavior, binding the privately verified key to the public broadcast.
In the protocol specification, $P_3$ obtains its reference~$K$ from the $\Fpub$ publications for~$\sid$ as follows:
$P_3$ requires that both $P_1$'s and $P_2$'s authenticated $\Fpub$ publications for $\sid$ agree;
if $P_3$ receives $(\mathsf{Recv},\sid,P_1,K^{(1)})$ and $(\mathsf{Recv},\sid,P_2,K^{(2)})$ from $\Fpub$,
it sets $K:=K^{(1)}$ only if $K^{(1)}=K^{(2)}$, and aborts registration otherwise.
This ensures that a corrupted sender alone cannot install a shifted key in the honest joiner's
KeyBox; the check $K^{(1)}=K^{(2)}$ enforces consistency with the honest sender's publication.

\paragraph{Post-accept KeyBox installation.}
Each honest party finalizes by installing its long-term NXK shares and the retained registration scalars
inside its local KeyBox via $\Load$ calls over the internal channel to $\FKeyBox^{(i)}$,
as soon as it has locally verified transcript consistency:

\begin{itemize}[leftmargin=*]
	\item $P_2$ installs immediately after completing its Round~3 verification and sending $(X_2,\pi_{\aff_2},K)$ to $P_1$ (within the same activation). At this point, $P_2$ has verified $P_1$'s affine AoKs and holds all required $\sigma$-scalars. $P_2$ invokes: $\FKeyBox^{(2)}.\Load(\KBsid{k2},g_{2},\langle \sigma_{1,2},\sigma_{2,2},\sigma_{3,2}\rangle)$, $	\FKeyBox^{(2)}.\Load(\KBsid{k32},g_{3,2}^{\textsf{reg}},\langle \sigma_{3,2}\rangle)$,\\ $	\FKeyBox^{(2)}.\Load(\KBsid{k23},g_{2,3}^{\textsf{reg}},\langle \sigma_{2,3}\rangle).$

	\item $P_1$ installs after receiving $P_2$'s Round~3 message and verifying the full transcript via $\AccSDKG$ (Algorithm~\ref{alg:AccSDKG}). $P_1$ invokes: $\FKeyBox^{(1)}.\Load(\KBsid{k12},g_{1,2},\langle \sigma_{1,1},\sigma_{2,1},\sigma_{3,1}\rangle)$, $\FKeyBox^{(1)}.\Load(\KBsid{k13},g_{1,3},\langle \sigma_{1,1},\sigma_{2,1},\sigma_{3,1},\sigma_{1,3}\rangle)$, \\$\FKeyBox^{(1)}.\Load(\KBsid{k31},g_{3,1}^{\textsf{reg}},\langle \sigma_{3,1}\rangle).$
	$P_1$ additionally publishes the verified key via $\Fpub$ at this point.
\end{itemize}
The two parties install at different adversary-scheduled times; $\FSDKG$'s atomic \textsf{TryFinalize} models the logical conjunction. If $P_1$ rejects, $P_2$'s installed shares are inert under NXK (unusable without $P_1$'s cooperation, since the access structure requires $P_1$ in every authorized pair).

\paragraph{On exposure of share-deriving material.}
The scalars $\sigma_{i,j}$ are protocol-internal and share-deriving in the sense of
Remark~\ref{rem:transport-vs-nxk}: learning the appropriate tuple(s) of $\sigma$-values suffices to recompute a party’s
installed long-term share via the public derivation routines (Definition~\ref{def:sdkg-deriv}).
Accordingly, $\sigma$-values are never placed in the adversary-visible transcript: whenever they are transmitted,
they are sent only over $\Fchan$, so the adversary learns at most the explicit leakage
modeled by $\Fchan$, unless it corrupts an endpoint. Beyond transport confidentiality, the protocol enforces an explicit erasure discipline compatible with adaptive corruptions
with secure erasures (Definition~\ref{KeyBox}): each honest party keeps share-deriving scalars in host memory only until their
last use, and then securely erases them. Concretely, the only times share-deriving material is intentionally
fed into long-term state are the internal $\FKeyBox.\Load$ calls that install signing shares and the three registration
scalars used for RDR; immediately after the corresponding $\Load$ calls return, the host erases the consumed values.
Per party, (i) ephemerals erased immediately, (ii) the precise $\sigma$-tuples passed into
$\FKeyBox.\Load$ and then erased, and (iii) the values retained only inside KeyBox slots $\KBsid{k31},\KBsid{k32},\KBsid{k23}$
for later use via $\SealToPeer/\OpenFromPeer$ (Algorithm~\ref{alg:reg}). Optionally, to avoid any exposure of ephemeral scalars in host memory, the KeyBox can be extended with a small non-exporting interface that (i) samples ephemerals internally and (ii) performs the required group/field operations on opaque handles and, for RDR, seals payloads TEE-to-TEE via $\SealToPeer/\OpenFromPeer$; more details in Appendix \ref{App1}.

\paragraph{Erasure discipline (adaptive corruptions).}
All local steps in the base run are atomic w.r.t.\ adaptive corruptions: whenever an honest
party outputs an outgoing message (including one carrying a UC-context NIZK), it immediately performs
the explicit secure erasures listed below before yielding control to the adversary/scheduler.
Consequently, any corruption that occurs after the message is emitted reveals none of the erased values
(Definition~\ref{KeyBox}).
\begin{itemize}[leftmargin=*]
	\item $P_1$ (Round 2): After computing $\pi_{\aff_1}=(\pi_{Y_1},\pi_{D_1})$ and sending
	$(\sid,X_1,M_1,B_1,\sigma_{1,2},\pi_{\aff_1})$, $P_1$ securely erases all randomness and
	intermediate state used by the UC-NIZK prover(s). It retains only the
	$\sigma$-values needed for the later $\FKeyBox^{(1)}.\Load(\cdot)$ calls.
	\item $P_2$ (Round 3): After computing $\pi_{\aff_2}=(\pi_{Y_2},\pi_{D_2})$ and sending
	$(X_2,\pi_{\aff_2},K)$, $P_2$ securely erases all UC-NIZK prover randomness and Fischlin prover scratch
	state for $\pi_{Y_2}$ and $\pi_{D_2}$. This includes per-trial rarity-search temporaries,
	which need not be retained beyond the current iteration under a streaming implementation,
	and retains only the $\sigma$-values needed for the later
	$\FKeyBox^{(2)}.\Load(\cdot)$ calls.
	\item All parties (post-install): Immediately after the final KeyBox installation step returns,
	each honest party securely erases all share-deriving scalars that were consumed by its $\Load$ calls
	(cf.\ Remark~\ref{rem:transport-vs-nxk}). The only share-derived values retained after this point are
	the NXK-confined KeyBox slots themselves.
\end{itemize}
In particular, after an honest party has sent its protocol messages, an adaptive corruption reveals at
most the retained $\sigma$-values (until installation) and never reveals any per-proof randomness beyond
what is already encoded in the public proof strings.

\paragraph{Secure erasure in practice (non-normative).}
Our UC execution model uses the standard idealization of adaptive corruption with secure erasures:
once the protocol explicitly erases a buffer, a later corruption reveals only the remaining (non-erased)
local state (Definition~\ref{KeyBox}). We do not claim that commodity host platforms provide this
property as a default. Rather, the intended interpretation is that NXK-restricted host-side material
(including UC-context AoK witnesses and the prover’s ephemeral randomness/scratch state) is handled in
a hardened process boundary and is engineered to (i) never reach persistent storage (swap/hibernation,
crash dumps, logs), and (ii) be overwritten promptly with compiler-resistant zeroization primitives.

The optimized Fischlin prover need not store a large rarity-search ``trace'' in memory:
for each repetition it can maintain only the current best candidate and overwrite per-challenge
temporaries immediately, so the long-lived working set is $O(1)$ scalars/group elements per repetition.
Operational mitigations commonly used to approximate the erasure assumption include locking sensitive
pages to avoid paging, disabling core dumps, and avoiding runtimes that may transparently copy or
intern sensitive buffers. If such controls cannot be ensured, then the host should be treated as
effectively ``always-on'' adversarial (outside our corruption-with-erasures model), or the design
should move witness-bearing computation into a hardened boundary with an explicit interface that
supports the required simulation/extraction view.

\subsubsection{Hardened deployment: what runs where (KeyBox vs.\ host).}
\label{subsubsec:hardened-placement}
Our security arguments rely on a precise boundary between computations that may run inside a
state-continuous KeyBox (or an enclave-resident profile adapter that is part of the same trusted boundary)
and computations that must remain outside on the host. By ``hardened'' we mean that long-term shares are
confined to the KeyBox and remain API-non-exportable even if the host is later corrupted (Definition~\ref{KeyBox}).
We still permit NXK-restricted ephemerals (including witnesses for UC-context AoKs) to exist transiently in host RAM
during an atomic local step and to be protected only by secure erasures against adaptive corruptions; we do not aim
to protect against an ``always-on'' RAM adversary on an otherwise honest host.
The key distinction below is extractability: proofs that the UC simulator must extract from
(via oracle-log inspection) cannot be generated inside a KeyBox, because \gROCRP\ has local-call semantics
and the simulator does not observe KeyBox-internal $\Query$ traces (Remark~\ref{rem:oracle-tape} and
Definition~\ref{def:GNPRO}). Table~\ref{tab:placement-hardened} summarizes the intended placement in hardened deployments and
the reason each placement is necessary.
\begin{table}[t]
	\centering
	\small
	\setlength{\tabcolsep}{4pt}
	\renewcommand{\arraystretch}{1.15}
	\begin{tabular}{p{0.30\linewidth}p{0.13\linewidth}p{0.52\linewidth}}
		\hline
		\textbf{Component} & \textbf{Runs in} & \textbf{Why this placement is necessary} \\
		\hline
		Long-term share generation/storage; $\Load$ derivations; $\GetPub$; $\SealToPeer/\OpenFromPeer$ &
		KeyBox &
		Enforces NXK non-exportability and key-opacity: long-term shares and any caller-invertible affine images
		never cross the KeyBox boundary; only restricted API outputs are visible. \\[3mm]
		
		NXK-restricted witnesses for UC-context consistency AoKs
		&
		Host
		&
		Must remain outside to enable oracle-log-based straight-line extraction for Fischlin-based UC-NIZK-AoKs.
		These scalars are treated as NXK-restricted material and are securely erased after proof generation and $\Load$.\\[0.4em]
		
		USV certificate generation $\Cert(\pp,m)$ &
		KeyBox &
		$\Cert$ is implemented inside the hardened boundary and
		only $(C,\zeta)$ is exported.\\[0.4em]
		
		USV verification/opening $\Vcert,\Derive,\Open_M$ and receipt binding
		$H(\USVrcpt,\langle \sid,\cid,P_s,P_r,C,M\rangle)$ where $M:=\Open_M(\pp,C,\zeta)$ &
		Host &
		These are public deterministic checks/derivations and must be transcript-defined for verifiers and the UC
		simulator; no secrets are required. Receipt binding is in a non-programmable context $\USVrcpt\in\CtxTEE$. \\[0.4em]
		
		UC-context consistency AoKs
		&
		Host &
		Must be outside to enable straight-line extraction from corrupted-party proofs by inspecting the
		adversary's \gROCRP\ query log.\\[0.4em]
		
		KeyBox-resident one-shot DL proofs &
		KeyBox &
		Used only under ZK/simulation; sealing targets one-shot responses (Assumption~\ref{assump:tee-continuity}), and deterministic nonce derivation ensures that even under rollback an attacker cannot obtain two different responses for the same commitment (Lemma~\ref{lem:rollback-robust}). \\[0.4em]
		
		\gROCRP\ programming $\SimProgramRO$ &
		Simulator only &
		Only the ideal-world simulator programs in contexts in $\CtxUC$; parties and KeyBoxes never program $H$. \\
		\hline
	\end{tabular}
	\caption{Threat/placement summary for hardened SDKG deployments (host vs. KeyBox).}
	\label{tab:placement-hardened}
\end{table}
This placement is not merely an implementation preference: it is required by the
combination of (i) state continuity (no rewinding/forking inside the KeyBox; Assumption~\ref{assump:tee-continuity})
and (ii) oracle-log-based straight-line extraction for Fischlin-based UC-NIZK-AoKs (Definition~\ref{def:NIZK-AoK} and
Remark~\ref{rem:oracle-tape}). If one wishes to move UC-extractable AoKs into a hardened boundary, the model must be
strengthened to expose an extractable interface (e.g., a simulator-visible oracle-log/receipt interface) or one must
switch to a proof mechanism whose UC extraction does not rely on observing the prover's \gROCRP\ queries.

\begin{algorithm}[t]
	\DontPrintSemicolon
	\caption{One-shot registration of recovery device $P_3$.}
	\label{alg:reg}
	\KwIn{$\sid$, public $K$ for this session (obtained by $P_3$ from the authenticated $\Fpub$ publications for~$\sid$: $P_3$ requires that both $P_1$'s and $P_2$'s $\Fpub$ publications agree, i.e., $P_3$ receives $(\mathsf{Recv},\sid,P_1,K^{(1)})$ and $(\mathsf{Recv},\sid,P_2,K^{(2)})$ from $\Fpub$ and sets $K:=K^{(1)}$ only if $K^{(1)}=K^{(2)}$, aborting otherwise), and parties $P_1,P_2$ with sealed registration scalars installed inside their respective KeyBoxes---i.e., $\KBsid{k13}, \KBsid{k31}$ for $P_1$ and $\KBsid{k32},\KBsid{k23}$ for $P_2$. $P_1$ and $P_2$ (along with any already-registered parties/devices) additionally hold the public session metadata point $K_{1,3}$ derived from the base-run transcript (Definition \ref{def:k13-derived}).}
	
	\smallskip
	\textbf{Associated data.} $P_3$ ignores any associated-data strings received from the network. Else, it recomputes the following locally from $(\sid,P_s,P_r,\texttt{tag})$ and supplies them to $\OpenFromPeer$:
	\[
	\mathsf{ad}_{31}:=\ADSDKGreg{\sid}{P_1}{P_3}{\texttt{k31}}, \quad
	\mathsf{ad}_{32}:=\ADSDKGreg{\sid}{P_2}{P_3}{\texttt{k32}}, \quad
	\mathsf{ad}_{23}:=\ADSDKGreg{\sid}{P_2}{P_3}{\texttt{k23}}.
	\]
		
	\smallskip
	$P_1$ retrieves the session metadata $K_{1,3}$; computes
	$\varpi_{1}\leftarrow \FKeyBox^{(1)}.\Use(\KBsid{k31},\SealToPeer,\langle P_3,\mathsf{ad}_{31}\rangle)$,
	and sends $(\sid,\varpi_1,K_{1,3},K)$.\;
		
	\smallskip
$P_2$ uses the session metadata $K_{1,3}$ and computes two ciphertexts: \\
$\varpi_{2a}\leftarrow \FKeyBox^{(2)}.\Use(\KBsid{k32},\SealToPeer,\langle P_3,\mathsf{ad}_{32}\rangle)$ and \\
$\varpi_{2b}\leftarrow \FKeyBox^{(2)}.\Use(\KBsid{k23},\SealToPeer,\langle P_3,\mathsf{ad}_{23}\rangle)$,
and sends $(\sid,\varpi_{2a},\varpi_{2b},K_{1,3})$.\;
	
	\smallskip
	\textbf{Transport.}
	The two registration messages ($P_1\!\to\!P_3$ and $P_2\!\to\!P_3$) are delivered over authenticated channels
	(modeled by $\Fchan$).
	
	\smallskip
	\textbf{Input consistency.}
	($P_3$ buffers whichever registration message arrives first and applies the following checks once both have been received.)
	Upon receiving $(\sid,\varpi_1,K_{1,3}^{(1)},K^{\mathsf{net}})$ from $P_1$, require $K^{\mathsf{net}}=K$ (otherwise abort). Upon receiving $(\sid,\varpi_{2a},\varpi_{2b},K_{1,3}^{(2)})$ from $P_2$, require $K_{1,3}^{(1)}=K_{1,3}^{(2)}$ (otherwise abort), and set $K_{1,3}\gets K_{1,3}^{(1)}$.
	
	\smallskip
	$P_3$ forwards $(\varpi_1,\varpi_{2a},\varpi_{2b},\mathsf{ad}_{31},\mathsf{ad}_{32},\mathsf{ad}_{23},K,K_{1,3})$ into $\FKeyBox^{(3)}$ which executes the following
	installation procedure (modeled as an atomic transaction):
	\begin{enumerate}[leftmargin=1.5em, label=(\arabic*), nosep]
		\item \textbf{Install sponsor state.} Open $\varpi_{2a}$ with $\mathsf{ad}_{32}$ and $\varpi_{2b}$ with $\mathsf{ad}_{23}$ internally via $\OpenFromPeer$ to obtain opaque handles
		$\tau_{3,2}^{\textsf{reg}}$ and $\tau_{2,3}^{\textsf{reg}}$, then invoke
		$\Load(\KBsid{k32},g_{3,2}^{\textsf{reg}},\langle \tau_{3,2}^{\textsf{reg}}\rangle)$ and
		$\Load(\KBsid{k23},g_{2,3}^{\textsf{reg}},\langle \tau_{2,3}^{\textsf{reg}}\rangle).$
		\item \textbf{Install the recovery share.} (Re-)open $\varpi_1$ with $\mathsf{ad}_{31}$, $\varpi_{2a}$ with $\mathsf{ad}_{32}$, and $\varpi_{2b}$ with $\mathsf{ad}_{23}$ internally via $\OpenFromPeer$
		to obtain fresh opaque handles $\tau_{3,1},\tau_{3,2},\tau_{2,3}$, and then invoke
		$\Load(\KBsid{k3},g_3,\langle \tau_{2,3},\tau_{3,2},\tau_{3,1},K,K_{1,3}\rangle).$
	\end{enumerate}
	Accept iff all $\Load$ invocations above return $\ok$.
\end{algorithm}

\begin{definition}[Transcript-derived $K_{1,3}$]\label{def:k13-derived}
	\emph{In any accepted base-run transcript with $\mathsf{T}_2=(X_1,M_1,B_1,\sigma_{1,2},\pi_{\aff_1})$, define the
	public metadata point
	\[
	K_{1,3}\ :=\ \DeriveK(\mathsf{T}_2)\ :=\ 2(M_1+B_1)-X_1 \in \mathbb G.
	\]
	Equivalently, $K_{1,3}=\Pub(k_{1,3})$ for $k_{1,3}:=2\sigma_{1,3}-x_1$, but the protocol uses only the public
	derivation above and never invokes $\GetPub$ on slot $\KBsid{k13}$.}
\end{definition}

The one-shot registration protocol in Algorithm~\ref{alg:reg} is modeled in $\FSDKG$ (Fig.~\ref{fig:FSDKG}) by
(i) leaking only the lengths of the two registration messages, denoted by $\ellregone$ for the $P_1\!\to\!P_3$ message and
$\ellregtwo$ for the sponsor-to-$P_3$ message, and (ii) installing the corresponding recovery-role share and sponsor-state
slots inside $\FKeyBox^{(3)}$ upon successful completion. When the receiver $P_3$ is corrupted, an authenticated confidential-channel functionality reveals the entire delivered message to the adversary. Therefore, in $\FSDKG$ we do not model registration transport by uniform dummy strings.
Instead, the registration payloads are treated as syntactically valid encodings of the same tuple-structured messages as in
Algorithm~\ref{alg:reg}, with ciphertext components sampled from the same distribution as real $\SealToPeer$ outputs. Concretely,
$\FSDKG$ conceptually samples sealing ciphertexts
\[
\varpi_1 \leftarrow \Enc_{\pk_{\mathrm{seal}}^{(P_3)}}(\mathsf{ad}_{31},\sigma_{3,1}),\quad
\varpi_{2a} \leftarrow \Enc_{\pk_{\mathrm{seal}}^{(P_3)}}(\mathsf{ad}_{32},\sigma_{3,2}),\quad
\varpi_{2b} \leftarrow \Enc_{\pk_{\mathrm{seal}}^{(P_3)}}(\mathsf{ad}_{23},\sigma_{2,3}),
\]
for the slot-bound associated data $(\mathsf{ad}_{31},\mathsf{ad}_{32},\mathsf{ad}_{23})$ defined in Algorithm~\ref{alg:reg}, and then
forms the delivered network messages as the self-delimiting encodings
\[
\langle \sid,\varpi_1,K_{1,3},K\rangle,\qquad
\langle \sid,\varpi_{2a},\varpi_{2b},K_{1,3}\rangle.
\]
This ensures that if $P_3$ is corrupted, the adversary observes well-formed messages and can invoke $\OpenFromPeer$ on the received
ciphertexts. When $P_3$ is honest, the externally visible
distribution of these sealed blobs remains simulatable from public information under our KeyBox key-opacity assumption
(Assumption~\ref{assump:keybox-opacity}), so the only explicit leakage retained by $\FSDKG$ is message length. Registration is not modeled as an atomic two-message transaction w.r.t. corruption: if $P_3$ is corrupted after receiving only one of the two registration payloads, then the adversary learns the already-delivered payload(s) upon corruption and may complete or abort registration at will by (not) delivering and/or forwarding the remaining payload into $\FKeyBox^{(3)}$.
 
\begin{observation}[Lagrange weights]\label{obs:sec7}
	\emph{It follows from direct computation that for some scalars $\alpha, \beta \in \Zp,$ it holds that \begin{align*}
			K &= k\mathcal{G} = k_{1,2} \mathcal{G} + k_2 \mathcal{G} = k_{1,3} \mathcal{G} + k_3 \mathcal{G} =\alpha x_1 \mathcal{G} - \beta x_2 \mathcal{G} \\&= (\alpha - \beta)(m_1 + m_2) \mathcal{G} + (2 \alpha - 3 \beta)(b_1 + b_2) \mathcal{G} + \alpha \sigma_{3,1}\mathcal{G} - \beta \sigma_{3,2}\mathcal{G}.
		\end{align*} For $u=2$ and $v=3$, the unique $\alpha,\beta$ satisfying
		$\alpha-\beta=1$ and $\alpha u-\beta v=0$ are $\alpha=3$ and $\beta=2$.
		Consequently, $
		\alpha x_1 - \beta x_2
		= (m_1+m_2) + \alpha\sigma_{3,1} - \beta\sigma_{3,2}.
		$
		More generally, for distinct $u,v\in\mathbb Z^*_p$, the unique choice is
		$\alpha=\frac{v}{v-u}$ and $\beta=\frac{u}{v-u}$.}
\end{observation}

For $\alpha=3,\beta=2$ (the Lagrange weights for $u=2,v=3)$, in any accepting transcript $\mathcal T$, the UC-NIZK(-AoK)s and the USV opening
uniquely determine $(x_1,x_2)$. Hence, the only $\mathcal T$-consistent shares
with $k=k_{1,2}+k_2$ are $(k_{1,2},k_2)=(3x_1,-2x_2)$. Since $f_1,f_2$ are linear, $m_1=f_1(0)=2f_1(1)-f_1(2)=2\sigma_{1,3}-\sigma_{1,1}$ and
$m_2=f_2(0)=2f_2(1)-f_2(2)=2\sigma_{2,3}-\sigma_{2,1}$. By Observation~\ref{obs:sec7}, $k=(m_1+m_2)+3\sigma_{3,1}-2\sigma_{3,2}$. Substituting and regrouping yields $k=(2\sigma_{1,3}-x_1)+2(\sigma_{2,3}-\sigma_{3,2}+2\sigma_{3,1}) = k_{1,3} + k_3$.

\begin{algorithm}[t]
	\caption{Acceptance predicate $\AccSDKG(\sid, P_s, P_r, \TSDKG)$}
	\label{alg:AccSDKG}
\textbf{Input:} session identifier $\sid$; USV committer $P_s$ and relying party $P_r$;
transcript $\TSDKG=(\mathsf{T}_1,\mathsf{T}_2,\mathsf{T}_3)$, $$\text{ where } \mathsf{T}_1 = (\cid_2, C_2, \zeta_2, B_2, \sigma_{2,1}, h_{3,2}, d), \mathsf{T}_2 = (X_1, M_1, B_1, \sigma_{1,2}, \pi_{\aff_1}), \mathsf{T}_3 = (X_2, \pi_{\aff_2}, K_{\mathsf{rec}}).$$
\textbf{Output:} $1$ iff all checks below pass; otherwise $0$.

\smallskip
\textbf{Parse / derive.} If parsing fails, return $0$.
\begin{enumerate}[leftmargin=*,label=\textbf{D\arabic*:}, nosep]
	\item Compute $M_2 \gets \Open_M(\pp,C_2,\zeta_2)\text{ (if }M_2=\perp,\text{ return }0\text{)} \text{ and }	d^\star \gets H\!\bigl(\USVrcpt,\langle \sid,\cid_2,P_s,P_r,C_2,M_2\rangle\bigr).$
	\item Derive affine-check auxiliaries: $$Y_1 \gets M_1 + 2B_1; \quad D_1 \gets X_1 - \sigma_{2,1}\G - Y_1; \quad Y_2 \gets M_2 + 3B_2; \quad D_2 \gets X_2 - \sigma_{1,2}\G - Y_2.$$
		\item Derive $K_{1,3}\gets 2(M_1+B_1)-X_1$ and key $\widehat K \gets 3X_1 - 2X_2$.
	\end{enumerate}
	
	\smallskip
	\textbf{Checks.} Accept iff all of the following hold:
	\begin{enumerate}[leftmargin=*,label=\textbf{C\arabic*:}, nosep]
		\item \textbf{USV receipt consistency:} $d = d^\star$.
		\item \textbf{USV certificate linkage:} $\sigma_{2,1}\G - 2B_2 = M_2$.
		\item \textbf{Affine AoK for $X_1$:} write $\pi_{\aff_1}=(\pi_{Y_1},\pi_{D_1})$ and require
		\[
		\sigma_{1,2}\G = M_1 + 3B_1
		\quad\wedge\quad
		\cV_{\mathsf{DL}}(\pp,\DLstmt{\sid}{\LblAffOneY}{Y_1},\pi_{Y_1})=1
		\quad\wedge\quad
		\cV_{\mathsf{DL}}(\pp,\DLstmt{\sid}{\LblAffOneD}{D_1},\pi_{D_1})=1.
		\]
		\item \textbf{Affine AoK for $X_2$:} write $\pi_{\aff_2}=(\pi_{Y_2},\pi_{D_2})$ and require $$	\cV_{\mathsf{DL}}(\pp,\DLstmt{\sid}{\LblAffTwoY}{Y_2},\pi_{Y_2})=1 \wedge \cV_{\mathsf{DL}}(\pp,\DLstmt{\sid}{\LblAffTwoD}{D_2},\pi_{D_2})=1.$$
		\item \textbf{Digest check for $h_{3,2}$:}  $h_{3,2} = H_{s32}(\langle \sid,\cid_2,D_2\rangle)$.
		\item \textbf{Key consistency:} check $K_{\mathsf{rec}} = \widehat K$.
	\end{enumerate}

	return $1$.
\end{algorithm}

\begin{readerbox}[rn:broadcast-guard]{$\Fpub$ agreement as a formal protocol requirement}
In the base run, $P_1$ additionally verifies that $\widehat K$ equals the value received from $\Fpub$ for this session,
i.e., the $(\mathsf{Recv},\sid,P_2,K_{\mathsf{pub}})$ delivered by $\Fpub$ satisfies $K_{\mathsf{pub}}=\widehat K$.
This supplementary check is a non-blocking local consistency guard evaluated by $P_1$ when the $\Fpub$ delivery arrives;
it does not gate $\AccSDKG$ acceptance and does not introduce an additional protocol round.
For registration, the $\Fpub$-agreement check is a \emph{formal protocol requirement}: any party $P_i$ initiating
registration must require that both $P_1$'s and $P_2$'s authenticated $\Fpub$ publications
for $\sid$ agree (Algorithm~\ref{alg:reg}). This prevents a corrupted sender from installing a shifted key
in an honest joiner's KeyBox by publishing $K_{\mathrm{fake}}\neq K$ via $\Fpub$.
For the $n$-device extension, all joiners check $P_1$'s and $P_2$'s base-run $\Fpub$ publications, regardless of the sponsor's identity. The immutability of $\Fpub$ (which provides no retract or re-publish interface) ensures that an honest base-run publication cannot be retroactively altered by a subsequent adaptive corruption.
The base-run acceptance predicate $\AccSDKG$ itself operates on the transcript $\TSDKG$ alone
(Theorem~\ref{thm:SDKG-UC}).
\end{readerbox}

\begin{lemma}[$\AccSDKG$ need not call $\Fusv.\Verify$]\label{lem:usv-verify-implicit}
	In any real execution of $\mathrm{\Psi}^{(3)}_{\mathsf{SDKG}}$ in the $(\FKeyBox,\Fusv,\Fchan)$-hybrid model,
	if $P_1$ is honest and the Round~2 message
	$\mathsf{T}_2=(X_1,M_1,B_1,\sigma_{1,2},\pi_{\aff_1})$ is delivered to $P_2$ over $\Fchan$, then
	$\Fusv.\Verify(\sid,\cid_2,P_2,C_2,\zeta_2)=1$ and the receipt-digest consistency check holds:
	\[
	d = H\!\bigl(\USVrcpt,\langle \sid,\cid_2,P_2,P_1,C_2,M_2\rangle\bigr)
	\quad\text{where }M_2:=\Open_M(\pp,C_2,\zeta_2).
	\]
\end{lemma}

\begin{proofsketch}
	An honest $P_1$ sends $\mathsf{T}_2$ only if $\Fusv.\Verify(\sid,\cid_2,P_2,C_2,\zeta_2)=1$ and the receipt digest
	check succeeds; otherwise it aborts before sending. Since $\Fchan$ is authenticated, an adversary cannot
	forge $\mathsf{T}_2$ on behalf of honest $P_1$. \qed
\end{proofsketch}

For any transcript that can arise with honest $P_1$ and authenticated channels, the receipt-digest check \textbf{C1}
	in Algorithm~\ref{alg:AccSDKG} already captures the effect of the explicit $\Fusv.\Verify$ performed by honest $P_1$
	in Round~2.

The KeyBox/$\Load$ step executes within a single honest-party activation (Definition~\ref{def:ace}), which is atomic with respect to corruption. When both senders are honest, the derivation routines are deterministic on well-formed inputs and all $\Load$ calls succeed; if a sender is corrupted and a $\Load$ fails, registration is not recorded as complete and any partially installed sponsor-state slots are inert (NXK-restricted and unusable for key recovery without the recovery-share slot $\KBsid{k3}$). In practice, this can be implemented by delaying the sealed-storage write until after validation. Registering devices $P_i$ for $i>3$ follows the same pattern: $P_1$ always seals $\sigma_{3,1}$ from slot $\KBsid{k31}$, and the sponsor leaf seals the sponsor-state scalars from its local slots $\KBsid{k32}$ and $\KBsid{k23}$. The joining device installs (i) its long-term recovery share in slot $\KBsid{k3}$ via $g_3$, and (ii) the sponsor-state slots $\KBsid{k32}$ and $\KBsid{k23}$ via
$g_{3,2}^{\textsf{reg}}$ and $g_{2,3}^{\textsf{reg}}$. Consequently, every already-registered leaf can serve as sponsor for future registrations.

\paragraph{Parameterized registration (general case).}
For any joiner $P_i$ ($i\ge 3$) with sponsor $P_j\in\{2\}\cup\mathsf{Reg}$, the real-world registration executes Algorithm~\ref{alg:reg} with the following parameter substitutions: $P_3\mapsto P_i$ and $P_2\mapsto P_j$ (as sponsor), associated-data strings
\[
\mathsf{ad}_{1i}:=\ADSDKGreg{\sid}{P_1}{P_i}{\texttt{k31}},\quad
\mathsf{ad}_{ji}^{32}:=\ADSDKGreg{\sid}{P_j}{P_i}{\texttt{k32}},\quad
\mathsf{ad}_{ji}^{23}:=\ADSDKGreg{\sid}{P_j}{P_i}{\texttt{k23}},
\]
and derivation routines $g_3,g_{3,2}^{\textsf{reg}},g_{2,3}^{\textsf{reg}}$ parameterized by $(P_i,P_j,\sid)$ per Definition~\ref{def:sdkg-deriv}. The $\Fpub$-agreement check uses $P_1$'s and $P_2$'s base-run publications (not the sponsor's). Transport uses $\Fchan$. The KeyBox installation procedure within $\FKeyBox^{(i)}$ is identical to Algorithm~\ref{alg:reg}'s steps~(1)--(2).

In the setting of Proposition~\ref{prop:nxk-resharing-impossible-updated}, any enrollment mechanism that requires a party to
export or externally derive a fresh share-deriving material from the public/external
view conflicts with NXK under key-opacity. Proposition~\ref{prop:nxk-resharing-impossible-updated} bounds the success
probability of any PPT strategy that computes $\widehat{k}_{\isnew}$ from $(\pp,\tau_{\mathrm{ext}})$.
SDKG's one-shot registration of $P_3$ does not attempt to compute $f(x_{\isnew})$ (or any share-derived plaintext) from
$\tau_{\mathrm{ext}}$. Instead, existing devices transfer the required registration scalars to $P_3$ only via
(attested) KeyBox-to-KeyBox sealing using $\SealToPeer/\OpenFromPeer$, and $P_3$'s KeyBox derives and installs $k_3$ internally
(via $g_3$). Externally, the transported ciphertexts are part of $\tau_{\mathrm{ext}}$ and remain simulatable under key-opacity. Beyond $i=3$, our scalable NXK-compatible RDR enrollment
registers additional devices as redundant front-ends for the recovery role
(with the same base-run-derived share guaranteed whenever at least one registration sender is honest).

\begin{figure}[!ph]
	\centering
	\setlength{\fboxrule}{0.2pt} 
	\fbox{%
		\parbox{\dimexpr\linewidth-2\fboxsep-2\fboxrule\relax}{%
			\ding{169} \textsf{State (per session $\sid$):}
			corruption set $\mathsf{Cor}\subseteq\{1,2,3\}$; flags $\mathsf{finalized},\mathsf{RegPending},\mathsf{RegDone}\in\{0,1\}$
			(init $0$); approvals $\mathsf{Auth}\subseteq\{1,2\}$ (init $\emptyset$); transcript slots
			$\mathsf{T}_1,\mathsf{T}_2,\mathsf{T}_3$ (init $\perp$); optional programmed values
			$(x_1,x_2,\sigma_{3,1},\sigma_{3,2})\in(\Zp\cup\{\perp\})^4\ \text{(init} \perp)$;
			retain the registration scalars $(\sigma_{3,1},\sigma_{2,3}, \sigma_{3,2})\in\Zp^3$ and $K_{1,3}\in\mathbb G$ after finalization;
			registration-channel state multiset $\mathsf Q_{\mathsf{reg}}$ of $(\rho,P_s,P_r,w,\phi)$, and delivered set $\mathsf D_{\mathsf{reg}}$ (init empty); delivered registration payload buffers $w_{\mathsf{reg},1}^{\mathsf{del}}, w_{\mathsf{reg},2}^{\mathsf{del}}\in\{0,1\}^*\cup\{\perp\}$ (init $\perp$).
						
			\smallskip						
			\ding{169} \textsf{Init:} Upon $(\mathsf{init},2,\sid)$ from $P_2$, send $(\mathsf{init},\sid)$ to $P_1,P_2,P_3$ and $\cA$.
			
		\smallskip
		\ding{169} \textsf{Transcript feed (ideal adversary/simulator interface; bookkeeping):}
		Upon receiving one of the following messages from $\cA$, if the corresponding transcript slot is empty then store it and run \textsf{TryFinalize}:
		\begin{itemize}[leftmargin=*,nosep]
			\item $(\mathsf{Tin}_1,\sid,\cid_2,C_2,\zeta_2,B_2,\sigma_{2,1},h_{3,2},d)$:
			if $\mathsf{T}_1=\perp$ then set $\mathsf{T}_1\gets(\cid_2,C_2,\zeta_2,B_2,\sigma_{2,1},h_{3,2},d)$.
			\item $(\mathsf{Tin}_2,\sid,X_1,M_1,B_1,\sigma_{1,2},\pi_{\aff_1})$:
			if $\mathsf{T}_2=\perp$ then set $\mathsf{T}_2\gets(X_1,M_1,B_1,\sigma_{1,2},\pi_{\aff_1})$.
			\item $(\mathsf{Tin}_3,\sid,X_2,\pi_{\aff_2},\widehat K)$:
			if $\mathsf{T}_3=\perp$ then set $\mathsf{T}_3\gets(X_2,\pi_{\aff_2},\widehat K)$.
		\end{itemize}
			
			\smallskip
			\ding{169} \textsf{TryFinalize:} If $\mathsf{finalized}=1$ or some $\mathsf{T}_i=\perp$, do nothing.
			Otherwise let $\mathcal T:=(\mathsf{T}_1,\mathsf{T}_2,\mathsf{T}_3)$.
			If $\AccSDKG(\sid,P_2,P_1,\mathcal T)=0$ (Algorithm \ref{alg:AccSDKG}), do nothing.
			If $\AccSDKG(\sid,P_2,P_1,\mathcal T)=1$, then:
			\begin{enumerate}[leftmargin=*,nosep]
\item If $x_1=\perp$ or $x_2=\perp$ or $\sigma_{3,1}=\perp$ or $\sigma_{3,2}=\perp$, do nothing and return. Let $\mathsf{T}_1=(\cid_2,C_2,\zeta_2,B_2,\sigma_{2,1},h_{3,2},d)$,
$\mathsf{T}_2=(X_1,M_1,B_1,\sigma_{1,2},\pi_{\aff_1})$,
and $\mathsf{T}_3=(X_2,\pi_{\aff_2},K_{\mathsf{rec}})$.
Define $Y_1 := M_1 + 2B_1$ and $D_1 := X_1 - \sigma_{2,1}\G - Y_1$.
Require $X_1 = x_1\G, X_2 = x_2\G, h_{3,2} = H_{\mathrm{s32}}(\langle \sid,\cid_2,\sigma_{3,2}\G\rangle), D_1 = \sigma_{3,1}\G$. Define $\sigma_{1,1} := x_1-\sigma_{2,1}-\sigma_{3,1}\bmod p, \sigma_{1,3} := 2\sigma_{1,1}-\sigma_{1,2}\bmod p, k_{1,2}:=3x_1, k_2:=-2x_2, k_{1,3}:=2\sigma_{1,3}-x_1 \bmod p, K_{1,3}:=2(M_1+B_1)-X_1, k:=k_{1,2}+k_2 \bmod p, \sigma_{2,2} \coloneqq x_2-\sigma_{1,2}-\sigma_{3,2} \bmod p.$ 
\item Define recovery-role and sponsor-state scalars: $k_3 := (k-k_{1,3}) \bmod p, \sigma_{2,3} := k_3\cdot 2^{-1}-2\sigma_{3,1}+\sigma_{3,2} \bmod p.$
				Output $K:=k\G$ to $P_1,P_2$ and $\cA$. Set $\mathsf{finalized}\gets 1$.
				(In the real protocol, $P_3$ obtains $K$ by requiring agreement between both $P_1$'s and $P_2$'s base-run authenticated $\Fpub$ publications for $\sid$ before initiating registration; since $\Fpub$ publications are immutable once made, at least one honest base-run publication anchors $K$ to the correct value. The registration payload from $P_1$ also carries $K$ as a redundant consistency check.)
			\item If $1\notin\mathsf{Cor}$, send to $P_1$ the internal KeyBox command\\
			$			\left(\KBcmd,\sid,\big[
			(\Load,\KBsid{k12},g_{1,2},\langle \sigma_{1,1},\sigma_{2,1},\sigma_{3,1}\rangle), 
			(\Load,\KBsid{k13},g_{1,3},\langle \sigma_{1,1},\sigma_{2,1},\sigma_{3,1},\sigma_{1,3}\rangle),
			(\Load,\KBsid{k31},g_{3,1}^{\textsf{reg}},\langle \sigma_{3,1}\rangle)
			\big]\right).
			$
		\item If $2\notin\mathsf{Cor}$, send to $P_2$ the internal KeyBox command
			\[
			\left(\KBcmd,\sid,\big[
			(\Load,\KBsid{k2},g_{2},\langle \sigma_{1,2},\sigma_{2,2},\sigma_{3,2}\rangle), 
			(\Load,\KBsid{k32},g_{3,2}^{\textsf{reg}},\langle \sigma_{3,2}\rangle), 
			(\Load,\KBsid{k23},g_{2,3}^{\textsf{reg}},\langle \sigma_{2,3}\rangle)
			\big]\right).
			\]
			\end{enumerate}
			
			\smallskip
\ding{169} \textsf{Registration of $P_3$:}
Upon first $(\mathsf{register},3,\sid)$ from $P_3$:
require $\mathsf{finalized}=1$ and $\mathsf{RegDone}=0$; 
set $\mathsf{RegPending}\gets 1$, $\mathsf{Auth}\gets\emptyset$;
set $w_{\mathsf{reg},1}^{\mathsf{del}}\gets\perp$ and $w_{\mathsf{reg},2}^{\mathsf{del}}\gets\perp$;
notify $P_1,P_2$ and $\cA$ with $(\mathsf{RegReq},3,\sid)$.

\medskip
Upon $(\mathsf{approve},3,\sid)$ from $P_j$ with $j\in\{1,2\}$ and $\mathsf{RegPending}=1$,
set $\mathsf{Auth}\gets \mathsf{Auth}\cup\{j\}$ and notify $\cA$.

\medskip
Upon $(\mathsf{RegGo},3,\sid,w_1^\star,w_2^\star)$ from $\cA$:
if $\mathsf{RegPending}=1$, $\mathsf{Auth}=\{1,2\}$, $\mathsf{RegDone}=0$, and $\mathsf Q_{\mathsf{reg}}=\emptyset$, then:
define slot-bound \textsf{ad} as in Algorithm \ref{alg:reg}:
$\mathsf{ad}_{31}:=\ADSDKGreg{\sid}{P_1}{P_3}{\texttt{k31}}$,
$\mathsf{ad}_{32}:=\ADSDKGreg{\sid}{P_2}{P_3}{\texttt{k32}}$,
$\mathsf{ad}_{23}:=\ADSDKGreg{\sid}{P_2}{P_3}{\texttt{k23}}$.
Sample $\rho_1,\rho_2\leftarrowdollar \{0,1\}^\lambda$.

\begin{itemize}[leftmargin=*,nosep]
	\item If $1\notin\mathsf{Cor}$, sample 
	$\varpi_1 \leftarrow \Enc_{\pk_{\mathrm{seal}}^{(P_3)}}(\mathsf{ad}_{31},\sigma_{3,1})$
	and set $w_{\mathsf{reg},1}:=\langle \sid,\varpi_1,K_{1,3},K\rangle$.
	If $1\in\mathsf{Cor}$, set $w_{\mathsf{reg},1}:=w_1^\star$.
	
	\item 
	If $2\notin\mathsf{Cor}$, sample 
	$\varpi_{2a} \leftarrow \Enc_{\pk_{\mathrm{seal}}^{(P_3)}}(\mathsf{ad}_{32},\sigma_{3,2})$ and
	$\varpi_{2b} \leftarrow \Enc_{\pk_{\mathrm{seal}}^{(P_3)}}(\mathsf{ad}_{23},\sigma_{2,3})$. Set $w_{\mathsf{reg},2}:=\langle \sid,\varpi_{2a},\varpi_{2b}, K_{1,3}\rangle$.
	If $2\in\mathsf{Cor}$, set $w_{\mathsf{reg},2}:=w_2^\star$.
\end{itemize}

For each $j\in\{1,2\}$ with $w_{\mathsf{reg},j}\neq\perp$:
let $\phi_j:=|w_{\mathsf{reg},j}|$; insert $(\rho_j,P_j,P_3,w_{\mathsf{reg},j},\phi_j)$ into $\mathsf Q_{\mathsf{reg}}$;
send $(\mathsf{Leak},\sid,P_j,P_3,\rho_j,\phi_j)$ to $\cA$.
If $j\in\mathsf{Cor}$, additionally reveal $w_{\mathsf{reg},j}$ to $\cA$.

\medskip
Upon receiving $(\mathsf{Deliver},\sid,\rho)$ from $\cA$:
if $(\rho,P_s,P_3,w,\phi)\in\mathsf Q_{\mathsf{reg}}$ and $\rho\notin\mathsf D_{\mathsf{reg}}$, delete it from $\mathsf Q_{\mathsf{reg}}$,
add $\rho$ to $\mathsf D_{\mathsf{reg}}$, and deliver $(\mathsf{Recv},\sid,P_s,w)$ to $P_3$.
If $3\in\mathsf{Cor}$, reveal $w$ to $\cA$ at delivery time.
If $P_s=P_1$, set $w_{\mathsf{reg},1}^{\mathsf{del}}\gets w$. If $P_s=P_2$, set $w_{\mathsf{reg},2}^{\mathsf{del}}\gets w$.

\medskip
If $\mathsf{RegPending}=1$, $\mathsf{Auth}=\{1,2\}$, $\mathsf{RegDone}=0$,
$\mathsf Q_{\mathsf{reg}}=\emptyset$, and $w_{\mathsf{reg},1}^{\mathsf{del}}\neq\perp$ and $w_{\mathsf{reg},2}^{\mathsf{del}}\neq\perp$, then:
\begin{itemize}[leftmargin=*,nosep]
   \item If $3\in\mathsf{Cor}$: do nothing further as after corruption, completion/abort is controlled by $\cA$.
	\item If $3\notin\mathsf{Cor}$:
	parse $w_{\mathsf{reg},1}^{\mathsf{del}}=\langle \sid,\varpi_1,K_{1,3}^{(1)},K^\star\rangle$ and
	$w_{\mathsf{reg},2}^{\mathsf{del}}=\langle \sid,\varpi_{2a},\varpi_{2b},K_{1,3}^{(2)}\rangle$.
	Require $K_{1,3}^{(1)}=K_{1,3}^{(2)}$ and set $K_{1,3}^\star\gets K_{1,3}^{(1)}$. If parsing fails then do nothing further.
	Require $K^\star = K$ (the stored finalized public key).
	In the real protocol, the honest joiner additionally requires agreement of both senders' $\Fpub$ publications before initiating registration; the simulator matches this by conditioning $\mathsf{RegGo}$ on $\Fpub$ consistency (see proof of Theorem~\ref{thm:SDKG-UC}). The $K^\star = K$ check here is a secondary payload-level guard; otherwise do nothing further.
	Otherwise, have $P_3$ attempt the KeyBox installation procedure from Algorithm~\ref{alg:reg} using $(\varpi_1,\varpi_{2a},\varpi_{2b},\mathsf{ad}_{31},\mathsf{ad}_{32},\mathsf{ad}_{23},K^\star,K_{1,3}^\star)$,
	i.e., by invoking $\OpenFromPeer$ and $\Load(\cdot)$ exactly as in Algorithm~\ref{alg:reg}.
	If all invoked $\Load$ calls return $\ok$, then set $\mathsf{RegDone}\gets 1$ and $\mathsf{RegPending}\gets 0$ and output $(\mathsf{registered},3,\sid)$ to $P_3$ and $\cA$. Otherwise do nothing further.
\end{itemize}
						
			\smallskip			
			\ding{169} \textsf{Corruptions:} Upon $(\mathsf{Corrupt},i)$ from $\cA$, add $i$ to $\mathsf{Cor}$ and reveal $P_i$'s local state.
			Thereafter, $\cA$ controls $P_i$ and may invoke $\FKeyBox^{(i)}.\Load/\Use$ directly;
			$\FSDKG$ does not mediate these calls.
						
			\smallskip
\smallskip
\ding{169} \textsf{Programming (simulator-only):} Before finalization, $\s$ may once send
$(\mathsf{Program},\sid,x_1^\star,x_2^\star,\sigma_{3,1}^\star,\sigma_{3,2}^\star)$; set
$x_1\gets x_1^\star$, $x_2\gets x_2^\star$, $\sigma_{3,1}\gets\sigma_{3,1}^\star$, $\sigma_{3,2}\gets\sigma_{3,2}^\star$.
Then run \textsf{TryFinalize}. Finalization occurs only if the transcript slots are filled and the programmed values satisfy the \textsf{TryFinalize} consistency checks.}}
	\caption{Transcript-driven ideal functionality $\FSDKG$.}
	\label{fig:FSDKG}
\end{figure}

\begin{remark}[$\Fpub$-agreement enforcement is simulator-mediated]
\label{rem:fpub-agreement-sim}
In the real protocol, honest~$P_3$ gates registration on receiving
consistent $\Fpub$ publications from both $P_1$ and~$P_2$
(Algorithm~\ref{alg:reg}).
$\FSDKG$ does not enforce this condition directly because $\Fpub$ is a
shared hybrid functionality whose state $\FSDKG$ cannot inspect.
Instead, the UC proof of Theorem~\ref{thm:SDKG-UC} shows that the
simulator~$\s$ conditions the $\mathsf{RegGo}$ message to~$\FSDKG$ on
$\Fpub$ consistency: if $\s$ detects that a corrupted sender published a
value on~$\Fpub$ for~$\sid$ different from the finalized~$K$, then $\s$
does not forward $\mathsf{RegGo}$ and registration does not complete in
the ideal world, matching the real-world abort behavior.
This delegation to the simulator is not a specification gap but a standard
UC modeling pattern for conditions involving shared hybrid resources.
\end{remark}

\section{UC Security}\label{subsec:SDKG-UC}
In the ideal functionality $\FSDKG$ (Fig.~\ref{fig:FSDKG}), to avoid distributional mismatches under adaptive corruption,
we use a transcript-driven finalization rule. In any accepting transcript, the UC-extractable affine AoKs and the
$s32$-digest check pin down unique values $(x_1,x_2)$, the value $\sigma_{3,2}$ bound into $h_{3,2}$, and (via the affine AoK for $X_1$) the value $\sigma_{3,1}$ bound by $D_1=X_1-\sigma_{2,1}\G-(M_1+2B_1)=\sigma_{3,1}\G$.
Accordingly, $\FSDKG$ does not sample $(x_1,x_2,\sigma_{3,1},\sigma_{3,2})$ at finalization time. Instead, finalization is gated on a single simulator-only \textsf{Program} message supplying these transcript-consistent values; $\FSDKG$ checks consistency against the stored transcript before outputting $K$ and installing KeyBox shares. In \(\FSDKG\), the base-run record \(\TSDKG=(\mathsf{T}_1,\mathsf{T}_2,\mathsf{T}_3)\) is supplied via the
ideal functionality’s adversary interface (i.e., by the ideal-world simulator) through the
\(\mathsf{Tin}_1/\mathsf{Tin}_2/\mathsf{Tin}_3\) messages. This is purely bookkeeping used to couple the
ideal execution to the simulated real execution; it should not be interpreted as adversary-visible leakage.
The slot tuples \(\mathsf{T}_i\) may include fields that are transmitted only over \(\Fchan\)
in honest executions (e.g., \(\sigma\)-scalars), even though the adversary learns at most the explicit
\(\Fchan\) leakage unless it corrupts an endpoint.

KeyBox installation effects in $\FSDKG$ are realized via the KeyBox-driver wrapper mechanism of Remark~\ref{rem:kb-wrapper}.

\begin{lemma}[Latent (marginal) uniformity of the SDKG key]
	\label{lem:sdkg-latent-unif}
	Fix any session identifier $\sid$ of $\widehat{\mathrm{\Psi}}^{(3)}_{\mathsf{SDKG}}$.
	Let $x_1:=\sigma_{1,1}+\sigma_{2,1}+\sigma_{3,1}$ and $x_2:=\sigma_{1,2}+\sigma_{2,2}+\sigma_{3,2}$ denote the
	(conceptual) scalars defined by the parties' sampled $\sigma$-values in the base run, and define the
	candidate key scalar and group element $\hat{k} := 3x_1-2x_2 \bmod p, \widehat K := \hat{k}\G \in \mathbb G.$
If at least one of $\sigma_{3,1}$ or $\sigma_{3,2}$ is sampled uniformly, then
$\widehat K$ is uniform in $\mathbb G$ marginally over that honest pad randomness.
\end{lemma}

\begin{proof}
	For the Lagrange weights $u=2$ and $v=3$, Observation~\ref{obs:sec7} gives
	\[
	\hat{k} \;=\; (m_1+m_2) + 3\sigma_{3,1} - 2\sigma_{3,2}\; \bmod p,
	\]
	and hence $\widehat K=\hat{k}\G$.
	If $\sigma_{3,1}$ is sampled honestly, then $3\sigma_{3,1}$ is uniform in $\Zp$ and therefore
	$\hat{k}=( (m_1+m_2)-2\sigma_{3,2}) + 3\sigma_{3,1}$ is uniform in $\Zp$.
	If instead $\sigma_{3,2}$ is sampled honestly, then $-2\sigma_{3,2}$ is uniform in $\Zp$ and therefore
	$$\hat{k}=( (m_1+m_2)+3\sigma_{3,1}) - 2\sigma_{3,2}$$ is uniform in $\Zp$.
	In either case $\hat{k}$ is uniform in $\Zp$, and since $z\mapsto z\G$ is a bijection, $\widehat K=\hat{k}\G$ is uniform in $\mathbb G$. Since we make no fairness/guaranteed-output-delivery claim, the distribution of the delivered key
	conditioned on completion may be biased by selective abort (cf.\ Cleve~\cite{Cleve[86]}).
	\qed
\end{proof}

\subsection{From $\FSDKG$ to the standard NXK-DKG interface}
Fix the 1+1-out-of-3 star access structure
\[
\mathrm{\Gamma}_0 \ :=\ \bigl\{\{P_1,P_2\},\{P_1,P_3\}\bigr\}.
\]
The functionality $\FSDKG$ (Fig.~\ref{fig:FSDKG}) is deliberately transcript-driven: besides producing the
public key $K$ and installing non-exportable KeyBox shares, it (i) records the public transcript
items that determine acceptance and (ii) offers a simulator-only one-shot \textsf{Program} hook used
exclusively to couple real and ideal executions under adaptive corruptions and straight-line extraction.

\begin{definition}[Standard NXK-star DKG functionality $\FDKGstarNXK$]\label{def:fdkg-star-nxk}
	\emph{
		Fix the following deterministic polynomial-time wrapper ITM $\WDKG$ that sits between the environment and an instance of $\FSDKG$ and
		exposes only the interface that we regard as the standard NXK-DKG API for $\mathrm{\Gamma}_0$, namely:
		\begin{enumerate}
			\item the public key output $K$ for each accepting session, with no output on abort;
			\item the induced KeyBox installation effects (via the surrounding $\FKeyBox$ instances); and
			\item if invoked, the post-finalization \textsf{registered} events for device registration (RDR).
		\end{enumerate}
		The wrapper suppresses $\FSDKG$'s internal transcript bookkeeping (e.g., the transcript slots
		$\mathsf{T}_1,\mathsf{T}_2,\mathsf{T}_3$) and does not expose the simulator-only $\mathsf{Program}$ port to the
		environment. Define
		\[
		\FDKGstarNXK \ :=\ \WDKG \circ \FSDKG.
		\]}
\end{definition}

As usual, $\FDKGstarNXK$ makes no fairness guarantee: a corrupted party and/or the adversary-controlled
scheduler may delay or prevent completion (selective abort), and may even condition its decision to abort on
partial information about the (would-be) key.\footnote{Accordingly, while we can prove latent/marginal
	uniformity of the key prior to conditioning on completion, the distribution of the output key conditioned on
	completion may be biased via selective abort. This limitation is inherent in DKG/coin-flipping style tasks
	without guaranteed output delivery; see, e.g., Cleve~\cite{Cleve[86]}.}
Consequently, the uniformity guarantee we target is the conventional UC-DKG latent (marginal) uniformity (prior
to conditioning on completion), not uniformity conditioned on completion.

\begin{lemma}[Closure under interface restriction]\label{lem:sdkg-interface-refinement}
	Let $\WDKG$ be the fixed wrapper from Definition~\ref{def:fdkg-star-nxk} and let
	$\FDKGstarNXK := \WDKG \circ \FSDKG$.
	If a protocol $\mathrm{\Psi}$ UC-realizes $\FSDKG$ in some model $\mathcal M$, then $\mathrm{\Psi}$ UC-realizes
	$\FDKGstarNXK$ in the same model $\mathcal M$.
\end{lemma}

\begin{proofsketch}
	This is UC closure under efficient local post-processing.
	Fix any PPT adversary $\cA$ for $\mathrm{\Psi}$ when the ideal functionality is $\FDKGstarNXK$.
	Define an adversary $\cA'$ for $\mathrm{\Psi}$ when the ideal functionality is $\FSDKG$ that runs $\cA$
	and locally applies the same forwarding/output-filtering that $\WDKG$ applies between the
	environment and $\FSDKG$ (i.e., it suppresses exactly the bookkeeping outputs/ports hidden by $\WDKG$).
	By the hypothesis that $\mathrm{\Psi}$ UC-realizes $\FSDKG$, there exists a PPT simulator $\s'$ for $\cA'$ in the
	$\FSDKG$-ideal execution. Composing $\s'$ with the same local filtering yields a simulator for $\cA$
	in the $\FDKGstarNXK$-ideal execution. \qed
\end{proofsketch}

\begin{observation}[DKG properties captured by $\FDKGstarNXK$]
	\label{lem:fdkg-properties}
	\emph{For $\mathrm{\Gamma}_0 := \{\{P_1,P_2\},\{P_1,P_3\}\}$ and any session $\sid$ of $\FDKGstarNXK$, the following properties hold:
	\begin{enumerate}[leftmargin=*,label=(\roman*),nosep]
		\item NXK:
		No non-corruption interface ever outputs the secret scalar $k$ or any share-deriving plaintext
		(the pre-install adaptive-corruption window, in which retained $\sigma$-scalars may be revealed, is addressed
		in the proof of Theorem~\ref{thm:SDKG-UC}).
		All long-term shares remain confined to the parties' $\FKeyBox$ instances, and the adversary has only
		black-box access via the fixed admissible KeyBox profile.
		\item Uniqueness / consistency of the induced key:
		If the session completes (outputs $K$), then there exists a unique scalar $k\in\Zp$ such that
		$K=k\G$, and the KeyBox-installed shares are consistent with a single global key under $\mathrm{\Gamma}_0$
		(i.e., the two authorized sets induce the same $k$). Concretely, the accepting transcript uniquely
		determines $x_1,x_2$ and hence $k=(3x_1-2x_2)\bmod p$ (Lemma~\ref{lem:SDKG-key-transcript}).
		\item Secrecy against unauthorized corruption sets:
		For any corruption set $B\subseteq\{P_1,P_2,P_3\}$ with $B\notin\mathrm{\Gamma}$, the functionality does not enable
		recovery of $k$: the adversary never obtains enough long-term shares to reconstruct $k$, and any leakage from
		interacting with corrupted parties' KeyBoxes is limited to the admissible KeyBox profile and hence
		is simulatable from public information under key-opacity (Assumption~\ref{assump:keybox-opacity}).
		\item Latent (marginal) uniformity (no fairness):
		As usual for UC-DKG, no fairness is guaranteed; however, if at least one of the designated pad scalars is honestly
		sampled, then the resulting public key $K$ is uniform in $\mathbb G$ marginally (prior to conditioning on completion)
		(Lemma~\ref{lem:sdkg-latent-unif}).
	\end{enumerate}}
\end{observation}

\begin{proofsketch}
	Item (i) follows from the NXK/$\FKeyBox$ model (Definition~\ref{KeyBox} and Remark~\ref{rem:transport-vs-nxk}).
	Item (ii) is Lemma~\ref{lem:SDKG-key-transcript} together with the deterministic KeyBox-installation logic
	in $\FSDKG$ / $\FDKGstarNXK$.
	Item (iii) is exactly the intended DKG secrecy statement under $\mathrm{\Gamma}$ in the NXK model:
	unauthorized sets lack a reconstructing share set, and admissible KeyBox interactions do not reveal
	share-deriving plaintexts under key-opacity.
	Item (iv) is Lemma~\ref{lem:sdkg-latent-unif}. \qed
\end{proofsketch}

For the star access structure $\mathrm{\Gamma}_0=\{\{P_1,P_2\},\{P_1,P_3\}\}$, in the $(\FKeyBox,\Fchan,\Fpub)$-hybrid model,
$\mathcal{W}_{\mathsf{DKG}}\circ \FSDKG$ matches the standard NXK-DKG semantics for $\mathrm{\Gamma}_0$ in the following concrete sense:
\begin{enumerate}
	\item Latent (marginal) uniformity under one honest pad:
	Define the transcript-derived candidate key $\widehat K := 3X_1-2X_2$
	whenever $X_1$ and $X_2$ are defined (cf. Algorithm~\ref{alg:AccSDKG}, Step~\textbf{D4}).
	If at least one of $\sigma_{3,1}$ or $\sigma_{3,2}$ is sampled by an honest party, then
	$\widehat K$ is uniform in $\mathbb G$, marginally over that honest pad randomness (i.e., prior to conditioning on session completion)
	(Lemma~\ref{lem:sdkg-latent-unif}).
	\item Consistency with a single key: Acceptance implies there exist unique transcript-defined $x_1,x_2\in\Zp$ and thus a unique
	$k\in\Zp$ such that $K=k\G$ and the KeyBox-installed shares are consistent with $k$ under $\mathrm{\Gamma}_0$
	(Lemma~\ref{lem:SDKG-key-transcript} and the definition of the installed slots in Fig.~\ref{fig:FSDKG}).
	\item Non-exportability: All long-term shares remain confined to $\FKeyBox$ by construction, and share-deriving material are NXK-restricted
	(Remark~\ref{rem:transport-vs-nxk}).
\end{enumerate}
As usual, no fairness guarantee is made: a corrupted party may selectively abort after learning enough to decide whether to proceed.

\begin{remark}[On transcript-driven finalization]
	\label{rem:why-transcript-driven-not-overfit}
	The simulator-only \textsf{Program} interface in $\FSDKG$ is a proof device used to handle
	adaptive corruptions with straight-line extraction: it delays committing to certain
	internal scalars until the simulator can derive them from the public transcript.
	\textsf{Program} does not give the ideal world extra freedom to choose outputs:
	finalization is gated by deterministic transcript checks, and in any accepting transcript
	the induced values are uniquely determined (Lemma~\ref{lem:SDKG-key-transcript}).
	The wrapper $\WDKG$ hides \textsf{Program} and all transcript slots from applications, yielding the
	standard NXK-DKG interface $\FDKGstarNXK$.
\end{remark}

\subsection{Formal necessity of USV for commit-only alternatives under hardened profiles}
\label{subsec:why-usv-straightline}
Section~\ref{subsec:usv-need} gave the short design-space motivation; this section
provides the formal obstruction and proof in the NXK/$\FKeyBox$ + \gROCRP\ model.

\noindent\emph{Scope.}
The concrete $\textsf{SDKG.LeafInit}$ of Section~\ref{SDKG}
(Fig.~\ref{fig:SDKGLeafInit}) returns $\sigma$-values from which $m_2$ is
host-recoverable (cf.\ the profile convention in Section~\ref{SDKG}).
In that profile, the host could in principle compute and publish $M_2=m_2\G$
directly (Option~1 below). The necessity argument below addresses a
design-space question: whether USV can be replaced by a plain hiding commitment
whose adversary-visible footprint is only $C_2=\Commit(m_2;r_2)$, without any
additional public material that determines $M_2$. This is relevant both to the
hypothetical hardened profile in which the KeyBox does not return share-deriving
$\sigma$-values at all, and---independently---to any deployment where the
host-recovered $m_2$ is erased before the transcript is finalized and no
equivalent opening-to-$\mathbb{G}$ material has been published.

The concrete SDKG base run already transmits a full USV instance $(C_2,\zeta_2)$, and hence
$M_2$ is deterministically computable from the transcript as $M_2=\Open_M(\pp,C_2,\zeta_2)$.
The purpose of this section is to address a natural alternative design in which one attempts to replace USV by a generic hiding commitment
to $m_2$---that is, a commit-only transcript that contains $C_2$ but no public material that fixes $m_2\G$---and then
tries to recover $M_2$ by straight-line extraction of $m_2$ under NXK/$\FKeyBox$. We show this fails unless the model
is strengthened or hiding is broken. Equivalently, in the NXK setting one must provide either explicit $M_2$
or USV-style public opening material that deterministically fixes $M_2$ without exporting $m_2$.

\subsubsection{What must be transcript-defined (and why).}
\label{subsec:design-space-M2}
In the SDKG base run, both honest verification and the UC simulator must be able
to compute certain group elements as deterministic functions of the public
transcript, in straight-line. Concretely, verification (Algorithm~\ref{alg:AccSDKG})
and transcript-driven programming of $\FSDKG$ (Fig.~\ref{fig:FSDKG}) require forming the
leaf-dependent auxiliary point
\[
Y_2 := M_2 + 3B_2 \qquad\text{where}\qquad M_2 := m_2\G,
\]
for a leaf scalar $m_2\in\Zp$ whose group image $M_2=m_2\G$ does not appear in the adversary-visible transcript $\tau_{\mathsf{pub}}$.
Equivalently, the transcript must determine $M_2$ via a deterministic PPT map
(an ``opening-to-$\mathbb G$'' map), so that $Y_2$ and the subsequent affine auxiliaries
are well-defined from the transcript alone.

\paragraph{Design options under NXK/$\FKeyBox$.}
Under our KeyBox/\gROCRP\ model, there are essentially three ways to make $M_2$
transcript-defined:

\begin{enumerate}
\item Publish $M_2$ directly:
If the leaf can compute $M_2=m_2\G$ outside the KeyBox, then it can simply transmit $M_2$.
This trivially makes $Y_2$ transcript-defined and removes the need for USV.
We treat this as a different design point: it bypasses the commit-shaped transcript setting and assumes the leaf
has a way to compute and publish $m_2\G$ without relying on straight-line extraction from a KeyBox-internal proof.
Our focus is the hardened profile-centric setting in which $m_2$ is generated inside the KeyBox/profile-adapter
and the admissible profile does not expose a generic ``export $m\G$'' interface for fresh ephemeral scalars. For example, when the admissible KeyBox profile corresponds to a cloud KMS role that permits signing/derivation but denies public-key retrieval~\cite{aws-kms-getpublickey,gcp-kms-roles-perms,azure-kv-keys-details}. In that
setting, an alternative attempt using only a hiding commitment (commit-only transcript) would fail.
(In the concrete $\textsf{SDKG.LeafInit}$ profile of Section~\ref{SDKG},
the returned $\sigma$-values do allow host-side recovery of $m_2$ and hence of
$m_2\G$; USV is retained so the transcript format is self-certifying without
relying on any host-side scalar-to-group computation.)

	\item Commit to $m_2$ and extract via an opening-AoK: One could have the leaf
	publish $C_2\gets\mathsf{Commit}(m_2;r_2)$ and an AoK $\pi_{\mathrm{open}}$ of an opening,
	then let the verifier/simulator extract $m_2$ and set $M_2=m_2\G$. In our setting this
	fails in straight-line: witness-bearing computation may be delegated to a state-continuous
	KeyBox (Assumption~\ref{assump:tee-continuity}), so rewinding/forking is unavailable; and
	straight-line extraction for our Fischlin-style UC-NIZK-AoKs relies on access to the prover's
	\gROCRP\ query log in the proof context (Remark~\ref{rem:oracle-tape}), which is hidden by
	local-call semantics (Definition~\ref{def:GNPRO}). Forcing $\pi_{\mathrm{open}}$ outside the KeyBox restores observability but requires the leaf to
	materialize $(m_2,r_2)$, or an equivalent caller-invertible affine image sufficient to derive
	$M_2$, in its non-KeyBox state, which contradicts the hardened/minimal-profile design point
	of this section: $m_2$ is not available outside the KeyBox/profile-adapter and the admissible
	profile exposes no $\mathsf{export}\text{-}m_2$ or $\mathsf{export}\text{-}(m_2\G)$ interface\footnote{If one allows such materialization, one can instead publish $M_2=m_2\G$ directly as mentioned previously.}.	Lemma~\ref{lem:commit-only-no-M2} formalizes this obstruction.
	
	\item Add public opening material that deterministically fixes $M_2$ without exporting $m_2$:
	in the hardened profile-centric setting, the leaf runs $(C_2,\zeta_2)\gets\Cert(\pp,m_2)$ inside the KeyBox boundary
	and releases only $(C_2,\zeta_2)$; everyone derives $M_2=\Open_M(\pp,C_2,\zeta_2)$ outside. USV provides this and makes $Y_2$
	(and the induced transcript-defined auxiliaries and shares) well-defined from the transcript,
	enabling both honest verification and straight-line UC simulation/programming.
\end{enumerate}

Hence, the SDKG verification predicate and the UC proof need $M_2$ (hence $Y_2$) to be
a deterministic function of the public transcript, while neither $m_2$ nor $m_2\G$ appears in $\tau_{\mathsf{pub}}$ (the $\sigma$-values from which $m_2$ is host-recoverable travel only over $\Fchan$ and are therefore transcript-private; cf.\ the profile convention in Section~\ref{SDKG}), and the KeyBox/local-call semantics prevent straight-line extraction of $m_2$ from a KeyBox-internal opening-AoK.
USV resolves this by providing a public ``opening-to-$\mathbb G$'' map
$M_2=\Open_M(\pp,C_2,\zeta_2)$ that is transcript-defined yet does not violate KeyBox's confidentiality assumptions. Our subsequent transcript analysis is stated in terms of $\Open_M(\pp,C_2,\zeta_2)$ rather
than in terms of extracted witnesses.

In our SDKG verification predicate (Algorithm~\ref{alg:AccSDKG}) and in the transcript-driven idealization (Fig.~\ref{fig:FSDKG}),
the value
\[
Y_2 := M_2 + 3B_2
\]
must be computable in straight-line from the transcript-visible view, where $B_2$ is public and $M_2:=m_2\G$ is induced by the leaf’s hidden scalar.
Accordingly, any design that replaces USV with a commit-only transcript must still enable a PPT straight-line derivation of $Y_2$ from the public transcript.

\begin{lemma}[Commit-only transcripts cannot define $M_2$ under NXK/$\FKeyBox$]
	\label{lem:commit-only-no-M2}
	Fix the NXK/$\FKeyBox$ setting in the \gROCRP\ model. Let $\mathsf{Com}=(\Commit,\Open)$ be a computationally
	hiding commitment scheme for messages in $\Zp$. Consider any protocol variant in which the leaf's \emph{transcript-visible} contribution contains only a commitment
	$C_2:=\Commit(m_2;r_2)$ to $m_2\in\Zp$ and a public point $B_2\in\mathbb G$, and the transcript contains no additional public material from which
	$M_2:=m_2\G$ is deterministically computable. Let $\tau_{\mathsf{pub}}$ denote the transcript-visible view.

	Assume further that the protocol variant satisfies \emph{transcript sampleability}:
	there exists a PPT algorithm $\mathsf{Sample}(\pp,C_2)$ that, given only public parameters
	and a commitment $C_2$, outputs $\tau_{\mathsf{pub}}^{\setminus C_2}$ with a distribution
	computationally indistinguishable from the real conditional distribution of
	$\tau_{\mathsf{pub}}^{\setminus C_2}$ given $C_2$ in an honest execution\footnote{This condition
	holds for commit-only variants where the non-$C_2$ transcript fields are generated independently
	of $m_2$, or where all $m_2$-dependence passes through zero-knowledge proofs or UC-NIZK
	instances for which a universal simulation interface is available.}.

	Suppose there exists a PPT straight-line algorithm $\Derive_Y$ such that
	\[
	\Pr\big[\Derive_Y(\tau_{\mathsf{pub}})= m_2\G + 3B_2\big]
	\]
	is non-negligible (over the protocol coins and $\Derive_Y$'s coins), where $\Derive_Y$ is given only $\tau_{\mathsf{pub}}$ together with black-box access to corrupted KeyBoxes via admissible profiles. Then the hiding property of $\mathsf{Com}$ is violated.
\end{lemma}

\begin{proof}
	Assume for contradiction that there exists a PPT straight-line simulator $\s$ that
	outputs the correct $Y_2$ with non-negligible probability.
	
	Let $\tau_{\mathsf{pub}}$ denote the \emph{transcript-visible} portion of an execution
	(\textit{Reader Note}~\ref{box:export-visibility}), i.e., everything outside honest KeyBoxes and outside
	authenticated confidential channels.
	In the \emph{commit-only} design point of this lemma, the leaf’s only $\tau_{\mathsf{pub}}$-dependence on $m_2$
	is through $C_2=\Commit(m_2;r_2)$ (by hypothesis: the transcript contains no additional public material that
	deterministically fixes $M_2=m_2\G$).
	By the transcript-sampleability hypothesis, the remaining public fields $\tau_{\mathsf{pub}}^{\setminus C_2}$ admit a PPT sampler
	$\mathsf{Sample}(\pp,C_2)$ that outputs $\tau_{\mathsf{pub}}^{\setminus C_2}$ distributed as in a real execution
	conditioned on the given $C_2$.
	
	\smallskip
	\noindent Case 1 (transcript-only derivation):
	Assume $\s$’s output $Y_2$ is determined by $\tau_{\mathsf{pub}}$ alone (i.e., it does not rely on extracting $m_2$
	from a KeyBox-internal interaction). Formally, fix the induced PPT map
	$\Derive_Y(\tau_{\mathsf{pub}})$ that outputs the same $Y_2$ value that $\s$ outputs on public view
	$\tau_{\mathsf{pub}}$.
	Since $B_2$ is public, define
	\[
	\Derive_M(\tau_{\mathsf{pub}})\ :=\ \Derive_Y(\tau_{\mathsf{pub}})\ -\ 3B_2 \ \in \mathbb G.
	\]
	By the non-negligible-success hypothesis,
	$\Derive_Y(\tau_{\mathsf{pub}})=Y_2=M_2+3B_2$, and thus $\Derive_M(\tau_{\mathsf{pub}})=M_2=m_2\G$,
	with non-negligible probability.
	
	\noindent
	We now build a hiding adversary $\mathcal B$ against $\mathsf{Com}$.
	On input a hiding challenge $C^\star=\Commit(m_b;r)$ for $b\in\{0,1\}$ (for chosen distinct messages $m_0,m_1\in\Zp$),
	$\mathcal B$ samples $\tau_{\mathsf{pub}}^{\setminus C_2}\leftarrow \mathsf{Sample}(\pp,C^\star)$ and sets
	$\tau_{\mathsf{pub}} := (C_2:=C^\star,\tau_{\mathsf{pub}}^{\setminus C_2})$.
	It computes $M^\star := \Derive_M(\tau_{\mathsf{pub}})$ and outputs $b'=0$ iff $M^\star=m_0\G$.
	Because $\tau_{\mathsf{pub}}$ is distributed as a real public transcript conditioned on $C_2=\Commit(m_b;r)$,
	the above correctness implies $\Pr[b'=b]\ge \tfrac12+\epsilon(\lambda)$ for some non-negligible $\epsilon$,
	contradicting computational hiding of $\mathsf{Com}$.
	
	\smallskip
	\noindent Case 2, $\s$ obtains $m_2$ (or $M_2$) by extraction from an opening-AoK:
	To avoid Case~1, the only generic route is to require an AoK $\pi_{\mathrm{open}}$ of an opening $(m_2,r_2)$ and extract
	$m_2$ to set $M_2=m_2\G$. If $\pi_{\mathrm{open}}$ is generated outside the KeyBox so that standard extraction applies, then
	the leaf must first obtain $(m_2,r_2)$, or any caller-invertible affine image of $m_2$ sufficient to
	compute $M_2$, in its \emph{non-KeyBox} state. This contradicts the hardened/minimal-profile premise
	of this subsection, in which $m_2$ is generated inside the KeyBox/profile-adapter and no admissible
	API exports $m_2$ (or $m_2\G$) in the clear\footnote{Again, if such materialization is allowed, the protocol can
	instead publish $M_2=m_2\G$ directly.}. If instead $\pi_{\mathrm{open}}$ is generated inside the KeyBox to preserve NXK,
	then straight-line extraction for our Fischlin-based UC-NIZK(-AoK) mechanisms requires access to the prover's \gROCRP\
	query log in the relevant proof context (Remark~\ref{rem:oracle-tape}), which is hidden by local-call semantics
	(Definition~\ref{def:GNPRO}); and Assumption~\ref{assump:tee-continuity} rules out rewinding/forking the KeyBox.
	Thus $\s$ cannot extract in straight-line in the NXK/$\FKeyBox$ model.
	
	\smallskip \noindent
	Combining the cases, Case~2 is ruled out in the NXK/$\FKeyBox$ model, so any successful straight-line derivation of $Y_2$ reduces to Case~1 and yields a non-negligible break of computational hiding for $\mathsf{Com}$. Therefore, an additional public-opening mechanism that
	deterministically maps commitment material to $M_2$ (e.g., USV or an explicit $M_2$) is necessary. \qed
\end{proof}

The above lemma is specific to the NXK/$\FKeyBox$ execution model, where local-call semantics hide the KeyBox's \gROCRP\ oracle-query tape (Remark~\ref{rem:oracle-tape}) and Assumption~\ref{assump:tee-continuity} forbids rewinding/forking a KeyBox. If one strengthens the model to permit such extraction capabilities, then the commit-only design may instead recover $M_2$ via the opening-AoK route discussed in Case~2 of the proof, and Lemma~\ref{lem:commit-only-no-M2} no longer applies.

\subsection{Main Theorem}
In $\FSDKG$, the state $\mathsf{Q}_{\mathsf{reg}}, \mathsf{D}_{\mathsf{reg}}$ and the (\textsf{Leak}/\textsf{Deliver}/\textsf{Recv}) scheduling are an inlined instance of $\Fchan$ (Fig. \ref{fig:Fchan}), specialized to the two registration payloads with leakage $\ellregone, \ellregtwo.$

\begin{lemma}[Transcript uniquely determines the key]\label{lem:SDKG-key-transcript}
	Let $\TSDKG$ be a base-run (1+1-out-of-3) transcript such that $\AccSDKG(\sid,P_2,P_1,\TSDKG)=1$, and let $M_2,Y_1,D_1,Y_2,D_2,\widehat K$ be the
	values derived by $\AccSDKG$ from $\TSDKG$. Then:
	\begin{enumerate}[label=(\roman*)]
		\item There exist unique $x_1,x_2\in\Zp$ such that $X_1=x_1\G$ and $X_2=x_2\G$.
		\item Except with negligible probability in $\lambda$, the (straight-line) AoK extractors
		applied to the verifying DL subproofs in $\pi_{\aff_1}=(\pi_{Y_1},\pi_{D_1})$ and
		$\pi_{\aff_2}=(\pi_{Y_2},\pi_{D_2})$ recover witnesses $\alpha_i,\delta_i\in\Zp$ such that
		$Y_i=\alpha_i\G$ and $D_i=\delta_i\G$ for $i\in\{1,2\}$, and in particular determine
		\[
		x_1 \;=\; \sigma_{2,1}+\alpha_1+\delta_1 \bmod p,\qquad
		x_2 \;=\; \sigma_{1,2}+\alpha_2+\delta_2 \bmod p.
		\]
		\item Except with negligible probability in $\lambda$,
		$\AccSDKG(\sid, P_2, P_1, \TSDKG)=1$ implies that $D_2$ equals the unique group element
		committed via $h_{3,2}$ in the non-programmable digest context $\CtxSDKG$.
	\end{enumerate}
\end{lemma}

\begin{proofsketch}
	Item (i) is unconditional because $\mathbb G$ is cyclic of prime order with generator $\G$, making $z\mapsto z\G$ a bijection on $\Zp$. Because $\AccSDKG(\sid, P_2, P_1, \TSDKG)=1$ (Algorithm~\ref{alg:AccSDKG}),
	the DL verifications for $Y_1,D_1,Y_2,D_2$ all accept. By the AoK property
	of $\mathrm{\Pi}_{\mathsf{DL}}$, except with negligible extraction error the straight-line extractor outputs
	$\alpha_i,\delta_i$ such that $Y_i=\alpha_i\G$ and $D_i=\delta_i\G$ for $i\in\{1,2\}$.
	From the definitions $Y_1:=M_1+2B_1$ and $D_1:=X_1-\sigma_{2,1}\G-Y_1$, we get
	$$X_1=\sigma_{2,1}\G+Y_1+D_1=(\sigma_{2,1}+\alpha_1+\delta_1)\G,$$ and similarly
	$$X_2=(\sigma_{1,2}+\alpha_2+\delta_2)\G,$$ yielding item (ii).
	
	For item (iii), consider the Round~1 value $h_{3,2}$ and the later derived point $D_2$ used in Check~\textbf{C5}.
	If the sender is honest, it computes
	$h_{3,2}=H_{\mathrm{s32}}(\langle \sid,\cid_2,\sigma_{3,2}\G\rangle)$ and thus fixes the point
	$U:=\sigma_{3,2}\G$ before $D_2$ is formed.
	If the sender is corrupted, there are two possibilities:
	(i) it queried $H_{\mathrm{s32}}(\langle \sid,\cid_2,U\rangle)$ for some point $U\in\mathbb{G}$ and set $h_{3,2}$ to the
	reply, or
	(ii) it outputs a fresh $\lambda$-bit guess for $h_{3,2}$ without querying.
	In case (ii), since $H_{\mathrm{s32}}(\langle \sid,\cid_2,D_2\rangle)$ is uniform conditioned on the adversary’s view,
	Check~\textbf{C5} holds with probability at most $2^{-\lambda}$.
	In case (i), Check~\textbf{C5} implies
	$H_{\mathrm{s32}}(\langle \sid,\cid_2,D_2\rangle)=H_{\mathrm{s32}}(\langle \sid,\cid_2,U\rangle)$.
	Thus, unless a collision/second-preimage occurs for $H_{\mathrm{s32}}$, we must have $D_2=U$.
	(In honest executions this specializes to $D_2=\sigma_{3,2}\G$.)
	Finally, since $\widehat K=3X_1-2X_2$, we obtain
	$\widehat K=(3x_1-2x_2)\G$ by algebra. \qed
\end{proofsketch}

\begin{theorem}[UC realization of \(\FSDKG\) by \(\mathrm{\Psi}^{(3)}_{\mathsf{SDKG}}\)]\label{thm:SDKG-UC}
	Fix Fischlin parameter functions $(t(\lambda),b(\lambda),r(\lambda),S(\lambda))$ satisfying
	Definition~\ref{FSdef}.
	Then assuming hardness of DL, PRF security of the nonce-derivation function (Definition~\ref{def:linos-prf}), and the seed integrity invariant (Assumption~\ref{assump:seed-integrity}), protocol $\mathrm{\Psi}^{(3)}_{\mathsf{SDKG}}$ UC-realizes the functionality $\mathcal{F}_{\mathsf{SDKG}}$
	in the $(\mathcal{F}_{\mathsf{KeyBox}},\Fusv,\Fchan,\Fpub)$-hybrid and \gROCRP\ models
	against arbitrary adaptive corruptions of any subset of parties in $\{P_i\}_{i \in [3]}$.
\end{theorem}

\begin{proof}
	Fix any PPT real-world adversary $\mathcal A$ and PPT environment $\mathcal Z$.
	We construct a PPT simulator $\s$ such that
	\[
	\mathsf{Exec}\!\left(\mathrm{\Psi}^{(3)}_{\mathsf{SDKG}},\mathcal A,\mathcal Z\right)
	\approx_c
	\mathsf{Ideal}\!\left(\FSDKG,\s,\mathcal Z\right)
	\]
	in the $(\FKeyBox,\Fusv,\Fchan,\Fpub)$-hybrid and \gROCRP\ models, against adaptive corruptions with secure erasures.
	
	\paragraph{Acceptance predicate.}
	A session is accepting iff the deterministic checks in the protocol description verify
	(USV validity/digest, affine consistency equations, and all UC-NIZK verifications).
	The simulator applies the same checks and mirrors aborts. By Lemma~\ref{lem:usv-verify-implicit}, when $P_1$ is honest through Round~2 the
	$\Fusv.\Verify$ gate is implied by the presence of $\mathsf{T}_2$, hence
	Algorithm~\ref{alg:AccSDKG} captures acceptance.
	If $P_1$ is corrupted, no such implication is required.
	
	\paragraph{UC-NIZK interfaces.}
	All UC-NIZK proofs that matter for extraction use the dedicated UC context(s).
	By the AoK property of the Fischlin-based UC-NIZK,
	there exists a straight-line PPT extractor $\Ext_{\mathsf{DL}}$ that, given a verifying proof
	for a DL statement and the prover's \gROCRP\ query/answer log under the UC proof context,
	outputs the corresponding witness except with negligible probability.
	For an affine proof $\pi_{\aff_i}=(\pi_{Y_i},\pi_{D_i})$, define
	$\Ext_{\mathsf{aff}}(\pi_{\aff_i}) := \bigl(\Ext_{\mathsf{DL}}(\pi_{Y_i}),\ \Ext_{\mathsf{DL}}(\pi_{D_i})\bigr)$.
	Moreover, by the ZK property of the same UC-NIZK in the (programmable) \gROCRP\ UC contexts, there exists a PPT
	simulator $\s_{\mathsf{UC}}$ that can produce accepting proofs without witnesses in those contexts.
	
\paragraph{Simulator \(\s\).}
The simulator \(\s\) runs \(\mathcal A\) as a subroutine and maintains, for each session identifier \(\sid\),
a session record containing the \emph{functionality transcript slots}
\(\TSDKG_\sid := (\mathsf{T}_1,\mathsf{T}_2,\mathsf{T}_3)\) (Fig.~\ref{fig:FSDKG}),
a stage variable, and (when defined) values \(x_1(\sid),x_2(\sid),\sigma_{3,1}(\sid),\sigma_{3,2}(\sid)\).
It forwards all of \(\mathcal A\)'s \gROCRP\ queries to the global oracle \(H\) and logs all
\((\mathsf{ctx},x,H(\mathsf{ctx},x))\) triples for the UC proof contexts used by
\(\mathrm{\Pi}^{\mathsf{UC}}_{\mathsf{DL}}\). For the non-programmable digest context \(\CtxSDKG\),
\(\s\) only forwards queries (no programming).
Whenever \(\s\)'s real-world emulation fixes the value that will populate one of the transcript slots
\(\mathsf{T}_1,\mathsf{T}_2,\mathsf{T}_3\) for session \(\sid\),
\(\s\) supplies it to \(\FSDKG\) by sending the corresponding
\((\mathsf{Tin}_i,\sid,\ldots)\) message on the ideal adversary/simulator interface.
These transcript-slot values are used only to drive \(\FSDKG\)'s deterministic acceptance/finalization logic and may
include fields that are not adversary-visible in honest executions (they can be carried only over \(\Fchan\));
they are never output by \(\FSDKG\) and are hidden from \(\cZ\) by the wrapper \(\WDKG\). $\FSDKG$’s stored transcript slots track the simulated real transcript under adversarial scheduling.

\paragraph{$\Fpub$ publications.}
Since $\Fpub$ is a hybrid functionality present in both real and ideal worlds, the simulator
must reproduce the $\Fpub$ publications that honest parties make in the real protocol.
Concretely, when $\s$ emulates an honest $P_2$ completing Round~3 (publishing $K$ via $\Fpub$), $\s$ sends
$(\mathsf{Publish},\sid,K)$ to $\Fpub$ on behalf of the simulated honest $P_2$.
Likewise, when $\s$ emulates an honest $P_1$ accepting the transcript, $\s$ sends
$(\mathsf{Publish},\sid,K)$ to $\Fpub$ on behalf of the simulated honest $P_1$.
These calls are to the real $\Fpub$ instance in the hybrid model, ensuring that $P_3$
(and any environment $\cZ$ listening on $\Fpub$ delivery ports) observes identical
$\Fpub$ publications in both real and ideal executions.

\paragraph{$\Fchan$ messages.}
Since $\Fchan$ is a hybrid functionality present in both real and ideal worlds, $\s$ must reproduce the $\Fchan$ interactions that honest parties perform in the real protocol.
Concretely, when $\s$'s honest-party emulation produces a Round~1 message from an honest~$P_2$, $\s$ sends $(\mathsf{Send},\sid,c_1)$ to the $\Fchan$ instance for the $(P_2,P_1)$ endpoint, where $c_1$ is the simulated Round~1 payload; this generates the corresponding $\mathsf{Leak}$ ticket that~$\cA$ observes.
Likewise for Round~2 (honest $P_1\to P_2$) and Round~3 (honest $P_2\to P_1$).
When $\cA$ issues $\mathsf{Deliver}$ commands, $\Fchan$ delivers the original honest-party messages, matching the real-world delivery semantics.
Since $\s$ sends the same payloads at the same simulation points, $\cA$ (and $\cZ$) observes identical $\Fchan$ leak tickets, scheduling metadata, and delivery patterns in both worlds.

\paragraph{$\Fusv$ interactions.}
Since $\Fusv$ is a hybrid functionality present in both real and ideal worlds, $\s$ must issue the corresponding $\Fusv$ calls as part of its honest-party emulation.
When $\s$ emulates an honest~$P_2$ in Round~1, $\s$ sends $(\Commit,\sid,\cid_2,P_1,C_2,\zeta_2)$ to~$\Fusv$ on behalf of the simulated honest~$P_2$ and records the receipt~$d$.
When $\cA$'s $\mathsf{Deliver}$ command causes $\Fchan$ to deliver $P_2$'s Round~1 message to the honest~$P_1$, $\s$ concurrently sends $(\mathsf{Deliver},\sid,\cid_2,P_2,P_1)$ to~$\Fusv$ to set $\mathsf{del}\gets 1$ for the corresponding entry, so that the subsequent $\Fusv.\Verify$ call does not return~$\perp$.
$\s$ then emulates honest~$P_1$ in Round~2 by calling $\Fusv.\Verify(\sid,\cid_2,P_2,C_2,\zeta_2)$ and checks the return value, aborting if it returns $0$ or~$\perp$ (matching the real-world $P_1$'s behavior).
These calls ensure that $\Fusv$'s internal table~$\mathsf{T}$ is populated and delivery-activated identically in both worlds, so any subsequent query by a corrupted counterparty to $\Fusv.\Verify$ or $\Fusv.\Open$ returns the same result as in the real execution.
Note that Lemma~\ref{lem:usv-verify-implicit} concerns only the acceptance predicate $\AccSDKG$: it says that $\AccSDKG$ need not separately check the $\Fusv.\Verify$ return value (because its success is implied by the existence of~$\mathsf{T}_2$).
It does not mean $\s$ omits the $\Fusv.\Verify$ call; the call is part of $\s$'s honest-party emulation and its result is guaranteed to match by construction.

	\paragraph{Hybrid argument.}
We define hybrids over the joint execution (potentially many concurrent sessions); $\s$ keeps a separate record per $\sid$.
Since the environment $\cZ$, adversary $\cA$, and all parties are PPT, the total number of sessions initiated and the total number of UC-context
proofs simulated/verified in the execution are bounded by $\poly(\lambda)$. Specifically, for each programmable proof context
$\mathsf{ctx}\in\CtxUC$, the simulator makes at most $m_{\mathsf{ctx}}(\lambda)=\poly(\lambda)$ total calls to
$\SimProgramRO(\mathsf{ctx},\cdot,\cdot)$ across all sessions. Therefore, Lemma~\ref{lem:grocrp-prequery} applies with
this global $m_{\mathsf{ctx}}(\lambda)$, and the resulting union bound over all programming attempts (across all sessions) remains negligible.
Indistinguishability is shown by a standard per-session hybrid argument together with this global union bound.
	
	\medskip
	\noindent PRF assumption for LinOS nonces: In the real execution, each honest KeyBox instance derives LinOS
	nonces via $\PRF(\seed,\cdot)$ (Fig.~\ref{LinOS}).
	By PRF security (Definition~\ref{def:linos-prf}) and the seed integrity
	invariant (Assumption~\ref{assump:seed-integrity}), these derived nonces
	are computationally indistinguishable from independently sampled uniform
	values in $\Zp$.
	Concretely, a standard PRF-to-random hybrid replaces all $\PRF(\seed,\cdot)$
	evaluations with outputs of a truly random function; the distinguishing
	advantage of this step is bounded by the PRF advantage $\negl(\lambda)$.
	After this replacement, the nonces have the same joint distribution as in
	the original (sample-and-store) construction, so all subsequent hybrids
	proceed identically.
	This step introduces an additive $\negl(\lambda)$ term from the PRF
	assumption, which is already absorbed into the asymptotic bound.
	By Lemma~\ref{lem:rollback-robust}, this deterministic derivation
	additionally ensures that even under a rollback of the KeyBox's mutable
	state (violating state continuity for the one-shot seal), the adversary
	cannot obtain two transcripts with the same commitment and different
	challenges, and therefore cannot extract the resident share via Schnorr
	special soundness.
	
	\medskip
	\noindent Hybrid $\Game_0$ (real execution): This is $\mathrm{\Psi}^{(3)}_{\mathsf{SDKG}}$ with honest parties, in the
		$(\FKeyBox,\Fusv,\Fchan,\Fpub)$-hybrid and \gROCRP\ models.
		
	\medskip
\noindent Hybrid $\Game_1$ (simulate honest UC-NIZKs):
		Modify $\Game_0$ as follows: whenever an honest party would output a UC-context UC-NIZK proof, replace it with a proof generated by $\s_{\mathsf{UC}}$ in the corresponding UC context.
		All other values/messages are unchanged.
		By the ZK property of the Fischlin-based UC-NIZK in the UC \gROCRP\ contexts, we have $\Game_1 \approx_c \Game_0$.
		Concretely, $\s$ uses $\SimProgramRO$ to realize the universal simulation interface; because $\SimProgramRO$ fails on already-defined points,
		the only divergence is the pre-query bad event of Lemma~\ref{lem:grocrp-prequery}, which occurs with negligible probability.
		
		\medskip
		\noindent\emph{Bookkeeping.}
		In $\Game_1$, $\s$ still samples and records (per session $\sid$) the honest-party
		local scalars that determine the public points $X_1,X_2$ and the digest check
(e.g., the values used to form $X_1=x_1\G$, $X_2=x_2\G$, $\sigma_{3,1}\G$, and $\sigma_{3,2}\G$),
		exactly as in $\Game_0$; only the sent UC-context proof strings are replaced by
		simulated proofs via $\s_{\mathsf{UC}}$.
		Importantly, $\s$ will never invoke the Fischlin/AoK extractor on any proof output by
		$\s_{\mathsf{UC}}$, which may have been produced using $\SimProgramRO$.
		
		\begin{lemma}[Fresh tagged statements for SDKG extraction]\label{lem:sdkg-fresh-tagged}
			Let $Q$ be the set of UC-context DL statements (across all sessions in the execution) for which the simulator invoked the UC-proof simulator
			$\s_{\mathsf{UC}}$. In any session $\sid$, SDKG verifies
			UC-context DL proofs only on tagged statements of the form $\DLstmt{\sid}{\ell}{M}$ where the label
			$\ell$ is fixed by the proof position (one of $\LblAffOneY,\LblAffOneD,\LblAffTwoY,\LblAffTwoD$).
			Then any UC-context DL proof contributed by a corrupted party in session $\sid$ and verified by an honest party
			is on a tagged statement $x$ satisfying $x\notin Q$.
		\end{lemma}
		
	\begin{proofsketch}
		Fix any session identifier $\sid$. By construction of Hybrid $\Game_1$, the simulator invokes the
		UC-proof simulator $\s_{\mathsf{UC}}$ only for UC-context DL proofs that are generated by honest parties
		in that session (i.e., in honest proof positions). Let
		\[
		Q_\sid \ :=\ \{\,x\in Q : x \text{ is of the form } \DLstmt{\sid}{\ell}{M}\,\}.
		\]
		Then $Q_\sid$ contains exactly the tagged statements for which $\s_{\mathsf{UC}}$ produced simulated
		UC-context proofs in session $\sid$. Consider any UC-context DL proof that is verified by an
		honest party and whose originating party is \emph{currently} corrupted at verification time.
		Two sub-cases arise:

		\smallskip\noindent
		\emph{Case~1: the party was corrupted at proof-generation time.}
		Then $\s$ did not invoke $\s_{\mathsf{UC}}$ for that proof position, so
		the corresponding tagged statement $x\notin Q_\sid$ and thus $x\notin Q$.
		Simulation-extractability applies and $\s$ extracts a valid witness.

		\smallskip\noindent
		\emph{Case~2: the party was honest at proof-generation time but was
		corrupted after the proof-bearing message was queued in $\Fchan$ and
		before delivery.}
		In this case $\s$ \emph{did} invoke $\s_{\mathsf{UC}}$ for the proof
		(since the party was honest at generation time), so the statement $x$
		is in $Q$.  However, $\s$ does not need to extract from this proof:
		it already holds the bookkeeping values
		$(x_1,x_2,\sigma_{3,1},\sigma_{3,2})$ that it committed during
		honest simulation, and these remain valid regardless of the party's
		later corruption status.  $\Fchan$ delivers the original message
		unchanged, so the transcript-defined group elements match the
		bookkeeping.  Therefore the simulator uses the bookkeeping
		path (not extraction) for such proofs, and freshness relative to $Q$
		is not required.

		\smallskip\noindent
		Since the tag $\sid$ is part of the statement, statements from different
		sessions cannot collide; injectivity of $\langle\cdot\rangle$ rules out
		any other collisions. \qed
	\end{proofsketch}
									
\medskip
\noindent Hybrid $\Game_2$ (extract from adversarial UC-NIZKs; abort on failure):
Modify $\Game_1$ as follows: whenever $\mathcal A$ delivers a verifying UC-context affine proof
$\pi_{\aff_i}=(\pi_{Y_i},\pi_{D_i})$ attributed to a party that was
corrupted \emph{at proof-generation time} (proofs generated by a
then-honest party that was later adaptively corrupted are handled via the
simulator's bookkeeping values without extraction; see
Lemma~\ref{lem:sdkg-fresh-tagged}, Case~2), $\s$ runs
$\Ext_{\mathsf{aff}}(\pi_{\aff_i})$ using $\mathsf{Log}_\cA$ in that UC context to obtain $(\alpha_i,\delta_i)$.
If extraction fails or the extracted witnesses do not satisfy the relation, abort the session.

Because $\Game_1$ may expose $\cA$ to simulated proofs produced using $\SimProgramRO$,
the appropriate guarantee here is \emph{simulation-extractability} (Definition~\ref{def:simext}),
not merely plain AoK soundness. Concretely, let $Q$ be the set of UC-context statements
for which $\s$ invoked the simulator $\s_{\mathsf{UC}}$.
By Lemma~\ref{lem:sdkg-fresh-tagged}, every corrupted-party UC-context DL statement on which $\s$ runs extraction
is fresh relative to the simulator-produced statement set $Q$. Therefore, the simulation-extractability
guarantee of Lemma~\ref{lem:gnpro-uc-nizk}(iii) applies to these proofs and we get the bad event
\[
\Bad_{\mathsf{ext/forge}} := \{\text{a corrupted party's UC-context proof verifies but no valid witness is extracted}\}
\]
occurs with probability at most $\negl(\lambda)$, and thus $\Game_2 \approx_c \Game_1$.
		
\medskip
\noindent Hybrid $\Game_3$ (switch to the ideal functionality via transcript-defined programming):
Modify $\Game_2$ by replacing the real-world key-derivation effect with interaction with $\FSDKG$ as follows.
$\s$ continues to simulate the network transcript towards $\mathcal A$ and $\mathcal Z$ and mirrors the
same accept/abort predicate. Whenever $\Game_2$ reaches an accepting transcript for some session $\sid$, $\s$ defines $(x_1(\sid),x_2(\sid),\sigma_{3,1}(\sid),\sigma_{3,2}(\sid))$ as follows,
without extracting from any simulated proof:

Fix a session $\sid$ and suppose the simulator's emulated transcript reaches a point where the
stored transcript slots $\TSDKG_\sid=(\mathsf{T}_1,\mathsf{T}_2,\mathsf{T}_3)$ satisfy
$\AccSDKG(\sid,P_2,P_1,\TSDKG_\sid)=1$ (Algorithm~\ref{alg:AccSDKG}).
Let $Y_1,D_1,Y_2,D_2$ be the derived points from $\TSDKG_\sid$, and write
$\pi_{\aff_i}=(\pi_{Y_i},\pi_{D_i})$. 

\smallskip
\noindent Define $\Bad_{\mathsf{ext}}(\sid)$ to be the event that, for some DL subproof in this session
that is attributed to a corrupted party and verifies under $\cV_{\mathsf{DL}}$,
straight-line extraction fails or yields a value not satisfying the DL relation.
Let $\Bad_{\mathsf{s32}}(\sid)$ be the event that the $s32$-digest check in $\AccSDKG$ holds by
a fresh guess without a prior query or by an oracle collision/second-preimage.
In Hybrid $\Game_2$, the simulator aborts the session on $\Bad_{\mathsf{ext}}(\sid)$, so below we
condition on $\neg\Bad_{\mathsf{ext}}(\sid)$; and by standard RO reasoning,
$\Pr[\Bad_{\mathsf{s32}}(\sid)]$ is negligible.

\smallskip
\noindent Conditioned on $\neg(\Bad_{\mathsf{ext}}(\sid)\lor\Bad_{\mathsf{s32}}(\sid))$, the simulator defines
$(x_1,x_2,\sigma_{3,1},\sigma_{3,2})$ by the following corruption-pattern analysis and then issues the
one-shot $\FSDKG$ programming command $(\mathsf{Program},\sid,x_1,x_2,\sigma_{3,1},\sigma_{3,2})$.

\begin{table}[t]
	\centering
	\small
	\setlength{\tabcolsep}{4pt}
	\renewcommand{\arraystretch}{1.15}
	\begin{tabular}{p{0.11\linewidth}p{0.45\linewidth}p{0.38\linewidth}}
		\hline
		\textbf{Pattern} &
		\textbf{How $\s$ fixes $(x_1,\sigma_{3,1})$} &
		\textbf{How $\s$ fixes $(x_2,\sigma_{3,2})$} \\
		\hline
		$P_1$ honest, $P_2$ honest &
		From the honest emulation record: $X_1=x_1\G$ and $D_1=\sigma_{3,1}\G$ were formed using sampled scalars; $\s$ sets $(x_1,\sigma_{3,1})$ to those recorded values. &
		From the honest emulation record: $X_2=x_2\G$ and $h_{3,2}=H_{\mathrm{s32}}(\langle\sid,\cid_2,\sigma_{3,2}\G\rangle)$ were formed using sampled $\sigma_{3,2}$; $\s$ sets $(x_2,\sigma_{3,2})$ to those recorded values. \\
		\hline
		$P_1$ corrupted, $P_2$ honest &
		Extract $(\alpha_1,\delta_1)$ from $\pi_{\aff_1}=(\pi_{Y_1},\pi_{D_1})$ and set
		$\sigma_{3,1}:=\delta_1$, $x_1:=\sigma_{2,1}+\alpha_1+\delta_1$. &
		As in the all-honest row (recorded from honest $P_2$ emulation). \\
		\hline
		$P_1$ honest, $P_2$ corrupted &
		As in the all-honest row (recorded from honest $P_1$ emulation). &
		Extract $(\alpha_2,\delta_2)$ from $\pi_{\aff_2}=(\pi_{Y_2},\pi_{D_2})$ and set
		$\sigma_{3,2}:=\delta_2$, $x_2:=\sigma_{1,2}+\alpha_2+\delta_2$. \\
		\hline
		$P_1$ corrupted, $P_2$ corrupted &
		Extract $(\alpha_1,\delta_1)$ and $(\alpha_2,\delta_2)$ and define
		$\sigma_{3,1}:=\delta_1$, $x_1:=\sigma_{2,1}+\alpha_1+\delta_1$ and
		$\sigma_{3,2}:=\delta_2$, $x_2:=\sigma_{1,2}+\alpha_2+\delta_2$. &
		(same as left cell) \\
		\hline
	\end{tabular}
	\caption{Explicit determination of the programmed values in Theorem~\ref{thm:SDKG-UC}.}
	\label{tab:sdkg-program-cases}
\end{table}

\smallskip
\noindent Hence, for corrupted parties, the simulator never needs to argue that the extracted witnesses equal some
hidden internal $\sigma$-variables of the adversary. It only needs them to satisfy
$Y_i=\alpha_i\G$ and $D_i=\delta_i\G$ for the transcript-defined points $Y_i,D_i$.
Because $z\mapsto z\G$ is a bijection, these witnesses (when extracted) are uniquely determined by
$Y_i$ and $D_i$, and therefore the programmed values are canonical given the accepting transcript.
Finally, by Lemma~\ref{lem:SDKG-key-transcript}, the resulting programmed key is
$K=(3x_1-2x_2)\G$, matching the real transcript-derived key $\widehat K$ whenever the session accepts.

Then $\s$ invokes the one-shot programming interface of $\FSDKG$:
\[
(\mathsf{Program},\sid,x_1(\sid),x_2(\sid),\sigma_{3,1}(\sid),\sigma_{3,2}(\sid)).
\]
By transcript mirroring, at the point $\s$ invokes $(\mathsf{Program},\sid,\ldots)$ the ideal $\FSDKG$ has already received the same public transcript items via $\mathsf{Tin}_1/\mathsf{Tin}_2/\mathsf{Tin}_3$ (up to adversarial scheduling), and since $\FSDKG$ runs \textsf{TryFinalize} on both transcript stores and \textsf{Program}, the ideal output event is triggered exactly when the simulated transcript becomes accepting.	

\paragraph{Ordering and slot consistency for registration.}
In Fig.~\ref{fig:FSDKG}, the registration handler is gated on $\mathsf{finalized}=1$.
By construction, $\mathsf{finalized}$ is set only inside \textsf{TryFinalize}, and
\textsf{TryFinalize} returns early unless the simulator has already supplied the one-shot
\textsf{Program} message fixing $(x_1,x_2,\sigma_{3,1},\sigma_{3,2})$.
Therefore, every registration attempt occurs
only after the above programming step has fixed these values.

Moreover, at the same point where \textsf{TryFinalize} sets $\mathsf{finalized}\gets 1$,
it deterministically defines the registration scalars
$(\sigma_{3,1},\sigma_{3,2},\sigma_{2,3})$ (Fig.~\ref{fig:FSDKG}, Steps~(1)--(2)) and, for each honest sender,
issues the corresponding KeyBox installation commands that load these scalars into the dedicated
slots $\KBsid{k31}$ (for $P_1$) and $\KBsid{k32},\KBsid{k23}$ (for $P_2$) (Fig.~\ref{fig:FSDKG}).
Consequently, in the real protocol the subsequent sealing calls in Algorithm~\ref{alg:reg}
encrypt exactly these same resident values, i.e.,
\[
\varpi_1 \overset{\mathrm d}{=} \Enc_{\pk_{\mathrm{seal}}^{(P_3)}}(\mathsf{ad}_{31},\sigma_{3,1}),\quad
\varpi_{2a} \overset{\mathrm d}{=} \Enc_{\pk_{\mathrm{seal}}^{(P_3)}}(\mathsf{ad}_{32},\sigma_{3,2}),\quad
\varpi_{2b} \overset{\mathrm d}{=} \Enc_{\pk_{\mathrm{seal}}^{(P_3)}}(\mathsf{ad}_{23},\sigma_{2,3}),
\]
which matches the ciphertext sampling performed by $\FSDKG$ on honest senders in the ideal world.
If a sender is corrupted, both the real protocol and $\FSDKG$ allow the adversary to supply the
delivered payload directly (via $w_1^\star,w_2^\star$), so no additional slot-consistency
condition is required in that case.	

	\noindent By the \gROCRP\ assumption for the non-programmable context $\CtxSDKG, \Pr[\Bad_{\mathsf{s32}}]\le \negl(\lambda)$.
	Conditioned on $\neg(\Bad_{\mathsf{ext/forge}}\lor \Bad_{\mathsf{s32}})$, the extracted scalars satisfy the relations in
	Lemma~\ref{lem:SDKG-key-transcript}, and in particular the real-world output satisfies
	$K_{\mathrm{real}}=(3x_1(\sid)-2x_2(\sid))\G$.
	After programming, $\FSDKG$ outputs
	\[
	K(\sid)=(3x_1(\sid)-2x_2(\sid))\G = K_{\mathrm{real}}.
	\]
	Thus, conditioned on no bad event, the public key output seen by $\mathcal Z$ is identical in $\Game_3$ and $\Game_2$,
	and the simulated transcript/abort behavior is unchanged. Hence $\Game_3 \approx_c \Game_2$.
	
	\paragraph{Registration of $P_3$.}
	In $\Game_3$, when $\FSDKG$ receives $(\mathsf{register},3,\sid)$ and $\mathsf{finalized}=1$, the resulting interaction with
	$\FKeyBox$ is exactly as specified by $\FSDKG$ (Fig.~\ref{fig:FSDKG}).
	The only observable outcome is whether the host obtains $K_3=\Pub(k_3)$ from $\FKeyBox^{(3)}.\Use$. Since $\FKeyBox$ is an ideal functionality whose internal state is never revealed, and since
	$\s$ forwards $\FKeyBox$ calls on behalf of corrupted parties exactly as in the real execution, $\mathcal Z$'s observable
	view of registration matches the real execution. Additionally, $\FSDKG$ emits two adversary-scheduled, length-leaking channel messages
	(via $\mathsf{Leak}/\mathsf{Deliver}/\mathsf{Recv})$ that model the sealed-payload transport in Algorithm~\ref{alg:reg}.
	Concretely, $\FSDKG$ sets the delivered payloads to
	$w_{\mathsf{reg},1}:=\langle \sid,\varpi_1,K_{1,3},K\rangle$ and
	$w_{\mathsf{reg},2}:=\langle \sid,\varpi_{2a},\varpi_{2b},K_{1,3}\rangle$,
	where $\varpi_1,\varpi_{2a},\varpi_{2b}$ are fresh sealing ciphertexts sampled as
	$\Enc_{\pk_{\mathrm{seal}}^{(P_3)}}(\mathsf{ad},\sigma)$.
	Lemma \ref{lem:sdkg-reg-sim} formalizes that: if $P_3$ is corrupted, $\cA$ observes well-formed ciphertexts and can invoke $\OpenFromPeer$ on them
	exactly as in the real protocol; and if $P_3$ is honest, the adversary learns only the explicit length
	leakage $\ellregone,\ellregtwo$.

	\smallskip\noindent\emph{$\Fpub$-disagreement abort.}
	If $\s$ detects that $\cA$ caused a corrupted sender to publish a value on $\Fpub$ for $\sid$ different from the finalized~$K$, then $\s$ does not forward $(\mathsf{RegGo},\ldots)$ to $\FSDKG$. This prevents registration from completing in the ideal world, matching the real-world abort behavior of honest~$P_3$ under $\Fpub$ disagreement (Algorithm~\ref{alg:reg}). More precisely, in the real protocol the honest joiner gates registration on receiving consistent $\Fpub$ publications from both senders; when the simulated $\Fpub$ publications disagree, $\s$ withholds $\mathsf{RegGo}$ so that $\FSDKG$ never reaches its completion check. Since $\Fpub$ is a shared hybrid functionality, the environment observes identical $\Fpub$ publications in both worlds, and in both worlds registration does not complete.
	
	\begin{lemma}[Registration/RDR simulation]\label{lem:sdkg-reg-sim}
		Fix a session $\sid$ after base finalization. Under $\Fchan$ and the KeyBox sealing interfaces
		$\SealToPeer/\OpenFromPeer$ (Fig.~\ref{Fdskg}), the registration phase of Algorithm~\ref{alg:reg} is
		indistinguishable from the registration subroutine implemented by $\FSDKG$ (Fig.~\ref{fig:FSDKG}),
		for any adaptive corruption pattern with secure erasures. 
		
		More precisely:
		(i) if the receiver $P_3$ is honest, the adversary learns only the explicit length leakage
		$\ellregone,\ellregtwo$, and $\FSDKG$ produces exactly the same leakage and
		(ii) if $P_3$ is corrupted, then the delivered payloads in $\FSDKG$ contain ciphertext components sampled
		as $\Enc_{\pk_{\mathrm{seal}}^{(P_3)}}(\mathsf{ad},\sigma)$, which matches exactly the distribution of real
		$\SealToPeer$ outputs, so the adversary’s view (including subsequent $\OpenFromPeer$ calls) is identically
		distributed.
	\end{lemma}
	
	\begin{proofsketch}
		By definition of $\FSDKG$ (Fig.~\ref{fig:FSDKG}), registration is enabled only once
		$\mathsf{finalized}=1$, which implies that $(\sigma_{3,1},\sigma_{3,2},\sigma_{2,3})$ have been
		fixed by \textsf{TryFinalize} (after the simulator’s one-shot \textsf{Program}) and, for honest senders,
		loaded into the corresponding KeyBox slots used by $\SealToPeer$.
		
		If a sender is corrupted, both the real protocol and $\FSDKG$ allow the adversary to supply arbitrary
		payloads (via $w_1^\star,w_2^\star$), yielding identical distributions. If the receiver is honest, $\Fchan$
		reveals only message lengths; $\FSDKG$ matches this via $(\mathsf{Leak},\cdot,\ellregone/\ellregtwo)$.
		If the receiver is corrupted, then in the real protocol each ciphertext is produced by $\SealToPeer$ as
		$c\leftarrow \Enc_{\pk_{\mathrm{seal}}^{(P_3)}}(\mathsf{ad},\sigma)$; $\FSDKG$ samples the same distribution
		explicitly, so the delivered messages are identically distributed. 
		
		If the sponsor (either $P_2$ or an already-registered leaf) is corrupted, it may try to disrupt registration by
		sending malformed or altered ciphertext components in $(\varpi_{2a},\varpi_{2b})$.
		However, $P_3$ recomputes the intended associated-data strings
		$\mathsf{ad}_{32},\mathsf{ad}_{23}$ locally (Algorithm~\ref{alg:reg}) and supplies them to $\OpenFromPeer$; by AEAD
		integrity in $\FKeyBox$ (Fig.~\ref{Fdskg}), any tampering or mix-and-match across sessions/slots causes
		$\OpenFromPeer$ to return $\perp$ and the KeyBox-side installation transaction aborts.
		If instead the corrupted sponsor crafts fresh ciphertexts under the correct associated data but encrypting arbitrary
		scalars, the recovery-share derivation $\Load(\KBsid{k3},g_3,\cdot)$ still uniquely determines the installed share
		$k_3$ via $K_{1,3}+k_3\G=K$ (Definition~\ref{def:sdkg-deriv}). Because $g_3$ constrains only the linear combination
		$\sigma_{2,3}-\sigma_{3,2}$ (not the individual values), a corrupted sponsor may shift both sponsor-state slots by a
		common offset $\mathrm{\Delta}$ without affecting $k_3$
		(cf.\ Theorem~\ref{thm:SDKG-UC-n}, cases~(b)--(c)).
		Thus, a malicious sponsor can at most cause denial of registration or shift the per-device sponsor-state slots, but
		cannot make an honest joiner install a recovery share inconsistent with the established public key~$K$.
		
		Finally, when $P_3$ is honest, both
		worlds attempt the same KeyBox-internal sequence of $\OpenFromPeer$ and $\Load$ calls, so the observable
		success/failure outcome (and any subsequent public $\GetPub$ output) matches. \qed
	\end{proofsketch}
				
\paragraph{Installation timing and \textsf{TryFinalize} activation semantics.}
In the real protocol, $P_2$ installs its KeyBox shares during its Round~3 activation (before $P_1$ receives the message), while $P_1$ installs after receiving and verifying~$P_2$’s Round~3 message. The simulator feeds $(\mathsf{Tin}_3,\sid,\ldots)$ when the emulated honest~$P_2$ produces its Round~3 message. This triggers \textsf{TryFinalize}, which issues $\KBcmd$ to both parties’ KeyBox-driver wrappers (Remark~\ref{rem:kb-wrapper}). Each wrapper processes the command upon the respective party’s activation: $P_2$’s wrapper executes within $P_2$’s current activation (matching $P_2$’s immediate real-world installation), and $P_1$’s wrapper executes when $P_1$ is activated to process the $\Fchan$ delivery of Round~3 (matching $P_1$’s real-world installation timing). Thus, the $\FKeyBox$ state at each party matches in both worlds at every corruption point.

\paragraph{Adaptive corruptions with secure erasures: explicit state consistency.}
We make explicit how the simulator answers adaptive corruptions under Definition~\ref{def:ace}.
For each session $\sid$ and each honest party $P_i$, the simulator maintains a \emph{shadow host state}
$\mathsf{st}_i^{\sid}$ consisting of exactly those host variables that the real protocol retains
(i.e., does not erase) after $P_i$’s most recent honest activation in that session.
Because honest activations are atomic w.r.t.\ corruption (Definition~\ref{def:ace}),
a real corruption reveals precisely the current $\mathsf{st}_i^{\sid}$ and nothing erased within
past activations. The simulator updates $\mathsf{st}_i^{\sid}$ in lockstep with its honest-party emulation
and, upon corruption of $P_i$, returns $\mathsf{st}_i^{\sid}$ (together with the already transcript-visible
messages) as the revealed host state.

Crucially, in the base run the only NXK-restricted values that remain in host RAM across activations
are the $\sigma$-scalars needed for the deferred post-accept $\FKeyBox.\Load$ calls; all UC-NIZK/Fischlin
prover randomness and rarity-search scratch state is erased in the same activation that emits
the proof-bearing message (as stated in the protocol’s erasure discipline), and therefore never appears
in $\mathsf{st}_i^{\sid}$ and is never revealed by later corruptions.

\begin{table}[t]
	\centering
	\small
	\setlength{\tabcolsep}{4pt}
	\renewcommand{\arraystretch}{1.15}
	\begin{tabular}{p{0.17\linewidth}p{0.42\linewidth}p{0.34\linewidth}}
		\hline
		& \textbf{Pre-install shadow state} & \textbf{Post-install shadow state} \\
		\hline
		Honest $P_1$ &
		$\{\sigma_{2,1},\sigma_{1,1},\sigma_{1,3},\sigma_{3,1}\}$ (plus public transcript items) &
		$\emptyset$ (all share-deriving scalars erased immediately after the $\Load$ calls return) \\
		Honest $P_2$ &
		$\{\sigma_{1,2},\sigma_{2,2},\sigma_{2,3},\sigma_{3,2}\}$ (plus public transcript items) &
		$\emptyset$ \\
		Honest $P_3$ & $\emptyset$ & $\emptyset$ \\
		\hline
	\end{tabular}
	\caption{NXK-restricted host variables retained across activations in $\mathrm{\Psi}^{(3)}_{\mathsf{SDKG}}$ under the stated erasure discipline.}
	\label{tab:sdkg-ace-shadow}
\end{table}
\medskip
\noindent\textbf{Consistency with $\FSDKG$.}
Whenever $\s$ supplies $\mathsf{Tin}_1/\mathsf{Tin}_2/\mathsf{Tin}_3$ to $\FSDKG$ for session $\sid$,
any scalar fields not transcript-visible in honest executions (e.g., $\sigma$-values carried over $\Fchan$)
are taken from the same shadow states $\mathsf{st}_i^{\sid}$.
Thus, if an honest party is corrupted before post-accept installation, the corruption reveals exactly the
values that (in both the real and ideal worlds) would be consumed by the pending $\Load$ calls; if it is
corrupted after installation, these values have been erased and both worlds expose only black-box access
to already-installed KeyBox slots (Definition~\ref{KeyBox}).
		
	Let $\Bad := \Bad_{\mathsf{ext/forge}} \cup \Bad_{\mathsf{s32}}$. We have $\Pr[\Bad]\le \negl(\lambda)$, and conditioned on
	$\neg\Bad$, the $\Game_2$-to-$\Game_3$ transition produces identical transcripts, accept/abort behavior, public key
	outputs, and identically distributed corruption views. Together with the
	computational losses absorbed in the earlier $\approx_c$ transitions---the PRF-to-random
	hybrid step (introducing $\Adv_{\PRF}(\lambda)$), the ZK simulation in
	$\Game_0\!\to\!\Game_1$ including the pre-query bound of
	Lemma~\ref{lem:grocrp-prequery}, and the simulation-extractability step in
	$\Game_1\!\to\!\Game_2$---the total simulation error is bounded by
	$\Adv_{\PRF}(\lambda)+\Pr[\Bad_{\mathsf{pre}}]+\Pr[\Bad_{\mathsf{ext/forge}}]+\Pr[\Bad_{\mathsf{s32}}]$,
	each term being $\negl(\lambda)$.
	The registration portion of the simulation follows from Lemma~\ref{lem:sdkg-reg-sim}. Therefore, $$\mathsf{Exec}\!\left(\mathrm{\Psi}^{(3)}_{\mathsf{SDKG}},\mathcal A,\mathcal Z\right)
	\approx_c
	\mathsf{Ideal}\!\left(\FSDKG,\s,\mathcal Z\right).$$
	This completes the proof for Theorem \ref{thm:SDKG-UC}. \qed
\end{proof}

\begin{remark}[Dominating bad events / failure terms]\label{rem:sdkg-bad-events}
	In the proof of Theorem~\ref{thm:SDKG-UC}, the simulation error is dominated by the union of the following explicit bad events:
	\begin{itemize}[leftmargin=*]
		\item $\Bad_{\mathsf{sc}}$: a violation of state continuity (rollback/fork) for any KeyBox instance relied upon by the
		protocol (Assumption~\ref{assump:tee-continuity}). In a concrete realization that approximates state continuity via
		a freshness mechanism that can fail (e.g., counter exhaustion or oracle unavailability), let
		$\varepsilon_{\mathsf{sc}}(\lambda)$ bound $\Pr[\Bad_{\mathsf{sc}}]$; then all distinguishing/forgery bounds in the paper
		should be read as $\negl(\lambda)+\varepsilon_{\mathsf{sc}}(\lambda)$.
		 \item $\Bad_{\mathsf{prf}}$: the PRF security of the nonce-derivation
		function keyed by $\seed$ (Definition~\ref{def:linos-prf}) is broken, enabling
		an adversary to predict or correlate LinOS nonces across sessions.
		Under PRF security and the seed integrity invariant
		(Assumption~\ref{assump:seed-integrity}), $\Pr[\Bad_{\mathsf{prf}}]\le\negl(\lambda)$.
		Notably, by Lemma~\ref{lem:rollback-robust}, even if $\Bad_{\mathsf{sc}}$
		occurs (a rollback of mutable KeyBox state), the deterministic nonce
		derivation prevents the catastrophic ``rollback $\Rightarrow$ key disclosure''
		path: replaying the same $(\sid,K)$ reproduces the same transcript and
		reveals no new information about~$k$.
		\item $\Bad_{\mathsf{rcpt}}$: a second-preimage/collision for a receipt-binding digest in a non-programmable context
		(e.g., for $H(\USVrcpt,\cdot)$ in $\Fusv$).
		\item $\Bad_{\mathsf{s32}}$: either (a) a successful $\lambda$-bit guess of
		$h_{3,2}=H_{\mathrm{s32}}(\langle \sid,\cid_2,D_2\rangle)$ without having queried that point, or (b) a collision/second-preimage
		under $H_{\mathrm{s32}}$ that allows an accepting transcript with $D_2\neq U$, where $U$ is the point committed by $h_{3,2}$.
		\item $\Bad_{\mathsf{pre}}$: a pre-query collision that causes a simulator programming attempt
		$\SimProgramRO(\mathsf{ctx},x,y)$ with $\mathsf{ctx}\in\CtxUC$ to return $\perp$
		(Lemma~\ref{lem:grocrp-prequery}).
		\item $\Bad_{\mathsf{fs}}$: a Fischlin knowledge/soundness failure for any verifying UC-context NIZK,
		i.e., acceptance of a proof for which straight-line extraction fails or yields no valid witness
		(Lemma~\ref{lem:fischlin-negl}).
		\item $\Bad_{\mathsf{hr}}$: honest rejection of an honest Fischlin proof due to a capped rarity-search miss. This is a liveness failure only (abort/retry)
		and is not a security break.
		\item $\Bad_{\mathsf{hdl}}$: an adversary guess of a live KeyBox opaque handle (e.g., in $\mathsf{buf}$),
		allowing unintended resolution/use; bounded by a union bound over polynomially many guesses as $\Pr[\Bad_{\mathsf{hdl}}]\le \poly(\lambda)\cdot 2^{-\lambda}$.
	\end{itemize}
	By a union bound, the security distinguishing advantage is upper-bounded by
	$$\Pr[\Bad_{\mathsf{rcpt}}]+\Pr[\Bad_{\mathsf{s32}}]+\Pr[\Bad_{\mathsf{pre}}]+\Pr[\Bad_{\mathsf{fs}}]+\Pr[\Bad_{\mathsf{hdl}}],$$
	and each term is negligible under the stated hypotheses of Theorem~\ref{thm:SDKG-UC}.
	When compiling out $\Fusv$ to the concrete USV protocol (Corollary~\ref{cor:compile-out-fusv}),
	additional negligible terms stemming from the USV/DLEQ-based instantiation are bounded under DDLEQ as in
	Theorem~\ref{thm:usv-real-ideal} and Lemmas~\ref{lem:usv-tag-sim} and \ref{lem:usv-eqv}.
\end{remark}

If it holds for the corrupted set that $B \in \mathrm{\Gamma}$,
no secrecy claims can be made on the conceptual signing key $k$.
Even then, the functionality and the real protocol guarantee non-exportability:
long-term secrets remain confined to the respective $\mathcal{F}^{(i)}_{\mathsf{KeyBox}}$ instances and the
adversary only obtains black-box access via $\mathsf{Use}$. In any session wherein at least one of $P_1$ or $P_2$ honestly samples its
auxiliary scalar $\sigma_{3,1}$ (resp.\ $\sigma_{3,2}$), the resulting key $K$ is uniformly distributed in
$\mathbb{G}$ (see Lemma~\ref{lem:sdkg-latent-unif}).

\begin{corollary}[Standard NXK-DKG interface]\label{cor:sdkg-standard-api}
	Under the hypotheses of Theorem~\ref{thm:SDKG-UC}, protocol $\mathrm{\Psi}^{(3)}_{\mathsf{SDKG}}$ UC-realizes
	$\FDKGstarNXK$ in the same model.
\end{corollary}

\begin{proof}
	By Theorem~\ref{thm:SDKG-UC}, $\mathrm{\Psi}^{(3)}_{\mathsf{SDKG}}$ UC-realizes $\FSDKG$ in the stated model.
	The claim follows by Lemma~\ref{lem:sdkg-interface-refinement}. \qed
\end{proof}

\begin{corollary}[Compiling out $\Fusv$]\label{cor:compile-out-fusv}
	Let $\mathrm{\Pi}_{\mathsf{USV}}$ be the concrete protocol of Section~\ref{subsec:USV-inst} implementing $\Fusv$, and define
	$\widehat{\mathrm{\Psi}}^{(3)}_{\mathsf{SDKG}} := \mathrm{\Psi}^{(3)}_{\mathsf{SDKG}}[\mathrm{\Pi}_{\mathsf{USV}}/\Fusv]$.
	Under the hypotheses of Theorem~\ref{thm:usv-real-ideal} and Theorem~\ref{thm:SDKG-UC},
	$\widehat{\mathrm{\Psi}}^{(3)}_{\mathsf{SDKG}}$ UC-realizes $\FSDKG$ in the $(\FKeyBox,\Fchan,\Fpub)$-hybrid and \gROCRP\ models.
\end{corollary}

\begin{proof}
	Immediate from UC composition: $\mathrm{\Pi}_{\mathsf{USV}}$ UC-realizes $\Fusv$ (Theorem~\ref{thm:usv-real-ideal}) and
	$\mathrm{\Psi}^{(3)}_{\mathsf{SDKG}}$ UC-realizes $\FSDKG$ in the $(\FKeyBox,\Fusv,\Fchan,\Fpub)$-hybrid (Theorem~\ref{thm:SDKG-UC}),
	with shared \gROCRP\ access via domain-separated contexts. \qed
\end{proof}

\section{UC Security of the 1+1-out-of-$n$ SDKG Extension}\label{sec:n-extension}
To enroll a new recovery device $P_i$ ($i\ge 3$), run the one-shot registration sub-protocol used for $P_3$,
sponsored by $P_2$ or any already registered leaf. When at least one of $P_1$ and the sponsor is honest,
the $\Fpub$-agreement check (Algorithm~\ref{alg:reg}) ensures that the honest joiner's reference~$K$
equals the stored finalized key, and $g_3$ then installs the replicated recovery-role share
$k_i=k_3$ satisfying $K_{1,3}+k_3\G=K$.
When both registration senders are corrupted,
$g_3$ (Definition~\ref{def:sdkg-deriv}) guarantees only that the installed share satisfies
$K_{1,3}+K_3=K$ for the adversary-supplied $K_{1,3}$;
since the corrupted senders already know~$k$, this confers no additional adversarial advantage,
but the installed $(k_{1,3},k_3)$ decomposition need not equal the transcript-derived one.
$\FSDKG^{(n)}$ (Fig. \ref{fig:FSDKG-n}) lifts $\FSDKG$ to 1+1-out-of-$n$ star access structure (for $n \ge 3)$ by keeping the same
transcript-driven base run for $(P_1,P_2)$ and allowing a polynomial number of post-finalization
RDRs for leaves $P_i$ with $i\in\{3,\ldots,n\}$, installing the same recovery-role share in $\FKeyBox^{(i)}$
whenever at least one registration sender is honest.

\begin{figure}[!ph]
	\centering
	\setlength{\fboxrule}{0.2pt} 	
	\fbox{%
		\parbox{\dimexpr\linewidth-2\fboxsep-2\fboxrule\relax}{%
			\ding{169} \textsf{State (per session $\sid$):}
			run the base-run state of $\FSDKG$ to completion.
			Maintain $\mathsf{Reg}\subseteq\{3,\ldots,n\}$ (init $\emptyset$), a pending table
			$\mathsf{Pend}[i]=(\mathsf{sponsor},\mathsf{ok}_1,\mathsf{ok}_s,\rho_1,\rho_s,w_1^{\mathsf{del}},w_s^{\mathsf{del}})$
			(init undefined, with $\rho_1=\rho_s=w_1^{\mathsf{del}}=w_s^{\mathsf{del}}=\perp$),
			registration-channel state multiset $\mathsf Q_{\mathsf{reg}}$ of $(\rho,i,P_s,P_r,w,\phi)$ and delivered set $\mathsf D_{\mathsf{reg}}$ (init empty).
			When an honest sponsor $P_j$ constructs a registration payload for a new joiner,
			$\FSDKG^{(n)}$ obtains the ciphertexts by commanding $P_j$'s KeyBox-driver wrapper
			(Remark~\ref{rem:kb-wrapper}) to invoke $\SealToPeer$ on
			$\FKeyBox^{(j)}$'s resident slots $\KBsid{k32}$ and $\KBsid{k23}$, rather than
			maintaining an explicit scalar table.
			
			\smallskip
			\ding{169}  \textsf{Base run:} Identical to $\FSDKG$ up to and including \textsf{TryFinalize}. 
			
			\smallskip
			\ding{169}  \textsf{Register:}
			Upon $(\mathsf{register},i,\sid,j)$ from $P_i$ with $i\in\{3,\ldots,n\}$:
			require $\mathsf{finalized}=1$, $i\notin\mathsf{Reg}$, and $j\in\{2\}\cup\mathsf{Reg}$.
			If $\mathsf{Pend}[i]$ is undefined, set $\mathsf{Pend}[i]\gets(j,0,0,\perp,\perp,\perp,\perp)$ and notify $\cA$ with $(\mathsf{RegReq},i,\sid,j)$.
			
			\medskip
			Upon $(\mathsf{approve},i,\sid)$ from $P_1$ (resp.\ from sponsor $P_j)$:
			if $\mathsf{Pend}[i]$ is defined then set $\mathsf{Pend}[i].\mathsf{ok}_1\gets 1$
			(resp.\ $\mathsf{Pend}[i].\mathsf{ok}_s\gets 1$); notify $\cA$.
			
			\medskip
			\ding{169} \textsf{RegGo (sender-controlled on corruption):}
			Upon $(\mathsf{RegGo},i,\sid,w_1^\star,w_s^\star)$ from $\cA$:
			if $\mathsf{Pend}[i]$ is defined, $\mathsf{Pend}[i].\mathsf{ok}_1=\mathsf{Pend}[i].\mathsf{ok}_s=1$, and
			$(\mathsf{Pend}[i].\rho_1,\mathsf{Pend}[i].\rho_s)=(\perp,\perp)$, then:
			\begin{itemize}[leftmargin=*,nosep]
				\item Let $j:=\mathsf{Pend}[i].\mathsf{sponsor}$.
				Define slot-bound associated data strings:
				$$\mathsf{ad}_{1i}:=\ADSDKGreg{\sid}{P_1}{P_i}{\texttt{k31}},
				\quad
				\mathsf{ad}_{ji}^{32}:=\ADSDKGreg{\sid}{P_j}{P_i}{\texttt{k32}},
				\quad
				\mathsf{ad}_{ji}^{23}:=\ADSDKGreg{\sid}{P_j}{P_i}{\texttt{k23}}.$$
				
				\item Reset delivered-payload buffers:
				set $\mathsf{Pend}[i].w_1^{\mathsf{del}}\gets\perp$ and $\mathsf{Pend}[i].w_s^{\mathsf{del}}\gets\perp$.
				
				\item Sample $\rho_1,\rho_s\leftarrowdollar \{0,1\}^\lambda$.
				
				\item If $1\notin\mathsf{Cor}$, sample
				$\varpi_1 \leftarrow \Enc_{\pk_{\mathrm{seal}}^{(P_i)}}(\mathsf{ad}_{1i},\sigma_{3,1})$
				and set $w_{\mathsf{reg},1}:=\langle \sid,\varpi_1,K_{1,3},K\rangle$ and $\phi_1:=\ellregone$.
				If $1\in\mathsf{Cor}$, set $w_{\mathsf{reg},1}:=w_1^\star$ and $\phi_1:=|w_{\mathsf{reg},1}|$.
				If $w_{\mathsf{reg},1}\neq\perp$, then:
				set $\mathsf{Pend}[i].\rho_1\gets \rho_1$;
				insert $(\rho_1,i,P_1,P_i,w_{\mathsf{reg},1},\phi_1)$ into $\mathsf Q_{\mathsf{reg}}$;
				send $(\mathsf{Leak},\sid,P_1,P_i,\rho_1,\phi_1)$ to $\cA$; and
				if $1\in\mathsf{Cor}$ additionally reveal $w_{\mathsf{reg},1}$ to $\cA$.
				
				\item If $j\notin\mathsf{Cor}$, have $P_j$ invoke
				$\varpi_{sa} \leftarrow \FKeyBox^{(j)}.\Use(\KBsid{k32},\SealToPeer,\langle P_i,\mathsf{ad}_{ji}^{32}\rangle)$
				and
				$\varpi_{sb} \leftarrow \FKeyBox^{(j)}.\Use(\KBsid{k23},\SealToPeer,\langle P_i,\mathsf{ad}_{ji}^{23}\rangle)$
				via the KeyBox-driver wrapper (Remark~\ref{rem:kb-wrapper}).
				Set $w_{\mathsf{reg},s}:=\langle \sid,\varpi_{sa},\varpi_{sb},K_{1,3}\rangle$ and $\phi_s:=\ellregtwo$.
				If $j\in\mathsf{Cor}$, set $w_{\mathsf{reg},s}:=w_s^\star$ and $\phi_s:=|w_{\mathsf{reg},s}|$.
				If $w_{\mathsf{reg},s}\neq\perp$, then:
				set $\mathsf{Pend}[i].\rho_s\gets \rho_s$;
				insert $(\rho_s,i,P_j,P_i,w_{\mathsf{reg},s},\phi_s)$ into $\mathsf Q_{\mathsf{reg}}$;
				send $(\mathsf{Leak},\sid,P_j,P_i,\rho_s,\phi_s)$ to $\cA$; and
				if $j\in\mathsf{Cor}$ additionally reveal $w_{\mathsf{reg},s}$ to $\cA$.
			\end{itemize}
			
			\medskip
			\ding{169} \textsf{Deliver:}
			Upon $(\mathsf{Deliver},\sid,\rho)$ from $\cA$:
			if $(\rho,i,P_s,P_i,w,\phi)\in\mathsf Q_{\mathsf{reg}}$ and $\rho\notin\mathsf D_{\mathsf{reg}}$, delete it from $\mathsf Q_{\mathsf{reg}}$,
			add $\rho$ to $\mathsf D_{\mathsf{reg}}$, and deliver $(\mathsf{Recv},\sid,P_s,w)$ to $P_i$.
			If $i\in\mathsf{Cor}$ reveal $w$ to $\cA$ at delivery time.
			If $\mathsf{Pend}[i]$ is defined then:
			\begin{itemize}[leftmargin=*,nosep]
				\item If $\rho=\mathsf{Pend}[i].\rho_1$, set $\mathsf{Pend}[i].\rho_1\gets\perp$ and $\mathsf{Pend}[i].w_1^{\mathsf{del}}\gets w$.
				If $\rho=\mathsf{Pend}[i].\rho_s$, set $\mathsf{Pend}[i].\rho_s\gets\perp$ and $\mathsf{Pend}[i].w_s^{\mathsf{del}}\gets w$.
				
				\item If $\mathsf{Pend}[i].\mathsf{ok}_1=\mathsf{Pend}[i].\mathsf{ok}_s=1$,
				$(\mathsf{Pend}[i].\rho_1,\mathsf{Pend}[i].\rho_s)=(\perp,\perp)$,
				and $\mathsf{Pend}[i].w_1^{\mathsf{del}}\neq\perp$ and $\mathsf{Pend}[i].w_s^{\mathsf{del}}\neq\perp$, then:
				\begin{itemize}[leftmargin=*,nosep]
					\item If $i\in\mathsf{Cor}$:
					set $\mathsf{Reg}\gets \mathsf{Reg}\cup\{i\}$, delete $\mathsf{Pend}[i]$, and output $(\mathsf{registered},i,\sid)$ to $P_i$ and $\cA$.
					
					\item If $i\notin\mathsf{Cor}$:
					let $j:=\mathsf{Pend}[i].\mathsf{sponsor}$ and define the same $\mathsf{ad}_{1i},\mathsf{ad}_{ji}^{32},\mathsf{ad}_{ji}^{23}$ as above.
				Parse
				$\mathsf{Pend}[i].w_1^{\mathsf{del}}=\langle \sid,\varpi_1,K_{1,3}^{(1)},K^\star\rangle$
				and
				$\mathsf{Pend}[i].w_s^{\mathsf{del}}=\langle \sid,\varpi_{sa},\varpi_{sb},K_{1,3}^{(2)}\rangle$.
				If parsing fails, do nothing further.
				Require $K_{1,3}^{(1)}=K_{1,3}^{(2)}$; otherwise do nothing further.
				Set $K_{1,3}^\star\gets K_{1,3}^{(1)}$.
				Require $K^\star = K$ (the stored finalized public key).
				In the real protocol, the honest joiner additionally requires agreement of $P_1$'s and $P_2$'s base-run $\Fpub$ publications; the simulator conditions $\mathsf{RegGo}$ on $\Fpub$ consistency. Otherwise do nothing further.
					Otherwise, have $P_i$ forward $(\varpi_1,\varpi_{sa},\varpi_{sb},\mathsf{ad}_{1i},\mathsf{ad}_{ji}^{32},\mathsf{ad}_{ji}^{23},K^\star,K_{1,3}^\star)$
					into $\FKeyBox^{(i)}$, which executes the same KeyBox-side installation procedure as Algorithm~\ref{alg:reg}
					(with sponsor $P_j$): open the ciphertexts under the corresponding $\mathsf{ad}$ strings and invoke the corresponding $\Load$ calls to install $\KBsid{k32},\KBsid{k23},\KBsid{k3}$.
					If all invoked $\Load$ calls return $\ok$, then set $\mathsf{Reg}\gets \mathsf{Reg}\cup\{i\}$;
					delete $\mathsf{Pend}[i]$, and output $(\mathsf{registered},i,\sid)$ to $P_i$ and $\cA$.
					(Future honest-sponsor payloads from $P_i$ are generated by invoking $\SealToPeer$ on
					$\FKeyBox^{(i)}$'s resident slots $\KBsid{k32}$ and $\KBsid{k23}$ directly via the
					KeyBox-driver wrapper, avoiding the need to record the installed scalar values explicitly.)
					Otherwise do nothing further.
				\end{itemize}
			\end{itemize}
			
			\smallskip
			\ding{169} \textsf{Corruptions/Programming:} As in $\FSDKG$ (Fig.~\ref{fig:FSDKG}).
		Before base finalization, $\s$ may once send
		$(\mathsf{Program},\sid,x_1^\star,x_2^\star,\sigma_{3,1}^\star,\sigma_{3,2}^\star)$; set
		$x_1\gets x_1^\star$, $x_2\gets x_2^\star$, $\sigma_{3,1}\gets\sigma_{3,1}^\star$, and $\sigma_{3,2}\gets\sigma_{3,2}^\star$.}}
	\caption{Transcript-driven ideal functionality $\FSDKG^{(n)}$.}
	\label{fig:FSDKG-n}
\end{figure}

\begin{readerbox}[rn:corr-joiner-upgrade]{Corrupted-joiner handling: $\FSDKG$ vs.\ $\FSDKG^{(n)}$}
In the base $\FSDKG$ (Fig.~\ref{fig:FSDKG}), if the joiner $P_3$ is corrupted after both payloads are
delivered, the functionality does nothing further and leaves completion to~$\cA$.
This suffices because the base case has a single joiner with no subsequent sponsor chain.
In $\FSDKG^{(n)}$, a corrupted joiner~$i$ is immediately added to $\mathsf{Reg}$ and
$(\mathsf{registered},i,\sid)$ is output, because in the real world the adversary
(controlling $\FKeyBox^{(i)}$) can signal completion and $i$ may subsequently serve as sponsor;
the ideal functionality must reflect this by granting sponsor eligibility.
This is a deliberate semantic upgrade, not an inconsistency.
\end{readerbox}

\begin{theorem}[UC security of $\FSDKG^{(n)}$]\label{thm:SDKG-UC-n}
	Fix Fischlin parameter functions $(t(\lambda),b(\lambda),r(\lambda),S(\lambda))$ satisfying Definition~\ref{FSdef}.
	Assume hardness of DL and DDLEQ, and the hypotheses used for Theorems~\ref{thm:SDKG-UC} and~\ref{thm:usv-real-ideal}.
	Then the compiled $n$-party protocol $\widehat{\mathrm{\Psi}}^{(n)}_{\mathsf{SDKG}}$ UC-realizes $\FSDKG^{(n)}$
	in the $(\FKeyBox,\Fchan,\Fpub)$-hybrid and \gROCRP\ models against adaptive corruptions with secure erasures.
\end{theorem}

\begin{proofsketch}
	Let $\mathcal A$ be any PPT real-world adversary and $\mathcal Z$ any PPT environment.
	We build a PPT simulator $\s^{(n)}$ for the ideal execution with $\FSDKG^{(n)}$.
	
	\smallskip\noindent
	\emph{Base run.}
	$\s^{(n)}$ simulates the base transcript exactly as in the proof of Theorem~\ref{thm:SDKG-UC}:
	it mediates $\mathcal A$'s \gROCRP\ queries under the UC-proof contexts, extracts witnesses from verifying
	UC-NIZK-AoKs, computes the unique transcript-defined $(x_1,x_2)$ in accepting sessions, computes the corresponding
	$\sigma_{3,1}$ and $\sigma_{3,2}$ as in Theorem~\ref{thm:SDKG-UC}, and programs	$\FSDKG^{(n)}$ once via $(\mathsf{Program},\sid,x_1,x_2,\sigma_{3,1},\sigma_{3,2})$ prior to finalization. This ensures that the public key output $K$ and the base-installed KeyBox shares for $P_1$ and $P_2$
	match the real execution distribution, up to negligible $\Bad$ events as in Theorem~\ref{thm:SDKG-UC}.
	
	\smallskip\noindent
	\emph{Registrations.}
	After finalization, $\FSDKG^{(n)}$ may receive any number of registration requests $(\mathsf{register},i,\sid)$
	for distinct $i\in\{3,\ldots,n\}$. In the real execution, each registration of an honest joiner affects only:
	(i) the transcript/messages of that registration session and
	(ii) the internal state of $\FKeyBox^{(i)}$ via three $\Load$ calls that install the two sponsor-state slots
	($\KBsid{k32}$ and $\KBsid{k23}$ via $g^{\textsf{reg}}_{3,2}$ and $g^{\textsf{reg}}_{2,3}$) and the replicated
	recovery-role share ($\KBsid{k3}$ via $g_3$), exactly as in Algorithm~\ref{alg:reg} with the
	appropriate joiner and sponsor parameters
	(Definition~\ref{def:sdkg-deriv}).
	For a corrupted joiner~$i$, the adversary controls $\FKeyBox^{(i)}$ directly and may or may not invoke $\Load$; correspondingly, $\FSDKG^{(n)}$ adds~$i$ to $\mathsf{Reg}$ and emits $(\mathsf{registered},i,\sid)$ without KeyBox installation, matching the real-world behavior where $\cA$ controls the completion signaling.
	Since $\FKeyBox^{(i)}$ state is never revealed upon corruption (only black-box access via $\Use$),
	$\s^{(n)}$ can simulate each device registration independently, while ensuring that $\FKeyBox^{(i)}$ ends up storing the correct
	shares and hence returns the same $\Pub(k_3)$ under $\Use$ as in the real protocol.
	Lemma~\ref{lem:sdkg-reg-sim} is stated for the $P_3$ base case; its proof extends to arbitrary joiner $P_i$ with sponsor $P_j$
	because the registration procedure, the associated-data binding, and the $g_3/g^{\textsf{reg}}$ routines are parameterized by
	$(P_r, P_{\mathsf{sp}}, \sid)$ (Definition~\ref{def:sdkg-deriv}).
	We distinguish four cases for the installed scalars:
	\begin{enumerate}[label=(\alph*),nosep,leftmargin=*]
		\item \emph{Both registration senders honest.}
		All installed scalars---including the individual sponsor-state
		slots $(\sigma_{3,2},\sigma_{2,3})$---are deterministic functions of the
		already-fixed values $(\sigma_{3,1},\sigma_{3,2},\sigma_{2,3},K,K_{1,3})$.
		\item \emph{$P_1$ honest, sponsor corrupted.}
		The recovery-role share $k_3$ is still uniquely determined by $g_3$
		(Definition~\ref{def:sdkg-deriv}), which enforces $K_{1,3}+k_3\G=K$;
		however, because $g_3$ constrains only the linear combination
		$\sigma_{2,3}-\sigma_{3,2}$ (not the individual values), a corrupted
		sponsor may shift both slots by a common offset $\mathrm{\Delta}$ without
		affecting~$k_3$.  The actually-installed values reside in $\FKeyBox^{(i)}$'s
		slots $\KBsid{k32}$ and $\KBsid{k23}$; the
		adversary supplies the sponsor payload $w_s^\star$ directly, so
		$\FSDKG^{(n)}$ correctly mirrors the real-world state. When $P_i$ later
		serves as an honest sponsor, $\FSDKG^{(n)}$ invokes $\SealToPeer$ on
		these resident slots via the KeyBox-driver wrapper, producing ciphertexts
		with the same distribution as in the real protocol.
		\item \emph{$P_1$ corrupted, sponsor honest.}
		The corrupted~$P_1$ controls $\FKeyBox^{(1)}$ and can seal an arbitrary scalar
		$\sigma_{3,1}'$ via $\SealToPeer$; it also controls its own $\Fpub$ publication.
		However, the honest joiner~$P_i$'s reference~$K$ is obtained by requiring
		agreement between $P_1$'s and $P_2$'s base-run authenticated
		$\Fpub$ publications for~$\sid$ (Algorithm~\ref{alg:reg}):
		$P_i$ sets $K:=K^{(1)}$ only if $K^{(1)}=K^{(2)}$, aborting otherwise.
		Since case~(c) presupposes that at least one of $P_1$ or $P_2$ was honest at the time of its
		base-run $\Fpub$ publication---if both were corrupted during the base run, we are in case~(d)---and
		$\Fpub$ publications are immutable once made (the functionality provides no retract interface),
		the honest publication anchors~$K$ to the correct value.
		Concretely, a corrupted~$P_1$ that published $K_{\mathrm{fake}}\neq K$ during the base run implies~$P_2$ was honest
		and published the correct~$K$, so $K^{(1)}\neq K^{(2)}$ and $P_i$ aborts.
		If $P_1$ was honest during the base run and corrupted later, $K^{(1)}=K$ is already committed.
		In either sub-case, if $P_1$ publishes the correct~$K$, then $g_3$ enforces $K_{1,3}+k_3\G=K$ with the
		correct $K$, and the recovery-role share is uniquely determined as in case~(a).
		As in case~(b), the sponsor-state slots may shift by a common offset without
		affecting~$k_3$; the actually-installed values reside in $\FKeyBox^{(i)}$'s
		slots and are accessed via $\SealToPeer$ when needed.
		The adversary supplies $w_1^\star$ directly, so $\FSDKG^{(n)}$ correctly mirrors
		the real-world state.
		\item \emph{Both senders corrupted} (requires $P_1\in\mathsf{Cor}$).
		The actually-installed values reside in $\FKeyBox^{(i)}$'s resident slots,
		and $\FSDKG^{(n)}$ invokes $\SealToPeer$ on these slots (via the KeyBox-driver
		wrapper) in subsequent honest-sponsor payloads, so real and ideal sponsor
		ciphertexts match.
	\end{enumerate}
	In every case, $\FSDKG^{(n)}$ performs the same $\FKeyBox$ installation
	procedure as Algorithm~\ref{alg:reg}, so the installed $\FKeyBox^{(i)}$
	state is identical in real and ideal
	executions.  Consequently, when a newly registered device later serves as
	sponsor, $\FSDKG^{(n)}$ invokes $\SealToPeer$ on the same resident slots
	that the real-world protocol would use, preserving the hybrid step's locality.
	
	\smallskip\noindent
	\emph{Composition over many registrations.}
	Registrations for different $i$ touch disjoint KeyBox instances $\FKeyBox^{(i)}$ and do not modify $(x_1,x_2,K)$.
	When a newly registered device later sponsors another device, $\FSDKG^{(n)}$ invokes $\SealToPeer$
	on the sponsor's $\FKeyBox$ resident slots via the KeyBox-driver wrapper; this dependency does not
	break hybrid locality because in every hybrid, $\FSDKG^{(n)}$ performs the same $\FKeyBox$-mediated
	installation as the real protocol, producing identical installed slot values.
	Therefore, switching registration $i$ from real to ideal leaves all state
	visible to subsequent registrations unchanged.
	We argue via a standard hybrid over $q = q (\lambda) \le \poly(\lambda)$ completed registrations: define hybrids $\Game^{(j)}$ for $j=0,\ldots,q$
	that replace the first $j$ real registrations by the $\FSDKG^{(n)}$-generated ones.
	Consecutive hybrids differ in only one registration instance, so Lemma~\ref{lem:sdkg-reg-sim}
	gives a negligible per-step change, and a union bound over $q$ steps keeps the total
	distinguishing advantage negligible.
	Thus, it follows from Theorem~\ref{thm:SDKG-UC} and Corollary~\ref{cor:compile-out-fusv} that $$\mathsf{Exec}(\widehat{\mathrm{\Psi}}^{(n)}_{\mathsf{SDKG}},\mathcal A,\mathcal Z)\approx_c
	\mathsf{Ideal}(\FSDKG^{(n)},\s^{(n)},\mathcal Z).$$ \qed
\end{proofsketch}

\begin{remark}[Why the two-party scoping is necessary]
	Under UC scheduling,
	a corrupted committer can equivocate by delivering distinct first commits to different relying parties.
	The compiled protocol $\mathrm{\Pi}_{\mathsf{USV}}$ necessarily yields recipient-local state in that case, so
	it can only UC-realize an ideal functionality with the same relying-party scoped semantics.
\end{remark}

To support scalable RDR with an arbitrary already-registered leaf as sponsor,
SDKG designates a small subset of $\sigma$-values as \emph{registration scalars} that are retained
inside KeyBoxes in dedicated slots (Definition~\ref{def:sdkg-deriv}).
These retained values remain NXK-restricted and are used only via
$\SealToPeer$. If an instantiation exposes remote-attestation transcripts in the clear, this can introduce extra visible structure; modeling such leakage is orthogonal to our NXK consistency enforcement.

\paragraph{Optional.} If external verifiability is desired: after KeyBox$_3$ installs $k_3$,
the host obtains $K_3:=\Pub(k_3)$ by invoking $\FKeyBoxOf{3}.\Use(\KBsid{k3},\GetPub,\cdot)$ and obtains a proof $\pi_{\DL_3}$ by running the in-KeyBox
optimized Fischlin prover:
\[
(\muFS,\mathbf a)\leftarrow \FKeyBoxOf{3}.\Use(\KBsid{k3},\textsf{FS.Start},\langle \sid,K_3\rangle),\qquad
\pi_{\DL_3} \leftarrow \FKeyBoxOf{3}.\Use(\KBsid{k3},\textsf{FS.Prove},\muFS).
\]
The proof $\pi_{\DL_3}=((a_i,e_i,z_i))_{i=1}^r$ is a UC-NIZK-AoK for the DL relation $\cR_{\mathsf{DL}}$
on statement $(\pp,K_3)$ in the \gROCRP\ model
under context $\mathsf{ctx}_{\mathsf{KeyBox}}$. By construction, $\mathsf{FS}.\mathsf{Prove}$ releases at most one distinct response per commitment---since even under rollback, deterministic nonce derivation prevents two different responses for the same commitment---and issues all rare-structure hash queries internally under $\mathsf{ctx}_{\mathsf{KeyBox}}$. Verification is public via $\mathsf{FS}.\Verify(K_3,\pi_{\DL_3})$.

\section{Complexity and Overhead}
\label{sec:complexity}

We summarize dominant computation and communication costs of SDKG (base $1{+}1$-out-of-$3$ run) and its
RDR extension to $1{+}1$-out-of-$n$. Let $\mathbb G$ be a prime-order group of size $p$ and write
$\kappa := \log p$ (e.g., $\kappa\approx 256$ for standard 128-bit ECC instantiations). All NIZK(-AoK)s are instantiated via the optimized Fischlin transform in the \gROCRP\ model with parameters
$(t,b,r,S)$ satisfying Definition~\ref{FSdef}. For our concrete estimates, we fix a target security level $\lambda=\lambda_0$ and instantiate
\[
(t,b,r,S) := (t(\lambda_0),b(\lambda_0),r(\lambda_0),S(\lambda_0)) = (13,8,32,32).
\]
Hence, for any prover making at most $Q$ distinct \gROCRP\ queries under the relevant proof context, the resulting
Fischlin knowledge/soundness error is $\le (Q+1)\cdot 2^{-195} + \negl(\lambda)$, which constitutes a security bound.
Separately, the early-break rarity search induces a per-proof honest-rejection probability bounded by
$p_{\mathrm{rej}}\le r\cdot(1-2^{-b})^{2^t}\approx 2^{-41}$ for these parameters
(liveness/completeness).
This $p_{\mathrm{rej}}$ is an operational failure probability of the proof-generation subroutine, not an
adversarial success probability; the prover retries proof generation upon rejection, with expected
attempts $1/(1-p_{\mathrm{rej}})\approx 1$. On rejection, the prover aborts and retries by choosing a fresh proof-session identifier $\sid'\neq\sid$ and invoking \textsf{FS.Start}$(\sid',K)$ to obtain a fresh handle and commitments because reusing the same $(\sid,K)$ deterministically regenerates the same commitments. Using standard compressed encodings, the dominant objects are:
\[
|\pi_{\mathsf{DL}}| \approx 2.1\text{ KiB},\qquad
|\pi_{\mathsf{aff}}| \approx 4.1\text{ KiB},\qquad
|(C,\zeta)| \approx 3.1\text{--}3.3\text{ KiB}.
\]
The base SDKG transcript contains one USV certificate and two affine AoKs, yielding a total transcript size of
$\approx 11$--$13$ KiB (excluding small constant headers). Verification of a Schnorr-DL Fischlin proof performs $\Theta(r)$ Schnorr checks / $\Theta(r)$ scalar multiplications: one fixed-base by $\G$, one variable-base by the statement element $K$, plus additions, up to standard multi-scalar optimizations. The SDKG base run verifies a constant number of proof objects---one DLEQ proof inside USV and two affine proofs implemented by two DL proofs---hence the verification and proving
costs are $\widetilde O(\kappa^{2.585})$ bit-operations in the Karatsuba model (dominated by a constant number of scalar multiplications).
Registration of an additional recovery device adds a constant-size exchange and an optional DL proof for external verifiability.

SDKG uses a constant number of proof objects in the base run and performs no resharing. Therefore, $\mathsf{Comm}_{\mathsf{SDKG}} = \widetilde{O}(\kappa)\ \text{bits}$ and $\mathsf{Comp}_{\mathsf{SDKG}} = \widetilde{O}(\kappa^{2.585})\ \text{bit-ops}.$ For the $1{+}1$-out-of-$n$ extension, RDR registers each additional device with $\mathsf{Comm}_{\mathsf{RDR,per\;device}} = \widetilde{O}(\kappa)\  \text{bits}$ and $\mathsf{Comp}_{\mathsf{RDR,per\;device}} = \widetilde{O}(\kappa^{2.585})\ \text{bit-ops},$
so in total: $\mathsf{Comm}(n)=\widetilde{O}(n\kappa)$ and $\mathsf{Comp}(n)=\widetilde{O}(n\kappa^{2.585})$.

\section{Conclusion}\label{sec:conclude}
We studied UC-secure Distributed Key Generation (DKG) in the Non-eXportable Key (NXK) setting wherein long-term signing shares are confined to state-continuous
trusted hardware (e.g., TEEs). In this NXK setting, the protocol cannot rely on classical VSE
(VSS/PVSS/AVSS or commitment-and-proof analogues) that exports and manipulates shares (or affine images of shares),
and state-continuity rule out rewinding/forking-style extraction once witness-bearing computation is delegated to
the hardware boundary. Our starting point is therefore to decouple confidentiality from consistency:
confidentiality can be delegated to the NXK/TEE boundary, while the protocol must still enforce transcript-defined
affine consistency and uniqueness of the induced sharing without any share opening or resharing. We introduced Unique Structure Verification (USV), a publicly verifiable certificate mechanism that yields a
deterministic public opening to a committed group element while keeping the underlying scalar hidden under DL and
DDLEQ hardness, providing canonical public structure without trapdoors or programmable setup. We combined USV with
UC-extractable NIZK arguments of knowledge (via an optimized Fischlin transform in our \gROCRP-hybrid model) to obtain
straight-line extraction without rewinding. Building on these tools, we constructed Star DKG (SDKG), a
constant-round UC-secure DKG for a $1{+}1$-out-of-$n$ star access structure motivated by threshold cryptocurrency wallets where a designated
service must co-sign yet cannot reconstruct the signing key alone (under key-opacity of the KeyBox profile).
More precisely, for the base authorized pairs $\{P_1,P_2\}$ and $\{P_1,P_3\}$, the transcript-derived shares always sum to~$k$ (Observation~\ref{lem:fdkg-properties}); for each subsequently registered device~$P_i$, the same holds whenever at least one registration sender ($P_1$ or the sponsoring leaf) is honest (Section~\ref{sec:n-extension}). When both senders are corrupted---which entails a corrupted~$P_1$ and hence a fully corrupted authorized pair---the installed decomposition satisfies $K_{1,3}+K_3=K$ but need not coincide with the transcript-derived split; because the adversary already controls an authorized pair, this confers no additional advantage. In all corruption scenarios, (ii)~no unauthorized set---including the center $P_1$ alone---can recover~$k$ from the KeyBox-confined shares under key-opacity (Observation~\ref{lem:fdkg-properties}).
The operational signing completeness guarantee---that authorized pairs can \emph{use} their NXK-confined shares to
produce valid signatures---depends on composing SDKG with a suitable threshold signing protocol compatible with
the KeyBox API and is left to future work.
SDKG supports post-DKG device enrollment by registering additional
recovery devices as replicated leaves, preserving the public key without resharing. Under DL and DDLEQ assumptions,
PRF security for deterministic nonce derivation, the seed integrity invariant (Assumption~\ref{assump:seed-integrity}),
and assuming KeyBox (TEE-resident NXK keystore) opacity, state continuity, and secure erasures, we prove that
SDKG UC-realizes a transcript-driven ideal functionality. The protocol achieves $\tilde O(n\log p)$ communication and
$\tilde O(n\log^{2.585}p)$ bit-operation cost in a prime-order group of size $p$, while each additional device
registration costs $\tilde O(\log p)$ communication and $\tilde O(\log^{2.585}p)$ computation (with a base transcript
of $\approx 11$--$13$~KiB in a 128-bit instantiation).

\section*{Acknowledgments} 
I thank Nolan Miranda for helpful discussions and feedback on early versions of the SDKG idea and protocol sketch.

\bibliographystyle{splncs04}
\bibliography{ref}

\appendix
\section{Programmable Secure Hardware Integration}\label{App1}
The main construction already assumes that USV certificate generation $\Cert(\pp, \cdot)$ is executed inside the KeyBox boundary and exports only $(C, \zeta)$. This appendix describes an optional additional hardening that further reduces exposure of other ephemeral scalars and group/field arithmetic in host memory by shifting (i) sampling of ephemerals and (ii) scalar multiplication with public points into the KeyBox/TEE. At a high level, the host is then restricted to routing public group elements, protocol transcripts, and ciphertexts, while the KeyBox performs ephemeral sampling and arithmetic and (optionally) authenticated encryption/decryption so plaintexts never leave the KeyBox boundary.

Although this hardening is mostly written for TEEs, because they are a convenient programmable substrate, the profile applies to any KeyBox realization that can expose the same non-exporting ephemeral-handle operations without violating key-opacity/state-continuity.

\paragraph{Scope / proof compatibility.}
This profile is an implementation hardening and does not change the protocol transcript or acceptance
predicate. Our UC analysis continues to apply to the baseline profile used in the main construction.
UC-extractable consistency AoKs (whose straight-line extraction relies on observing
the prover's \gROCRP\ query log) remain generated outside the KeyBox as in the main protocol;
moving those AoKs inside the KeyBox would require a different extraction mechanism/model.

\paragraph{Illustrative hardened interface (informal).}
Model this by augmenting the admissible KeyBox operation set
$\mathcal F_{\mathrm{adm}}$ (Definition~\ref{def:keybox-profile}) with ephemeral-handle arithmetic
that outputs only public group elements.
Let $\tau$ denote a type-tagged opaque handle to an ephemeral scalar $s\in\Zp$ stored inside the KeyBox, similar in spirit to internal $\mathsf{buf}$ handles in $\FKeyBox$ (Fig.~\ref{Fdskg}).
Handles refer only to KeyBox-internal ephemeral state and are not interchangeable with long-term key slots.

\begin{itemize}
	\item $\mathsf{GenScalar}() \rightarrow \tau$:
	sample $s\leftarrow \Zp$ internally and return only the handle $\tau$.
	
	\item $\mathsf{Mul}(\tau, P) \rightarrow Q$:
	for public $P\in\mathbb G$, output $Q:=sP$ (and optionally consume $\tau$).
	
	\item $\mathsf{MulGen}(\tau) \rightarrow J$:
	output $J:=s\G$ (a wrapper for $\mathsf{Mul}(\tau,\G)$).
	
	\item (Optional convenience) $\mathsf{LinComb}(\{(\tau_i,P_i)\}_{i=1}^t, \{\alpha_i\}_{i=1}^t) \rightarrow Q$:
	for public $\alpha_i\in\Zp$ and $P_i\in\mathbb G$, output
	$Q := \sum_{i=1}^t \alpha_i \cdot s_i P_i$.
	
	\item (Optional; needed for RDR / any KeyBox-to-KeyBox transport of share-derived payloads)
	retain the sealing interfaces $\SealToPeer$ and $\OpenFromPeer$ so the host supplies only peer identity
	and associated data and plaintexts never leave the KeyBox boundary.
\end{itemize}

\paragraph{Key-opacity intuition.}
Since the hardened interface returns only public group elements derived from fresh ephemerals and public inputs,
its externally visible outputs remain simulatable given public information, aligning with the key-opacity
assumption used throughout the paper.

\section{Candidate KeyBox Implementations and Profile-Capture Checklist}
\label{app:keybox-impls}
This appendix provides (i) concrete implementation classes that plausibly realize the KeyBox abstraction under a
pinned profile, and (ii) an operational checklist for deriving/enforcing such a profile from vendor APIs. Table \ref{tab:keybox-candidates} summarizes candidate deployment families.

\subsection{Candidate classes of implementations}
The following implementation patterns are plausible realizations of the abstraction in
Fig.~\ref{Fdskg}, provided they are configured and constrained to enforce the intended KeyBox
profile (Definition~\ref{def:keybox-profile}) and do not silently expose mechanisms that violate
key-opacity (Assumption~\ref{assump:keybox-opacity}) or state continuity
(Assumption~\ref{assump:tee-continuity}). 

\begin{table}[t]
	\centering
	\small
	\setlength{\tabcolsep}{3.5pt}
\begin{tabular}{@{}p{2.5cm} | p{5.9cm} | p{5.3cm} | p{3.4cm}@{}}
		\toprule
		\textbf{Family} &
		\textbf{What matches $\FKeyBox$} &
		\textbf{Profile adapter must forbid} &
		\textbf{Freshness / anti-rollback} \\
		\midrule
		TEE as restricted keystore &
		Minimal enclave exports only $\Load/\Use$; sealing to attestation-bound peer keys &
		Wrap/export under host-chosen keys; APIs returning share-deriving outputs of resident keys$^{\dagger}$; high-entropy failure channels &
		TPM/TEE counters, trusted time, or server freshness oracle \\
		
		Attested enclave$\leftrightarrow$KMS &
		KMS authorizes use/sealing conditioned on enclave measurement/attestation &
		KMS export/wrap features; host-selected recipients; policy gaps &
		KMS-side monotonic state + server oracle \\
		
		HSM (PKCS\#11 allowlist) &
		Non-extractable objects; sealing/wrap only to trusted/pinned recipients &
		Caller-decryptable wrapping and key-deriving ops &
		HSM counters or external freshness \\
		
		Endpoint keystore &
		Hardware-bound NXK keys + restricted built-in ops; optional attestation &
		Limited programmability; operation menu may not support custom subroutines &
		Platform counters/time + server oracle \\
		
		TPM-assisted freshness &
		Provides monotonic primitives to bind sealed state to freshness &
		Integration complexity; counter exhaustion/availability &
		TPM NV counters / PCR-bound state \\
		\bottomrule
	\end{tabular}
	\caption{Non-normative summary: common substrates for enforcing a KeyBox API profile.\\
	{\footnotesize $^{\dagger}$This refers to vendor-API operations that export caller-invertible affine images of an already-installed long-term share. Key-independent admissible operations (e.g., $\textsf{SDKG.LeafInit}\in\mathcal{F}_{\mathrm{KI}}$) that output fresh share-deriving material for transient protocol use are part of the protocol's own admissible profile and are not subject to this restriction; their outputs are NXK-restricted transient host state subject to secure erasure (Remark~\ref{rem:transport-vs-nxk}).}}
	\label{tab:keybox-candidates}
\end{table}

\begin{enumerate}[leftmargin=*,label=\arabic*.]
	\item Dedicated TEEs as restricted keystores:
	Implement the KeyBox as a minimal enclave exposing only $\Load/\Use$ endpoints.
	Approximate $\SealToPeer$ by pinning recipient keys to attested identities/measurements (e.g., HPKE-style sealing),
	and treat rollback/fork protection as an explicit subsystem (counters / trusted time / server freshness).
	
	\item Attested enclave$\leftrightarrow$KMS / KMS-backed enclaves:
	Split ``KeyBox logic'' (enclave) from ``policy enforcement'' (KMS) so that key usage and sealing are
	authorized by attestation evidence and measurement-based policy, rather than host-supplied keys.
	
	\item HSMs under a strict allowlist profile (PKCS\#11-style):
	Use HSM non-exportability, but restrict mechanisms to a small allowlist that rules out
	share-deriving outputs and caller-decryptable wrap/export; if wrapping exists, constrain it to
	trusted/pinned recipients only (to preserve key-opacity).
	
	\item Endpoint hardware-backed keystores:
	Platform keystores often match the ``NXK + restricted operations'' model, but typically expose a fixed
	operation menu; they therefore realize only those profiles reducible to built-in operations plus pinned sealing.
	
	\item TPM-assisted freshness/anti-rollback:
	Even when the KeyBox is a TEE/enclave, TPM-backed monotonic primitives can supply the freshness mechanism
	needed to approximate state continuity in hostile host environments.
\end{enumerate}

\section{Illustrative application: NXK-compatible commit--reveal randomness beacons}\label{App2}
Many practical protocols--- such as lotteries, leader election, and randomness beacons \cite{Cleve[86],Blum[83],Randao[19]}---require a commit--reveal
discipline and/or VSS/VDF to prevent adaptive choice/grinding: if parties were to publish their contribution $M_i := m_i \G$ immediately,
then a last-moving adversary can sample many candidate scalars $m_i$ after seeing the honest $M_j$'s and select one
that biases a downstream predicate of the beacon (e.g., a target prefix), where the beacon is
defined as $\pounds := \sum_i M_i$ (and typically $\rho := H(\mathsf{beacon},\langle \sid, \pounds\rangle)$). In a classical commit--reveal
design one would commit to $m_i$ and later open by revealing $m_i$, but in the NXK setting the scalar is non-exportable. USV provides an alternative compatible under the NXK model in which the opening is to the induced group element rather than to the
scalar. Concretely, each party samples $m_i \leftarrowdollar \Zp$ inside its KeyBox and computes a USV certificate
$(C_i,\zeta_i) \gets \textsf{Cert}(\pp,m_i)$.
\begin{itemize}\setlength\itemsep{0.25em}
	\item Commit: broadcast only $C_i$ (keep $\zeta_i$ private).
	\item Reveal: broadcast $\zeta_i$. Anyone checks $\Vcert(\pp,C_i,\zeta_i)=1$ and derives the canonical
	public contribution
	\[
	M_i := \Open_M(\pp,C_i,\zeta_i) = m_i \G .
	\]
	\item Output: set $\pounds := \sum_i M_i$ and $\rho := H(\mathsf{beacon},\langle \sid, \pounds\rangle)$.
\end{itemize}
Hence, each $M_i$ is a deterministic function of the public transcript via $(C_i,\zeta_i)$, so a verifier (and, in UC, the simulator/extractor) can compute the exact statement points used by the protocol in straight-line without ever extracting or exporting $m_i$. At the same time, withholding $\zeta_i$ until the
reveal phase restores the usual commit--reveal discipline, since the adversary cannot choose $m_i$ as a function of others'
revealed contributions, while remaining fully compatible with non-exportable scalars.

\end{document}